\documentclass{article}

\usepackage{amssymb}
\usepackage{amsmath}
\usepackage{amsfonts}
\usepackage{amsbsy}
\usepackage{amsthm}
\usepackage[parfill]{parskip}
\usepackage{graphicx}
\usepackage{epsfig}

\numberwithin{equation}{section}

\newtheorem{theorem}{Theorem}[section]
\newtheorem{lemma}[theorem]{Lemma}
\newtheorem{proposition}[theorem]{Proposition}
\newtheorem{corollary}[theorem]{Corollary}

\theoremstyle{definition}

\newtheorem{problem}{Problem}

\theoremstyle{remark}
\newtheorem{remark}{Remark}[section]
\newcommand{\Real}{\mathop{\mathrm{Re}}}
\newcommand{\Imag}{\mathop{\mathrm{Im}}}
\newcommand{\Arg}{\mathop{\mathrm{Arg}}}
\newcommand{\Tr}{\mathop{\mathrm{Tr}}}

\newcommand{\cut}{\mathop{(\mathrm{cut})}}

\newcommand{\emphsl}[1]{{\sl #1}}

\newcommand{\rmi}{\mathrm {i}}
\newcommand{\rmd}{\mathrm {d}}
\newcommand{\rmf}{\mathrm {f}}

\newcommand{\C}{\mathbb{C}}
\newcommand{\R}{\mathbb{R}}
\newcommand{\N}{\mathbb{N}}
\newcommand{\I}{\mathbb{I}}
\newcommand{\Z}{\mathbb{Z}}

\newcommand{\cC}{\mathcal{C}}
\newcommand{\cD}{\mathcal{D}}
\newcommand{\cE}{\mathcal{E}}
\newcommand{\cJ}{\mathcal{J}}
\newcommand{\cK}{\mathcal{K}}
\newcommand{\cL}{\mathcal{L}}
\newcommand{\cV}{\mathcal{V}}
\newcommand{\cN}{\mathcal{N}}

\newcommand{\bxi}{{\boldsymbol{\xi}}}
\newcommand{\bbeta}{{\boldsymbol{\eta}}}

\newcommand{\bx}{\boldsymbol{x}}
\newcommand{\by}{\boldsymbol{y}}
\newcommand{\bz}{\boldsymbol{z}}
\newcommand{\bR}{\boldsymbol{R}}

\newcommand{\stronglimit}{\mathop{\mathrm{s\!-\!lim}}}

\newcommand{\beq}{\begin{equation}}
\newcommand{\eeq}{\end{equation}}

\newcommand{\sds}{\strut\displaystyle}
\newcommand{\z}{\zeta}

\newcommand{\we}{w^{(\varepsilon)}}
\newcommand{\wea}{{\we\!,\alpha}}

\newcommand{\sku}{\vspace*{0.1cm}}
\newcommand{\skd}{\vspace*{0.2cm}}
\newcommand{\skt}{\vspace*{0.3cm}}

\begin{document}

\title{\bf\sc Nonlocal Potentials and \\ Complex Angular Momentum Theory}

\author{
\textbf{Jacques Bros} \\
\normalsize CEA, Service de Physique Th\'eorique, \\
CE-Saclay, F-91191 Gif-sur-Yvette Cedex, France \\
\texttt{jacques.bros@cea.fr}
\and
\\ \textbf{Enrico De Micheli} \\
\normalsize Consiglio Nazionale delle Ricerche \\
\normalsize Via De Marini, 6 - 16149 Genova, Italy \\
\texttt{demicheli@ge.cnr.it}
\and
\\ \textbf{Giovanni Alberto Viano} \\
\normalsize Dipartimento di Fisica - Universit\`a di Genova \\
\normalsize Istituto Nazionale di Fisica Nucleare - Sezione di Genova \\
\normalsize Via Dodecaneso, 33 - 16146 Genova, Italy \\
\texttt{viano@ge.infn.it}}

\date{\today}

\maketitle

\begin{abstract}
The purpose of this paper is to establish meromorphy properties of the 
\emphsl{partial scattering amplitude} $T(\lambda,k)$ associated with
physically relevant classes $\cN_\wea^{\,\gamma}$ of nonlocal potentials 
in corresponding domains $D_{\gamma,\alpha}^{(\delta)}$ of the space $\C^2$ 
of the complex angular momentum $\lambda$ and of the complex momentum $k$ 
(namely, the square root of the energy). The general expression of $T$ as 
a quotient $\Theta(\lambda,k)/\sigma(\lambda,k)$ of two holomorphic functions 
in $D_{\gamma,\alpha}^{(\delta)}$ is obtained by using the Fredholm--Smithies 
theory for complex $k$, at first for $\lambda=\ell$ integer, and in a second 
step for $\lambda$ complex ($\Real\lambda>-1/2$). Finally, we justify the 
``Watson resummation" of the partial wave amplitudes in an angular sector of 
the $\lambda$--plane in terms of the various components of the polar manifold 
of $T$ with equation $\sigma(\lambda,k)=0$. While integrating the basic Regge 
notion of interpolation of \emphsl{resonances} in the upper half--plane of $\lambda$,
this unified representation of the singularities of $T$ also provides an 
attractive possible description of \emphsl{antiresonances} in the lower 
half--plane of $\lambda$. Such a possibility, which is forbidden  in the usual 
theory of local potentials, represents an enriching alternative to the standard 
Breit--Wigner hard--sphere picture of antiresonances.
\end{abstract}

\section{Introduction}
\label{se:introduction}

In the standard Breit--Wigner theory of scattering the notions 
of ``time delay" and ``time advance", corresponding respectively 
to the increasing or decreasing character of a given phase--shift 
as a function of the energy (see, e.g., \cite[pp. 110 and 111]{Nussenzveig})
are not described in a symmetric way: while the former is described 
by a pole singularity of the scattering amplitude, the latter relies 
on the model of scattering by an impenetrable sphere. A question
then arises: can the time advance still be evaluated in terms of 
hard--sphere scattering in collisions between composite particles, 
when Pauli exclusion principle comes into play? We recall, indeed, 
that when two composite particles collide, the fermionic character 
of the components emerges and the antisymmetrization of the whole 
system generates repulsive exchange--forces, which can produce 
antiresonances, and are thus responsible for the occurrence of 
time advance.

\vskip 0.2cm

From a more formal viewpoint,
if antiresonances can be defined in a way which parallels 
the description of resonances, namely as bumps in the 
energy--dependance plot of the cross--section associated with
the \emphsl{downward} (instead of \emphsl{upward}) passage through 
$\frac{\pi}{2}$ of a given phase--shift, their status has remained
rather obscure in terms of the analyticity properties in energy
of the scattering functions, in comparison with the success of 
the standard \emphsl{pole--resonance} conceptual correspondance.
The reason for this discrepancy lies in the fact
that the Breit--Wigner poles associated with resonances 
necessarily occur in the ``second--sheet" of the energy 
variable $E$, namely in the lower
half--plane of the usual momentum variable $k= \sqrt E$, 
near the positive real axis; this is
a situation which readily accounts for the upward passage of
the phase--shift through $\frac{\pi}{2}$. 
But then the natural candidacy of
poles in the upper half--plane  for accounting the downward passage  
characterizing the antiresonances 
turns out to be strictly forbidden
in any formalism of scattering theory
(as well in  nonrelativistic potential theory as in relativistic 
quantum field theory) as a constraint imposed by the causality principle.
However one also knows that 
for any given phase--shift
the antiresonances are present together with the resonances.  
In fact, the physically observed occurrence of resonance--antiresonance pairs 
finds its firm theoretical basis 
in the Levinson theorem (see \cite[p. 206]{Nussenzveig})
which prescribes the equality of each phase--shift at zero and infinite energies 
in the absence of bound states.

\vskip 0.2cm

Considering the present sum of knowledge that is available in 
the general theory of quantum scattering processes, it
seems to us that 
there remains a problem of global understanding of the 
organization of sequences of
resonances and associated antiresonances for successive values
of the angular momentum $\ell$ and of the corresponding relationships between
the energy variable $E$ (or $k$) and $\ell$.
Let us recall the historical situation from both theoretical and experimental aspects.
At first,
the standard theory of scattering does not attempt to group
resonances in families: each resonance is simply described by a fixed pole singularity
in the energy variable. As for antiresonances, they are individually 
parametrized in a very rough way and from the outset by the 
hard--sphere picture.
However, phenomenological data clearly show that
the resonances often appear in ordered sequences, such
as rotational bands.
Typical examples can be observed in $\alpha$--$\alpha$,
$\alpha$--$^{40}$Ca, $^{12}$C--$^{12}$C, $^{28}$Si--$^{28}$Si,
and other heavy--ion collisions
(see \cite{Cindro} and references quoted therein).
In these examples, the resonances are ordered in
rotational bands of levels
whose energy spectrum can be fitted by an expression of the
form $E_\ell=A+B\ell(\ell+1)$, where $\ell$ is
the angular momentum of the level, and $A$ and $B$
are constants. Furthermore, the widths of the resonances
increase as a function of the energy.

\vskip 0.2cm

Since 1960, a wide opening on the previous problem
was given by Regge's basic 
formalism of the holomorphic interpolation of scattering 
partial waves in the complex angular momentum 
(CAM) variable $\lambda$ (see \cite{DeAlfaro}).
In this formalism, the concept of pole reached his 
full fecondity through the association of
sequences of resonances with \emphsl{polar manifolds}
$\lambda= \lambda(k)$ in the complex space $\C^2$ of the
variables $(\lambda,k)$.
Each such polar manifold gives rise to a ``trajectory"
on which a sequence of complex energy values $k_\ell$
with $\Imag k_\ell <0$ satisfying the equation 
$\lambda(k_\ell)= \ell$ corresponds to a sequence of resonances.
In view of the complex geometrical framework, a new possibility 
offered by that description
was the study of the trajectories  
$\lambda= \lambda(k)$ for $k$ real and positive, 
realized as curves in the complex upper half--plane of $\lambda$.
Coming back to the general problem that we have raised above,
it now seems interesting to inquire whether this enlarged framework 
of complex geometry in the CAM and energy variables
might give some new insight on the description of antiresonances 
and also on their global
organization. In other words, are there some polar manifolds 
$\lambda= \lambda(k)$ which have  
something to do with the description of antiresonances?

A tentative answer to this question was proposed by various authors 
(see \cite{Newton4}) by considering in particular the case
of scattering by Yukawa--type potentials.
In such cases each pole trajectory in the  upper half--plane of
$\lambda$ behaves as follows. After having traveled forward 
(i.e., with $\rmd\Real\lambda/\rmd k>0$) 
and described thereby a sequence of resonances at $\Real \lambda (k) =\ell$,     
it starts going up and turning backwards, and then passes again 
but \emphsl{in decreasing order} (and with larger values of $\Imag \lambda$)
through points whose real parts 
$\Real \lambda =\ell$ had been previously visited.
It was then proposed that this second part of such trajectories 
be associated with antiresonances.
However it turns out that such a global ordering 
is inconsistent with most of the experimental data,
which exhibit the alternance of resonances 
and antiresonances for successive values of $\ell$
when one follows the increase of the energy variable.  
From a formal viewpoint, 
a scenario which would appear to be much more appropriate
would be the suitable coupling of a \emphsl{pure resonance trajectory} 
in the upper half--plane with a \emphsl{pure antiresonance trajectory} in the 
lower half--plane $\Imag\lambda <0$.
However, this latter possibility has been excluded 
from the whole theory of local potentials by a theorem proved by Regge
(see \cite{DeAlfaro}).

\vskip 0.2cm

It is at this point of our considerations that we wish to advocate  
for the necessity of enlarging the framework of Schr\"odinger theory 
so as to include the occurrence of nonlocal potentials
$V(\bR,\bR')$ and thereby invalidate the application 
of Regge's no--go theorem. As a matter of fact,
whenever the antiresonances are generated by exchange
forces due to Pauli's exclusion principle, the
many--body structure of the colliding particles
is involved and leads one to introduce nonlocal potentials.
This class of potentials has been indeed
considered long time ago, mainly in connection
with the theory of nuclear matter \cite{Bethe,Wildermuth1}.
The procedures commonly used for treating the
many--body dynamics, like the resonating group method \cite{Tang},
the complex generator coordinate technique \cite{LeMere},
the cluster coordinate method \cite{Wildermuth2}, all lead to
an extension of the standard Schr\"odinger equation,
which now becomes an integro--differential
equation of the same form as that obtained in the study of the 
nonlocal potentials.
In this paper, we intend to study the singularities of the 
partial scattering amplitude for appropriate classes of nonlocal potentials,
both in the complex momentum $k$--plane and in the complex angular
momentum $\lambda$--plane. Concerning the possible location of these singularities,
one of the results that we have obtained is a
modified extension of the previous Regge theorem
to the case of nonlocal potentials, which now opens the way to 
singularities in the lower half--plane of $\lambda$ (at $k$ real and positive).
For local potentials
(notably, Yukawian potentials) all the
singularities of the partial scattering function, such as those
which manifest themselves as resonances, must be located at $ k>0$
in the region $\Imag\lambda>0$. 
Here we have considered as an example a class of
nonlocal potentials $V(\lambda;R,R')=V_*(R,R'){\tilde F}(\lambda)$,
where ${\tilde F}(\lambda)$ is a function holomorphic, bounded, and of Hermitian--type
in the half--plane 
$\C^+_{-\frac{1}{2}}\doteq\{\lambda\in\C\,:\,\Real\lambda>-\frac{1}{2}\}$.
For such a class of potentials, we denote by $\cL_j$ ($j\in\Z$) the
set of lines in $\C^+_{-\frac{1}{2}}$ where $\Imag {\tilde F}(\lambda)=0$.
All the points of a curve $\cL_j$ with $j>0$ (resp., $j<0$) belong to
the region $\Imag\lambda>0$ (resp., $\Imag\lambda<0$), $\cL_0$
being along the real positive axis. We prove that no singular
pairs $(\lambda,k)$ can occur with $\lambda$ in any line $\cL_j$
in the lower half--plane
($j\leqslant 0$), and with $\Imag k>0$ and $\Real k>0$. But there 
is a possible occurrence of singular manifolds containing
branches in the region $\Imag k\geqslant 0$, $\Real k\geqslant 0$
and $\Real\lambda>0$, $\Imag\lambda<0$, always located in strips
of the fourth quadrant of the $\lambda$--plane, well separated 
from one another by the set of lines $\cL_j$ ($j\leqslant 0$). These
strips therefore set the ground for 
possible antiresonance trajectories.
Although it is still premature to conclude on the existence of the latter
before some numerical exploration be performed, we think that this promising result 
may suggest how to fill a gap between phenomenological and theoretical analysis
of scattering data. Indeed, 
in refs. \cite{DeMicheli1,DeMicheli2} two of us have performed
an extensive phenomenological analysis, fitting the scattering data
of $\alpha$--$\alpha$ and $\pi^+$--p elastic scattering by using two pole
trajectories in the CAM--plane.
The resonances have been described by poles in the first quadrant,
the antiresonances by poles in the fourth quadrant.

\vskip 0.2cm

As a further possible motivation to the investigation of the 
theory of nonlocal potentials,
we remind the reader that in a parallel presentation
of the two--particle scattering theory in Schr\"odinger wave--mechanics
and in Quantum Field Theory (QFT), it is a suitable Bethe--Salpeter--type
kernel which plays the role of the potential \cite{Bros1},
provided the latter be understood as a generalized potential
of nonlocal--type (and also energy--dependent). 
Recently, the CAM formalism has been introduced
in QFT \cite{Bros2}, and it has been used to obtain a corresponding  
CAM--diagonalization of the Bethe--Salpeter equation \cite{Bros3}. Therefore,
also from this viewpoint, exploring the CAM--singularities generated 
by appropriate classes of nonlocal potentials deserves some interest.
One can even suggest that the possible occurrence of antiresonances in a 
QFT with baryon and meson fields, associated with composite particles
(with respect to the elementary level of quark--structure),
could be described in the Bethe--Salpeter framework, as a natural relativistic QFT
extension of the phenomena permitted by nonlocal potential theory. 

As far as we know, the connection between nonlocal
potentials and complex angular
momentum theory is totally missing in the literature.
Let us quote what is
written in the classical textbook by De Alfaro and
Regge \cite{DeAlfaro} on this question:
``Our philosophy in regard to these potentials
(i.e., nonlocal potentials)
is the following: maybe they are there,
but if they are there we do not know what to do with them.
They will not be discussed in this book anymore.''
Even in excellent texts on scattering theory,
like the one by Reed and Simon \cite{Reed}, this problem
has not been considered at all.
Long time ago, one of the authors,
in collaboration with others, treated the scattering
theory for large classes of nonlocal
potentials in a series of papers \cite{Bertero},
where, however, the analysis in the CAM--plane
was not considered.
The present paper can be regarded as a
continuation and completion of these
works, in particular of \cite[IV]{Bertero} (which in Section \ref{se:outline} is referred to as I). 
So, before considering the interesting problem
of describing resonances and antiresonances, which will be outlined   
in the last section of this paper, one needs to settle on a firm basis the
whole theory of scattering functions with respect to the complex variables 
$(\lambda,k)$ in the framework of appropriate classes of nonlocal potentials.

The paper is organized as follows:
in Section \ref{se:outline} we recall the spectral properties
associated with the Schr\"odinger two--particle
operators for a large class of potentials
$U$, which include a local part called $V_0$
and a nonlocal part called $V$. All
these properties concern bound states and
scattering solutions, and are treated in the complex
half--plane $\Imag k\geqslant 0$.
In Section \ref{se:rotationally}
we study a class $\cN_{w,\alpha}$ of rotationally invariant nonlocal potentials $V$,
which is characterized by a positive parameter
$\alpha$; for such a class the previous treatment
is extended to the half--plane $\Imag k\geqslant-\alpha$, so as to include
the analysis of resonances.
In this section we use Smithies' theory \cite{Smithies} of
Fredholm--type integral equations for studying
various properties of the resolvent when $k$ belongs to the half--plane
$\{k\in\C:\Imag k\geqslant-\alpha\}$
and the angular momentum $\ell$ is a non--negative integer.
The Smithies formalism produces modified Fredholm formulae,
whose advantage is to fit rigorously  with the convenient framework
of Hilbert--Schmidt--type kernels.
It is therefore in this formalism that we
investigate in detail bound states,
spurious bound states (or bound states embedded in the continuum), 
anti--bound states, resonances, scattering solutions, and partial scattering amplitudes
in the strip $\Omega_\alpha=\{k\in\C:|\Imag k|<\alpha\}$.
In Section \ref{se:interpolation} we study the interpolation
of the so--called \emphsl{partial potentials} (i.e., the coefficients
of the Fourier--Legendre expansion of the potentials)
in the plane of the complexified angular momentum variable $\lambda$.
The potentials which admit this interpolation and satisfy an exponential 
decrease of the form $e^{-\gamma\Real\lambda}$ ($\gamma>0$) 
in $\C^+_{-\frac{1}{2}}$
will be called ``Carlsonian potentials 
with CAM--interpolation $V(\lambda;R,R')$
and rate of decrease $\gamma$'', in view of the fact that 
the uniqueness of the interpolation is guaranteed by Carlson's theorem.
The class of such potentials which belong to $\cN_{w,\alpha}$
will be denoted $\cN_{w,\alpha}^{\,\gamma}$. 
Then, by using the bounds on the complex angular momentum
Green function (i.e., the extension of the Green function
from integral non--negative values of the angular momentum $\ell$
to complex values $\lambda$) one can extend the Fredholm--Smithies formalism 
of the previous section in terms of 
vector--valued and  operator--valued holomorphic functions of the two
complex variables $\lambda$ and $k$. Correspondingly, one then
studies the properties of the partial scattering amplitude
$T(\lambda,k)$ as a meromorphic function of  
$\lambda$ and $k$ in appropriate
domains $D_{\gamma,\alpha}^{(\delta)}$ of the space $\C^2$.
This sets the basis for analyzing the notions of resonances and of antiresonances
in terms of the polar singularities of $T$ in either part $\Imag\lambda \gtrless 0$ of
the domain $D_{\gamma,\alpha}^{(\delta)}$.
The first part of Section \ref{se:AG.5} is devoted to the 
Watson resummation of the partial wave amplitudes in an angular sector of the
complex $\lambda$--plane for all positive values of $k$ (and also in some
domain of the complex $k$--plane). In the second part of this section,
the notion of a dominant pole in the Watson representation of the scattering amplitude
is introduced and a symmetric analysis of resonances and antiresonances is proposed
in this framework.
Finally, in order to
make simpler the presentation of our results,
we have reserved two appendices for
mathematical ingredients to be used in the text:
while Appendix \ref{appendix:a} is devoted to
the derivation of bounds on the complex angular momentum
Green function and on the Bessel and Hankel functions,
Appendix \ref{appendix:b} deals with continuity and holomorphy properties of
vector--valued and operator--valued functions.

\section{A review of spectral properties for a class of Schr\"odinger two--particle operators
including local and nonlocal potentials}
\label{se:outline}

We consider Schr\"odinger operators of the following form:
\beq
(H\psi)(\bx)=-\Delta \psi(\bx)+(U \psi)(\bx)=
-\Delta \psi(\bx)+V_0(\bx)\psi(\bx)+\int V(\bx,\by)\psi(\by)\,\rmd\by,
\label{f:1}
\eeq
where $\bx=(x_1,x_2,x_3)$ is a three--dimensional real vector,
$|\bx|^2=\sum_{k=1}^3 x_k^2$, $\Delta=\sum_{k=1}^3
\partial^2/\partial x_k^2$ is the Laplace operator, and
integration is on the whole $\R^3$ space.
The integro--differential operator $H$ represents,
in the center of mass system, the energy operator of two
particles interacting through a local plus a nonlocal potential.
Our assumptions on the potential functions $V_0$ and $V$ are the following:
\begin{itemize}
\item[(a)] $V_0$ and $V$ are real--valued; $V$ is symmetric: i.e., $V(\bx,\by)=V(\by,\bx)$.
\item[(b)] $\sds 
A_p=\left[\int(1+|\bx|)^p \ |V_0(\bx)|^p\,\rmd\bx\right]^{1/p} < +\infty \qquad (p=1,2)$.
\item[(c)] $\sds 
B_p=\left[\int(1+|\bx|)^p\left(\int|V(\bx,\by)|\,\rmd\by\right)^p\,\rmd\bx\right]^{1/p}<+\infty \qquad (p=1,2)$.
\end{itemize}
Note that condition (a) guarantees that the operator $H$
is a time--reversal invariant and formally Hermitian operator.
We then have (see I):
\sku
\begin{proposition}
\label{propo:1}
If conditions {\rm (a)}, {\rm (b)} and {\rm (c)} hold, then the operator $H$ defined by \eqref{f:1} is self--adjoint
in $L^2$ with domain $\cD_H=W^{2,2}$.
\end{proposition}

In this section $L^2$ and $W^{2,2}$ will indicate $L^2(\R^3)$ and $W^{2,2}(\R^3)$, respectively;
moreover, we recall that $W^{2,2}$ is the Sobolev space of all the square integrable functions which
have square integrable distributional derivatives up to the second order.

\vskip 0.2cm

We now consider the following problems.

\skd

\begin{problem}[Bound state problem]
\label{problem:1}
Suppose that the potential functions $V_0$ and $V$ satisfy conditions (a), (b), and (c),
then we study the solutions $\psi$ with
\beq
\psi\in W^{2,2}\quad \mathrm{and} \quad \|\psi\|_{L^2}=1,
\label{f:2}
\eeq
of the Schr\"odinger equation
$(H\psi)(\bx) = k^2 \psi(\bx)$, namely:
\beq
\Delta\psi(\bx)+k^2\psi(\bx)=(U\psi)(\bx)=V_0(\bx)\psi(\bx)+\int V(\bx,\by)\,\psi(\by)\,\rmd\by,
\label{f:3}
\eeq
where $k$ is a complex number with $\Imag k\geqslant 0$.
\end{problem}

The following theorem can then be proved.

\begin{theorem}
\label{theorem:1}
Let $B$ be the set of values of $k$ $(\Imag k \geqslant 0)$ such that a non trivial solution of Problem
\ref{problem:1} exists; then:
\begin{itemize}
\item[\rm (i)] for any $k\in B$ the number of linearly independent solutions of Problem \ref{problem:1}
(i.e., the multiplicity of $k$) is finite;
\item[\rm (ii)] $B$ is contained in the union of the real axis and the positive imaginary axis;
\item[\rm (iii)] $B$ is bounded;
\item[\rm (iv)] $B$ is countable with no limit points except $k=0$;
\item[\rm (v)] if a real value of $k$ belongs to $B$, then $(-k)\in B$ and Problem \ref{problem:1}
has precisely the same solutions associated with $k$ and $-k$.
\end{itemize}
\end{theorem}

\begin{proof}
See I (proof of Theorem 2.1).
\end{proof}

We introduce the index $n$ ($n=1,2,3,\ldots$) to label the imaginary and the positive real values
of $k$ which belong to $B$, and we understand to count any such value of $k$ as many times as its multiplicity.
Then the numbers $E_n=k_n^2$ ($n=1,2,3,\ldots$) are the eigenvalues of the operator $H=-\Delta+U$;
with each eigenvalue $E_n$ we can associate one and only one
solution $\psi_n(\bx)$ of Problem \ref{problem:1},
i.e., an eigenfunction of $H$. Since the operator $H$ is self--adjoint
(see Proposition \ref{propo:1}), the functions $\psi_n(\bx)$ can be regarded as forming an orthonormal system.

\skt

\begin{problem}[Scattering problem]
\label{problem:2}
We study solutions $\Psi_\bxi(\bx)$ of the Schr\"odinger equation
$(H\Psi_\bxi)(\bx) = |\bxi|^2 \Psi_\bxi(\bx)$, which are of the following form:
\beq
\Psi_\bxi(\bx)=e^{\rmi\langle\bxi,\bx\rangle}+\Phi(\bxi,\bx),
\label{f:4}
\eeq
where $\bxi=(\xi_1,\xi_2,\xi_3)$ is a given
three--dimensional vector, $\langle\bxi,\bx\rangle=\sum_{j=1}^3\xi_j x_j$
and $\Phi(\bxi,\bx)$ satisfies the following properties:
\begin{itemize}
\addtolength{\itemsep}{0.1cm}
\item[(i)] $\Phi(\bxi,\cdot) \in W^{2,2}_\mathrm{loc}$;
\item[(ii)] $\sds \Phi(\bxi,\bx) = O\left(\frac{1}{|\bx|}\right) \qquad (|\bx|\to+\infty)$;
\item[(iii)] $\sds \int_{|\bx|=r}\left|\frac{\partial\Phi}{\partial |\bx|}(\bxi,\bx)-
\rmi |\bxi|\Phi(\bxi,\bx)\right|^2\mu(\rmd\bx) \xrightarrow[r\to +\infty]{} 0$,
$\mu$ denoting the Lebesgue measure on the sphere;
\item[(iv)]$\sds \Delta\Phi(\bxi,\bx)+|\bxi|^2\Phi(\bxi,\bx)=
V_0(\bx)\Phi(\bxi,\bx)+\int V(\bx,\by)\Phi(\bxi,\by)\,\rmd\by+V_0(\bx)e^{\rmi\langle\bxi,\bx\rangle}$ \\
\null\hspace{0.5cm} $\sds+\int V(\bx,\by)e^{\rmi\langle\bxi,\by\rangle}\,\rmd\by$, \quad
(Schr\"odinger equation written in terms of $\Phi$).
\end{itemize}
\end{problem}

Note that $\Phi\in W^{2,2}_\mathrm{loc}$
means that $f\,\Phi\in W^{2,2}$ for any $f\in C_0^\infty(\R^3)$. Besides,
we remark that the third condition involves the notion of trace in the sense of Sobolev of the function $\Phi$ and of
its first derivatives. To this purpose, let us note that if $u \in W^{2,2}$ and $\Sigma$ is a smooth surface,
then the trace of $u$ on $\Sigma$ can be defined as the restriction of $u$ on $\Sigma$. This definition is meaningful
because $u$ is a continuous function, but in the case of the first derivatives, the previous definition
breaks down. However, the linear operator relating every $\Phi\in C^\infty_0(\R^3)$ to
$\partial\Phi/\partial x_k$, restricted on the sphere of radius $r$, is a continuous mapping
of a subspace of $W^{2,2}$ into the space of the functions which are square integrable on the sphere of radius
$r$. Since $C^\infty_0(\R^3)$ is dense in $W^{2,2}$, this operator can be extended in a unique way to the
whole $W^{2,2}$ (see the Appendix of I). Furthermore, any function $\Phi\in W^{2,2}_\mathrm{loc}$ is
continuous (Lemma A.3 of I), so that from the second
condition of Problem \ref{problem:2} it follows that $\Psi_\bxi$
is a bounded and continuous function of $\bx$ in $\R^3$. Then the following theorem holds.

\begin{theorem}
\label{theorem:2}
For any $\bxi\in\R^3$, $\bxi\not = 0$, a solution of Problem \ref{problem:2} exists. The solution is unique if and only if
$|\bxi|\not\in B$ (the set defined in Theorem \ref{theorem:1}). Any solution of Problem \ref{problem:2} has the
following asymptotic behavior:
\begin{align}
\Phi(\bxi,\bx) &= \frac{e^{\rmi|\bxi||\bx|}}{|\bx|}F\left(|\bxi|\frac{\bx}{|\bx|},\bxi\right)+o\left(\frac{1}{|\bx|}\right)
\qquad (|\bx|\to +\infty), \nonumber \\
F(\bx,\bxi) &= -\frac{1}{4\pi}\int e^{-\rmi\langle\bx,\by\rangle}\left(U\Psi_\bxi\right)(\by)\,\rmd\by. \nonumber
\end{align}
The quantity $F\left(|\bxi|\frac{\bx}{|\bx|},\bxi\right)$, the so--called scattering amplitude,
is uniquely defined for any $\bxi\in\R^3$.
\end{theorem}

\begin{proof}
See I (proof of Theorem 2.2).
\end{proof}

Besides, let $\Omega^R$ be the set of vectors $\bxi\in\R^3$ such that $|\bxi|\not\in B$, $\bxi\neq 0$;
then the functions $\Omega^R\ni\bxi\mapsto\Phi(\bxi,\bx)$ are equicontinuous, i.e.:
\beq
\sup_{\bx\in\R^3}\left|\Phi(\bxi,\bx)-\Phi(\bbeta,\bx)\right|\longrightarrow 0,
\eeq
where $\bxi\in\Omega^R$ is fixed, and $|\bxi-\bbeta|\to 0$.

We can rapidly sketch the approach to Problems \ref{problem:1} and \ref{problem:2}.
One first introduces the following operator $L(k)$:
\beq
L(k)f(\bx)=-\int\! V_0(\bx)\frac{e^{\rmi k|\bx-\by|}}{4\pi|\bx-\by|}f(\by)\,\rmd\by
-\int\!\left(\int \!V(\bx,\bz)\frac{e^{\rmi k|\bz-\by|}}{4\pi|\bz-\by|}\,d\bz\right)\,f(\by)\,\rmd\by,
\label{f:7}
\eeq
acting on the Hilbert space $X^2$ defined by
\beq
X^2=\left\{f\in L^2\,:\,\|f\|_{X^2} \doteq \left(\int(1+|\bx|)^2|f(\bx)|^2\,\rmd\bx\right)^{1/2}<+\infty\right\}. \nonumber
\label{f:8}
\eeq
Note that the function ${e^{\rmi k|\bx-\by|}}/{4\pi|\bx-\by|}$ (in definition \eqref{f:7} of the operator
$L(k)$), is the Green function associated with the classical radiation problem. Next, one
applies the Riesz--Schauder theory \cite{Yosida} to the resolvent $R(k)=[1-L(k)]^{-1}$, and then the following
alternative holds:
\begin{itemize}
\addtolength{\itemsep}{0.1cm}
\item[(a)] either $R(k)=[1-L(k)]^{-1}$
exists and is a bounded operator on $X^2$;
\item[(b)] or the space spanned by the eigenfunctions of $L(k)$ has dimension $n\geqslant 1$.
\end{itemize}
In case (a) the integral equation
\beq
v(\bxi,\cdot)=v_0(\bxi,\cdot)+L(|\bxi|)\,v(\bxi,\cdot),
\label{f:9}
\eeq
where
\beq
v_0(\bxi,\bx)=V_0(\bx)\,e^{\rmi\langle\bxi,\bx\rangle}+\int V(\bx,\by)\,e^{\rmi\langle\bxi,\by\rangle}\,\rmd\by,\nonumber
\label{f:10}
\eeq
has a unique solution in $X^2$, i.e.:
\beq
v(\bxi,\cdot)= R(|\bxi|) \, v_0(\bxi,\cdot).\nonumber
\label{f:11}
\eeq
In case (b), Eq. \eqref{f:9} has a solution $v(\bxi,\cdot)\in X^2$ if and only if $v_0(\bxi,\cdot)$
is orthogonal in $X^2$ to $n$ linearly independent eigenfunctions of the adjoint operator $L^\dagger(|\bxi|)$.
For this purpose the following theorem can be proved.

\begin{theorem}
\label{theorem:3}
{\rm (i)} If the function $f$ is a solution of the problem
\beq
f \in X^2, \quad f = L(|k|)f, \quad f \not\equiv 0 \qquad (k\neq 0,\, \Imag k \geqslant 0),
\label{f:12}
\eeq
then the function
\beq
\psi(\bx)=-\int\frac{e^{\rmi k|\bx-\by|}}{4\pi|\bx-\by|}\,f(\by)\,\rmd\by \nonumber
\label{f:13}
\eeq
is a solution of the problem
\beq
\psi \in W^{2,2}, \quad \Delta\psi+k^2\psi=U\psi \qquad (\psi\not\equiv 0).
\label{f:14}
\eeq
{\rm (ii)} Conversely, if $\psi(\bx)$ is a solution of problem \eqref{f:14} with $\Imag k\geqslant 0$, then the function
\beq
f(\bx)=\Delta\psi(\bx)+k^2\psi(\bx)\nonumber
\label{f:15}
\eeq
is a solution of problem \eqref{f:12}.\\
{\rm (iii)} If $\psi(\bx)$ is a solution of problem \eqref{f:14} with $\Imag k=0$, then the function
\beq
{f^*}(\bx)=(1+|\bx|)^{-2}\psi(\bx)\nonumber
\label{f:16}
\eeq
is an eigenfunction in $X^2$ of the adjoint operator $L^\dagger(k)$.
\end{theorem}

\begin{proof}
See I (proof of Theorem 4.2).
\end{proof}

Let now $\psi_1,\ldots,\psi_n$ be 
$n$ linearly independent solutions of problem \eqref{f:12}; by
statement (iii) of the theorem above it follows that the functions
$f^*_j(\bx)=(1+|\bx|)^{-2}\psi_j(\bx)$ 
are linearly independent eigenfunctions of the adjoint operator
$L^\dagger(|\bxi|)$. 
Since the following equality holds (see I, Eq. (6.5)):
\beq
\left\langle f^*_j,v_0(\bxi,\cdot)\right\rangle_{X^2} = 0,\nonumber
\label{f:17}
\eeq
where $\langle\cdot,\cdot\rangle_{X^2}$ denotes the scalar product in $X^2$,
then a solution to Eq. \eqref{f:9} in $X^2$ exists but it is no longer unique. Nevertheless,
the scattering amplitude $F\left(|\bxi|\frac{\bx}{|\bx|},\bxi\right)$ is unique (see I).

Once the solutions of the time--independent Schr\"odinger equation have been obtained one can write,
following a procedure which goes back to T. Ikebe\footnote{It must be mentioned that there was a subtle error
in the original paper by Ikebe \cite{Ikebe1}, which was subsequently corrected by B. Simon \cite{Simon};
the interested reader is referred to Simon's monograph \cite{Simon}.} \cite{Ikebe1,Ikebe2},
an expansion formula for an arbitrary function $f\in L^2$ in terms of the eigenfunctions $\psi_n(\bx)$
and of the scattering solutions $\Psi_\bxi(\bx)$ (see I, Theorem 2.3).
By the use of this expansion one obtains a spectral representation of the operator $H$:
\beq
\left(Hu(\cdot,t)\right)(\bx) =
-\Delta u(\bx,t)+V_0(\bx)u(\bx,t)+\int V(\bx,\by)u(\by,t)\,\rmd\by=
\rmi\,\frac{\partial u(\bx,t)}{\partial t},\nonumber
\label{f:18}
\eeq
and of functions of $H$. In particular, a representation of the evolution operator
$\exp(-\rmi tH)$ is obtained, and then the solution of the time--dependent Schr\"odinger equation can be studied.
In the case of the nonlocal potentials being considered, we face an additional problem:
the possible existence of positive
energy eigenvalues and, correspondingly, the non--uniqueness of the scattering solutions.
Nevertheless the scattering amplitude exists and
is unique for any $\bxi\in\R^3$ (see Theorem \ref{theorem:2}) even
if the scattering solution is not unique. This allows us to define uniquely the scattering operator
$S=W_+^\dagger W_-$, where $W_\pm=\,\,\sds\stronglimit_{t\to\pm\infty}\,\exp(\rmi tH)\exp(\rmi t\Delta)$
($\stronglimit\equiv$ strong limit), and prove the unitarity of $S$ in a very general setting (see I).


\section{Rotationally invariant nonlocal potentials;
analyticity in the $\boldsymbol{k}$--plane of the resolvents
and of the partial scattering amplitudes at fixed angular momentum} 
\label{se:rotationally}

Hereafter we shall be concerned 
with a class of \emphsl{nonlocal} potentials,
which represents a natural generalization
of the central character of the local interaction.
These potentials, denoted by $V(\bR,\bR')$,
are assumed to depend only on the lengths $R,R'$  of the vectors
$\bR$ and $\bR'$, and on the
angle $\eta$ between them, or equivalently,
on the shape and dimension of the triangle $(O,\bR,\bR')$, but not
on its orientation. 
We then rewrite Eq. \eqref{f:3} in the following form:
\beq
(H\psi)(\bR)=-\Delta\psi(\bR)+g\int_{\R^3}V(\bR,\bR')\,\psi(\bR')\,\rmd\bR'
= k^2 \, \psi(\bR).
\label{f:19}
\eeq
This equation can be seen to represent
the two--body Schr\"odinger equation in its reduced form with respect
to the coordinate $\bR$ of the relative motion between the two interacting
particles; accordingly, 
$\psi(\bR)$ represents the relative motion wavefunction, and $g$ is the
coupling constant of the interacting particles. The Planck constant
$\hbar$ and the reduced mass $\mu$ do not
appear in Eq. \eqref{f:19} corresponding to a simple choice of units ($\hbar=2 \mu=1$).
In view of the assumptions on $V(\bR,\bR')$, we can write the
following formal expansion:
\beq
V(\bR,\bR') \doteq
V(R,R';\cos\eta)
=\frac{1}{4\pi RR'}\sum_{\ell=0}^\infty(2\ell+1)\,V_\ell(R,R')\,P_\ell(\cos\eta),
\label{f:20}
\eeq
where $\cos\eta=(\bR\cdot\bR')/(RR')$,
and the $P_\ell(\cdot)$ are the Legendre polynomials.
The Fourier--Legendre coefficients $V_\ell(R,R')$ of $V(\bR,\bR')$ are given by:
\beq
V_\ell(R,R')=2\pi RR'\int_{-1}^{+1}V(R,R';\cos\eta)\,P_\ell(\cos\eta)\,\rmd(\cos\eta).
\label{f:21}
\eeq
Next, from the current conservation law it follows that $V(\bR,\bR')$ is a real and symmetric function:
$V(\bR,\bR')=\overline{V(\bR,\bR')}=V(\bR',\bR)$.
We can thus conclude that, \emphsl{provided the coupling constant $g$ is restricted to real
values, the Hamiltonian $H$ is a formally Hermitian and rotationally invariant operator.}
The relative motion wavefunction $\psi(\bR)$ can now be expanded in the form:
\beq
\psi(\bR)=\frac{1}{R}\sum_{\ell=0}^\infty \psi_\ell(R)\, P_\ell(\cos\theta),
\label{f:22}
\eeq
where $\ell$ is now the relative angular momentum between the interacting particles.

Representing the unit vectors
$(\bR/R)$ and $(\bR'/R')$ respectively by
the angles $(\theta,\varphi)$ and $(\theta',\varphi')$, we have:
$\cos\eta=\cos\theta\cos\theta'+\sin\theta\sin\theta'\cos(\varphi-\varphi')$.
Then, using the following addition formula for the Legendre polynomials:
\beq
\int_0^\pi\int_0^{2\pi}P_s(\cos\eta)P_\ell(\cos\theta')\sin\theta'\,\rmd\theta'\rmd\varphi'=
\frac{4\pi}{2\ell+1}P_\ell(\cos\theta)\,\delta_{s\ell},
\label{f:23}
\eeq
one readily obtains from formulae \eqref{f:19}--\eqref{f:23} the 
following {\bf nonlocal Schr\"odinger--type
integro--differential equation at fixed angular momentum}:
\beq
D_{\ell,k}\psi_\ell(R)\doteq 
\psi_\ell^{''}(R)+k^2\psi_\ell(R)-\frac{\ell(\ell+1)}{R^2}\psi_\ell(R)=
g\int_0^{+\infty}\!\!V_\ell(R,R')\,\psi_\ell(R')\,\rmd R',
\label{f:24}
\eeq
where $k^2=E$ is the relative kinetic energy of 
the two particles in the center of mass system,
and (for all integer $\ell$):
\beq
V_\ell (R,R')
= \overline{V_\ell(R,R')}
= V_\ell(R',R).
\label{f:24'}
\eeq
We study two types of solutions of 
the Schr\"odinger--type equation  \eqref{f:24}:

\vskip 0.2cm

\noindent
{\bf (S--a) Bound state solutions}, which satisfy the following conditions:
\begin{itemize}
\addtolength{\itemsep}{0.1cm}
\item[(i)] $\psi_\ell'(R)$ is absolutely continuous;
\item[(ii)] $\psi_\ell(0)=0$;
\item[(iii)] $\sds\int_0^{+\infty}|\psi_\ell(R)|^2\,\rmd R <+\infty$.
\end{itemize}

\vskip 0.2cm

\noindent
{\bf (S--b) Scattering solutions}, denoted by $\Psi_\ell(k;R)$, 
which satisfy the following conditions:
\begin{itemize}
\addtolength{\itemsep}{0.1cm}
\item[(i$^\prime$)] $\Psi_\ell'(k;R)$ is absolutely continuous;
\item[(ii$^\prime$)]
$\Psi_\ell(k;R)$ can be written as
\beq
\Psi_\ell(k;R)=kR j_\ell(kR)+\Phi_\ell(k;R),
\label{f:25}
\eeq
where $j_\ell(\cdot)$ denotes the spherical Bessel function
(note that $R j_\ell(kR)$ is a solution of the differential equation 
$D_{\ell,k} [Rj_\ell(kR)]=0$ which 
vanishes at $R=0$),  
and $\Phi_\ell(k;R)$ satisfies the following conditions:
\beq
\null\qquad\Phi_\ell(k;0)=0,
\quad
\lim_{R\to+\infty}
\left[\frac{\rmd}{\rmd R}\Phi_\ell(k;R)-\rmi k \Phi_\ell(k;R)\right]=0 
\quad (k\in\R^+).
\label{f:26}
\eeq
The second equality in \eqref{f:26} is the so--called \emphsl{Sommerfeld radiation condition}.
\end{itemize}

\noindent
It will be shown below in Theorem \ref{theorem:7'} 
that
(in accordance with the study given in refs. \cite[I,II,III]{Bertero})
the research of solutions of type (S-b)
of the \emphsl{Schr\"odinger--type}  
equation \eqref{f:24} reduces to the problem of solving a 
{\bf Lippmann--Schwinger--type linear integral equation},
namely the following \emphsl{inhomogeneous} Fredholm equation:

\begin{subequations}
\label{f:27}
\begin{align}
& v_\ell(k,g;R) = v_{\ell,0}(k;R)+g\int_0^{+\infty}L_\ell(k;R,R')v_\ell(k,g;R')\, \rmd R',
\label{f:27a} \\  
\intertext{where:}
& v_{\ell,0}(k;R) = \int_0^{+\infty} V_\ell(R,R')\, kR'j_\ell(kR')\,\rmd R', \label{f:27b} \\
& L_\ell(k;R,R') = \int_0^{+\infty} V_\ell(R,R'')\,G_\ell(k;R'',R')\,\rmd R''. \label{f:27c}
\end{align}
\noindent 
In the latter, $G_\ell(k;R,R')$ satisfies the ``Green function" distributional identity 
\beq
D_{\ell,k}(R)\  G_\ell(k;R,R')= \delta(R-R'),
\label{f:27c'}
\eeq
and is explicitly given by the following formula:
\beq
G_\ell(k;R,R') = 
G_\ell(k;R',R) = 
-\rmi kRR'j_\ell\left[k\min(R,R')\right]\,h_\ell^{(1)}\!\left[k\max(R,R')\right], 
\label{f:27d}
\eeq
\end{subequations}
where $h_\ell^{(1)}(\cdot)$ denotes the spherical Hankel function.

\noindent
Similarly , as it will be shown in Theorem \ref{theorem:7}, the solutions of
type (S-a) of Eq. \eqref{f:24} are associated with 
solutions of the corresponding 
\emphsl{homogeneous} Fredholm equation 
(obtained by replacing $v_{\ell,0}$ by 0 in Eq. \eqref{f:27a}).

\vskip 0.1cm

\noindent
{\bf Symmetry properties in the complex $\boldsymbol{k}$--plane.}
\, In view of the parity and complex conjugation properties satisfied by the
spherical Bessel and Hankel functions, namely
$\overline{j_\ell(z)} =  j_\ell(\overline z)
= (-1)^\ell j_\ell(-\overline z)$,
$\overline{h_\ell^{(1)}(z)} =(-1)^\ell  h_\ell^{(1)}(-\overline z)$
(see \cite[formulae (9.1.35), (9.1.39), (9.1.40)]{Abramowitz}),  the following symmetry
properties readily follow from the 
reality condition \eqref{f:24'} on the potentials $V_\ell$ and 
from Eqs. \eqref{f:27b}, \eqref{f:27c}, \eqref{f:27d}:
\begin{align}
& \overline{G_\ell(k;R,R')} = G_\ell(-\overline k;R,R'), \label{f:symd} \\
& \overline{L_\ell(k;R,R')} = L_\ell(-\overline k;R,R'), \label{f:symc} \\
& \overline{v_{\ell,0}(k;R)}= v_{\ell,0}(\overline k;R). \label{f:symb}
\end{align}

\noindent
The present section is organised as follows. After having defined 
an appropriate class of nonlocal potentials
together with the corresponding Hilbert--space framework,
we give a complete study of the Fredholm--resolvent integral equation
at fixed angular momentum $\ell$,
associated with Eqs. \eqref{f:27}; this is done in Subsection \ref{subse:properties-L}.
Then the general meromorphy properties of this resolvent
with respect to the complex momentum variable $k$ are presented in
Subsection \ref{subse:mero}, where we also outline the algebraic correspondence
between the pole structure of the latter and bound--state--type 
solutions of the nonlocal Schr\"odinger equation.  
Complete results concerning the relationship between the 
integral equation formalism and the Schr\"odinger--type formalism
are then given in Subsection \ref{subse:corresp}, including the introduction
and analyticity properties in $k$ of the partial scattering amplitudes
$T_\ell(k;g)$. Finally, a short Subsection \ref{subse:partial-wave} is devoted to  
the partial wave expansion of the total scattering amplitude 
$F(k,\cos\theta;g)$ and to its general analyticity properties in $k$
and $\cos\theta$.

\subsection{Classes $\boldsymbol{{\cN}_{w,\alpha}}$
of nonlocal potentials. Properties of the functions $\boldsymbol{v_{\ell,0}(k,\cdot)}$, 
of the operators $\boldsymbol{L_\ell(k)}$, and of the resolvents
$\boldsymbol{R_\ell(k;g)}$ in the $\boldsymbol{k}$--plane}
\label{subse:properties-L}

\noindent
{\bf Definitions.} In what follows, all the functions 
of the real positive variable $R$ are considered
as defined \emphsl{for almost every (a.e.)} $R$ with respect to an 
appropriate measure on $\R^+$.
For each positive number $\alpha$, we introduce:

\noindent 
(1) The Hilbert space 
\beq
X_{w,\alpha} = \left\{x(R)\,:\,\|x\|_{w,\alpha} \doteq
\left[\int_0^{+\infty} {w}(R)\, e^{2\alpha R}
|x(R)|^2\,\rmd R\right]^{1/2} < +\infty\right\}.
\label{f:28} 
\eeq 
In \eqref{f:28} $w$ denotes a given
continuous and strictly positive \emphsl{weight--function} on the interval
$[0,+\infty)$. For any bounded operator $A$ on $X_{w,\alpha}$, the corresponding
norm will be denoted by $\|A\|_{X_{w,\alpha}}$, or simply
$\left\|A \right\|$. A subspace of bounded operators equipped
with an appropriate Hilbert--Schmidt (HS) norm
$\left\|A\right\|_{\mathrm{HS}_{w,\alpha}}$, or simply
$\left\|A\right\|_\mathrm{HS}$, (such that 
$\left\|A\right\|_\mathrm{HS} \geqslant \left\|A \right\|$), will be introduced below.

\vskip 0.3cm 

\noindent 
(2) The class $\cN_{w,\alpha}$ of \emphsl{rotationally invariant nonlocal potentials}
$V(\bR,\bR')$ (defined for a.e. $\bR,\bR'$), which satisfy the conditions 
$V(\bR,\bR')=\overline{V(\bR,\bR')}=V(\bR',\bR)$ (or equivalently conditions \eqref{f:24'}), 
together with the following condition: 
\beq 
{\bf C}(V) \doteq \left[\int_{\R^3} w(R) \, e^{2\alpha R} \,\rmd\bR
\int_{\R^3} w(R') \, e^{2\alpha R'} \, V^2(\bR,\bR') \, \rmd\bR'\right]^{1/2} < +\infty. 
\label{f:29} 
\eeq
In view of Parseval's equality, we also have: 
\beq
{\bf C}(V)=\left[\int_0^{+\infty}\!\!\! w(R)\,e^{2\alpha R}\,\rmd R
\int_0^{+\infty}\!\!\! w(R')\,e^{2\alpha R'}
\sum_{\ell=0}^{+\infty}(2\ell+1)\,V_\ell^2(R,R')\,\rmd R'\right]^{1/2}\!\!,
\label{f:29'} 
\eeq 
so that the partial potentials $V_\ell(R,R')$ satisfy 
(for all $\ell\geqslant 0$) the condition 
\beq
C(V_\ell)\doteq\left[\int_0^{+\infty} w(R)\,e^{2\alpha R}\,\rmd R
\int_0^{+\infty} w(R')\, e^{2\alpha R'}
V_\ell^2(R,R')\,\rmd R'\right]^{1/2} < +\infty, 
\label{f:30} 
\eeq
or, in terms of the function (defined for a.e. $R$)
\beq
V_\ell^{(w)}(R)\doteq 
\left(\int_0^{+\infty} w(R')\,e^{2\alpha R'} \, V_\ell^2(R,R')\,\rmd R'\right)^{1/2},
\label{f:30'}
\eeq
which belongs to $X_{w,\alpha}$, 
\beq 
C(V_\ell) = \left\|V_\ell^{(w)}\right\|_{w,\alpha}
\leqslant\frac{{\bf C}(V)}{\sqrt{2\ell+1}}.  
\label{f:31} 
\eeq 
Our aim is to consider the integral equation \eqref{f:27a} as a
linear equation in $X_{{w},\alpha}$ depending on the complex
parameters $k$ and $g$ and on the integer $\ell$ ($\ell \geqslant 0$), which
we rewrite in operator form as follows: 
\beq
\left[\I-gL_\ell(k)\right] v_\ell(k,g;\cdot)=
v_{\ell,0}(k;\cdot). 
\label{f:32} 
\eeq 
In the latter, $\I$ denotes the identity operator in $X_{{w},\alpha}$,
$v_{\ell,0}(k;\cdot)$ is the function defined by 
\eqref{f:27b}, and $L_\ell(k)$ denotes the integral operator with
kernel $L_\ell(k; R,R')$ (see \eqref{f:27c}). In the
complex plane of $k$, we consider the strip
$\Omega_\alpha\doteq\{k\in\C \,:\, |\Imag k| < \alpha\}$ 
and the half--plane 
$\Pi_\alpha\doteq\{k\in\C \,:\, \Imag k > -\alpha\}$.
$\overline{\Pi}_\alpha$ and $\overline{\Omega}_\alpha$ will
denote the closures of $\Pi_\alpha$ and $\Omega_\alpha$, respectively. 
We shall then specify classes of nonlocal
potentials $\cN_{{w},\alpha}$, with appropriate conditions on
the weight--function $w$ in such a way that the following
properties can be established:

\begin{itemize}
 \item[(i)] the functions $v_{\ell,0}(k;\cdot)$  belong to
$X_{w,\alpha}$ for all $k\in \overline{\Omega}_\alpha$, ($\ell=0,1,2,\ldots$);\\[-5pt]
\item[(ii)] the operators $L_\ell(k)$ are compact operators of
Hilbert--Schmidt--type in $X_{w,\alpha}$ for all
$k\in\overline{\Pi}_\alpha$, ($\ell=0,1,2,\ldots$).
\end{itemize}
In fact, for all such classes of potentials, the kernel
$L_{\ell}(k;R,R')$ will be majorized by an appropriate kernel
of rank one, and this will then allow us to apply
Smithies' refined version of the Fredholm
theory \cite{Smithies} for describing and discussing
the solutions of Eq. \eqref{f:32}.

Note that a similar study, which was more
based on the results of the Riesz--Schauder theory \cite{Yosida},
had been performed for the particular case $\ell=0$ (s--wave)
and a slightly different class of potentials in \cite[I]{Bertero}.

\subsubsection{Properties of the vector--valued functions
$\boldsymbol{k}\boldsymbol{\mapsto}
\boldsymbol{v}_{\boldsymbol{\ell},\boldsymbol{0}}
\boldsymbol{(}\boldsymbol{k};\boldsymbol{\cdot}\boldsymbol{)}$}
\label{subsubse:v-ell-0}  

We shall rely on the fact that the spherical Bessel functions 
$j_\ell(z)$ are entire functions for all integers $\ell$ ($\ell \geqslant 0$),
which satisfy bounds of the form \eqref{a:0}, valid
for all integers $\ell$ ($\ell \geqslant 0$) and for all $k\in \C$ (see \cite{Newton2}). 
For $k\in \overline\Omega_\alpha$ we shall use the norm of the function $kRj_\ell(kR)$ 
in the dual space $X^*_{w,\alpha}$ of $X_{w,\alpha}$, namely
\beq
\left\|{k\cdot} \,\, j_\ell(k\cdot)\right\|^*_{w,\alpha}
\doteq\left(\int_0^{+\infty}
\frac{|kR j_\ell(kR)|^2}{w(R)\,e^{2\alpha R}}\,\rmd R\right)^{1/2}.
\label{f:j0}
\eeq
In fact, in view of \eqref{a:0}, we have for all $k\in \overline\Omega_\alpha$ 
and $\ell\geqslant 0$:
\beq
\frac{1}{c_\ell} \ \left\|{k\cdot} \,\, j_\ell(k\cdot)\right\|^*_{w,\alpha}
\leqslant
\left[\int_0^\infty  \left(\frac{|k|R}{1+|k|R}\right)^2 \frac{\rmd R}{w(R)}\right]^{1/2} 
\leqslant A_w(|k|),
\label{f:j1}
\eeq
where the last inequality expresses a requirement on the 
weight--function $w$, namely the existence of
a positive and non--decreasing function 
$|k| \mapsto A_w (|k|)$ to be defined on $\R^+$.
We shall then prove the following lemma.

\begin{lemma}
\label{lemma:v} 
For every potential $V$ in a class
$\cN_{{w},\alpha}$ such that $w$ satisfies a condition of the
type \eqref{f:j1}, the corresponding functions 
$k\mapsto v_{\ell,0}(k;\cdot)$ (formally defined in \eqref{f:27b}) are 
well--defined for all integers $\ell$, $\ell\geqslant 0$,
as functions on $\overline {\Omega}_\alpha$ with
values in $X_{w,\alpha}$; for each $\ell$  the corresponding
norm $\|v_{\ell,0}(k;\cdot)\|_{w,\alpha}$ admits the following 
bound for $k$ varying in $\overline {\Omega}_\alpha$: 
\beq
\|v_{\ell,0}(k;\cdot)\|_{w,\alpha}
\leqslant
\frac{{\bf C}(V)}{(2\ell + 1)^{1/2}} 
\left\|{k\cdot} \,\, j_\ell(k\cdot)\right\|^*_{w,\alpha}
\leqslant c_\ell \frac{{\bf C}(V)}{(2\ell + 1)^{1/2}} \, A_{w}(|k|).
\label{f:v3}
\eeq
Moreover, this vector--valued function is 
continuous in $\overline{\Omega}_\alpha$ and 
holomorphic in $\Omega_\alpha$.
\end{lemma}

\begin{proof}
Starting from Eq. \eqref{f:27b}, using the Schwarz inequality and taking into account
Eqs. \eqref{f:30'}, \eqref{f:j0} and condition \eqref{f:j1},
we can write for a.e. $R$:
\beq
|v_{\ell,0}(k;R)|=\left|\int_0^{+\infty}V_\ell(R,R')kR' j_\ell(kR')\,\rmd R'\right|
\leqslant
V_\ell^{(w)}(R) \,\, \left\|{k\cdot} \,\, j_\ell(k\cdot)\right\|^*_{w,\alpha}.
\label{f:v2}
\eeq
It then follows from \eqref{f:31} that the function 
$v_{\ell,0}(k;\cdot)$ belongs to ${X_{{w},\alpha}}$ and 
the bounds \eqref{f:v3} readily follow from 
\eqref{f:31}, \eqref{f:v2}, and \eqref{f:j1} for  
$k\in\overline{\Omega}_\alpha$.

\vskip 0.1cm

\noindent
\emph{Proof of the last statement:}
For $k$ varying in any bounded domain of the form
$\Omega_\alpha^{(K)}\doteq\{k\in\Omega_\alpha\,:\, |k| < K\}$ 
(or in its closure $\overline{\Omega}_\alpha^{(K)}$),
Eq. \eqref{f:v2} yields the following bound 
$|v_{\ell,0}(k;R)|\leqslant c_\ell A_{w}(K) V_\ell^{(w)}(R)$. 
One concludes that $v_{\ell,0}(k;R)$ belongs to a class 
${\cC}(D,\mu,p)$ (see Lemma \ref{lemma:B8}), with 
$D=\Omega_\alpha^{(K)}$ (or $\overline{\Omega}_\alpha^{(K)}$),
$\mu(R)= w(R)\,e^{2\alpha R}$, and $p=2$.
One moreover checks that the function
$v_{\ell,0}(k;R)$ is continuous in $\overline{\Omega}_\alpha^{(K)}$,
holomorphic in $\Omega_\alpha^{(K)}$ for a.e. $R$ in view of
Lemma \ref{lemma:B9}, since the integrand of \eqref{f:27b},
holomorphic with respect to $k$ in $\C$, can be uniformly bounded 
in $\overline{\Omega}_\alpha^{(K)}$
(in view of \eqref{a:0}) by the integrable function:
\beq
R' \longmapsto c_\ell V_\ell(R,R') w^{1/2}(R')\,e^{\alpha R'} \,
\left[\left(\frac{KR'}{1+KR'}\right) w^{-1/2}(R')\right] \nonumber
\eeq
(the integrability property of the latter being obtained by the Schwarz
inequality in view of \eqref{f:30'}) and \eqref{f:j1}). 
Lemma \ref{lemma:B8} is thus applicable and allows one
to state that the vector--valued function
$k\mapsto v_{\ell,0}(k;\cdot)\in X_{w,\alpha}$ is
continuous in $\overline{\Omega}_\alpha^{(K)}$,
holomorphic in $\Omega_\alpha^{(K)}$.
Since the argument is valid for any value of $K$, the
last statement of the lemma is thus established.
\end{proof}

\noindent
{\bf Choice of the weight--function $\boldsymbol{w}$.}
\ Our choice of relevant weight--functions ${w}$
satisfying a condition of the type \eqref{f:j1} will obey the
following criteria:

\vskip 0.1cm 

(a) ${w}(R)$ should be chosen as small as possible
near $R=0$ in order to include in the class $\cN_{w,\alpha}$ 
potentials $V(R,R')$, whose behaviour is as much singular as
possible near $R=0$ and $R'=0$.

Note that the behaviour of $w(R)$ for $R$ tending to infinity
is not so relevant, since the dominant behaviour of $V(R,R')$
at large $R$ and $R'$ is in fact dictated by the factors
$e^{\alpha R}$ and $e^{\alpha R'}$, which make the classes
$\cN_{w,\alpha}$ look like nonlocal versions of Yukawa--type
potentials. The convergence of the integral \eqref{f:j1} at
$R\to \infty$ will only serve to ensure the validity of the
boundedness of $\|v_{\ell,0}(k;\cdot)\|_{w,\alpha}$ for
$k$ lying on the boundary of the strip $\Omega_\alpha$.

\vskip 0.1cm 

(b) the behaviour of the majorant $A_{w}(|k|)$ may
keep some flexibility, according to whether one is interested
in improving the behaviour of
$\|v_{\ell,0}(k;\cdot)\|_{w,\alpha}$ at $k\to 0$ or at
$k\to \infty$.

\vskip 0.2cm 

In view of these considerations, we can propose
the following three specifications of the weight--function
$w$, which will correspond respectively to the following
convenient majorizations of the integrand of \eqref{f:j1}:

\begin{itemize}
\item[(i)] $\frac{|k|R}{1+|k|R} < 1$ leads one to a choice $w=w_0$ with
$A_{w_0}^2(|k|) = \int_0^{+\infty}\frac{\rmd R}{w_0(R)}$. 
An appropriate choice is:
$w_0 = w_0^{(\varepsilon)} \doteq R^{1-\varepsilon}\,(1+R)^{2\varepsilon}$, so that
\begin{equation}
A_{w_0}^2 = A_\varepsilon^2\doteq
\int_0^{+\infty} \frac{\rmd R}{R^{1-\varepsilon}(1+R)^{2\varepsilon}}.
\label{f:e0}
\end{equation}
\item[(ii)] $\frac{|k|R}{1+|k|R} \leqslant \frac{\sqrt{|k|R}}{2}$
leads one to a choice $w=w_1$ with 
$A_{w_1}^2(|k|)=\frac{|k|}{4} \int_0^{+\infty} \frac{R\,\rmd R}{w_1(R)}$.
An appropriate choice is:
$w_1 = w_1^{(\varepsilon)} \doteq R^{2-\varepsilon}(1+R)^{2\varepsilon}$, so that one has:
$A_{w_1}(|k|)=\frac{\sqrt{|k|}}{2} A_\varepsilon$,
(with $A_\varepsilon$ given by formula \eqref{f:e0}). \\[-4pt]
\item[(iii)] $\frac{|k|R}{1+|k|R} \leqslant |k|R$ \,
leads one to a choice $w=w_2$ with 
$A_{w_2}^2(|k|) =|k|^2 \int_0^{+\infty}\!\! \frac{R^2\,\rmd R}{w_2(R)}$. 
An appropriate choice is: 
$w_2 = w_2^{(\varepsilon)} \doteq R^{3-\varepsilon}(1+R)^{2\varepsilon}$, so that one has:
$A_{w_2}(|k|)=|k|\, A_\varepsilon$.
\end{itemize}

\noindent
As a consequence of the previous analysis
and of majorization \eqref{f:v3}, we can now give the following

\vskip 0.5cm

\noindent
{\bf Complement to Lemma $\boldsymbol{\ref{lemma:v}}$.}
\emph{The following bounds hold for each vector--valued function}
$k \mapsto v_{\ell,0}(k;\cdot)$:
\begin{itemize}
\it
\item[\rm (i)] for $w = {w}_0^{(\varepsilon)}=R^{1-\varepsilon}\,(1+R)^{2\varepsilon}$, \qquad
$\|v_{\ell,0}(k;\cdot)\|_{w,\alpha}\leqslant c_\ell\, C(V_\ell)\, A_{\varepsilon}$; \\[-5pt]
\item[\rm (ii)] for $w = w_1^{(\varepsilon)}= R^{2-\varepsilon} \, (1+R)^{2\varepsilon}$, \qquad
$\|v_{\ell,0}(k;\cdot)\|_{w,\alpha}\leqslant \frac{\sqrt{|k|}}{2}\, c_\ell\, C(V_\ell) \, A_{\varepsilon}$; \\[-5pt]
\item[\rm (iii)] for $w = w_2^{(\varepsilon)}=R^{3-\varepsilon}\,(1+R)^{2\varepsilon}$, \qquad 
$\|v_{\ell,0}(k;\cdot)\|_{w,\alpha}\leqslant |k|\, c_\ell\, C(V_\ell)\, A_{\varepsilon}$.
\end{itemize}

\subsubsection{Properties of the operator--valued functions
$\boldsymbol{k}\, \boldsymbol{\mapsto}\, 
\boldsymbol{L}_{\boldsymbol{\ell}}\boldsymbol{(}\boldsymbol{k}\boldsymbol{)}$}
\label{subsubse:L-ell} 

In Appendix \ref{appendix:a} we have derived  bounds
on the angular--momentum Green functions
$G_\ell(k;R,R')$ which imply global majorizations of the following
form for $k$ varying in $ \overline{\Pi}_\alpha$:
\begin{equation}
|G_\ell(k;R,R')|\leqslant h_M(\ell,|k|) \ e^{\alpha(R+R')}\  M(R)M(R'),
\label{f:G}
\end{equation}
with the following three specifications: \\[-5pt]

\begin{itemize}
\item[(i)] $M(R) = \sqrt{R}$, \quad $h_M(\ell,|k|)= \frac{1}{2}\sqrt{\frac{\pi}{2\ell+1}}$ 
\quad (implied by \eqref{a:5});
\item[(ii)] $M(R) = 1$, \quad $h_M(\ell,|k|)= \left(\frac{1+\ell\pi}{|k|}\right)$ \quad (implied by \eqref{a:12});
\item[(iii)] $M(R) =\sqrt{1+R}$,\quad 
$h_M(\ell,|k|)=\min\left(\frac{1+\ell\pi}{|k|}, \frac{1}{2}\sqrt{\frac{\pi}{2\ell+1}}\right)$ 
\quad (implied by \eqref{a:13}).
\end{itemize}

\vskip 0.3cm

\noindent 
We then introduce a condition of the following type on the weight--function $w$:
\begin{equation}
B(M,w) \doteq \left[\int_0^\infty  M^2(R) \frac{\rmd R}{w(R)}\right]^{1/2} < \infty,
\label{f:mu2}
\end{equation}
which allows one to prove

\skd

\begin{theorem}
\label{theorem:4} 
For every potential $V$ in a class
$\cN_{{w},\alpha}$ such that ${w}$ satisfies a condition of the
type \eqref{f:mu2}, the corresponding kernels $L_\ell(k;R,R')$
(formally defined in \eqref{f:27c}) are well--defined as compact
operators  $L_\ell(k)$ of Hilbert--Schmidt--type acting in the
Hilbert space $X_{{w},\alpha}$ for all non--negative integral
value of $\ell$ and for all $k\in\overline{\Pi}_\alpha \setminus \{0\}$. 
More precisely, $|L_\ell(k;R,R')|$ is bounded (for each $\ell$ and $k$)
by a kernel of rank one,
and the corresponding Hilbert--Schmidt norm
$\left\|L_\ell(k)\right\|_{\rm HS}$ of $L_\ell (k)$ 
in a Hilbert space called $\widehat{X}_{w,\alpha}$ admits the
following majorization for $k\in\overline{\Pi}_\alpha \setminus \{0\}$:
\begin{equation}
\left\|L_\ell(k)\right\|_{\rm HS} \leqslant B^2(M,{w})\
\frac{{\bf C}(V)}{\sqrt{2\ell +1}} \ h_M(\ell,|k|).
\label{f:th4}
\end{equation}
Moreover, for each $\ell$,
the {\rm HS}--operator--valued function $k \mapsto L_\ell(k)$, taking its values in
$\widehat{X}_{w,\alpha}$, is continuous in $\overline{\Pi}_\alpha \setminus \{0\}$ 
and holomorphic in ${\Pi}_\alpha \setminus \{0\}$. 
\end{theorem}

\begin{proof}
Assuming that $L_\ell(k)$ exists as an operator in
$X_{{w},\alpha}$, let $L^\dagger_\ell(k)$ be its adjoint, given
by the standard definition $\langle L^\dagger_\ell(k)x,y
\rangle_{w,\alpha}= \left\langle x, L_\ell(k)y
\right\rangle_{w,\alpha}$, ($x,y\in X_{{w},\alpha}$),
where $\langle\cdot,\cdot\rangle_{w,\alpha}$ denotes the
scalar product in $X_{{w},\alpha}$. $L^\dagger_\ell(k)$ is the
integral operator with kernel: 
\beq
L^\dagger_\ell(k;R,R')=\frac{{w}(R')}{{w}(R)} \, e^{2\alpha(R'-R)} \,
\overline{L_\ell(k;R',R)}. \nonumber 
\label{f:39} 
\eeq
Therefore, the Hilbert--Schmidt norm of $L_\ell(k)$
in $\widehat{X}_{w,\alpha}$, whose 
finiteness has to be proven, is given by the following
double--integral:
\beq
\left\|L_\ell(k)\right\|^2_{\rm HS}
\doteq \Tr[L^\dagger_\ell(k)L_\ell(k)] = 
\!\! \int_0^{+\infty} \!\! \frac{e^{-2\alpha R}}{{w}(R)}\,\rmd R
\int_0^{+\infty} \!\!\!\! {w}(R')e^{2\alpha R'} |L_\ell(k;R',R)|^2 \,\rmd R'.
\label{f:40}
\eeq
Let us first show that the integral on the r.h.s. of \eqref{f:27c}
is absolutely convergent, and therefore defines $L_\ell(k;R,R')$
for a.e. $R$. In view of \eqref{f:G}, this integral is bounded 
in modulus by
\beq  
h_M(\ell,|k|) \left[\int_0^{+\infty}|V_\ell(R,R'')|\ e^{\alpha R''} M(R'')\,\rmd R''\right]
e^{\alpha R'} M(R'),
\label{f:41}
\eeq
and thereby, in view of the Schwarz inequality and of \eqref{f:30'} and \eqref{f:mu2},
one obtains for a.e. $R$ (since $V_\ell^{(w)} \in X_{w,\alpha}$)
the following majorization by a kernel of rank one:
\beq 
|L_\ell(k;R,R')|\leqslant  
h_M(\ell,|k|) \ B(M,w) \, V_\ell^{(w)}(R)\ e^{\alpha R'} M(R').
\label{f:42}
\eeq
In view of \eqref{f:42}, we now obtain the 
following majorization for the expression \eqref{f:40} of 
$\|L_\ell(k)\|_{\rm HS}^2$:
\beq
\|L_\ell(k)\|_{\rm HS}^2 \leqslant h_M^2(\ell,|k|) \ B^2(M,{w})\
\!\!\!\int_0^{+\infty} \frac{e^{-2\alpha R}}{w(R)}
\left[e^{2\alpha R} M^2(R)\right]\,\rmd R
\!\!\int_0^{+\infty} \!\!\! w(R') \, e^{2\alpha R'}
\!\left|V_\ell^{(w)}(R')\right|^2\ \rmd R',
\label{f:HSmaj}
\eeq
which yields (in view of \eqref{f:mu2}):
\beq
\left\|L_\ell(k)\right\|_{\rm HS}^2 \leqslant
B^4(M,{w}) \ h_M^2(\ell,|k|)\ \|V_\ell^{(w)}\|_{w,\alpha}^2,
\label{f:th4sq}
\eeq
and therefore
(in view of \eqref{f:31}) the majorization \eqref{f:th4}.

\vskip 0.1cm

\noindent 
\emph{Proof of the last statement:} 
We note that the space $X_{{w},\alpha}$, defined in \eqref{f:28},
is a Hilbert space of the type $X_\mu$ introduced in Appendix
\ref{appendix:b} (see formula \eqref{b:5}), with 
$\mu(R) = {w}(R)\,e^{2\alpha R}$. 
For each $k\in \overline{\Pi}_\alpha \setminus \{0\}$,
$L_\ell(k)$ is an element of the
corresponding space  $\widehat{X}_\mu = \widehat{X}_{w,\alpha}$ 
introduced in \eqref{b:8}, with the coincidence of notations 
$\|L_\ell(k)\|_{\rm HS}=\|L_\ell(k)\|_{(\mu)}$. It can be
checked that for $k$ varying in any given domain
$\Pi_\alpha^{(K)} \doteq \{k\in \Pi_\alpha\,:\, |k| > K\}$,
and for each $\ell$, the function $(k;R,R') \mapsto L_\ell(k;R,R')$ 
satisfies all the assumptions of Lemma
\ref{lemma:B10} with (in view of \eqref{f:27c} and \eqref{f:G}): $\z=k$,
$F_1(k;R,R')=G_1(R,R')=|V_\ell(R,R')|$,
$F_2(k;R,R')=G_\ell(k;R,R')$, and 
$G_2(R,R')= h_M(\ell,K) e^{\alpha(R+R')} M(R) M(R')$.          
The fact that the majorizing kernel  
\beq
G(R,R')= h_M(\ell,K)
\left[\int_0^{+\infty}|V_\ell(R,R'')|\ e^{\alpha R''} M(R'')\,\rmd R''\right] \,
e^{\alpha R'} M(R') 
\label{e:11}
\eeq
coincides with the expression \eqref{f:41}, taken for $|k|=K$, implies that
$\|G\|_{\rm HS}=
\|G\|_{(\mu)}$ is finite in view of the previous 
HS--norm majorization that yielded the r.h.s. of \eqref{f:HSmaj}.
Since $F_1$ and $F_2$ are continuous (resp., holomorphic) with respect to
$k$ in $\overline{\Pi}_\alpha$ (resp., $\Pi_\alpha$), 
Lemma \ref{lemma:B10} implies that the HS--operator--valued function
$k \mapsto L_\ell(k)$ is continuous (resp., holomorphic) in 
$\overline{\Pi}_\alpha^{(K)}$ (resp., $\Pi_\alpha^{(K)}$), and therefore in 
$\overline{\Pi}_\alpha\setminus \{0\}$ (resp., $\Pi_\alpha\setminus \{0\}$), 
since the argument is valid for any $K>0$.
\end{proof}

\vskip 0.1cm

\noindent 
{\bf Complement to the choice of the weight--function $\boldsymbol{w}$.}
For the three given specifications of $M(R)$ in
majorization \eqref{f:G}, one can always obtain the
equality $B(M,{w})= A_{\varepsilon}$, with $A_{\varepsilon}$
given by Eq. \eqref{f:e0}, provided one chooses respectively
the weight--functions ${w}_1^{ (\varepsilon)}$, 
${w}_0^{(\varepsilon)}$ (see the complement to Lemma \ref{lemma:v}), 
and the weight--function 
\beq
{w}^{(\varepsilon)}(R) \doteq 
R^{1-\varepsilon} (1+R)^{1+2\varepsilon}, 
\label{wf:2} 
\eeq
which is such that
${w}^{(\varepsilon)} \geqslant \max ({w}_0^{(\varepsilon)},{w}_1^{(\varepsilon)})$. 
In view of this remark, the majorization \eqref{f:th4} can thus be specified as follows:

\vskip 0.2cm 

\noindent 
{\bf Complement to Theorem $\boldsymbol{\ref{theorem:4}}$.} 
\emph{The following bounds hold for each operator--valued function} $k \mapsto L_{\ell}(k)$:
\begin{itemize}
\item[(i)] {\em for} 
${w} = {w}_0^{(\varepsilon)}= R^{1-\varepsilon}\,(1+R)^{2\varepsilon}$,\quad
$\left\|L_\ell(k)\right\|_{\rm HS} \leqslant
\frac{1+\ell\pi}{|k|}\ \frac{{\bf C}(V)}{\sqrt{2\ell +1}} \, A_{\varepsilon}^2$\,; \\[-3pt]
\item[(ii)] {\em for} 
${w} = {w}_1^{(\varepsilon)}=R^{2-\varepsilon}\,(1+R)^{2\varepsilon}$,\quad
$\left\|L_\ell(k)\right\|_{\rm HS} \leqslant 
\frac{1}{2} \sqrt{\frac{\pi}{2\ell+1}} \, \frac{{\bf C}(V)}{\sqrt{2\ell +1}} \, A_{\varepsilon}^2$\,; \\[-3pt]
\item[(iii)] {\em for}
${w} = {w}^{(\varepsilon)}= R^{1-\varepsilon}\,(1+R)^{1+2\varepsilon}$,
\beq
\left\|L_\ell(k)\right\|_{\rm HS} \leqslant
\min\left(\frac{1+\ell\pi}{|k|}, \frac{1}{2}\sqrt{\frac{\pi}{2\ell+1}} \right) \, 
\frac{{\bf C}(V)}{\sqrt{2\ell +1}} \ A_{\varepsilon}^2.
\label{wf:1}
\eeq
\end{itemize}

\vskip 0.2cm

\noindent
It will appear in the following that the choice
$w={w}^{(\varepsilon)}$ (see Eq. \eqref{wf:2}),
for any positive $\varepsilon$, allows one to obtain the most
interesting results, in view of the following 
corollary of the ``Complements to Lemma \ref{lemma:v} and to Theorem \ref{theorem:4}".

\vskip 0.2cm 

\noindent
{\bf Corollary $\boldsymbol{\ref{lemma:v}}$--$\boldsymbol{\ref{theorem:4}}$.}
{\em
For any weight--function ${w}^{(\varepsilon)}(R)$, one has:

\vskip 0.2cm 

{\rm (a)} \, $\|v_{\ell,0}(k;\cdot)\|_{\wea} \leqslant
\min\left(1, \frac{\sqrt {|k|}}{2}\right)\, c_\ell \ C(V_\ell) \, A_{\varepsilon}$, \ 
for $k\in {\overline\Omega}_\alpha$;

\vskip 0.2cm

{\rm (b)} \, for each $\ell$, the {\rm HS}--operator--valued function $k \mapsto L_\ell(k)$,
taking its values in $\widehat{X}_\wea$,
is also holomorphic at the origin and therefore in the whole half--plane ${\Pi}_\alpha $.
Moreover, in view of \eqref{wf:1}, the function 
$\left\|\widehat L\right\|_{\rm HS}(r) \doteq \sup_{\max (|k|, \ell) \geqslant r} \left\|L_\ell(k)\right\|_{\rm HS}$
is uniformly bounded for all $r\geqslant 0$, and tends to zero when $r$ tends to $+\infty$.
}

\subsubsection{Smithies' formalism for the resolvents
$\boldsymbol{R}_{\boldsymbol{\ell}}\boldsymbol{(}\boldsymbol{k}\boldsymbol{;}\boldsymbol{g}\boldsymbol{)}$}
\label{subsubse:smithies}

Let us introduce the resolvent associated with equation \eqref{f:32}, i.e.,
\beq
R_\ell(k;g)=\left[\I-gL_\ell(k)\right]^{-1}.
\label{f:45}
\eeq
In this formalism, $g$ is treated as a general complex parameter, keeping in mind
that each ``physical" theory is obtained by fixing $g$ at a real value interpreted
as a coupling constant.
The fact that $L_\ell(k)$ is a Hilbert--Schmidt operator on the Hilbert space $X_{w,\alpha}$
allows us to use Smithies' formulae and bounds \cite{Smithies}, which all make sense
in terms of Hilbert--Schmidt kernels. According to Theorem 5.6 of Ref. \cite{Smithies}, 
we have (with the identification of notations $g\leftrightarrow\lambda$, $L_\ell\leftrightarrow K$,
$N_\ell \leftrightarrow H$, $\sigma_\ell\leftrightarrow\delta$,
$R_\ell\leftrightarrow\frac{\Delta}{\delta}$): 
\newcommand{\block}{
\framebox[1.1in][c]{
$
\begin{array}[b]{ccccc}
& & ~ & & \\
& & ~ & & \\
& & (Q_\ell)_n(k) & & \\
& & ~ & & \\
& & ~ & &
\end{array}
$
}
}
\newcommand{\roo}{
$
\begin{array}[b]{c}
L_\ell(k) \\
L_\ell^2(k) \\
\cdots \\
\cdots \\
L_\ell^n(k)
\end{array}
$
}
\begin{align}
& R_\ell(k;g)=\I +g\,\frac{N_\ell(k;g)}{\sigma_\ell(k;g)} \qquad (\ell=0,1,2,\ldots), \label{f:46} \\
\intertext{where:}
{\boldsymbol{\rm (i)}} \qquad &\sigma_\ell(k;g)=\sum_{n=0}^{+\infty}(\sigma_\ell)_n(k) \ g^n, \label{f:47} \\
& (\sigma_\ell)_0(k)=1, \qquad (\sigma_\ell)_n(k)=\frac{(-1)^n}{n!}(Q_\ell)_n(k)\qquad (n\geqslant 1), \label{f:48} \\[+8pt] 
(Q_\ell)_n(k) =& \left|
{\setlength\arraycolsep{4pt}
\begin{array}{ccccccc}
0 & n-1 & 0 & \cdots & 0 & 0 & 0 \\
(\rho_\ell)_2(k) & 0 & n-2 & \cdots & 0 & 0 & 0     \\
(\rho_\ell)_3(k) & (\rho_\ell)_2(k) & 0 & \cdots & 0 & 0 & 0 \\
\cdots & \cdots & \cdots & \cdots & \cdots & \cdots & \cdots \\
(\rho_\ell)_{n-1}(k) & (\rho_\ell)_{n-2}(k)
& (\rho_\ell)_{n-3}(k) & \cdots & (\rho_\ell)_2(k) & 0 & 1 \\
(\rho_\ell)_n(k) & (\rho_\ell)_{n-1}(k) & (\rho_\ell)_{n-2}(k) & \cdots & (\rho_\ell)_3(k) & (\rho_\ell)_2(k)&0
\end{array}
}
\right|, \label{f:49} \\
\intertext{with:}
& (\rho_\ell)_n(k)=\Tr\left[L_\ell^n(k)\right] \qquad (n\geqslant 2). \label{f:50}
\intertext{(Note that $(Q_{\ell})_1(k) = (\sigma_{\ell})_1(k) =0$).}
{\boldsymbol{\rm (ii)}} \qquad & N_\ell(k;g)=\sum_{n=0}^{+\infty} (N_\ell)_n(k) \ g^n, 
\label{f:51}
\intertext{with:}
& (N_\ell)_n(k)=L_\ell(k)\ (\Delta_\ell)_n(k)
=(\Delta_\ell)_n(k)\  L_\ell(k),
\label{f:52}
\end{align}
\beq
(\Delta_\ell)_0 = \I, \quad
(\Delta_\ell)_n(k) = \frac{(-1)^n}{n!} \left|
\begin{array}{cccccc}
\I & n & \cdots & 0 & \cdots & 0 \\
\roo & \multicolumn{5}{c}{\block}
\end{array}
\right| \quad (n\geqslant 1).
\label{f:53}
\eeq
Note that there holds for all $n\geqslant 1$ the
following recurrence relation 
between the bounded operators $(\Delta_\ell)_n(k)$ 
and $(\Delta_\ell)_{n-1}(k)$ (see formula (5.4.1) of \cite{Smithies}): 
\beq
(\Delta_\ell)_n(k)=  
(\sigma_\ell)_n(k)\,\I + 
L_\ell(k) (\Delta_\ell)_{n-1}(k) 
=(\sigma_\ell)_n(k)\,\I + 
(\Delta_\ell)_{n-1}(k)\, L_\ell(k). 
\label{f:53'}
\eeq
The latter directly yields the following identity
in the sense of power series
of $g$: 
\beq
\Delta_\ell(k;g)\doteq  
\sum_{n=0}^{+\infty} (\Delta_\ell)_n(k) \ g^n = 
\sigma_\ell(k;g) \, \I + g\, \Delta_\ell (k;g)\, L_\ell(k), 
\label{f:53"} 
\eeq
which then yields:
\beq
\Delta_\ell (k;g) [\I -g\, L_\ell(k)]= 
\sigma_\ell(k;g) \, \I, 
\quad \mathrm{i.e.,} \quad
\Delta_\ell(k;g)=\sigma_\ell(k;g) \, R_\ell(k;g),
\label{f:53a}
\eeq
and also, in view of \eqref{f:52}:
\beq
\Delta_\ell(k;g)  
=\sigma_\ell(k;g) \, \I + g N_\ell (k;g). 
\label{f:53b}
\eeq
By putting Eqs. \eqref{f:53a} and \eqref{f:53b} together, 
one concludes that Eq. \eqref{f:46} is then
satisfied in the sense of formal series
by the functionals $\sigma_\ell(k;g)$ and  
$N_\ell(k;g)$, defined as functionals of $L_\ell(k)$ by
Eqs. \eqref{f:47} to \eqref{f:53}.

\vskip 0.2cm

\noindent
Since for each $\ell$ the function $k \mapsto L_{\ell}(k)$ is
holomorphic in $\Pi_{\alpha}$ and continuous in $\overline{\Pi}_{\alpha}$ 
as a function taking its values in $\widehat{X}_{w,\alpha}$ 
(see Theorem \ref{theorem:4} and Corollary 
\ref{lemma:v}--\ref{theorem:4}), it follows 
from Lemma \ref{lemma:B6} that the same property holds for all 
the corresponding power 
functions $k \mapsto L_{\ell}^n(k)$.
Moreover one also has:

\vskip 0.2cm

\noindent
(i) \, In view of Lemma \ref{lemma:B3} (applied to $K[\z] = L_\ell(k)$ and 
$K'_t[\z] = L_\ell(k)$), the function $k \mapsto \Tr[L_{\ell}^2(k)]$,
and similarly all functions 
$k \mapsto (\rho_{\ell})_n(k)= \Tr[L_{\ell}^{n-1}(k)\  L_\ell(k)] $
(see Eq. \eqref{f:50}) are holomorphic in $\Pi_{\alpha}$
and continuous in $\overline{\Pi}_{\alpha}$. 
Since (in view of \eqref{f:48}, \eqref{f:49})
$(Q_{\ell})_n(k)$ and $(\sigma_{\ell})_n(k)$
are polynomials of the variables $(\rho_{\ell})_p(k)$
($p\leqslant n$), all these functions are also holomorphic in
$\Pi_{\alpha}$ and continuous in $\overline{\Pi}_{\alpha}$. 

\vskip 0.2cm

\noindent
(ii) \, In view of \eqref{f:52} and \eqref{f:53},
the kernels $(N_{\ell})_n(k)$
are polynomials of $L_{\ell}(k)$ of the form
$(N_{\ell})_n(k) = \sum_{j=1}^{n+1} a_j(k) L_{\ell}^j(k)$,
with complex coefficients $a_j(k)$, which are themselves polynomials of the
variables $(\rho_{\ell})_p(k)$.
Then, in view of Lemma \ref{lemma:B7}, each function 
$k \mapsto (N_{\ell})_n(k)$ is (like $k \mapsto L_{\ell}(k)$) holomorphic in
$\Pi_{\alpha}$ and continuous in $\overline{\Pi}_{\alpha}$
as a function taking its values in $\widehat{X}_{w,\alpha}$. 

\vskip 0.2cm

From Smithies' theory it follows that,
for any $k$ in $\overline{\Pi}_\alpha$, and any non--negative
integral value of $\ell$, $\sigma_\ell(k;g)$
is an entire function of $g$, and $g \mapsto N_\ell(k;g)$ is an
entire Hilbert--Schmidt operator--valued function
(the series \eqref{f:51} being convergent in the norm of $\widehat{X}_{w,\alpha}$).
More precisely, we can prove the following theorems.

\skd

\begin{theorem}
\label{theorem:5}
For every potential $V$ in a class 
$\cN_{\wea}$, the functions $\sigma_\ell(k;g)$ satisfy the 
following properties, for all integers $\ell\geqslant 0$:  
\begin{itemize}
\addtolength{\itemsep}{0.15cm}
\item[\rm (a)]  
$\sigma_\ell(k;g)$ is defined and uniformly bounded in modulus  by
a function $\Phi_{\ell} (|g|)$ in $\overline{\Pi}_{\alpha}\times\C$; 
it is continuous in $\overline{\Pi}_\alpha\times\C$
and holomorphic in ${\Pi}_\alpha\times\C$;
\item[\rm (b)] $\overline{\sigma_\ell(k;g)}=\sigma_\ell(-\overline{k};\overline{g})$;
\item[\rm (c)] for any fixed value of $g$, there holds:
$\sup_{\max (|k|, \ell) \geqslant r}|\sigma_\ell(k;g)-1|
\xrightarrow[r\to +\infty]{} 0$.
\end{itemize}
\end{theorem}

\begin{proof}
By combining the basic inequalities of Smithies' theory
(see \cite[Lemma 5.4]{Smithies}) with the uniform bounds \eqref{wf:1} 
on $\|L_{\ell}(k)\|_{\rm HS}$,
one obtains the following majorizations, valid for all integers $n\geqslant 1$,\ 
$\ell \geqslant 0$ and for all $k$ in $\overline{\Pi}_{\alpha}$:
\beq
|(\sigma_\ell)_n(k)| \leqslant
\left(\frac{e}{n}\right)^{{n/2}}\left\|L_\ell(k)\right\|_{\rm HS}^n
\leqslant
\left(\frac{e}{n}\right)^{{n/2}}
\left[A_{\varepsilon}^2\,{\bf C}(V)\,h(\ell,|k|)\right]^n (2\ell+1)^{-{n/2}},
\label{f:54}
\eeq
where
\beq
h(\ell,|k|) \doteq
\min\left(\frac{\ell\pi+1}{|k|},\,\frac{1}{2}\sqrt{\frac{\pi}{2\ell+1}}\,\right).
\label{e:12}
\eeq
In view of \eqref{f:54}, the series \eqref{f:47}
defining $\sigma_\ell(k;g)-1$ is dominated for all $k \in \overline{\Pi}_{\alpha}$
by a convergent
series with positive terms. 
It is convenient to associate with this series the entire function
\beq
\Phi(z)\doteq
\sum_{n=1}^{\infty} \left(\frac{e}{n}\right)^{n/2}\ z^n,
\label{f:Phi}
\eeq
which is a positive and increasing function 
of $z$ for $z>0$, such that $\Phi(0) =0$.
From \eqref{f:54}, one then concludes that
$\sigma_\ell(k;g)$ is for
all $k \in \overline{\Pi}_{\alpha}$ an entire function of $g$,
which satisfies the following uniform majorization:
\begin{equation}
|\sigma_\ell(k;g)-1|\ \leqslant \  \Phi\left(|g| \ \|L_{\ell}(k)\|_{\rm HS}\right)\ 
\leqslant \
\Phi\left(|g|\, A_{\varepsilon}^2\, \frac{{\bf C}(V)}{\sqrt{2\ell+1}}\, h(\ell,|k|)\right).
\label{f:59'}
\end{equation}
Moreover, in view of the holomorphy (resp., continuity) property of the functions
$(\sigma_\ell)_n(k)$ in $\Pi_{\alpha}$ (resp., $\overline{\Pi}_\alpha$),
Lemma \ref{lemma:B0} can be applied to the sequence of functions
$\{(k,g) \mapsto (\sigma_\ell)_n(k) \, g^n; n\in \N\}$;
it follows that the sum of the series \eqref{f:47}
defines $\sigma_\ell(k;g)$ as a holomorphic function of $(k,g)$ in
${\Pi}_{\alpha}\times \C $, which is moreover continuous in
$\overline{\Pi}_\alpha\times \C$. 

\vskip 0.2cm

The symmetry relation \eqref{f:symc} implies analogous relations for the
quantities $(\rho_\ell)_n(k)$, $(Q_\ell)_n(k)$, $(\sigma_\ell)_n(k)$,
and therefore property (b).

\vskip 0.2cm

In view of the behaviour of $h(\ell,|k|)$
(given by \eqref{e:12}), one checks that the quantity
$\widehat{h}(r,g)\doteq \sup_{\max (|k|, \ell) \geqslant r}
\left(|g|\,A_{\varepsilon}^2\,\frac{{\bf C}(V)}{\sqrt{2\ell+1}}\, h(\ell,|k|)\right)$ 
is finite and tends to zero for $r$ tending to infinity (for each fixed $g$). 
Correspondingly, one has:
\beq
\lim_{r\to +\infty}  \ \sup_{\max (|k|, \ell) \geqslant r}
\Phi\left(|g|\,A_{\varepsilon}^2\,\frac{{\bf C}(V)}{\sqrt{2\ell+1}}\, h(\ell,|k|)\right) 
=\lim_{r\to +\infty} \Phi (\widehat{h}(r,g))= \Phi(0)=0.
\label{f:59"}
\eeq
Property (c) is then readily implied by inequality \eqref{f:59'}.
\end{proof}

\begin{theorem}
\label{theorem:6}
For every potential $V$ in a class $\cN_{\wea}$,
the operators $N_\ell(k;g)$ exist as Hilbert--Schmidt operators
acting on ${X}_{\wea}$, 
for all integers $\ell\geqslant 0$ and for all $(k,g)$
in $\overline{\Pi}_{\alpha}\times\C$;
in this set, the function
$(k,g) \mapsto \|N_{\ell}(k;g)\|_{\rm HS}$  
is uniformly bounded in $k$ by a function $\Psi_{\ell} (|g|)$.
Moreover, the following properties hold:
\begin{itemize}
\addtolength{\itemsep}{0.15cm}
\item[\rm (a)] 
The {\rm HS}--operator--valued function $(k,g) \mapsto N_{\ell}(k;g)$,
taking its values in $\widehat{X}_{\wea}$, 
is continuous in $\overline{\Pi}_\alpha \times \C$
and holomorphic in ${\Pi}_\alpha \times \C$;
\item[\rm (b)] $\overline{N_\ell(k;g;R,R')}=N_\ell(-\overline{k};\overline{g};R,R')$;
\item[\rm (c)]
$\sup_{\max (|k|, \ell) \geqslant r} \left\|N_\ell(k;g)\right\|_{\rm HS}
\xrightarrow[r\to +\infty]{} 0$.
\end{itemize}
\end{theorem}

\begin{proof}
In view of Smithies' theory, there hold the following inequalities
(see  Lemmas 2.6 and 5.4 and the proof of Theorem 5.6 of 
\cite{Smithies}) for all integers $n \geqslant 1$:
\beq
\|(N_\ell)_n(k)\|_{\rm HS}\!\ 
\leqslant\!\ \|(\Delta_\ell)_n(k)\| \ \|L_\ell(k)\|_{\rm HS}\!
\ \leqslant\!\ \frac{e^{(n+1)/2}}{n^{n/2}} \ \|L_\ell(k)\|_{\rm HS}^{n+1}.
\label{f:60'}
\eeq
In view of the latter, the series \eqref{f:51} is
dominated term--by--term in the HS--norm by a convergent series; the sum 
of the latter is therefore well--defined as a HS--operator
$N_\ell(k;g)$ for all values of $(k,g)$ in
$\overline{\Pi}_{\alpha}\times\C$. The entire function
\beq 
\Psi(z) \doteq z\  [1+e^{1/2}\,\Phi(z)],
\label{f:Psi}
\eeq
(with $\Phi$ given by Eq. \eqref{f:Phi}), 
is like $\Phi$ an increasing function of $z$, for $z\geqslant 0$. 
It follows from \eqref{f:51}, \eqref{f:60'} and
from the bound \eqref{wf:1} on $\|L_{\ell}(k)\|_{\rm HS}$
(used as in the proof of Theorem \ref{theorem:5}) 
that the norm of $N_\ell(k;g)$ in
$\widehat{X}_{\wea}$ satisfies the bound:
\beq
\|N_{\ell}(k;g)\|_{\rm HS}\ \leqslant \
\frac{1}{|g|}\Psi(|g| \ \|L_{\ell}(k)\|_{\rm HS})
\leqslant
\frac{1}{|g|} \ \Psi\left(|g|\,A_{\varepsilon}^2\,\frac{{\bf C}(V)}{\sqrt{2\ell+1}}\, h(\ell,|k|)\right).
\label{f:Psi'}
\eeq
Moreover, since the functions
$(k,g) \mapsto (N_\ell)_n(k)$, taking their values in  
$\widehat{X}_{\wea}$ are
continuous in $\overline{\Pi}_{\alpha}\times \C$ and
holomorphic in $\Pi_{\alpha}\times \C$,
Lemma \ref{lemma:B0} can be applied
to the sequence of functions
$\{(k,g) \mapsto (N_\ell)_n(k)\ g^n;\, n\in\N\}$;
it follows that for each $g\in \C$ the sum of
the series \eqref{f:51} defines $(k,g)\mapsto N_\ell(k;g)$
as a function satisfying property (a).

Property (b) follows from the symmetry relation \eqref{f:symc}, through
all the analogous symmetry relations satisfied by the quantities
$(Q_\ell)_n(k)$, $(\Delta_\ell)_n(k)$, $(N_\ell)_n(k)$.

By using the fact that $\lim_{z\to 0}\, \Psi(z) \, =\, 0$,
and taking into account expression \eqref{e:12} of
$h(\ell,|k|)$, one obtains property (c) as a by--product of 
inequality \eqref{f:Psi'}.
\end{proof}

\subsection{Meromorphy properties of the resolvent
and their physical interpretation}
\label{subse:mero}

It is convenient to rewrite Eq. \eqref{f:45} in terms of the 
\emphsl{Fredholm resolvent kernel} or
\emphsl{truncated\footnote{We use here the same terminology as in 
relativistic quantum field theory, in which the truncated four--point function
plays the same role as the truncated resolvent in the present nonrelativistic framework.} 
resolvent}
$R_\ell^{(\mathrm{tr})}(k;g) \doteq \frac{1}{g}\,[R_\ell(k;g)-\I]$,
as the following \emphsl{Fredholm resolvent equation}:
\beq
R_{\ell}^{(\mathrm{tr})}(k;g) = \, L_\ell(k) + g L_\ell(k)\, R_\ell^{(\mathrm{tr})}(k;g),
\label{f:454}
\eeq
whose solution is given in view of \eqref{f:46} by 
\beq
R_\ell^{(\mathrm{tr})}(k;g)=\,
\frac{N_\ell(k;g)}{\sigma_\ell(k;g)}.
\label{f:46'}
\eeq
We now give an analysis of the meromorphy properties of
the operator--valued function $(k,g)\mapsto R_\ell^{(\mathrm{tr})}(k;g)$,
which follow from Theorems \ref{theorem:5} and \ref{theorem:6}
and formula \eqref{f:46'}.

\subsubsection{Meromorphy in $\boldsymbol{(k,g)}$ and meromorphy 
in $\boldsymbol{k}$ at each fixed $\boldsymbol{g}$} 
\label{subsubse:meromorphy}

A singularity (more precisely, a pole) of
the function $(k,g) \mapsto R_\ell^{(\mathrm{tr})}(k;g)$
is generated by a zero of the modified Fredholm
determinant $\sigma_\ell(k;g)$, namely a connected component in
$\Pi_\alpha\times\C$ of the complex manifold with
equation $\sigma_\ell(k;g)=0$.
An essential property of this manifold to be checked at first
is the fact that it cannot contain components of the form $g-g_0=0$.
In fact, this would imply $\sigma_\ell(k;g_0)=0$ for all $k$, 
which (for $|k| \to \infty$) would contradict
property (c) of Theorem \ref{theorem:5}.
So, for each fixed value of $g\in \C$, 
the corresponding restriction of the function
$\sigma_\ell(k;g)$ is a non--zero holomorphic function of $k$
in $\Pi_\alpha$. Then, in view of \eqref{f:46'} and of
Theorem \ref{theorem:6}, we can conclude that
$R_\ell^{(\mathrm{tr})}(k;g)$ is defined for each $\ell$ ($\ell=0,1,2,\ldots$)
and for each $g \in \C$ as a \emphsl{meromorphic} HS--operator--valued
function of $k$ in $\Pi_\alpha$, which takes its values in 
$\widehat{X}_{\wea}$.
At fixed $g$, all the possible poles of the function
$k\to R_\ell^{(tr)}(k;g)$ can thus be generically
defined as solutions $k=k^{(j)}(\ell,g)$ of the implicit
equation $\sigma_\ell(k;g)=0$ (at points where 
$\partial\sigma_\ell /\partial k \neq 0$,
considering the generic case). The complement of this discrete set,
namely the set of all points $k$ in $\Pi_\alpha$
(resp., ${\overline\Pi}_\alpha$)
such that $\sigma_\ell(k,g)\neq 0$ 
will be denoted by $\Pi_{\alpha,\ell}(g)$
(resp., ${\overline \Pi}_{\alpha,\ell}(g)$).
We therefore have (in view of Theorems \ref{theorem:5} and \ref{theorem:6}):

\skd

\begin{theorem}
\label{theorem:6'}
For every potential $V$ in a class $\cN_{\wea}$,
the function $(k,g) \mapsto R_\ell^{(\mathrm{tr})}(k;g)$ 
is meromorphic in $\Pi_\alpha \times \C$ as a 
{\rm HS}--operator--valued function. More precisely,
the operators $R_\ell^{(\mathrm{tr})}(k;g)$ exist
for all $(k,g)$ such that
$k\in \overline{\Pi}_{\alpha,\ell}(g),\ g\in\C$,
as Hilbert--Schmidt operators
acting on ${X}_{\wea}$,
and for any fixed $g\in \C$, the {\rm HS}--operator--valued function
$k \mapsto R_\ell^{(\mathrm{tr})}(k;g)$,
taking its values in $\widehat{X}_{\wea}$,
is holomorphic in $\Pi_{\alpha,\ell}(g)$.
\end{theorem}

\noindent
As a corollary, we also have:

\begin{theorem}
\label{theorem:6"}
The function $(k,g) \mapsto R_\ell(k;g)$ 
is meromorphic in $\Pi_\alpha \times \C$, 
and for any fixed $g$ in $\C$ 
the function $k \mapsto R_\ell(k;g)$ 
is holomorphic in $\Pi_{\alpha,\ell}(g)$, 
as operator--valued functions taking their values in the
space of {\rm bounded} operators in  
${X}_{\wea}$.
\end{theorem}

\noindent
In fact, the holomorphy properties of $R_\ell^{(\mathrm{tr})}(k;g)$ as 
a HS--operator--valued function imply the same holomorphy properties 
\emphsl{as a bounded operator--valued} function (since
$\|R_\ell^{(\mathrm{tr})}(k;g)\|\leqslant \|R_\ell^{(\mathrm{tr})}(k;g)\|_\mathrm{HS}$). 
By adding the constant operator $\I$ (holomorphic as a
bounded operator), one thus concludes that
$R_\ell(k;g)=\I + R_\ell^{(\mathrm{tr})}(k;g)$
has the same holomorphy (and meromorphy) properties in $k$ as
$R_\ell^{(\mathrm{tr})}(k;g)$, but in the sense of
a bounded \emphsl{(not Hilbert--Schmidt)}--operator--valued function.

\subsubsection{Poles of the resolvent and solutions of the Schr\"odinger--type equation}
\label{subsubse:poles-integer}

All the possible poles $k=k^{(j)}(\ell,g)$
correspond to situations in which there exists
a non--zero solution $x=x(R)$ of the homogeneous equation $g \, L_\ell(k) \, x=x$.
In fact, Eqs. \eqref{f:454} and \eqref{f:46'} imply the following identity between
HS--operator--valued functions, which is valid for all $(k,g)\in \Pi_{\alpha}\times \C$:
\beq
N_{\ell}(k;g) = \sigma_{\ell}(k;g) \ L_\ell(k) + g \, L_\ell(k)\ N_\ell(k;g).
\label{f:4546}
\eeq
A value $k=k^{(j)}(\ell,g)$ corresponds to a pole of the function
$k \to R_\ell^{(tr)}(k;g)$ iff 
the previous equation
reduces to the homogeneous equation
\beq
N_{\ell}(k;g) = g \, L_\ell(k)\, N_\ell(k;g), 
\qquad {\rm with} \qquad 
N_{\ell}(k^{(j)}(\ell,g);g) \neq 0.
\label{f:4647}
\eeq
Then it follows from Fredholm's theory that  
the latter kernel is of the form
$N_{\ell}(k^{(j)}(\ell,g);g)(R,R') = \sum_{i\in I}   x_i(R) y_i(R')$,
with $x_i\in X_{\wea}$ and $y_i\in X_{\wea}^*$, $I$ denoting a finite set.
The existence of a pole of the function $k\to R_\ell^{(tr)}(k;g)$ 
is therefore equivalent to the existence of 
at least one (non--zero) solution $x=x(R)$ in $X_{\wea}$
of the equation $g\,L_\ell(k)\,x=x$.

As shown below (see Lemma \ref{lemma:vii}), one can associate with any function
$x(R) $ in $X_{\wea}$, for every $\ell$, and for $k\in {\overline \Pi}_\alpha$, 
the function
\beq
\psi(R)= g \int_0^{+\infty}G_\ell(k;R,R')\,x(R')\,\rmd R',
\label{xtopsi}
\eeq
which satisfies the equation
$D_{\ell,k}\,  \psi = g \, x$, since 
$G_\ell(k;R,R')$ is a Green function of the differential operator
$D_{\ell,k}$. Now, in view of Eq. \eqref{xtopsi}, the definition
\eqref{f:27c} of $L_\ell(k)$ implies the following equality:
\beq
g \, [L_\ell(k) x](R) = \int_0^{+\infty} V_\ell(R,R') \, \psi(R')\, \rmd R'.
\label{psitoLx}
\eeq
So, if $x$ is a non--zero solution of the homogeneous Fredholm equation 
$gL_\ell(k)x=x$ (associated with a value of $k$ which is a  pole
of the function $k \mapsto R_\ell(k;g)$), then the function $\psi$
defined by Eq. \eqref{xtopsi} satisfies the following relations:
\beq
D_{\ell,k} \, \psi (R) =\, g x= \,
g \int_0^{+\infty} V_\ell(R,R') \psi(R')\, \rmd R',
\label{FrSch}
\eeq
and therefore $\psi$ is a non--zero solution of the
Schr\"odinger--type equation \eqref{f:24}.

\begin{figure}[t]
\begin{center}
\leavevmode
\psfig{file=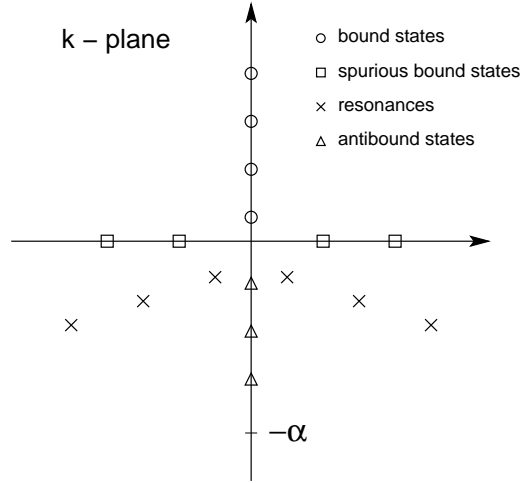,width=7cm}
\caption{\label{figure:1} Representation of bound states and resonances in the complex $k$--plane.}
\end{center}
\end{figure}

\subsubsection{Some results on the location of the poles 
and their physical interpretation (see Fig. $\boldsymbol{\ref{figure:1}}$)} 
\label{subsubse:someresults}

For $g$ \emphsl{real} (interpreted physically as a coupling constant)
and for $k$ such that $\Imag k\geqslant 0$,
Eq. \eqref{xtopsi} defines a one--to--one correspondence between 
the solutions of the homogeneous equation  $g\,L_\ell(k)\,x=x$ 
and a class of square--integrable solutions of the
Schr\"odinger--type equation \eqref{f:24}, which will be fully described in
Theorem \ref{theorem:7} below.
According to the latter, all the possible zeros $k=k^{(j)}(\ell,g)$
of $\sigma_\ell(k;g)$ which lie in $\Imag k\geqslant 0$
correspond to ``bound state solutions" of \eqref{f:24}: 
these solutions are the contributions
\emphsl{with a given angular momentum} $\ell$
to the set of solutions of Problem 1,
whose general properties have been listed in
Theorem \ref{theorem:1} (see Section 2). In particular, 
as a general by--product of Theorem \ref{theorem:1}, it follows that
for each real value of $g$, all the possible zeros $k=k^{(j)}(\ell,g)$
in the closed half--plane $\Imag k\geqslant 0$ can be located only
either on the imaginary axis or on the real axis. These two situations,
which will be analyzed in detail below (see Theorem \ref{theorem:7}),
correspond respectively to \emphsl{bound states} and to \emphsl{spurious bound states}
(i.e., ``bound states embedded in the continuum"). Furthermore,
the zeros on the real axis are distributed in pairs
symmetric with respect to $k=0$ (in view of statement (b)
of Theorem \ref{theorem:5}).
Concerning the possible poles of $R_\ell^{(\mathrm{tr})}(k;g)$ 
in the strip $-\alpha<\Imag k<0$,
we can also say that (for the same reason)
they occur either on the imaginary axis (\emphsl{anti--bound states})
or in pairs symmetric with respect to the imaginary axis
(\emphsl{resonances}).

Since (at fixed $g$) $\sigma_\ell(k;g)$ is holomorphic in
$\Pi_\alpha$ ($\ell=0,1,2,\ldots$), and since the function 
$(k,\ell) \mapsto \sigma_\ell(k;g)$
tends uniformly to 1 for $\max\,(|k|,\ell) \to +\infty$
(see statement (c) of Theorem \ref{theorem:5}), there holds 
a finiteness property of the set of zeros of all the functions
$\sigma_\ell$ in the domain $\Pi_\alpha $ of the $k$--plane,
which can be stated as follows in terms of the corresponding
physical interpretation:

\skd

\begin{proposition}
\label{proposition:2}
For any nonlocal potential $V$ in a class $\cN_{\wea}$,
and for each fixed real value of the coupling constant $g$:
\begin{itemize}
\item[\rm (a)] there exists an integer $L =L(V,g)$ such that, for every $\ell \geqslant L$,
there occur no bound states, anti--bound states and resonances of 
angular momentum $\ell$;
\item[\rm (b)] for each integer $\ell$ such that $0\leqslant \ell <L$, there occur at most 
a finite number $n_\ell$ of (normal or spurious) bound states
and a finite number of resonances and anti--bound states of angular momentum $\ell$
in any strip
$\Omega_{\alpha -\varepsilon}^- \doteq \{k\in \C\,:\, -(\alpha-\varepsilon)\leqslant\Imag k <0\}$, 
all of them being localized in a finite disk of the form $|k|< k_{V,g}$.
\end{itemize}
\end{proposition}

\subsection{Correspondence between
the solutions of the Fredholm equation with kernel $\boldsymbol{L_\ell(k)}$ and  
those of the nonlocal Schr\"odinger--type equation}
\label{subse:corresp}

\subsubsection{Preliminary properties} 
\label{subsubse:preliminary}

{\rm We need to state four lemmas.} 

\vskip 0.2cm

\begin{lemma}
\label{lemma:vi} 
For every function $x$ in ${X}_{\wea}$, there hold the following properties: 
\begin{itemize}
\item[\rm (a)] 
the function $R \mapsto kR\,j_\ell(kR)\,x(R)$ is integrable on $[0,+\infty)$
for all $k$ in $\overline{\Omega}_\alpha$, and satisfies the following majorization:
\beq
\left|\int_R^{+\infty} \!\!\! kR'\,j_\ell(kR')\, x(R')\, \rmd R'\right| \leqslant 
\frac{c_\ell}{[2(\alpha-|\Imag k|)]^{1/2}}\, e^{-(\alpha-|\Imag k|)\,R}
\|x\|_{\wea},
\label{f:60a}
\eeq
which is valid for all $k$ in ${\Omega}_\alpha$ and $R\geqslant 1$;
\item[\rm (b)] 
the function 
$R \mapsto kR\,h_\ell^{(1)}(kR)\,x(R)$ is integrable on any interval $[R,+\infty)$ ($R>0$)
for all $k$ in $\overline{\Pi}_\alpha$, and satisfies the following majorization:
\beq
\left|\int_R^{+\infty} \!\! kR'\,h_\ell^{(1)}(kR')\, x(R')\,\rmd R'\right| \leqslant 
\frac{c_\ell^{(1)} (1+|k|^{-1})^\ell}{[2(\alpha+ \Imag k)]^{1/2}}
\, e^{-(\alpha +\Imag k)\, R} \|x\|_{\wea},
\label{f:60a'}
\eeq
which is valid for all $k$ in ${\Pi}_\alpha$ and $R\geqslant 1$.
\end{itemize}
\end{lemma}

\begin{proof}
We make use of the Schwarz inequality in the space $X_{\wea}$: 

(a)\, By using bound \eqref{f:j1}, 
with $w=\we >\we_0 $ (see \eqref{wf:2}), 
which allows us to take  
$A_{\we}(k)=A_{\we_0}= A_\varepsilon$ 
(given by \eqref{f:e0}), 
we obtain for all $k\in \overline{\Omega}_\alpha$:
\beq
\int_0^{+\infty} |kR\,j_\ell(kR)|\,|x(R)|\, \rmd R \, \leqslant \,
\left\|{k\cdot} \ j_\ell(k\cdot)\right\|^*_{\wea} 
\ \|x\|_{\wea}
\,\leqslant\, c_\ell \, A_\varepsilon \ \|x\|_{\wea}.
\label{f:vi1}
\eeq
The majorization \eqref{f:60a} of the remainder of this integral from $R$ to $+\infty$
is obtained similarly (for $k\in {\Omega}_\alpha$ and $R\geqslant 1$)
by using the $(e^{|\Imag k| R})$--dependence 
of the bound \eqref{a:0} on $[kR j_\ell(kR)]$, 
and the inequality $[\we(R)]^{-1} < 1$. 

(b) By using the 
bound \eqref{a:0'} on $[kR h_\ell^{(1)}(kR)]$, 
one obtains a Schwarz inequality similar to \eqref{f:vi1}, except
for the replacement of the integration interval $[0,+\infty)$
by $[R,+\infty)$ (for all $R>0$), which proves 
the corresponding integrability property for all $k$ in 
$\overline{\Pi}_\alpha$. The bound \eqref{f:60a'} is obtained by using 
again the inequality $[\we(R)]^{-1} < 1$
together with a majorant of bound \eqref{a:0'} for $R\geqslant 1$.
\end{proof}

\begin{lemma}
\label{lemma:vii} 
For every function $x$ in ${X}_{\wea}$, the corresponding function 
\beq
\psi_{x;\ell,k}(R)= g \int_0^{+\infty} \!\! G_\ell(k;R,R') \, x(R')\,\rmd R'
\label{f:62"}
\eeq
is well--defined for all $k$ with $\Imag k \geqslant -\alpha$,
and enjoys the following properties:
\begin{itemize}
\item[\rm (a)] there hold majorizations of the following form:\\[-6pt]
\begin{itemize}
\item[\rm (i)] for $\Imag k\geqslant 0$,
\beq
\left|\psi_{x;\ell,k}(R)\right| \leqslant |g|\ {\widehat c_{\ell,\varepsilon}} \ 
\|x\|_{\wea} \ R\, e^{-\min(\alpha,\Imag k)\,R};
\label{f:64"}
\eeq
\item[\rm (ii)] for $-\alpha\leqslant \Imag k< 0$,
\beq
\left|\psi_{x;\ell,k}(R)\right|\leqslant |g|\ {\widehat c_{\ell,\varepsilon}} \ 
\|x\|_{\wea} \ R\,  e^{|\Imag k|\,R};
\label{f:65"}
\eeq
\end{itemize}
\item[\rm (b)] for every $k$ such that $-\frac{\alpha}{2}< \Imag k <\alpha,\ k\neq 0$, there
holds the following limit:
\begin{align}
\left|\psi_{x;\ell,k}(R)+\rmi g\, b_{x;\ell,k}\, R h_\ell^{(1)}(kR)\right|
\xrightarrow[R\to +\infty]{} 0, \label{f:vii'} \\
\intertext{where:}
b_{x;\ell,k} \doteq \int_0^{+\infty}kR\,j_\ell(kR)\,x(R)\,\rmd R,
\label{f:vii}
\end{align}
with the specifications listed below. There exists a function ${\rm c}_\ell(k)$
such that, for $R\geqslant 1$, there hold the following inequalities:\\[-6pt]
\begin{itemize}
\item[\rm (i)] if \ $0\leqslant \Imag k <\alpha$, 
\beq
\null\qquad\left|\psi_{x;\ell,k}(R)+\rmi g\,  b_{x;\ell,k}\ R h_\ell^{(1)}(kR)\right|
\leqslant {\rm c}_\ell(k)\,  |g|\ \|x\|_{\wea} \ e^{-\alpha R}; 
\label{f:viia}
\eeq
\item[\rm (ii)] if \ $-\frac{\alpha}{2}\leqslant \Imag k <0$,
\beq
\left|\psi_{x;\ell,k}(R)+\rmi g\,  b_{x;\ell,k}\ R h_\ell^{(1)}(kR)\right| 
\leqslant {\rm c}_\ell(k)\, |g|\ \|x\|_{\wea} \ e^{-(\alpha-2|\Imag k|) R};
\label{f:viib}
\eeq
\end{itemize}
\item[\rm (c)] for all $k$ with $\Imag k \geqslant -\alpha$, $k\neq 0$,
the derivative $\psi'_{x;\ell,k}(R)$ of $\psi_{x;\ell,k}(R)$ is 
well--defined, absolutely continuous and bounded for $R$ tending to zero. \\
For $\Imag k \geqslant 0$, there holds a majorization of the following form:
\beq
|\psi'_{x;\ell,k}(R)| 
\leqslant |g|\ {\widehat c\,'_{\ell,\varepsilon}} \ 
\|x\|_{\wea}\ e^{-\min(\alpha,\Imag k)R},
\label{f:64''}
\eeq
and the following limit is valid for all $k$ such that $0\leqslant \Imag k <\alpha$, $k\neq 0$:
\beq
e^{R\Imag k} \
\left|\psi'_{x;\ell,k}(R) -\rmi k\psi_{x;\ell,k}(R) \right|
\xrightarrow[R\to +\infty]{} 0.
\label{f:vii"}
\eeq
\end{itemize}
\end{lemma}

\begin{proof}
In view of expression \eqref{f:27d} of the Green function $G_\ell$,
we can rewrite Eq. \eqref{f:62"} under the following form, whose validity 
is established below:
\beq
\psi_{x;\ell,k}(R) = -\rmi g \left[Rh_\ell^{(1)}(kR)
\int_0^{R} kR' j_\ell(kR')x(R')\,\rmd R'
+R j_\ell(kR)\int_R^{+\infty} kR' h_\ell^{(1)}(kR')x(R')\,\rmd R'\right].
\label{e:15}
\eeq
In fact, for every $k$ such that $\Imag k \geqslant -\alpha$,
the convergence of the integrals in \eqref{e:15} 
is obtained, together with appropriate majorizations on the latter
by using the bounds \eqref{a:0} and \eqref{a:0'} on the functions
$j_\ell$ and $h_\ell^{(1)}$, in the following way:
\beq
\begin{split}
& |\psi_{x;\ell,k}(R)| \\
&\quad\leqslant |g|\frac{c_\ell c_\ell^{(1)}}{|k|}
\left[\left(\frac{1+|k|R}{|k|R}\right)^\ell\,e^{-R\Imag k}
\int_0^R\left(\frac{|k|R'}{1+|k|R'}\right)^{\ell+1}e^{R'|\Imag k|}\,|x(R')|\,\rmd R'\right. \\
&\qquad\left. +\left(\frac{|k|R}{1+|k|R}\right)^{\ell+1} e^{R|\Imag k|}
\int_R^{+\infty}\left(\frac{1+|k|R'}{|k|R'}\right)^\ell e^{-R'\Imag k} \, |x(R')|\,\rmd R'\right],
\end{split}
\label{e:16}
\eeq
which yields (by using the increase property of the function $\frac{|k|R}{1+|k|R}$):
\beq
\frac{|\psi_{x;\ell,k}(R)|}{|g|{c_\ell}c_\ell^{(1)} \, R} 
\leqslant
e^{-R\Imag k} \int_0^R e^{R'|\Imag k|}\,|x(R')|\,\rmd R' 
+e^{R|\Imag k|}\int_R^{+\infty}\!\!\! e^{-R'\Imag k} \, |x(R')|\,\rmd R',
\label{e:17}
\eeq
By using the assumption that $x$ belongs to 
${X}_{\wea}$, the two integrals on the r.h.s. of the
latter can be seen to be convergent for all $k$ in $\overline{\Pi}_\alpha$
in view of the Schwarz inequality, which therefore
implies the existence (and analyticity in $k$ for every $R\geqslant 0$)
of $\psi_{x;\ell,k}(R)$ in this domain of the $k$--plane. 
We now exhibit these convergence properties together with
the majorizations listed under (a).

\vskip 0.3cm

For $k$ in the half--plane $\Imag k \geqslant 0$, majorization \eqref{e:17}
also implies the following one: 
\beq
\frac{|\psi_{x;\ell,k}(R)|}{|g|{c_\ell}c_\ell^{(1)}R} 
\leqslant 
e^{-R\min (\alpha, \Imag k)}\int_0^R e^{\alpha R'}\,|x(R')|\,\rmd R' 
+e^{-\alpha R}\int_R^{+\infty}e^{\alpha R'}\, |x(R')|\,\rmd R',
\label{e:18}
\eeq
which then yields \eqref{f:64"} 
(with $\widehat c_{\ell,\varepsilon} = c_\ell c_\ell^{(1)} A_{\varepsilon}$)
by directly using the Schwarz inequality.

\vskip 0.2cm

For $k$ in the strip $-\alpha\leqslant \Imag k< 0$,
the majorization \eqref{e:17} readily implies that
$|\psi_{x;\ell,k}(R)| \leqslant |g|{c_\ell}c_\ell^{(1)} R \, e^{|\Imag k| R}
\int_0^{+\infty}e^{\alpha R'}\,|x(R')|\,\rmd R'$, which then yields \eqref{f:65"}.

\vskip 0.5cm

(b) Let $b_{x;\ell,k}$ be given by the integral in \eqref{f:vii},
whose convergence has been established in 
Lemma \ref{lemma:vi} (a), provided $|\Imag k| \leqslant \alpha$.
Equation \eqref{e:15} can then be rewritten as follows:
\beq
\begin{split}
& \psi_{x;\ell,k}(R)+\rmi g\,  b_{x;\ell,k}\, R h_\ell^{(1)}(kR) \\
& \quad=\rmi g\! \left[
Rh_\ell^{(1)}(kR) \!\! \int_R^{+\infty} \!\!\!\! kR' j_\ell(kR')x(R')\,\rmd R'
-R j_\ell(kR) \!\! \int_R^{+\infty} \!\!\!\! kR'h_\ell^{(1)}(kR')x(R')\,\rmd R'\right].
\end{split}
\label{e:15'}
\eeq
Let us show that for $-\frac{\alpha}{2}<\Imag k<\alpha$ (and $k\neq 0$),
each term in the bracket on the r.h.s. of the latter
tends to zero in the limit $R\to +\infty$, with the specifications
(i), (ii) listed under (b).

\vskip 0.2cm

In view of Lemma \ref{lemma:vi} (a) and of bound \eqref{a:0'}, the first term 
on the r.h.s. of \eqref{e:15'} can be majorized by 
$|g|\mathrm{c}_{\ell,1}(k)\, e^{-R\Imag k}\, e^{-(\alpha -|\Imag k|) R}\,\|x\|_{\wea}$, 
where we have put:
$\mathrm{c}_{\ell,1}(k)\doteq c_\ell c_\ell^{(1)}(1+|k|^{-1})^\ell \ [2(\alpha- |\Imag k|)]^{-{1/2}}$.
This bound correspond to the following two regimes:
\begin{itemize}
\item[(i)] \ $|g|\,\mathrm{c}_{\ell,1}(k)\ e^{-\alpha R } \ \|x\|_{\wea}$,
\ if $0\leqslant \Imag k<\alpha$, \\[-5pt]
\item[(ii)] \ $|g|\,\mathrm{c}_{\ell,1}(k) \ e^{-(\alpha-2|\Imag k|) R} \ \|x\|_{\wea}$,
\ if $-\frac{\alpha}{2}<\Imag k<0$.
\end{itemize}

\vskip 0.2cm

Similarly, in view of Lemma \ref{lemma:vi} (b) and of bound \eqref{a:0},
the second term on the r.h.s. of \eqref{e:15'} can be majorized by 
$|g|\mathrm{c}_{\ell,2}(k)\, e^{|\Imag k| R }\, e^{-(\alpha +\Imag k) R}\,\|x\|_{\wea}$, 
where we have put: 
$\mathrm{c}_{\ell,2}(k)\doteq c_\ell c_\ell^{(1)}(1+|k|^{-1})^\ell \ [2(\alpha+\Imag k)]^{-{1/2}}$.
Here again, this bound corresponds to the previous two regimes (i) and (ii), except for the substitution 
$\mathrm{c}_{\ell,1}(k) \to \mathrm{c}_{\ell,2}(k)$.
The inequalities \eqref{f:viia} and \eqref{f:viib} are therefore established 
with $\mathrm{c}_{\ell}(k)\doteq\mathrm{c}_{\ell,1}(k)+\mathrm{c}_{\ell,2}(k)$.

\vskip 0.3cm

(c) We now consider the derivative $\psi'_{x;\ell,k}$ of $\psi_{x;\ell,k}$,
which can be obtained for every positive $R$
by a direct computation from the r.h.s. of \eqref{e:15} 
(since for all $k\in\C$, $k\neq 0$, and $R>0$,
$j_\ell(kR)$ and $h_\ell^{(1)}(kR)$ define analytic functions of $R$ on $\R^+$).
This yields:
\beq
\begin{split}
& \psi'_{x;\ell,k}(R) =
-\rmi gk \frac{\rmd}{\rmd R}\left[Rh_\ell^{(1)}(kR)\right]
\int_0^R R'j_\ell(kR')\, x(R')\,\rmd R' \\
&\quad -\rmi gk \frac{\rmd}{\rmd R}\left[Rj_\ell(kR)\right]
\int_R^{+\infty}R'h_\ell^{(1)}(kR')\, x(R')\,\rmd R'.
\end{split}
\label{e:15"}
\eeq
The fact that $\psi'_{x;\ell,k}$ is absolutely continuous  
is then an immediate consequence of \eqref{e:15"}, in view of 
the convergence of the integral factors established above for 
$\Imag k \geqslant -\alpha$. The fact that 
$\psi'_{x;\ell,k}(R)$ remains a bounded and absolutely continuous function
\emphsl{in the limit} $R\to 0$ is also implied by Eq. \eqref{e:15"}
by taking into account the fact that the holomorphic functions $j_\ell$  
and $h_\ell^{(1)}$ respectively admit a 
zero of order $\ell$ and a pole of order $\ell +1$ at the origin. 

\vskip 0.2cm

The majorization \eqref{f:64''} is then deduced from
\eqref{e:15"} in the same way as \eqref{f:64"} is deduced 
from \eqref{e:15}; this is because,
in view of bounds \eqref{a:3b} and \eqref{a:6b}
(respectively similar to \eqref{a:0} and \eqref{a:0'}),
a majorization of the form \eqref{e:17} is equally valid for 
$\psi'_{x;\ell,k}$.

\vskip 0.2cm

Let us finally establish limit \eqref{f:vii"}. We note that 
the second term on the r.h.s. of \eqref{e:15"} tends to zero 
as a constant times $e^{-\alpha R}$, like the second term on 
the r.h.s. of \eqref{e:15} divided by $R$ 
(in view of the majorization of the latter used in 
\eqref{e:18}). If we now form the expression
$e^{R\Imag k}\times\left|\psi'_{x;\ell,k}(R) -\rmi k \psi_{x;\ell,k}(R)\right|$
and consider it as given by the corresponding difference of the r.h.s. of 
Eqs. \eqref{e:15"} and \eqref{e:15}, we conclude that its limit for $R\to +\infty$
is the same as the limit of the expression 
\beq 
e^{R\Imag k}\times 
(-\rmi gk) \left\{\frac{\rmd}{\rmd R}\left[Rh_\ell^{(1)}(kR)\right]
-\rmi kRh_\ell^{(1)}(kR)\right\}\int_0^R R'j_\ell(kR')\, x(R')\, \rmd R'.
\label{e:16'} 
\eeq
Now, in view of \eqref{a:8'}, the latter 
can be bounded by a quantity of the following form: 
$O(\frac{1}{R^2})\times\int_0^R e^{R'|\Imag k|}\,|x(R')|\,\rmd R' 
\leqslant O(\frac{1}{R^2})\times A_\varepsilon\,\|x\|_{\wea}$,
which proves the limit \eqref{f:vii"} for all $k$ such that 
$0\leqslant\Imag k< \alpha$, $k\neq 0$.
\end{proof}

\begin{lemma}
\label{lemma:vi'} 
For every function $\psi$ on $[0,+\infty[$ such that 
$|\psi(R)| \leqslant c(\psi)\, R$, or more generally for every
$\psi$ in the dual space ${X}^*_{\wea}$ 
of ${X}_{\wea}$, the following properties hold:
\begin{itemize}
\item[\rm (a)]\, the function $R \mapsto \psi(R)\,v_{\ell,0}(k;R)$, where 
$v_{\ell,0}(k;R)$ is the function defined by Eq. \eqref{f:27b} and Lemma \ref{lemma:v},
is integrable on $[0,+\infty)$ for all $k$ in $\overline{\Omega}_\alpha$;
\item[\rm (b)] \, the function $x_\psi(R)\doteq\int_0^{+\infty}V_\ell(R,R')\,\psi(R')\,\rmd R'$
is well--defined as an element of ${X}_{\wea}$, and one has
(for all $k$ in $\overline{\Omega}_\alpha$):
\beq
\int_0^{+\infty}kR\,j_\ell(kR)\,x_\psi(R)\,\rmd R
=\int_0^{+\infty}\psi(R)\,v_{\ell,0}(k;R)\,\rmd R;
\label{f:60}
\eeq
\item[\rm (c)] \, the following double integral is well--defined:
\beq
\int_0^{+\infty} \!\! \rmd R\int_0^{+\infty} \!\! \rmd R'\, 
\overline{\psi} (R)\,V_\ell(R,R')\,\psi (R') \, <\,+\infty. 
\label{f:69'}
\eeq
\end{itemize}
\end{lemma}

\begin{proof}
(a) The assumption $|\psi(R)| \leqslant c(\psi) \, R$ implies:
\beq
\int_0^{+\infty} |\psi(R)|\,|v_{\ell,0}(k;R)|\, \rmd R
\leqslant \frac{c(\psi)}{(2\alpha)^{1/2}}\ \|v_{\ell,0}(k;\cdot)\|_{\wea}.
\label{f:vi2}
\eeq
(b) Similarly, formula \eqref{f:30'} implies:
\beq
|x_\psi(R)|\leqslant \int_0^{+\infty} |V_\ell(R,R')|\ |\psi(R')|\, \rmd R'
\leqslant \frac{c(\psi)}{(2\alpha)^{1/2}} \ V_\ell^{(\we)}(R),
\label{f:vi3}
\eeq
and therefore, in view of \eqref{f:31}:
\beq
\|x_\psi\|_{\wea} \leqslant
\frac{c(\psi)}{(2\alpha)^{1/2}}\ C(V_\ell).
\label{f:vi4}
\eeq
Then, in view of \eqref{f:27b} and of the symmetry relation 
$V(R,R')=V(R',R)$, one readily obtains equality  
\eqref{f:60}, together with the following inequality (in view of 
\eqref{f:vi1} and \eqref{f:vi4}):
\beq
\left|\int_0^{+\infty}kR\,j_\ell(kR)\,x_\psi(R)\,\rmd R \right|\leqslant 
\frac{c(\psi)}{(2\alpha)^{1/2}}\ c_\ell \, C(V_\ell)\ A_{\varepsilon}. 
\label{f:vi5}
\eeq
(c) Finally, the double integral in \eqref{f:69'} can be rewritten and majorized
as follows:
\beq
\left|\int_0^\infty\overline{\psi}(R)\, x_\psi(R)\, \rmd R\right|
\leqslant \frac{c(\psi)}{(2\alpha)^{1/2}}\ \|x_\psi\|_{\wea}  
\leqslant \frac{c^2(\psi)}{2\alpha} \, C(V_\ell).
\label{f:vi6} 
\eeq
Note that inequalities similar to \eqref{f:vi2}--\eqref{f:vi6} are obtained 
by using the more general assumption
$\psi \in {X}^*_{\wea}$ and replacing  
$\frac{c(\psi)}{(2\alpha)^{1/2}}$ by $\left\|\psi\right\|_{\wea}^*$.
\end{proof}

\begin{lemma}[Wronskian Lemma]
\label{lemma:vii'}
Let $\psi(R)$ be a function on $\R^+$ which enjoys the following conditions:
\begin{itemize}
\item[\rm (i)] its derivative is absolutely continuous and bounded for $R$ tending to zero; \\[-6pt]
\item[\rm (ii)] it satisfies a bound of the form $|\psi(R)| \leqslant c\,R$; \\[-6pt]
\item[\rm (iii)] for given values of $\ell$ and $k\in \Pi_\alpha$, 
it is a solution of the integro--differential equation:
$D_{\ell,k} \psi(R) =  g \int_0^{+\infty} V_\ell(R,R') \psi(R')\, \rmd R'$, where $g$ is real.
\end{itemize}
Then one has:
\beq
\lim_{R\to +\infty}\left[\overline{\psi} (R)\,\psi'(R)-\overline{\psi'}(R)\,\psi(R)
-({\overline k}^{\,2}-k^2)\int_0^R {\overline \psi}(R')\, \psi(R')\, \rmd R'\right] = 0.
\label{f:wr}
\eeq
\end{lemma}

\begin{proof}
Equation \eqref{f:24} directly implies the following equalities for $g$ real:
\beq
\begin{split}
& \overline{\psi}(R)\,\psi''(R)-\overline{\psi''}(R)\,\psi(R)+(k^2-{\overline k}^{\,2}){\overline \psi}(R)\,\psi(R)
=\overline{\psi}(R)\,[D_{\ell,k}\psi](R)-[\overline{D_{\ell,k}\psi}](R)\,\psi(R) \\
&\quad =g \int_0^{+\infty} V_\ell(R,R') \left[\overline{\psi}(R)\,\psi(R')-\overline{\psi}(R')\,\psi(R)\right]\, \rmd R'.
\end{split}
\label{f:wr1}
\eeq
We note that, in view of Lemma \ref{lemma:vi'} (b) and (c) (and by taking condition (ii) into account),
the r.h.s. of the latter is well--defined as an integrable function $I(R)$ on $[0,+\infty)$,
and that, in view of the symmetry condition on the potential $V_\ell(R,R')$,
one then has $\lim_{\widehat{R} \to +\infty} \int_0^{\widehat{R}} I(R)\,\rmd R =0$.
Therefore, by integrating Eq. \eqref{f:wr1} side by side over $R$ between $0$ and $\widehat{R}$
and taking into account the fact that $\psi(0)=0$ with $\psi'(0)$ bounded, one readily obtains 
Eq. \eqref{f:wr}.
\end{proof}

\subsubsection{Homogeneous integral equation and bound state solutions} 
\label{subsubse:homogeneous}

We shall now focus on the basic relationship between 
(i) the solutions of the homogeneous Fredholm equation
$gL_\ell(k)x_\ell=x_\ell$ associated with the zeros 
$k^{(j)}(\ell,g)$ of $\sigma_\ell$
in the closed upper half--plane of $k$, and
(ii) the bound state solutions of the corresponding
nonlocal Schr\"odinger equation.
It is worthwhile to study this relationship specially 
\emphsl{for real values of the coupling} $g$, which 
correspond to the Hermitian character
of the Hamiltonian \eqref{f:19}. The results are
described by the following Theorem \ref{theorem:7}, 
whose proof relies basically on Lemmas \ref{lemma:vi},
\ref{lemma:vii}, \ref{lemma:vi'}, and \ref{lemma:vii'}.

\skd

\begin{theorem}
\label{theorem:7}
{\rm (a)} Let $x_{\ell,k}\in {X}_{\wea}$ 
be a non--zero solution of the homogeneous integral equation:
\beq
gL_\ell(k)\,x_{\ell,k}=x_{\ell,k},
\label{f:61}
\eeq
for fixed values of $\ell$ (non--negative integer), $g$ real, 
and $k$ such that $\Imag k\geqslant 0$, $k\neq 0$; 
then there exists a corresponding non--zero solution $\psi_{\ell,k}(R)$ 
of the Schr\"odinger--type integro--differential equation 
\eqref{f:24}, which is defined by
\beq
\psi_{\ell,k}(R)= g \int_0^{+\infty}G_\ell(k;R,R')\ x_{\ell,k}(R')\,\rmd R',
\label{f:62}
\eeq
and enjoys the following properties:
\begin{itemize}
\addtolength{\itemsep}{0.1cm}
\item[\rm (i)] $\psi_{\ell,k}'(R)$ is absolutely continuous and bounded in $[0,+\infty)$;
\item[\rm (ii)] $\psi_{\ell,k}(0)=0$ and there exists a constant $c(\psi)$
such that $|\psi(R)| \leqslant c(\psi)\ R$;
\item[\rm (iii)] $\|\psi_{\ell,k}\| 
\doteq \left[\int_0^{+\infty}|\psi_{\ell,k}(R)|^2\,\rmd R\right]^{1/2} < +\infty$.
\end{itemize}
Moreover, there holds the following inversion formula:
\beq
x_{\ell,k}(R)=\int_0^{+\infty}V_\ell(R,R')\ \psi_{\ell,k}(R')\,\rmd R'.
\label{f:63}
\eeq
{\rm (b)} If $k$ is such that $0\leqslant\Imag k \leqslant \alpha$, there is a finite constant  
$b_{\ell,k}$ such that:
\beq
b_{\ell,k}= \int_0^{+\infty}kR\,j_\ell(kR)\,x_{\ell,k}(R)\,\rmd R
=\int_0^{+\infty}\psi_{\ell,k}(R)\,v_{\ell,0}(k;R)\,\rmd R.
\label{f:63'}
\eeq
{\rm (c)} The respective cases $\Imag k > 0$ (ordinary bound states) 
and $k$ real, $k\neq 0$ (bound states embedded in the continuum
or ``spurious bound states") are distinguished from each other
by the following additional properties:
\begin{itemize}
\addtolength{\itemsep}{0.1cm}
\item[\rm (i)] if $\Imag k >0$, one necessarily has $\Real k =0$;
moreover, the solution $|\psi_{\ell,k}|$  satisfies the following 
global majorization:
\beq
|\psi_{\ell,k}(R)| \leqslant |g|\,{c_\ell}\,c_\ell^{(1)}\,A_{\varepsilon}
\|x_{\ell,k}\|_{\wea} \ R\ e^{-\min(\alpha,\Imag k)R};
\label{f:64}
\eeq
and
$|\psi'_{\ell,k}(R)|$ also tends to zero for $R$ tending to $+\infty$;
\item[\rm (ii)] if $\Imag k =0$, $k\neq 0$, there holds the following majorization: 
\beq
|\psi_{\ell,k}(R)| \leqslant |g|\,\widehat{c}_{\ell,k,\varepsilon}\
\|x_{\ell,k}\|_{\wea}\ R\ e^{-\alpha R}, 
\label{f:64'}
\eeq
where
$\widehat{c}_{\ell,k,\varepsilon}$ denotes a suitable constant,
and the previous relations \eqref{f:63'} become orthogonality relations, 
namely there holds the implication:
\beq
\Imag k=0 \quad \Longrightarrow \quad b_{\ell,k} =0. 
\label{f:63"}
\eeq
\end{itemize}
{\rm (d)} Conversely, if for fixed values of $\ell$ and $k$ 
($\ell \geqslant 0$, $\Imag k \geqslant 0$, $k\neq 0$), 
there exists a solution $\psi_{\ell,k}(R)$ of the  
integro--differential equation \eqref{f:24} which satisfies
properties (i), (ii), and (iii) listed in (a), then 
Eq. \eqref{f:63} defines a corresponding solution $x_{\ell,k}(R)$ in
${X}_{\wea}$ of the homogeneous equation \eqref{f:61}; moreover,
$\psi_{\ell,k}$ is reconstructed from $x_{\ell,k}$ by Eq. \eqref{f:62}
and all the properties described under (b) and (c) are valid.
\end{theorem}

\begin{proof}
As proved in Lemma \ref{lemma:vii}, the fact that the function $\psi_{\ell,k}$ 
is well--defined by formula \eqref{f:62} and satisfies properties (i) and (ii) 
listed in (a) is simply ensured by the assumption that 
$x_{\ell,k}\in X_{\wea}$.
It follows that the action on $\psi_{\ell,k}$ of
the second--order differential operator $D_{\ell,k}$
is well--defined and since 
$D_{\ell,k}(R) \, G_\ell(k;R,R')= \delta(R-R')$,  
Eq. \eqref{f:62} yields: 
\beq
D_{\ell,k}\  \psi_{\ell,k} = g \, x_{\ell,k}.
\label{f:66}
\eeq
In particular, since $x_{\ell,k}$ is non--zero,
$\psi_{\ell,k}$ is also non--zero. 
On the other hand, property (ii) of $\psi_{\ell,k}$ implies
that this function satisfies the assumptions of 
Lemma \ref{lemma:vi'}, and therefore (in view of Lemma \ref{lemma:vi'} (b)) 
the integral $\int_0^{+\infty} V_\ell(R,R') \, \psi_{\ell,k}(R')\, \rmd R'$ 
is convergent (for a.e. $R$) and defines an element of $X_{\wea}$. 
Now, by plugging the
expression \eqref{f:62} of $\psi_{\ell,k}$ in this integral, 
applying the Fubini theorem, and recognizing the definition
\eqref{f:27c} of $L_\ell(k;R,R')$, one obtains the following equality:
\beq
g \, [L_\ell(k) x_{\ell,k}](R) = 
\int_0^{+\infty} V_\ell(R,R') \,\psi_{\ell,k}(R')\, \rmd R'.
\label{f:65}
\eeq
Then, in view of Eqs. \eqref{f:66} and \eqref{f:65},
the assumption \eqref{f:61} of (a) implies the two equalities 
\eqref{FrSch} (for $\psi= \psi_{\ell,k}$); in other words, 
the integro--differential equation
\eqref{f:24} and the inversion formula \eqref{f:63} are satisfied
respectively by $\psi_{\ell,k}$ and $ x_{\ell,k}$.

\vskip 0.2cm

Another result of Lemma \ref{lemma:vi'} (b) (namely, Eq. \eqref{f:60})
implies that $\psi_{\ell,k}$ and $ x_{\ell,k}$ satisfy
our statement (b).

\vskip 0.2cm

The proof of the statements listed in (c) will rely crucially
on the Wronskian lemma (Lemma \ref{lemma:vii'}). We distinguish the two cases:

\vskip 0.2cm

(i) $\Imag k>0$:
The bound \eqref{f:64} coincides with 
\eqref{f:64"}, which has been established in
Lemma \ref{lemma:vii} under the same assumption for $x_{\ell,k}$. 
As a by--product, property (iii) of $\psi_{\ell,k}$ 
is established for the case $\Imag k>0$. Besides, Eq. \eqref{f:64''} 
implies that the function $\psi'_{\ell,k}(R)$ also tends to zero 
for $R$ tending to infinity. \\
Moreover, the uniform boundedness of $\psi'_{\ell,k}$ together with
the bound \eqref{f:64} on $\psi_{\ell,k}$ imply that the Wronskian
$W(R) \doteq 
[\overline{\psi_{\ell,k}}(R)\,\psi_{\ell,k}'(R)-
\overline{\psi_{\ell,k}'}(R)\,\psi_{\ell,k}(R)]$ 
tends to zero for $R$ tending to infinity.
Now, since $\psi_{\ell,k}$ is a solution of  
\eqref{f:24}, we can apply Lemma \ref{lemma:vii'}; 
Eq. \eqref{f:wr} then entails that $({\overline k}^2-k^2) 
\int_0^{+\infty}{\overline \psi}_{\ell,k}(R)\,\psi_{\ell,k}(R)\,\rmd R =0$,
which is only possible (for $\Imag k >0$) if $\Real k =0$, since 
$\psi_{\ell,k}$ is non--zero. 

\vskip 0.2cm

(ii) $\Imag k=0$: The following steps can be taken:

\sku

\begin{itemize}
\item[(1)] $|\psi_{\ell,k}(R)| $ is uniformly bounded.
This results from the inequalities  
\eqref{f:64"} and \eqref{f:viia} of 
Lemma \ref{lemma:vii} 
together with the bound \eqref{a:0'} on
$|kRh_\ell^{(1)}(kR)|$. 
\item[(2)] $\lim_{R\to \infty} |\psi_{\ell,k}(R)| =0$.
This is implied (for $k\neq 0$) by the following two statements.
At first, Eq. \eqref{f:wr} of Lemma \ref{lemma:vii'} now implies 
(since $k=\overline k$) that $\lim_{R\to \infty} W(R)=0$.
Secondly, in view of (1) and of
Eq. \eqref{f:vii"} of 
Lemma \ref{lemma:vii}  
(for $\Imag k=0$), it follows that the expression
$[2\rmi k {\overline\psi}_{\ell,k}(R) \, \psi_{\ell,k}(R)-W(R)]
={\overline\psi}_{\ell,k}(R)[\rmi k \psi_{\ell,k} -\psi'_{\ell,k}(R)]
-\psi_{\ell,k}(R)[\overline{\rmi k \psi_{\ell,k}}-
\overline{\psi'_{\ell,k}(R)}]$
tends itself to zero for $R$ tending to infinity. 
\item[(3)] There holds the orthogonality relation \eqref{f:63"},
i.e. $b_{\ell,k}=0$.
In fact, it also follows from Eq. \eqref{f:vii'} of Lemma 
\ref{lemma:vii} that
the difference $|\psi_{\ell,k}(R)| -|g||b_{\ell,k}| |Rh_\ell^{(1)}(kR)|$
tends to zero for $R$ tending to infinity.
But we know from \eqref{a:6'} that
in the limit $R\to \infty$, the function  
$kRh_\ell^{(1)}(kR)$ behaves like 
$\exp\{\rmi[kR-(\ell+1)\frac{\pi}{2}]\}$. 
Therefore the validity of (2) ($|\psi_{\ell,k}(R)| \to 0$)
necessitates that $b_{\ell,k}=0$.
\end{itemize}

\vskip 0.2cm
Let us now consider again Eq. \eqref{f:viia} of Lemma \ref{lemma:vii}.
Since $b_{\ell,k}=0$, this bound reduces to the following one: 
\beq
\left|\psi_{\ell,k}(R)\right|
\leqslant \mathrm{c}_\ell(k)\, |g|\, \|x\|_{\wea}
\ e^{-\alpha R} \qquad \mathrm{for} \quad R\geqslant 1.
\label{f:71}
\eeq
But, in view of Eq. \eqref{f:64"} of Lemma \ref{lemma:vii} (for $\Imag k=0$),
$\psi_{\ell,k}(R)$ also satisfies (for all $R\geqslant 0$) the bound
\beq
|\psi_{\ell.k}(R)|\leqslant |g|\, \widehat{c}_{\ell,\varepsilon} \, R\,\|x_{\ell,k}\|_{\wea}.
\label{f:72}
\eeq
Then, it is clear that the two bounds \eqref{f:71} and 
\eqref{f:72} can be replaced by a unique bound of the form 
\eqref{f:64'}, valid for all $R\geqslant 0$.
The latter also implies property (iii) of $\psi_{\ell,k}$ for the case $\Imag k =0$.

\vskip 0.5cm

\noindent
\emph{Proof of {\rm (d)}.} Let $\psi(R)\equiv\psi_{\ell,k}(R)$ be a solution of  
Eq. \eqref{f:24} satisfying properties (i), (ii), (iii) of (a); 
one then has:
$D_{\ell,k} \psi(R) = g\int_0^{+\infty} V_\ell(R,R')\,\psi(R')\,\rmd R'\doteq g \, x_\psi(R)$,
where $x_\psi$ belongs to $X_{\wea}$ in view of 
Lemma \ref{lemma:vi'} (b). If we now introduce the function 
$\widehat{\psi}(R)\doteq g\,\int_0^\infty G_\ell(k;R,R')\,x_\psi(R')\,\rmd R'$,
to which the previous study of \eqref{f:62} can be applied, 
we can assert that this function satisfies equation \eqref{f:66}, namely,
$D_{\ell,k}\widehat{\psi}(R) = g\,x_\psi(R)$, together with
properties (i), (ii) of (a), property (iii) being only obtained 
for the case $\Imag k >0$. Therefore the function 
$y(R)\doteq [\psi(R) -\widehat{\psi}(R)]$, which is such that
$D_{\ell,k} \, y(R) =0$ and satisfies properties (i) and (ii) of (a),
is a constant multiple of $kRj_\ell(kR)$; this multiple vanishes if
property (iii) of (a) is also satisfied and therefore, for the case $\Imag k>0$, 
we readily obtain that
\beq
\psi(R) = \widehat{\psi}(R) = g \int_0^\infty G_\ell(k;R,R')\,x_\psi(R')\,\rmd R'.
\label{f:73'}
\eeq
It requires a further argument to obtain the corresponding result for the case
$\Imag k=0$, since we can only write \emphsl{a priori} that: 
\beq
\psi(R) = g\int_0^\infty G_\ell(k;R,R')\,x_\psi(R')\,\rmd R' + \rho \, kRj_\ell(kR),
\label{f:78}
\eeq
where $\rho$ denotes a constant factor. Now, we can again deduce from Lemma \ref{lemma:vii} 
(applied to the function $\widehat{\psi}(R)$) that $\psi(R)$ admits an asymptotic behaviour 
of the following form (given by \eqref{f:vii'}):
\beq
\lim_{R\to \infty}\,
\left|\psi(R)-\rho\,kR j_\ell(kR)+\rmi \, g\,b\,Rh_\ell^{(1)}(kR)\right| =  0.
\label{f:75'}
\eeq
The proof that both constants $\rho $ and $b$ also vanish in this case relies on 
the fact that $\psi (R)$ is assumed to satisfy property (iii). 
In fact, the functions $kRh_\ell^{(1)}(kR)$ and $kR j_\ell(kR)$
behave at infinity, respectively, as $\exp\{\rmi[kR-(\ell+1)\frac{\pi}{2}]\}$ 
(in view of \eqref{a:6'}) and $\frac{\pi}{2}\cos[kR-(\ell+1)\frac{\pi}{2}]$ 
(see, e.g., \cite{Tichonov}), and therefore the finiteness of the quantity  
$\int_0^\infty |\psi(R)|^2 \rmd R$ is consistent with \eqref{f:75'} if and 
only if $\rho=b=0$. To conclude, we have obtained that Eq. \eqref{f:73'} is valid
for all cases ($\Imag k \geqslant 0$).
First, this implies that the function $x_\psi$ is non--zero and, moreover,
by applying to both sides of Eq. \eqref{f:73'} the integral operator 
associated with $V_\ell(R,R')$, one gets in view of \eqref{f:27c} (in operator form):
\beq
x_\psi  = V_\ell \, \psi = g \, V_\ell \, G_\ell(k) \, x_\psi = g L_\ell(k) x_\psi.
\label{f:76''}
\eeq
So, by starting from the assumptions of point (d), we have
been able to derive Eqs. \eqref{f:73'} and \eqref{f:76''}, namely,
we have reproduced the basic assumptions \eqref{f:61} and \eqref{f:62}
of (a) for all cases $\Imag k \geqslant 0$, which ends the proof of the theorem. 
\end{proof}

\skd

\begin{remark}
Had property (iii) not been imposed on $\psi$ in the assumptions of point (d),
one would have obtained from the general form \eqref{f:78}
(by applying the operator $V_\ell$ to both sides of the latter
and also accounting for \eqref{f:27b}):
\beq
x_\psi(R) = g\int_0^\infty L_\ell(k;R,R')\, x_\psi(R')\,\rmd R' + \rho \, v_{\ell,0}(k;R). 
\label{f:77}
\eeq
In particular, the latter form is relevant for $\rho =1$,
since it coincides with Eq. \eqref{f:27a}, whose solution 
$x_\psi(R) \doteq v_\ell(k,g;R)$ will be used below
for describing the \emphsl{scattering solution} of 
Eq. \eqref{f:24} (see Theorem \ref{theorem:7'} below).
\label{rem:1}
\end{remark}

\subsubsection{Inhomogeneous integral equation and scattering solutions; the partial 
scattering amplitude $\boldsymbol{T}_{\boldsymbol{\ell}}\boldsymbol{(k;g)}$} 
\label{subsubse:inhomogeneous}

\begin{theorem}
\label{theorem:3'}
Being given any potential $V$ in a class $\cN_{\wea}$ 
and any given complex number $g$, let $\Omega_{\alpha,\ell}(g)$
be the set of all points $k\in \Omega_\alpha$ such that
the corresponding Fredholm--Smithies denominator $\sigma_\ell(k,g)$ 
of the resolvent $R_\ell(k;g)$ does not vanish.
Then the inhomogeneous integral equation
\beq
[1-g L_\ell(k)]v_\ell(k,g;\cdot)=
v_{\ell,0}(k;\cdot) \qquad (\ell=0,1,2,\ldots),
\label{f:71'}
\eeq
admits for every $g\in \C$ and $k\in\Omega_{\alpha,\ell}(g)$  
a unique solution $v_\ell(k,g;R)$ in ${X}_{\wea}$, 
which is well--defined by the formula
\beq
v_\ell(k,g;\cdot)= R_{\ell}(k;g) \, v_{\ell,0}(k;\cdot).
\label{f:71"}
\eeq 
Furthermore, the function $(k,g) \mapsto v_\ell(k,g;\cdot)$ is meromorphic in
$\Omega_\alpha \times \C$, and for any $g$ in $\C$, the function 
$k \mapsto v_\ell(k,g;\cdot)$ is holomorphic in 
$\Omega_{\alpha,\ell}(g)$, in the sense of the vector--valued functions
taking their values in $X_{\wea}$. 
\end{theorem}

\begin{proof}
The solution \eqref{f:71"} of \eqref{f:71'}, which follows from
\eqref{f:45}, defines $v_\ell(k,g;\cdot)$ as an element of
$X_{\wea}$ in view of the fact that 
$v_{\ell,0}(k;\cdot)$ belongs to $X_{\wea}$
(see Lemma \ref{lemma:v}, more precisely Corollary \ref{lemma:v}--\ref{theorem:4})
and that $ R_{\ell}(k;g)$ is a bounded operator in 
$X_{\wea}$ (see Theorem \ref{theorem:6"}).
Moreover, the meromorphy properties in $(k,g)$
and the holomorphy properties in $k$ at fixed $g$ of
$R_{\ell}(k;g)$ and $v_{\ell,0}(k;\cdot)$,
established respectively in Theorem \ref{theorem:6"} 
and Lemma \ref{lemma:v}, imply the corresponding properties
for $v_\ell(k,g;\cdot)$ in view of Lemma \ref{lemma:B1'} (ii).
\end{proof}

Let us now define for every $g\in \C$  and $k$ in $\Omega_{\alpha,\ell}(g)$ 
the following functions:
\begin{align}
& \Psi_\ell(k,g;R) = kR j_\ell(kR) + \Phi_\ell(k,g;R), \label{f:74} \\
& \Phi_\ell(k,g;R) = g\int_0^{+\infty} \!\! G_\ell(k;R,R') \, v_\ell(k,g; R')\,\rmd R'. \label{f:74'}
\end{align}
Since $v_\ell(k,g;\cdot)\in X_{\wea}$,
the function $\Phi_\ell(k,g;\cdot)$ is well--defined 
by Lemma \ref{lemma:vii}, and we have:
\beq
D_{\ell,k}\,\Psi_\ell(k,g;R) =
D_{\ell,k}\,\Phi_\ell(k,g;R) = g \, v_\ell(k,g; R)
= g^2 \, [L_\ell(k)\, v_\ell(k,g;\cdot)](R)+g\,v_{\ell,0}(k;R).
\label{f:74"}
\eeq
But, by taking Eqs. \eqref{f:27c} and \eqref{f:74'} into account 
(and interchanging convergent integrals) along with
Eq. \eqref{f:27b} (and Lemma \ref{lemma:v}), the r.h.s. of 
\eqref{f:74"} can be rewritten as follows:
\beq
g\int_0^\infty V_\ell(R,R') \, \Phi_\ell(k,g;R') \,\rmd R'
+g\int_0^\infty V_\ell(R,R') \, kR'j_\ell(kR') \,\rmd R', \nonumber
\eeq
so that (in view of Eqs. \eqref{f:74} and \eqref{f:74"}), one has:
\beq
D_{\ell,k} \, \Psi_\ell(k,g;R) =  
g \int_0^\infty V_\ell(R,R') \, \Psi_\ell(k,g;R') \,\rmd R'.
\label{f:75"}
\eeq
We then conclude that Eq. \eqref{f:74} reproduces the form \eqref{f:25}
of the scattering solution of the Schr\"odinger--type equation \eqref{f:24}.
We shall now see that Lemma \ref{lemma:vii} not only allows one to obtain 
the properties of this solution for $k$ and $g$ real, which have been 
listed under (S--b) and include in
particular the Sommerfeld radiation condition (given in \eqref{f:26}),
but also implies extensions of these properties to the complex domain
$\Omega_{\alpha,\ell}(g)$ of the $k$--plane (for each $g\in \C$). 
Moreover, one will show that there also hold meromorphy properties 
of this solution with respect to $(k,g)$ in the domain 
$\Omega_\alpha \times \C$. This will be the scope of the 
following two theorems.

\skd

\begin{theorem}
\label{theorem:7'}
For every $g\in \C$, $k \in\Omega_{\alpha,\ell}(g)$ and for any non--negative 
integer $\ell$, there exists a solution of Eq. \eqref{f:24} which is of the form
$\Psi_\ell(k,g;R)=kRj_\ell(kR)+\Phi_\ell(k,g;R)$, and satisfies the following 
properties: 
\begin{itemize}
\item[\rm (i)] $\frac{\rmd}{\rmd R}\Psi_\ell(k,g;R)$ is absolutely continuous for all 
$R$ and bounded for $R$ tending to zero; \\[-6pt]
\item[\rm (ii)] the function $\Phi_\ell(k,g;\cdot)$ is expressed in terms of the solution 
$v_\ell(k,g;\cdot)$ of the integral equation \eqref{f:71'} by the formula \eqref{f:74'},
and satisfies a bound of the following form:
\begin{align}
& |\Phi_\ell(k,g;R)| \leqslant |g|\, {\widehat{c}_{\ell,\varepsilon}} \,
\|v_\ell(k,g;\cdot)\|_{\wea} \ R\  e^{-R\Imag k}, \label{f:76} \\
\intertext{in particular one has:}
& \Phi_\ell(k,g;0) =0; \label{f:76'}
\end{align}
\item[\rm (iii)] for $0\leqslant\Imag k <\alpha $, there holds the following limit:
\beq
\lim_{R\to +\infty}\, 
e^{R\Imag k} \ \left[\frac{\rmd}{\rmd R}\Phi_\ell(k,g;R)-\rmi k\Phi_\ell(k,g;R)\right]=0; 
\label{f:76"}
\eeq
\item[\rm (iv)] for $-\frac{\alpha}{2}< \Imag k <\alpha$, there holds the 
following majorization (containing a suitable constant ${\widehat c}_\ell(k,g)$):
\beq
\left|\Phi_\ell(k,g;R)-T_\ell(k;g) \ \rmi kRh_\ell^{(1)}(kR)\right|
\leqslant {\widehat c}_\ell(k,g)\  e^{-\alpha R}\ \max \left({\rm 1},
e^{-2(\Imag k)R}\right),
\label{f:80} \\
\eeq
{where:}
\beq T_\ell(k;g)=-g\int_0^{+\infty} Rj_\ell(kR) \, 
v_\ell(k,g;R)\,\rmd R, \label{f:81}
\eeq
is well--defined as an analytic function of $k$ in the domain $\Omega_{\alpha,\ell}(g)$; \\[-6pt]
\item[\rm (v)] for $k\in\R^+$ and $g\in\R$, the function $\Psi_\ell(k,g;R)$ 
satisfies all the properties listed under {\rm (S--b)} of the scattering 
solution of Eq. \eqref{f:24}, and the corresponding function $T_\ell(k;g)$,
which is the physical partial scattering amplitude, can also be defined as
\beq
T_\ell(k;g)=\lim_{R \to +\infty}
\left|\frac{\Phi_\ell(k,g;R)}{\rmi kRh_\ell^{(1)}(kR)}\right|.
\label{psa}
\eeq
\end{itemize}
\end{theorem}

\begin{proof}
By applying Lemma \ref{lemma:vii} with $x(R) = v_\ell(k,g;R)\in {X}_{\wea}$, \,
$\psi_{x;\ell,k}(R) = \Phi_\ell(k,g;R)$ and
$b_{x;\ell,k}= -\frac{k}{g} \, T_\ell(k;g)$, one readily obtains
property (i) (given by (c)) and
properties (ii), (iii), (iv) ( since formulae
\eqref{f:76}, \eqref{f:76"}, \eqref{f:80}, and \eqref{f:81}, 
correspond respectively to 
[\eqref{f:64"}, \eqref{f:65"}], 
\eqref{f:vii"}, [\eqref{f:viia}, \eqref{f:viib}] and \eqref{f:vii}).
Property (v) is directly obtained by inspection,   
the limit \eqref{psa} being also directly implied by
\eqref{f:80} in view of the asymptotic behaviour \eqref{a:6'} of $Rh_\ell^{(1)}(kR)$ 
for $R$ tending to infinity. 
The proof of the analyticity property of the function $k \mapsto T_\ell(k;g)$
in the domain $\Omega_{\alpha,\ell}(g)$ is left to the following theorem.
\end{proof}

Now, by exploiting the holomorphy properties of $\sigma_\ell(k;g)$ and 
$N_\ell(k;g)$, obtained respectively in Theorems \ref{theorem:5} and \ref{theorem:6},
we can derive \emphsl{meromorphy properties} in the complex variables $(k,g)$
of the scattering solution $\Psi_\ell(k,g;R)$ and of the partial scattering
amplitude $T_\ell(k;g)$. For this purpose, by taking into account Eq. \eqref{f:46},
we can now re--express $v_\ell(k,g;R)$ as follows, for all $k\in\Omega_{\alpha,\ell}$
and $g\in \C$:
\beq
v_\ell(k,g;\cdot)=\frac{u_\ell(k,g;\cdot)}{\sigma_\ell(k;g)},
\label{f:72'}
\eeq
where:
\beq
u_\ell(k,g;\cdot)\doteq
\left[\sigma_\ell(k;g)+g N_\ell(k;g)\right] \, v_{\ell,0}(k;\cdot).
\label{f:73}
\eeq
We then have:

\skd

\begin{theorem}
\label{theorem:8}
{\rm (i)} $T_\ell(k;g)$ can be written as follows:
\begin{align}
& T_\ell(k;g)= -\frac{g}{\sigma_\ell(k;g)}
\int_0^{+\infty} R'j_\ell(kR')\,u_\ell(k,g;R')\,\rmd R'=
\frac{\Theta_\ell(k;g)}{\sigma_\ell(k;g)}, \label{f:82} \\
\intertext{where:}
\begin{split}
& \Theta_\ell(k;g) =
-g\,\sigma_\ell(k;g)\int_0^{+\infty} \!\! R'j_\ell(kR')\,v_{\ell,0}(k;R')\,\rmd R' \\
&\quad -g^2\int_0^{+\infty} \!\! R' j_\ell(kR')\,\rmd R'
\int_0^{+\infty} \!\! N_\ell(k;g;R',R)\,v_{\ell,0}(k;R)\,\rmd R.
\end{split}
\label{f:83}
\end{align}
The following properties hold for any non--negative integral value of $\ell$:
\begin{itemize}
\item[\rm (i.a)] the function $k \mapsto \sigma_\ell(k;g)$ is
defined and uniformly bounded in $\overline{\Pi}_\alpha$,
holomorphic in $\Pi_\alpha$;
\item[\rm (i.b)] the function $k \mapsto \Theta_\ell(k;g)$ is
defined and uniformly bounded in $\overline\Omega_\alpha$,
holomorphic in $\Omega_\alpha$;
\item[\rm (i.c)] $k \mapsto T_\ell(k;g)$ is a meromorphic 
function in $\Omega_\alpha$.
\end{itemize}

\vskip 0.2cm

\noindent
{\rm (ii)} The function $k \mapsto S_\ell(k;g)$, given by
\beq
S_\ell(k;g) = 1+2\rmi \, T_\ell(k;g),
\label{f:84}
\eeq
is meromorphic in $\Omega_\alpha$. It satisfies the condition
of elastic unitarity for $k$ real and can be written as follows:
\beq
S_\ell(k;g) = e^{2\rmi\delta_\ell(k;g)} \qquad (\ell=0,1,2,\ldots),
\label{f:85}
\eeq
where $\delta_\ell(k;g)$ is a real--valued function of $k$ on $\R^+$
(which is defined modulo $\pi$). Accordingly, the following 
representation of $T_\ell(k;g)$ holds:
\beq
T_\ell(k;g) = e^{\rmi\delta_\ell(k;g)}\sin\delta_\ell(k;g) \qquad (k\in\R^+).
\label{f:86}
\eeq
\end{theorem}

\begin{proof}
(i) Formulae \eqref{f:82}, \eqref{f:83} obviously
follow from \eqref{f:81}, \eqref{f:72'}, and \eqref{f:73}.
Statement (i.a) has been proved in Theorem \ref{theorem:5}(a). \\[+4pt]
\emph{Proof of} (i.b). For each $k\in \overline\Pi_\alpha$, $N_\ell(k;g)$ acts on
$X_{\wea}$ as a bounded (and even Hilbert--Schmidt) operator
and, in view of Theorem \ref{theorem:6}(a), the function $k \mapsto N_\ell(k;g)$ is 
holomorphic in $\Pi_{\alpha}$ as a bounded--operator--valued function (since it is
holomorphic as a HS--operator--valued function). Since $v_{\ell,0}(k;\cdot)$ is 
holomorphic in $\Omega_{\alpha}$ as a function with values in $X_{\wea}$ 
(see Lemma \ref{lemma:v} and Corollary \ref{lemma:v}--\ref{theorem:4}), 
it then follows from Lemma \ref{lemma:B1'}(ii) 
that the function $k\to u_\ell^{(N)}(k;\cdot)$, defined by
\beq
u_\ell^{(N)}(k,g;R)=
\int_0^{+\infty} \!\! N_\ell(k;g;R,R')\,v_{\ell,0}(k;R')\,\rmd R,
\label{f:87}
\eeq
is also holomorphic in $\Omega_{\alpha}$ as a function with values in 
$X_{\wea}$. In view of Eq. \eqref{f:83} and of (i.a),
proving the analyticity of $\Theta_\ell(k;g)$ in $\Omega_{\alpha}$ amounts to
proving that the functions of $k$, defined by the integrals
\beq
\int_0^{+\infty} \!\! Rj_\ell(kR)\,v_{\ell,0}(k;R)\,\rmd R, 
\quad \mathrm{and} \quad
\int_0^{+\infty} \!\! Rj_\ell(kR) \, u_{\ell}^{(N)}(k,g;R)\,\rmd R,
\label{f:87'}
\eeq
are holomorphic in $\Omega_{\alpha}$. But this follows from
Lemma \ref{lemma:B2}(ii) by noting that:
\begin{itemize}
\item[(a)] the function $R \mapsto Rj_\ell(kR)$ takes its values in the dual space
$X_{\wea}^*$ of $X_{\wea}$, 
for all $k\in\overline\Omega_\alpha$; in fact (in view of \eqref{a:0} and \eqref{f:e0}) 
the following bound holds:
\beq
\begin{split}
& \hspace{-0.4cm} \mathrm{For\,\,} k\in\overline\Omega_\alpha, \\
& \!\!\!\!\!\left[\left\|\cdot \, j_\ell(k\cdot)\right\|^*_{\wea}\right]^2
= \int_0^{\infty}\frac{|Rj_\ell(kR)|^2 \, e^{-2\alpha R}}{R^{1-\varepsilon}(1+R)^{1+2\varepsilon}}\,\rmd R
\leqslant\frac{c_\ell^2}{|k|^2}\int_0^{\infty}
\frac{|k|R}{R^{1-\varepsilon}(1+R)^{1+2\varepsilon}}\,\rmd R
\leqslant\frac{c_\ell^2}{|k|}\,A_\varepsilon^2;
\end{split}
\label{f:88}
\eeq
\item[(b)] the holomorphic function $k \mapsto Rj_\ell(kR)$
satisfies (for $k\in \Omega_\alpha$) the conditions of 
Lemma \ref{lemma:B8}(ii), and therefore defines a 
\emphsl{holomorphic vector--valued function} of $k$ in $\Omega_\alpha$
taking its values in ${X}_{\wea}^*$. 
\end{itemize}
Moreover, since $\frac{|v_{\ell,0}(k,\cdot)|}{|k|^{1/2}}$ is
uniformly bounded in $X_{\wea}$ 
for $k\in \overline\Omega_{\alpha}$ (see Corollary \ref{lemma:v}--\ref{theorem:4}), and 
since $\|N_\ell(k;g)\|_\mathrm{HS}$ is uniformly
bounded (at fixed $g$) in $\overline\Pi_\alpha$ 
(see Theorem \ref{theorem:6}), it follows from \eqref{f:87} that
$\frac{|u_\ell^{(N)}(k,g;\cdot)|}{|k|^{1/2}}$ is also
uniformly bounded in $X_{\wea}$ 
for $k\in \overline\Omega_{\alpha}$. Then, it results from
the uniform bound \eqref{f:88} in $X_{\wea}^*$ 
that the functions defined by the two integrals in \eqref{f:87'}, 
which are respectively bounded by
$\left\|\cdot \, j_\ell(k\cdot)\right\|^*_{\wea}
\times\|v_{\ell,0}(k,\cdot)\|_{\wea}$
and
$\left\|\cdot \, j_\ell(k\cdot)\right\|^*_{\wea} 
\times\left\|u_\ell^{(N)}(k;\cdot)\right\|_{\wea}$, 
are uniformly bounded for $k\in \overline\Omega_{\alpha}$.
In view of \eqref{f:83} and \eqref{f:87} (and of (i.a)), this implies that
the functions $\Theta_\ell(k;g)$ are uniformly bounded for
$k\in \overline\Omega_{\alpha}$, which ends the proof of (i.b).

\vskip 0.2cm

\noindent
Statement (i.c) obviously follows from \eqref{f:82} and statements (i.a) and (i.b).

\vskip 0.2cm

\noindent
(ii) From \eqref{f:84} and (i.c) it follows that 
$S_\ell(k;g)$ is a meromorphic function of $k$ for
$k\in\Omega_\alpha$. The condition of elastic unitarity  
$S_\ell(k;g)\overline{S_\ell(k,g)} ={\rm 1}$ follows
from the unitarity of the scattering operator $S$
proved in \cite[IV]{Bertero} in a very general
setting (see Section \ref{se:outline}).
This implies the representations \eqref{f:85} and \eqref{f:86}.
\end{proof}

\subsubsection{Complement on spurious bound states: the corresponding 
properties of scattering solutions, partial scattering amplitudes and 
phase--shifts}
\label{subsubse:spurious}

The phenomenon of the ``bound states embedded in the continuum"
(or ``positive--energy bound states" or ``spurious bound states") traces back to a classical
paper of Wigner and von Neumann \cite{VonNeumann}.
It can be qualitatively explained as follows:
bumps in a potential will reflect a wave, and well--arranged bumps can act constructively
and prevent a wave from reaching infinity:
i.e., stationary waves with falloff can be formed \cite{Simon}.
For example, a bound state embedded 
in the continuum seems to appear in the negative helium ion. The
level $^4\mathrm{P}_{5/2}$ of this system lies in fact in the continuum, and it is not
liable to auto--ionization. An analogous state seems to be present in the helium atom
(both of these examples have been indicated by Wigner to 
the author of Ref. \cite{Fonda}, as written
in a footnote of that paper).

For the class of nonlocal potentials considered here,
we can say that the partial scattering amplitude $T_\ell(k;g)$,
and therefore also the phase--shift $\delta_\ell(k;g)$, remain
well--defined
at all the values of $k$ corresponding to spurious bound states.
This is a consequence of the previous theorem, namely of
the {\it meromorphy} property of $T_\ell(k;g)$
in $\Omega_\alpha$ and of its \emphsl{boundedness} for all real $k$,
which is implied by Eq. \eqref{f:86} (expressing the unitarity condition).
Putting these two properties together entails that
$T_\ell(k;g)$ is finite 
and holomorphic at all real values of $k$, and therefore
in particular at those for which $\sigma_\ell(k;g)=0$,
corresponding to the occurrence of spurious bound states
(note that it is thus necessary that any such value of $k$ be also
a zero of the holomorphic function
$\Theta_\ell(k;g)$ introduced in Theorem \ref{theorem:8}).

We are now going to show that not only the partial scattering 
amplitude $T_\ell(k;g)$ but also the scattering solution
$\Psi_\ell(k,g;R)$ remains finite at all pairs $(k,g)$ which 
correspond to spurious bound states. More precisely we can state:

\skd

\begin{theorem}
\label{theorem:8'}
Let $g=\widehat{g}(k)$ be any solution of the equation $\sigma_\ell(k;g)=0$
considered as an analytic curve $\cC$ in a complex neighborhood
${\cN}_0$ in $\C^2$ of a certain real point $(k_0,g_0)$,
such that $g_0=\widehat{g}(k_0)$ with $k_0 \in \R^+$.
Let us also assume that the curve $\cC$ is associated with a
``simple pole" of the Fredholm resolvent kernel 
$R_\ell^{(\mathrm{tr})}(k;g)$.
Then the following properties are valid: \\[-6pt]

{\rm (i)} For $(k,g)\in{\cN}_0$, there exists a decomposition
of the following form of $R_\ell^{(\mathrm{tr})}(k;g)$:
\beq
R_\ell^{(\mathrm{tr})}(k;g;R,R') = \frac{p_\ell(k;R,R')}{\widehat{g}(k)-g}
+ \widehat{R}_\ell^{(\mathrm{tr})}(k;g;R,R'),
\label{f:89}
\eeq
where $p_\ell(k;R,R')$ is (for each $k$) a kernel of finite rank $r$
in $\widehat{X}_{\wea}$, which depends holomorphically
on $k$ in a suitable neighborhood ${\cV}_0$ of $k_0$, and satisfies the 
functional relation
\beq
\int_0^{+\infty} \!\! p_\ell(k;R,R'') \, p_\ell(k;R'',R')\,\rmd R'' = p_\ell(k;R,R'),
\label{f:90}
\eeq
while the kernel $\widehat{R}_\ell^{(\mathrm{tr})}(k;g)$ is uniformly bounded  
in $\widehat{X}_{\wea}$, for $(k,g)\in {\cN}_0$; 
in this open set, it defines a {\rm HS}--operator--valued
holomorphic function of $(k,g)$ taking its values in $\widehat{X}_{\wea}$.

\vskip 0.2cm

{\rm (ii)} The following equalities hold for all $k$ in ${\cV}_0$:
\beq
\int_0^{+\infty}\!\! Rj_\ell(k R)\,p_\ell(k;R,R')\,\rmd R = 0, 
\quad 
\int_0^{+\infty} \!\! p_\ell(k;R,R') \, v_{\ell,0}(k;R')\,\rmd R' = 0.
\label{f:89"}
\eeq

\vskip 0.2cm

{\rm (iii)} The solution $v_\ell(k,g;R)$ of the inhomogeneous equation \eqref{f:71'},
which is well--defined as an element of ${X}_{\wea}$ by the formula 
$v_\ell(k,g;\cdot)= R_{\ell}(k;g)  v_{\ell,0}(k;\cdot)$ for $\sigma(k;g)\neq 0$, 
tends to a finite limit $v_\ell(k,\widehat{g}(k);\cdot)$ in 
${X}_{\wea}$, when $g$ tends to $\widehat{g}(k)$
and the vector--valued function $(k,g) \mapsto v_\ell(k,g;\cdot)$
is holomorphic in ${\cN}_0$. 

\vskip 0.2cm

{\rm (iv)} The function $\Phi_\ell(k,g;R)$ defined in Eq. \eqref{f:74'}, 
the scattering solution $\Psi_\ell(k,g;R)= k Rj_\ell(k R)+\Phi_\ell(k,g;R)$,
and the partial scattering amplitude $T_\ell(k;g)$, defined in \eqref{f:81},
remain finite and satisfy all the properties listed in Theorem \ref{theorem:7'} 
at all points $(k,g) \in {\cN}_0$.
\end{theorem}

\begin{proof}
(i) For convenience, we will choose 
${\cN}_0 \doteq \{(k,g)\,:\, k\in {\cV}_0; \ |g-\widehat{g}(k)|< a\}$,
for a suitable choice of ${\cV}_0$ containing $k_0>0$, and of $a>0$.
For any $k$ fixed in ${\cV}_0$, the existence of a decomposition of the form
\eqref{f:89} for the Fredholm resolvent kernel of $L(k)$
is a standard result (see, e.g., \cite{Goursat}). A simple presentation given in 
Subsection III--3 of \cite{Bros-Pesenti} (see, in particular, formulae (77) through (83),
Lemma 1 of the latter, which we here consider for the simple--pole case $n=1$) 
allows one to write the following formula for $p_\ell(k;R,R')$:
\beq
p_\ell(k;R,R')= -\frac{1}{2\pi\rmi}\int_\gamma R_\ell^{(\mathrm{tr})}(k;g;R,R')\, \rmd g,
\label{f:91}
\eeq
where $\gamma$ denotes a closed (anticlockwise) contour surrounding the point $\widehat{g}(k)$ 
inside the disk $|g- \widehat{g}(k)|<a$ of the $g$--plane
(note that for ${\cV}_0$ sufficiently small, $\gamma$ may be considered as independent of $k$).
This integral is meaningful in the sense of HS--kernels in 
$\widehat{X}_{\wea}$, since the function
$g \mapsto R_\ell^{(\mathrm{tr})}(k;g)=\frac{N_\ell(k;g)}{\sigma_\ell(k;g)}$
is holomorphic in the complement of $\cC$ as a HS--operator--valued function (see 
Theorem \ref{theorem:6} and the last page of Appendix \ref{appendix:b}).
Moreover, in view of the fact that the function
$k \mapsto R_\ell^{(\mathrm{tr})}(k;g;R,R')$ is holomorphic in ${\cV}_0$
for all $g\in \gamma$ (as a result of Theorem \ref{theorem:6'}), it follows that 
$k \mapsto p_\ell(k;R,R')$ is itself a HS--operator--valued holomorphic function 
(of finite rank $r$) in ${\cV}_0$. The projector formula \eqref{f:90}
(true for all $k\in {\cV}_0$) refers to formula (80) of \cite{Bros-Pesenti}.
Finally, in view of \eqref{f:89} and \eqref{f:91}, one also has (for $k\in {\cV}_0$):
\beq
\int_\gamma \widehat{R}_\ell^{(\mathrm{tr})}(k;g;R,R')\, \rmd g = 0,
\label{f:91'}
\eeq
which proves (also in view of Theorem \ref{theorem:6'})
that the function $(k,g) \mapsto \widehat{R}_\ell^{(\mathrm{tr})}(k;g;R,R')$ 
is holomorphic in ${\cN}_0$ as a HS--operator--valued function.

\vskip 0.2cm

\noindent
(ii) Since (see \cite{Goursat,Bros-Pesenti}), 
the kernel $p_\ell(k_0)$ of rank $r$ is such that
\beq
p_\ell(k_0)=g_0 \, L_\ell(k_0) \, p_\ell(k_0)=g_0 \, p_\ell(k_0)\, L_\ell(k_0),
\label{f:92-ex}
\eeq
and since $p_\ell(k_0)\in\widehat{X}_{\wea}$,
there exist $r$ linearly independent solutions 
$x^{(j)}(R) \in {X}_{\wea}$ 
of the homogeneous Fredholm equation
$g_0 \, L_\ell(k_0) \, x =x$,
and $r$ linearly independent solutions 
$\psi^{*(j)}(R) \in {X}_{\wea}^*$ 
of the corresponding equation
$g_0 \,\psi^* \, L_\ell(k_0)=\psi^*$ such that:
\beq
p_\ell(k_0;R,R')= \sum_{j=1}^r x^{(j)}(R)\,\psi^{*(j)}(R').
\label{f:89'}
\eeq
But we are in the case when formula \eqref{f:63"} of Theorem \ref{theorem:7} applies, 
namely we have, for $1\leqslant j \leqslant r$,  
$\int_0^{+\infty} Rj_\ell(k_0 R)\,x^{(j)}(R)\,\rmd R =0$,  entailing the orthogonality 
$\int_0^{+\infty}Rj_\ell(k_0 R)\, p_\ell(k_0;R,R')\, \rmd R =0$. 

Let us now associate with $\psi^{*(j)}$ ($1\leqslant j\leqslant r$) the function 
$x^{*(j)}\doteq V_\ell\, \psi^{*(j)}$, which belongs to  
${X}_{\wea}$ since
$\psi^{*(j)}\in {X}^*_{\wea}$ (see Lemma \ref{lemma:vi} (c)).
We then have: $\psi^{*(j)} = g_0 \,\psi^{*(j)}\,L_\ell(k_0)=g_0 \, \psi^{*(j)}\,V_\ell \, G_\ell(k_0)$ 
(in view of \eqref{f:27c}), which also yields:
$V_\ell\, \psi^{*(j)}=g_0 \, V_\ell\,[G_\ell(k_0)\,V_\ell\,\psi^{*(j)}]=g_0\,L_\ell(k_0)[V_\ell\,\psi^{*(j)}]$, 
(by using the symmetry relations \eqref{f:27d} and \eqref{f:24'}), 
i.e., $x^{*(j)}=g_0 \, L_\ell(k_0)\,x^{*(j)}$.
It then follows from Lemma \ref{lemma:vi} (c) and  
Eq. \eqref{f:63"} (applied to the eigenfunction $x^{*(j)}$ of $L_\ell(k_0)$) 
that: $\int_0^{+\infty}\psi^{*(j)}(R)\, v_{\ell,0}(k;R)\,\rmd R
=\int_0^{+\infty} Rj_\ell(k_0 R)\,x^{*(j)}(R)\,\rmd R =0$,
for $1\leqslant j \leqslant r$, which entails the orthogonality relation 
$\int_0^{+\infty} p_\ell(k_0;R,R')\, v_{\ell,0}(k;R)\, \rmd R =0$.  

The previous argument can of course be applied to any  \emphsl{real}
neighbouring point $(k,g)$ of $(k_0,g_0)$ on the curve $\cC$, thus 
implying that relations \eqref{f:89"} hold for all $k$ in ${\cV}_0 \cap \R^+$.
Therefore, they also hold in ${\cV}_0$ by the principle of analytic continuation.

\vskip 0.2cm

\noindent
{\rm (iii)} In view of Eqs. \eqref{f:89} and \eqref{f:89"}, one has
for all $(k,g)\in {\cN}_0\setminus \cC$:
\beq
v_\ell(k,g;\cdot)= v_{\ell,0}(k;\cdot) 
+ R^{(\mathrm{tr})}_{\ell}(k;g)\, v_{\ell,0}(k;\cdot)=v_{\ell,0}(k;\cdot)
+\widehat{R}^{(\mathrm{tr})}_{\ell}(k;g) \,  v_{\ell,0}(k;\cdot). 
\label{f:92}
\eeq
But, since the function $k \mapsto \widehat{R}^{(\mathrm{tr})}_{\ell}(k;g)$ 
is holomorphic at fixed $g$ for $(k,g)\in {\cN}_0$ as a HS--operator--valued 
function taking its values in $\widehat{X}_{\wea}$,
it follows from Lemmas \ref{lemma:v} and \ref{lemma:B1'}(ii)
(as in the proof of Theorem \ref{theorem:3'}) that the 
r.h.s. of Eq. \eqref{f:92} defines the analytic continuation of
the function $k \mapsto v_\ell(k,g;\cdot)$ at fixed $g$ from  
$\Omega^+_{\alpha,\ell}(g) $ to the set $\{k\,:\,(k,g)\in {\cN}_0\}$
(as a vector--valued function taking its values in ${X}_{\wea}$).

\vskip 0.2cm

\noindent
{\rm (iv)} In view of the properties of $v_\ell(k,g;\cdot)$ established in (iii),
the finiteness and holomorphy properties of $\Phi_\ell(k,g)$,
$\Psi_\ell(k,g)$, and $T_\ell(k;g)$, and all the bounds and asymptotic limits
of the latter, which have been established in Theorem \ref{theorem:7'}, 
are directly extended to the set ${\cN}_0$, since the arguments given in the
proof of that theorem remain valid there without modification.
\end{proof}

\skd

\begin{remark}
In view of Theorem \ref{theorem:7}, one can say that at each pair 
$(k,g=\widehat{g}(k))$, $k>0$, corresponding to a spurious bound state, 
there exists a finite--dimensional affine subspace of functions 
$[v_\ell]_\mu(k,g;\cdot)\in {X}_{\wea}$ 
of the form 
$[v_\ell]_\mu(k,g;R) = v_\ell(k,g;R) + \sum_{j=1}^r \mu_j\, x^{(j)}(R)$,
($\mu\doteq(\mu_1,\ldots,\mu_r)\in \R^r$), which all are solutions 
of the inhomogeneous equation \eqref{f:71'}, while the corresponding 
affine subspace of functions
$[\Psi_\ell]_\mu(k,g;R)= k Rj_\ell(k R)+[\Phi_\ell]_\mu(k,g;R)$, 
where we have put 
$[\Phi_\ell]_\mu(k,g;R)= g\int_0^{+\infty} G_\ell(k;R,R') [v_\ell]_\mu(k,g; R')\,\rmd R'$,
satisfy all the properties of scattering  solutions of the Schr\"odinger--type equation \eqref{f:24}.
However, in view of Eqs. \eqref{f:81} and \eqref{f:63"}, all these solutions lead to the 
(unique) scattering amplitude
$T_\ell(k;g)=-g\int_0^{+\infty}R'j_\ell(kR')[v_\ell]_\mu(k,g;R')\,\rmd R'$.
Moreover, we can say that the particular solutions 
$v_\ell(k,\widehat{g}(k);\cdot)$ and $\Psi_\ell(k,\widehat{g}(k);\cdot)$
are distinguished from all the others by their property of being the restrictions 
to the curve $\cC$ ($g=\widehat{g}(k)$) of the (respective) solutions 
$v_\ell(k,g;\cdot)$ and $\Psi_\ell(k,g;\cdot)$, 
which depend holomorphically on $g$ and $k$ 
in a complex neighborhood of $\cC$ in $\C^2$.
\label{rem:2}
\end{remark}

\vskip 0.3cm
 
We shall now complete this study of 
bound states embedded in the continuum by recalling the
following two interesting properties, 
which concern the behavior of the
phase--shift $\delta_\ell(k;g)$ for $k$ varying between zero and infinity.

\vskip 0.3cm

{\bf (a)}\ The following proposition was proved by
Gourdin, Martin and Chadan and quoted in \cite{Chadan}
for the case of \emphsl{separable} nonlocal potentials.

\skd

\begin{proposition}[G.M.C. \cite{Chadan}]
\label{propo:2}
{\rm (i)} Positive energy bound states correspond to those energies at which the phase--shift
crosses a value of $n\pi$, $n=0,\pm 1, \ldots$, downward, (i.e., with a negative slope)
and conversely. \\[+2pt]
{\rm (ii)} The phase--shift never crosses $n\pi$ upward. \\[+2pt]
{\rm (iii)} The phase--shift may become tangent to $n\pi$ either from below or from above.
\end{proposition}

\skd

{\bf (b)}\ An extension of Levinson's theorem has been proved, in the case of nonlocal
separable potentials, in Ref. \cite{Martin1};
see also Ref. \cite[II]{Bertero}, where Levinson's theorem has been proved, 
in the case $\ell=0$,
for a class of potentials very close to that considered here. The extension to any integral value of $\ell$
is straightforward.
One can then state that the total variation of the phase--shift in the interval $0\leqslant k<+\infty$
is given by:
\beq
\delta_\ell(0)-\delta_\ell(\infty)=\pi(N_\ell + N'_\ell) \qquad (\ell~{\rm integer}),
\label{f:104}
\eeq
where $N_\ell$ is the number of negative 
energy bound states and $N'_\ell$ is the number
of positive energy bound states; for simplicity, 
one assumes that all the bound states are represented 
by simple poles, and that there are no bound states at $k=0$.

\subsection{Partial wave expansion of the scattering amplitude}
\label{subse:partial-wave}

The absence of bound states for $\ell $ sufficiently large 
was already noted earlier as being a corollary of property (c) of
$\sigma_\ell(k;g)$ stated in Theorem \ref{theorem:5}. 
It can also be presented in a simpler and more precise
way as follows.
In view of the complement of theorem 6, the norm $\|L_\ell(k)\|$
of $L_\ell(k)$, considered as a bounded operator on
$X_{\wea}$, satisfies the  majorization \eqref{wf:1}
for all $k$ in $\Pi_\alpha$:
\beq
\|L_\ell(k)\| \leqslant
\|L_\ell(k)\|_\mathrm{HS} \leqslant
\frac{\sqrt{\pi}}{2} \, A_\varepsilon^2 \, \frac{C(V)}{2\ell+1}.
\label{k:2}
\eeq
Let $\ell=\ell_0(g)$ be the smallest integer such that the r.h.s. of
\eqref{k:2} is majorized by $1/|g|$ and let 
$\kappa_g \doteq |g|\frac{\sqrt{\pi}}{2} \, A_\varepsilon^2 \, \frac{C(V)}{2\ell_0(g)+1} < 1$.
Then for all integers $\ell $ such that $\ell \geqslant \ell_0(g)$,
the operator $gL_\ell(k)$ is a contraction in  
$X_{\wea}$ and therefore 
\emphsl{there exist no bound state or spurious bound state and even no resonance} 
in the strip $-\alpha\leqslant \Imag k <0$.
Next we prove the following proposition.

\skd

\begin{proposition}
\label{pro:k1}
For $k$ real and positive, the asymptotic behavior 
of the partial scattering amplitude $T_\ell(k;g)$,
for large values of $\ell$, is governed by the following majorization  
(including a suitable constant $\widehat{C}(k,g)$):
\beq
|T_\ell(k;g)| \leqslant \frac{\widehat{C}(k,g)}{(2\ell +1)^\frac{1}{2}} 
\frac{\Gamma (\ell +1)}{\Gamma(\ell +\frac{3}{2})} \,e^{-\beta(k)\ell} 
= O \left(\ell^{-1}\,e^{-\beta(k)\ell}\right),
\ {\rm with}\ \,\cosh\beta(k)=1+\frac{2\alpha^2}{k^2}.
\label{f:117}
\eeq
\end{proposition}

\begin{proof}
In view of Theorem \ref{theorem:3'}, we know that, for
$k\in \Omega_{\alpha,\ell}(g)$, $v_\ell(k,g;\cdot)$ belongs to 
$X_{\wea}$, and we therefore have in view of
\eqref{f:81}:
\beq
\frac{|T_\ell(k;g)|}{|g|} =
\left|\int_0^{+\infty}R j_\ell(kR)v_\ell(k,g;R)\,\rmd R\right| 
\leqslant
\frac{\|v_\ell(k,g;\cdot)\|_{\wea}}{|k|}
\ \left\|{k\cdot} \,\, j_\ell(k\cdot)\right\|^*_{\wea}.
\label{k:3}
\eeq
But, for $\ell\geqslant \ell_0(g)$, 
$\Omega_{\alpha,\ell}(g)$ contains the whole real set $k > 0$, 
and from Eq. \eqref{f:27a} and bound \eqref{k:2} we can also write:
\beq
\left\|v_\ell(k,g;\cdot)\right\|_{\wea} \leqslant 
\|v_{\ell,0}(k;\cdot)\|_{\wea}
+\frac{\sqrt{\pi}}{2} \,
\frac{|g|A_\varepsilon^2 C(V)}{2\ell+1} \,
\|v_\ell(k,g;\cdot)\|_{\wea},
\label{k:6-ex}
\eeq
which yields: 
\beq
\|v_\ell(k,g;\cdot)\|_{\wea}
\leqslant (1- \kappa_g)^{-1} \, 
\|v_{\ell,0}(k;\cdot)\|_{\wea}.
\label{k:7}
\eeq
So, for all $k>0$ and $\ell\geqslant\ell_0(g)$, we have 
(in view of \eqref{k:3} and \eqref{k:7}):
\beq
|T_\ell(k;g)| \leqslant \frac{|g|}{k} (1-\kappa_g)^{-1} 
\|v_{\ell,0}(k;\cdot)\|_{\wea}
\ \left\|{k\cdot} \,\, j_\ell(k\cdot)\right\|^*_{\wea},
\label{k:3'}
\eeq
and therefore, by applying Lemma \ref{lemma:v}  
(formula \eqref{f:v3}) for $w=w^{(\varepsilon)}$: 
\beq
|T_\ell(k;g)| \leqslant \frac{|g|}{k} (1-\kappa_g)^{-1} 
\frac{C(V)}{(2\ell +1)^{1/2}}
\left[\left\|{k\cdot} \ j_\ell(k\cdot)\right\|^*_{\wea}\right]^2.
\label{k:3"}
\eeq
We shall now introduce a majorization of
$\left\|{k\cdot} \ j_\ell(k\cdot)\right\|^*_{\wea}$,
which is \emph{uniform with respect to  all values of $\ell$ };
for this purpose, we shall use the following equality (valid for all $\ell$ and $k>0$):
\beq
\int_0^{+\infty}\!\! e^{-2\alpha R}J_{\ell+1/2}^{\,2}(kR)\,\rmd R
=\frac{1}{\pi k}\,Q_\ell\left(1+\frac{2\alpha^2}{k^2}\right),
\label{k:4}
\eeq
$Q_\ell(\cdot)$ denoting the second kind Legendre function (see \cite{Watson}).
In fact, we can write, by Eq. \eqref{a:1b'}:
\beq
\begin{split}
& \left[\left\|{k\cdot} \ j_\ell(k\cdot)\right\|^*_{\wea}\right]^2
\doteq 
\int_0^{+\infty} \!\! e^{-2\alpha R}\,\frac{|kR j_\ell(kR)|^2}
{R^{1-\varepsilon}(1+R)^{1+2\varepsilon}}\,\rmd R \\
&\qquad\leqslant\frac{\pi k}{2}
\int_0^{+\infty} \!\! e^{-2\alpha R} \, |J_{\ell+1/2}(kR)|^2 \,
\frac{R^\varepsilon}{(1+R)^{1+2\varepsilon}}\,\rmd R,
\end{split}
\label{k:5}
\eeq
and therefore, by making use of \eqref{k:4}
(since $J_{\ell + {1/2}}(kR)$ is real--valued for $k>0$):
\beq
\left[\left\|{k\cdot} \ j_\ell(k\cdot)\right\|^*_{\wea}\right]^2
\leqslant \frac{1}{2} \, Q_\ell\left(1+\frac{2\alpha^2}{k^2}\right).
\label{k:6}
\eeq
In view of the latter, majorization \eqref{k:3"} then yields:
\beq
|T_\ell(k,g)|\leqslant\frac{|g|}{2 k}(1- \kappa_g)^{-1} \,
\frac{C(V)}{(2\ell+1)^{1/2}} \,
Q_\ell\left(1+\frac{2\alpha^2}{k^2}\right).
\label{k:8}
\eeq
We now use the following uniform bound on the second--kind Legendre 
functions \cite{Bateman}\cite[III]{Bertero}:
\beq
|Q_\ell(\cosh\beta)|\leqslant \sqrt{\pi} \,
\frac{\Gamma(\ell+1)}{\Gamma(\ell+\frac{3}{2})} \, e^{-\beta(\ell+1)} \,
(1-e^{-2\beta})^{-1/2} \qquad (\beta>0).
\label{k:10}
\eeq
By taking the latter into account at the r.h.s. of Eq. \eqref{k:8}, 
one then readily obtains the majorization \eqref{f:117},
with the following definition of the constant $\widehat{C}(k,g)$:
\begin{equation*}
\widehat{C}(k,g)= \frac{\sqrt \pi \,|g|}{2k}(1-\kappa_g)^{-1}\  C(V) \,
\frac{e^{-\beta(k)}}{\left(1-e^{-2\beta(k)}\right)^\frac{1}{2}}.
\end{equation*}
\end{proof}

We can now introduce the so--called partial waves:
\beq
a_\ell(k;g)=\frac{T_\ell(k;g)}{k}=\frac{e^{2\rmi\delta_\ell(k;g)}-1}{2\rmi k}.
\label{f:118}
\eeq
Then the \emphsl{total scattering amplitude}
$F(k,\cos\theta;g)$ can be expanded as a Fourier--Legendre
series in terms of the $a_\ell(k;g)$, 
usually called the \emphsl{partial wave expansion}:
\beq
F(k,\cos\theta;g)=\sum_{\ell=0}^\infty(2\ell+1)a_\ell(k;g)P_\ell(\cos\theta),
\label{f:119}
\eeq
where the $P_\ell(\cdot)$ are the Legendre polynomials,
and $\theta$ is the scattering angle
in the center of mass system.
We have thus obtained results very similar to those
which hold for the class of local Yukawian potentials. In particular,
in view of \eqref{f:117}, \emphsl{expansion \eqref{f:119}
converges in an ellipse $\cE_{\beta(k)}$ with foci $\pm 1$ and half major axis
$\cosh\beta(k)=(1+2\alpha^2/k^2)$. This convergent expansion defines the total
scattering amplitude $F(k,\cos\theta;g)$ as a holomorphic function of 
$\cos \theta$ in the ellipse $\cE_{\beta(k)}$ for each positive value of $k$.}

\vskip 0.2cm

\noindent
Finally, we study the behaviour 
of the partial scattering amplitude $T_\ell(k;g)$
for small and large values of $k$. We prove the following proposition.

\skd

\begin{proposition}
\label{pro:k2}
{\rm (i)} For any real $g$ and any integer $\ell$ ($\ell \geqslant 0)$, 
the asymptotic behaviour of
the partial scattering amplitude $T_\ell(k;g)$ for $k$ tending to
infinity in $\Omega_\alpha$ is such that 
\beq
T_\ell(k;g) = \frac{1}{k} \,\widehat{T}_\ell(k;g),
\label{f:120as}
\eeq
where $\widehat{T}_\ell$ is bounded in $k$ at infinity.

{\rm (ii)} Considering the generic case of values 
of $g$ and $\ell$ such that  $\sigma_\ell(0;g) \neq 0$,  
the corresponding threshold behaviour (for $k\to 0$) 
of the partial scattering amplitude $T_\ell(k;g)$ is such that
\beq
\label{f:120}
T_\ell(k;g) = k^{2\ell+1} \, \widehat{T}^{(\mathrm{loc})}_\ell(k;g),
\eeq
where $\widehat{T}^{(\mathrm{loc})}_\ell$ is bounded in $k$ in a complex neighbourhood of $k=0$.
\end{proposition}

\begin{proof}
In view of Eqs. \eqref{f:82}, \eqref{f:83},  
and of the corresponding norm inequalities that have been used in
the proof of Theorem \ref{theorem:8}, we can write:
\beq
|T_\ell (k;g)| \leqslant \frac{|g|}{|k|} 
\left[1 +\frac{|g|}{|\sigma_\ell(k;g)|} \, \|N_\ell(k;g)\|_\mathrm{HS}\right] 
\ \left\|{k\cdot} \ j_\ell(k\cdot)\right\|^*_{\wea}
\ \|v_{\ell,0}(k;\cdot)\|_{\wea},
\label{k:14}
\eeq
and therefore, in view of Lemma \ref{lemma:v} (Eq. \eqref{f:v3}) and of \eqref{f:Psi'}:
\beq
|T_\ell (k;g)| \leqslant \frac{|g|}{|k|} 
\left[1 + \frac{\Psi(|g|\,\|L_\ell(k)\|_\mathrm{HS})}{|\sigma_\ell(k;g)|}\right] \
\frac{{\bf C}(V)}{(2\ell +1)^{1/2}} \,
\left[\left\|{k\cdot} \ j_\ell(k\cdot)\right\|^*_{\wea}\right]^2.
\label{k:14'}
\eeq

\vskip 0.1cm

\noindent
(i) For all $k\in {\overline \Omega}_\alpha$,
we can apply the global bounds \eqref{wf:1} and  
\eqref{f:j1} with $A_{\we}(k) = A_\varepsilon$
(see Eq. \eqref{f:e0}), which allow one to replace 
the majorization \eqref{k:14'} by
\beq
|T_\ell (k;g)| \leqslant \frac{|g|}{|k|} 
\left[1 + 
\frac{\Psi\left(|k|^{-1}\,|g|\frac{1+\ell\pi}{(2\ell+1)^{1/2}}\,{\bf C}(V) A_\varepsilon^2\right)}
{|\sigma_\ell(k;g)|}\right]
\
\frac{{\bf C}(V)}{(2\ell +1)^{1/2}} \, c_\ell^2 \, A_\varepsilon^2.
\label{k:14"}
\eeq
Now, we know that, for $k$ tending to infinity, 
one has $\Psi(c/|k|) = O(1/|k|)$ (in view of \eqref{f:Psi} and, moreover, 
$|\sigma_\ell(k;g)|$ tends uniformly to $1$ (see Theorem \ref{theorem:5} (c)).
Then, majorization \eqref{k:14"} implies:
\beq
|T_\ell (k;g)| \leqslant \frac{|g|}{|k|}\left[1 + O\left(\frac{1}{|k|}\right)\right] 
\ \frac{{\bf C}(V)}{(2\ell +1)^{1/2}} \, c_\ell^2 \, A_\varepsilon^2,
\label{k:13}
\eeq
which directly yields the asymptotic behaviour \eqref{f:120as}.

\vskip 0.1cm

\noindent
(ii) We now consider a (real) value $g_0$ of $g$ such that 
$\sigma_\ell(0;g_0) \neq 0$; moreover, we
know from Theorem \ref{theorem:5} that $\sigma_\ell(k;g_0)$ 
is holomorphic as a function of
$k$ in $\Pi_\alpha $. It follows that 
there exists a neighborhood of $k=0$,
say $S_\eta=\{k \,:\, |k|\leqslant\eta\}$, and a constant $M_\ell$ such that
$|\sigma_\ell(k;g_0)|^{-1}\leqslant M_\ell$ for all $k$ in $S_\eta$.

\vskip 0.2cm

\noindent
Here we shall make use of a majorization on the
norm of the function $kR j_\ell(kR)$ in $X^*_{\wea}$,
which is different from those obtained in \eqref{f:j1} and \eqref{k:5}. 
For this purpose, we shall use the bound \eqref{a:0}, which 
yields for all $R\geqslant 0$ and $k\in S_\eta$:
$|kR j_\ell(kR)|\leqslant c_\ell \,|k|^{\ell +1}\,R^{\ell +1}\,e^{\eta R}$.
If $\eta$ is chosen such that $\eta<\alpha$,  
the latter implies the following majorization:
\beq
\left\|{k\cdot} \ j_\ell(k\cdot)\right\|^*_{\wea}
\leqslant c_\ell \, |k|^{\ell+1}
\left[\int_0^{+\infty}\frac{R^{2\ell+1+\varepsilon} \, e^{-2(\alpha-\eta)R}}
{(1+R)^{1+2\varepsilon}}\right]^{1/2}
= \widehat{c}_\ell(\alpha,\eta,\varepsilon) \, |k|^{\ell +1},
\label{k:12'}
\eeq
where $\widehat{c}_\ell$ is a new constant (also depending of $\alpha, \eta, \varepsilon$).

By taking \eqref{k:12'} into account and also the $k$--independent option of 
the global bound \eqref{wf:1} on $\|L_\ell (k)\|_\mathrm{HS}$, the majorization 
\eqref{k:14'} can then be replaced by the following one, valid for all $k$ in $S_\eta$:
\beq
|T_\ell (k;g)| \leqslant \frac{|g|}{|k|} 
\left[1 + M_\ell \, 
\Psi\left(|g|\,\frac{\pi^{1/2}}{2}\frac{{\bf C}(V)}{2\ell +1} \, A_\varepsilon^2\right)\right] \
\frac{{\bf C}(V)}{(2\ell +1)^{1/2}}\,\widehat{c}_\ell^{\,2} \, |k|^{2\ell +2},
\label{k:12"}
\eeq
which directly yields the threshold behaviour \eqref{f:120}.
\end{proof}

\skd

\begin{remark}
In view of formula \eqref{f:84}, it follows from the previous 
propositions that $S_\ell(k)\to 1$ both for $k$ tending to zero and 
for $|k|$ tending to infinity. \\
Condition (ii) also implies that all the partial waves
$a_\ell(k)= {T_\ell(k,g)/k}$ are finite at $k=0$.
\label{rem:3}
\end{remark}

\sku

Finally, we can also re--express 
the statement
(i) of Proposition \ref{pro:k2} 
in terms of the physically interesting
quantities $\delta_\ell$ (namely, the phase--shifts
defined by Eq. \eqref{f:85}) as the following 

\skd

\begin{corollary}
\label{cor:phsh}
For any real $g$ and any integer $\ell$ ($\ell \geqslant 0)$,
the limit of the phase--shift $\delta_\ell(k;g)$ for $k$ tending to
infinity in $\Omega_\alpha$ exists and is equal to zero (mod. $\pi$).
\end{corollary}

\section{Nonlocal potentials with interpolation in the CAM--plane and analyticity
properties of the partial scattering amplitude}
\label{se:interpolation}

\subsection{Interpolation of the partial potentials $\boldsymbol{V_\ell}$ 
in the CAM--plane. Classes $\boldsymbol{\cN^{\,\gamma}_{\wea}}$
of nonlocal potentials}
\label{subse:interpolation}

Let us recall that the complex angular momentum (CAM) theory in potential scattering
has been rigorously stated for potentials which do not depend on the angular
momentum, namely for a large class of local potentials, including the Yukawian potentials. 
For this class of potentials one can show the existence of a distinguished
\emphsl{meromorphic interpolation} $T(\lambda,k)$ of the sequence of partial
waves $T_\ell(k)$ in the half--plane $\C^+_{-\frac{1}{2}}$ ($\Real\lambda >-\frac{1}{2}$),
which allows for performing a \emph{Watson resummation} of the partial wave expansion.
Concerning the notion of \emph{analytic interpolation of a sequence}
$\{f_\ell\}_{\ell=0}^\infty$ in the half--plane $\C^+_{-\frac{1}{2}}$,
namely the existence of a function $F(\lambda)$, holomorphic in  
$\C^+_{-\frac{1}{2}}$ and  such that $F(\ell) =f_\ell$ for all integers $\ell \geqslant 0$,
it is worthwhile to recall the important notion of \emphsl{Carlsonian interpolation}.
An analytic interpolation $F$ is said to be Carlsonian if it satisfies a global bound of the 
form $|F(\lambda)|\leqslant C e^{(\pi-\varepsilon)|\lambda|}$, with $\varepsilon >0$, 
in $\C^+_{-\frac{1}{2}}$ (which of course requires that the given sequence $\{f_\ell\}$ itself
satisfies such inequalities). Carlson's theorem \cite{Boas} then asserts that
$F(\lambda)$ is the \emphsl{unique} Carlsonian interpolation of the sequence
$\{f_\ell\}$; in fact, any other analytic interpolation 
of this sequence, such as the one obtained by adding to $F$
the non--Carlsonian function $\sin \pi \lambda$, is a non--Carlsonian interpolation. 
It is a remarkable feature of the theory of Yukawian potentials (and of a more general class 
of local potentials $V(R)$ enjoying suitable analyticity properties in $R$) that
all the relevant CAM analytic interpolations of the angular momentum formalism
can be performed in the sense of Carlsonian interpolations. 
In this connection, let us also mention that the introduction of so--called ``exchange potentials"
leads one to split the set of partial waves into two separate subsets, namely
the even--$\ell$ partial waves,
and the odd--$\ell$ partial waves.
When it is attempted to perform the CAM interpolation of the partial waves,
one is faced with the problem of handling the factor $(-1)^\ell = \cos\pi\ell$; indeed,
the function $\cos\pi\lambda$, whose restriction to integers gives $(-1)^\ell$,
is not a Carlsonian function. Therefore, in that case, one is led 
to perform \emphsl{Carlsonian analytic interpolations} of the two sets (even--$\ell$ and
odd--$\ell$) separately. Let us also note that such a separation
was also found to occur at a more fundamental level 
in our approach of complex angular momentum analysis in the general framework of 
relativistic Quantum Field Theory, 
with the notion of Bethe--Salpeter kernel playing the 
role of a generalized nonlocal potential (see \cite{Bros3}).
This separation is also exhibited in a striking way when one considers the scattering
of identical particles, since in view of the symmetrization (resp.,
antisymmetrization) properties of the boson (resp., fermion) quantum description,
one can only obtain an analytic interpolation of even--$\ell$ waves 
with all the odd--$\ell$ waves identically equal to zero for  the case of bosons, 
and the converse for the case of fermions. 

\vskip 0.1cm

Here we wish to exhibit the mechanism of analytic interpolation
of the partial waves in the framework of nonlocal rotationally invariant
potentials belonging to appropriate subclasses of the classes
considered previously in Section \ref{se:rotationally}: in view of their expansion in
partial potentials $V_\ell$, they can also be called (with a small abuse of language)
``angular--momentum--dependent potentials".
The first step in order to introduce such a relevant subclass 
of nonlocal potentials consists in finding suitable conditions 
which allow one to perform the 
\emphsl{analytic interpolation of the partial potentials} $V_\ell(R,R')$ 
(see Eqs. \eqref{f:20} and \eqref{f:21}) from integral $\ell$--values to
complex $\lambda$--values in a specific domain of the CAM--plane, containing
the real positive semi--axis. For the sake of simplicity and without loss of
generality, we only consider the case of an interpolation with respect to the set of
\emphsl{all integers} $\ell$; the case of separate interpolations with respect to the subsets
of even integers and of odd integers can be treated similarly.

\vskip 0.1cm

As we shall now explain it below, the existence of a Carlsonian 
analytic interpolation in $\lambda$ for the sequence 
$\{V_\ell(R,R')\}_{\ell=0}^\infty$ can be established either 
from appropriate bounds to be satisfied by multiple differences 
of the sequential elements $V_\ell$ themselves, or
from appropriate analyticity and increase properties of the complete potential
$V({\bR,\bR'})=V(R,R';\cos\eta)$ with respect to the complexified angular variable 
$\cos\eta$.

\vskip 0.2cm

\noindent
{\bf (i)\  Hausdorff--type bounds on the 
$\boldsymbol{V_\ell}$:} \,
Let us suppose that the partial potentials $V_\ell(R,R')$ ($\ell=0,1,2,\ldots$)
form, for arbitrary values of $R$ and $R'$, a sequence of real numbers which are
constrained as follows. Denote by $\Delta$ the difference operator:
$\Delta V_\ell(\cdot,\cdot) = V_{\ell+1}(\cdot,\cdot)-V_\ell(\cdot,\cdot)$;
we thus have:
\beq
\label{g:1}
\Delta^k V_\ell(\cdot,\cdot) =
\underbrace{\Delta \times \Delta \times \cdots \times \Delta}_{k~{\rm times}} \, V_\ell(\cdot,\cdot) =
\sum_{m=0}^k (-1)^m \binom{k}{m} \, V_{\ell+k-m}(\cdot,\cdot), \nonumber
\eeq
(for every $k \geqslant 0$); $\Delta^0$ is the identity operator, by definition.
Now, let us suppose that the sequence $\{V_\ell(\cdot,\cdot)\}_{\ell=0}^\infty$ is
constrained by the following Hausdorff--type bound:
\beq
\label{g:2}
(\ell+1)^{(1+\epsilon)} \sum_{i=0}^\ell \binom{\ell}{i}^{(2+\epsilon)}
\left|\Delta^i V_{(\ell-i)}(\cdot,\cdot)\right|^{(2+\epsilon)} < M 
\qquad (\ell=0,1,2,\ldots;\,\epsilon>0),
\eeq
where $M$ and $\varepsilon$ are given positive constants ($\varepsilon$ being as small as wanted).
It can be proved \cite{Widder} that condition \eqref{g:2}
is necessary and sufficient to represent the sequence $\{V_\ell(\cdot,\cdot)\}_{\ell=0}^\infty$ as:
\beq
\label{g:3}
V_\ell(\cdot,\cdot) = \int_0^1 x^\ell\,u(x;\cdot,\cdot)\,\rmd x \qquad (\ell=0,1,2,\ldots),
\eeq
where $u(x;\cdot,\cdot)$ belongs to $L^{(2+\epsilon)}(0,1)$.
Let us now put $x = e^{-v}$ into the integral on the r.h.s of \eqref{g:3}; we thus obtain:
\beq
\label{g:4}
V_\ell(\cdot,\cdot) = \int_0^{+\infty} e^{-\ell v}\,e^{-v}\,u(e^{-v};\cdot,\cdot)\,\rmd v
\qquad (\ell=0,1,2,\ldots). \nonumber
\eeq
The numbers $V_\ell(\cdot,\cdot)$ can be regarded as the restriction to the integers
of the following Laplace transform:
\beq
\label{g:5}
V(\lambda;\cdot,\cdot) = \int_0^{+\infty}
e^{-(\lambda+1/2)v}\,e^{-v/2}\,u(e^{-v};\cdot,\cdot)\,\rmd v
\qquad \left(\lambda\in\C,\Real\lambda\geqslant -\frac{1}{2}\right).
\eeq
In the latter, the function 
$v \mapsto U(v;\cdot,\cdot)\doteq e^{-v/2}\,u(e^{-v};\cdot,\cdot)$ can be shown 
to be in $L^1(0,+\infty)$ as a consequence of the inclusion
$u(x;\cdot,\cdot)\in L^{(2+\epsilon)}(0,1)$ (whatever small $\varepsilon$ is).
Therefore, for all $R,R'$, the Laplace transform $V(\lambda;R,R')$
of $U(v,R,R')$, defined by \eqref{g:5}, is holomorphic in the half--plane
$\Real\lambda > -\frac{1}{2}$. In view of the Riemann--Lebesgue theorem,
completed by the analysis of \cite[p. 125]{Hoffman} (and also by taking the 
Carlson theorem into account), we can state

\skd

\begin{proposition} 
\label{pro:25}
Let the partial potentials $V_\ell(R,R')$ ($\ell=0,1,2,\ldots$) 
of a given rotationally invariant nonlocal potential $V(\bR,\bR')$ satisfy
Hausdorff--type conditions of the form \eqref{g:2}.
Then there exists a unique Carlsonian interpolation $V(\lambda;R,R')$
of the sequence $\{V_\ell(R,R')\}_{\ell=0}^\infty$, which is analytic in
$\C^+_{-\frac{1}{2}}$, continuous at $\Real \lambda = -\frac{1}{2}$ and tends
uniformly to zero for $|\lambda|\to\infty$ inside any fixed
half--plane $\C^+_{-\frac{1}{2}+\delta}$, with $\delta>0$.
\end{proposition}

\skt

\noindent
{\bf (ii)\   Analyticity and behaviour at infinity 
of $\boldsymbol{V(R,R';\cos\eta)}$ in the
complex $\boldsymbol{\cos\eta}$--plane:} \,
Since every rotationally--invariant nonlocal potential $V(\bR,\bR')$ defines
for every $(R,R')$ an invariant kernel on the sphere, namely a function
of the scalar product $\frac{\bR}{R}\cdot \frac{\bR'}{R'} \doteq \cos\eta$, whose Fourier--Legendre coefficients
$V_\ell(R,R')$ are given by \eqref{f:21}, we can apply the results of  
\cite{Bros4} (see Theorem 1 of the latter) concerning
the Fourier--Laplace--type transformation of holomorphic invariant kernels on
cut--domains of the complexified sphere. These results can be summarized in the following

\skd

\begin{proposition} 
\label{pro:26}
Let $V(R,R';\zeta)$, initially defined for $-1\leqslant\zeta\equiv\cos\eta\leqslant +1$, 
satisfy (for each $(R,R')$ fixed) the following conditions:
\begin{itemize}
\addtolength{\itemsep}{0.1cm}
\item[\rm (a)] it admits an analytic continuation $\widehat{V}(R,R';\zeta)$ with respect to $\zeta$
in a given cut--plane $\Pi_\gamma = \C\setminus\{\zeta\in\R\,:\,\zeta\geqslant\gamma\}$, 
where $\gamma\geqslant 0$;
\item[\rm (b)] there exists a real number $m$, with $m>-1$, 
and a function $g(v)$ in $L^1(\R^+)$ such that for $\zeta=\cos(u+iv)$ 
varying in the closure of the cut--plane $\Pi_\gamma$ (represented by the set 
$0\leqslant u\leqslant 2\pi,\, v\geqslant 0$),
there holds a uniform bound of the following form
$|\widehat{V}(R,R';\cos(u+iv))|\leqslant (RR')^{-1}\, V_*(R,R') \, g(v) \, e^{mv}$.
\end{itemize}
Then there exists a function 
$V(\lambda; R,R')$, holomorphic in the complex half--plane 
$\C^+_m$ such that for all integers  $\ell >m$ one has:
$V(\ell; R,R')=V_\ell(R,R')$. Moreover, 
this interpolation of the sequence 
$\{V_\ell(R,R')\}_{\ell=0}^\infty$ is Carlsonian: in fact, it satisfies  
a global bound of the following form in
$\C^+_m$ (with a suitable constant $\cK$):
\beq
|V(\lambda; R,R')|\leqslant \cK V_*(R,R')\,  e^{-\gamma\Real \lambda}.
\label{g:5'}
\eeq
\end{proposition}

\vskip 0.2cm

\noindent
It also results from Theorem 3 of \cite{Bros4} that, conversely,
if the sequence $\{V_\ell(R,R')\}_{\ell=0}^\infty$ admits a 
Carlsonian interpolation $V(\lambda;R,R')$ in $\C^+_m$
satisfying a global bound of the previous type, then the potential
$V(R,R';\cos\eta)$ admits (for each $R,R'$) an analytic continuation in
the cut--plane $\Pi_\gamma$ of the complex variable $\zeta=\cos(u+iv)$,
which behaves at infinity like $e^{mv}$ (up to a power of $v$).

\skd

\begin{remark}
In the following, we shall consider a framework in which the
majorizing potential $V_*(R,R')$ satisfies an appropriate $L^2$--norm. 
The previous conditions of analyticity 
of ${\widehat V}(R,R';\cos\eta)$ in $\cos \eta$ and of 
$V(\lambda; R,R')$ in $\lambda$ must then be understood to hold
\emph{for a.e. $(R,R')$, with respect to the chosen $L^2$--norm.} 
\label{rem:4}
\end{remark}

\skd

In view of these results, one can exhibit examples of such nonlocal potentials, 
namely those of the form 
$\widehat{V}(R,R';\cos\eta) = (RR')^{-1}\,V_*(R,R') \, \rmf(\cos(u+iv))$,
where $\rmf$ denotes a holomorphic function in the cut--plane $\Pi_\gamma$, 
which behaves at infinity like $e^{mv}$. As a basic explicit example of
this type, for which $m=-1$, one can take $\rmf(\cos\eta)$ of the following form:
$\rmf(\cos\eta) = (e^\gamma - \cos\eta)^{-1}$, $(\gamma > 0)$.
Then, in view of formula \eqref{f:21}, one has:
$V_\ell(R,R')= 2\pi V_*(R,R')\, \rmf_\ell$, where $\rmf_\ell$ is expressed in terms of the
second--kind Legendre function $Q_\ell$
(see \cite[Vol. 2, p. 316, formula (17)]{Bateman2}) by the formula:
\beq
\label{g:6}
\rmf_\ell = \int_{-1}^{+1}\frac{P_\ell(\cos\eta)}{e^\gamma-\cos\eta}\,\rmd(\cos\eta)
= 2 \, Q_\ell(e^\gamma).
\eeq
It is known that the function $\lambda \mapsto Q_\lambda(e^\gamma)$ is 
holomorphic in the half--plane $\Real\lambda > -1$,
and tends to zero uniformly as $e^{-\gamma(\lambda+1)}$
for $|\lambda|\to\infty$ in the 
half--plane $\Real\lambda\geqslant -\frac{1}{2}$. 
It therefore represents the unique Carlsonian
interpolation of the sequence $\{Q_\ell(e^\gamma)\}_{\ell=0}^\infty$.
This example therefore illustrates in a typical way the previous proposition
(including a bound of the form \eqref{g:5'}). 

\noindent
Another example of a function $V(\lambda;R,R')$ which satisfies \eqref{g:5'}
is obtained in the previous approach (i) (Hausdorff--type bounds) by imposing  
the support condition $u(e^{-v};\cdot,\cdot)=0$ for $v\in[0,\gamma)$
into the Laplace representation \eqref{g:5} of $V(\lambda;R,R')$. 

\vskip 0.2cm 

\noindent
{\bf Classes $\boldsymbol{\cN^{\,\gamma}_{\wea}}$ of nonlocal potentials:} \,
For the purpose of the present Section \ref{se:interpolation},
we shall introduce subclasses $\cN^{\,\gamma}_{\wea}$,
of the previously considered classes $\cN_{\wea}$ of rotationally--invariant 
nonlocal potentials $V(\bR,\bR')$ (see Subsection \ref{subse:properties-L}), 
by imposing the following additional 

\vskip 0.1cm

\noindent
\textbf{Assumption:} the Fourier--Legendre coefficients $V_\ell(R,R')$ of 
$V(\bR,\bR')$ admit a Carlsonian interpolation $V(\lambda;R,R')$ in the half--plane 
$\Real\lambda > -\frac{1}{2}$ satisfying a global bound of the form \eqref{g:5'}; 
the given number $\gamma$ is supposed to be positive and 
the analyticity property of $V$ with respect to $\lambda$ is supposed to hold 
for a.e. $(R,R')$, according to the choice of the majorizing potential $V_*(R,R')$ 
(see our previous remark and the Hilbertian requirement on $V_*$ specified below). 
For brevity, we shall call these potentials \emphsl{Carlsonian potentials with 
CAM--interpolation $V(\lambda;R,R')$ and rate of decrease $\gamma$.}

\vskip 0.1cm

\noindent
Regarding the Hilbert space on which the potentials $V(\lambda;R,R') $ are acting,
we thus consider the space $X_{w,\alpha}$ defined in \eqref{f:28}  
with the choice \eqref{wf:2} for the weight--function $w$, namely:
$w^{(\varepsilon)}(R)=R^{(1-\varepsilon)}(1+R)^{(1+2\varepsilon)}$ ($\varepsilon>0$), 
which already played the main role in Section \ref{se:rotationally}.

Next, we wish to ensure the following condition on the potential
(note that we use the same notation as in Section \ref{se:rotationally},
formula \eqref{f:30} for the constant $C(\cdot)$): for all $\lambda$
with $\Real\lambda >-\frac{1}{2}$: 
\beq
\begin{split}
& C(V(\lambda;..)) \doteq \left\{\int_0^{+\infty} 
R^{(1-\varepsilon)}(1+R)^{(1+2\varepsilon)}\,e^{2\alpha R}\,\rmd R\right. \\
&\quad \left.\times\int_0^{+\infty} \!\!
(R')^{(1-\varepsilon)}(1+R')^{(1+2\varepsilon)}\,e^{2\alpha R'}
\left|V(\lambda;R,R')\right|^2\,\rmd R'\right\}^{1/2}<\infty,
\end{split}
\label{AG.7}
\eeq
Since $V(\lambda;R,R') $ is assumed to satisfy a bound of the form \eqref{g:5'}, it is 
natural to impose the following condition on the kernel $V_*(R,R')$: 
\beq
\begin{split}
& C(V_*) \doteq
\left\{\int_0^{+\infty} \,\, 
R^{(1-\varepsilon)}(1+R)^{(1+2\varepsilon)}\,e^{2\alpha R}\,\rmd R\right. \\
&\quad \left.\times\int_0^{+\infty} \!\! 
(R')^{(1-\varepsilon)}(1+R')^{(1+2\varepsilon)} \, e^{2\alpha R'}
\left|V_*(R,R')\right|^2\,\rmd R'\right\}^{1/2} <\infty,
\end{split}
\label{AG.9}
\eeq
or (as in \eqref{f:30'} and \eqref{f:31}) in terms of the function
\beq
V_*^{\left(\we\right)}(R) \doteq 
\left(\int_0^{+\infty} \we(R') \, e^{2\alpha R'} \, V_*^2(R,R')
\,\rmd R'\right)^{1/2},
\label{AG.6'}
\eeq
which belongs to $X_{\wea}$, 
\beq
C(V_*) = \left\|V_*^{\left(\we\right)}\right\|_{\wea} < \infty.
\label{AG.8}
\eeq 
In view of \eqref{g:5'}, one therefore has the following global majorization:
\beq
C(V(\lambda; ..)) \leqslant \cK \, e^{-\gamma\Real\lambda} \, C(V_*) 
\qquad \left(\Real\lambda > -\frac{1}{2}\right).
\label{AG.10}
\eeq
The set of conditions \eqref{g:5'}, \eqref{AG.7}, \eqref{AG.9}, and \eqref{AG.10} characterize the class 
$\cN^{\,\gamma}_{\wea}$ of nonlocal potentials for which we are going to study the analyticity properties of 
the partial scattering amplitudes. 

\vskip 0.3cm

\noindent
\textbf{Simple examples of potentials in
$\boldsymbol{\cN^{\,\gamma}_{\wea}}$:}
The simplest examples which can be seen to satisfy the previous conditions
are obtained by choosing $V_*$ as a kernel of rank one $V_*= \mathrm{v}\otimes\mathrm{v}$ 
and $V(\lambda;R,R')= \mathrm{v}(R)\,\mathrm{v}(R')\,e^{-\gamma \lambda}$
(or $\mathrm{v}(R)\,\mathrm{v}(R')\,Q_\lambda(e^\gamma)$). 
Concerning $\mathrm{v}$, one saturates bound \eqref{AG.9}
with a ``Yukawa--type" function such as $\mathrm{v}(R)= (1+R)^{-\frac{3}{2}-\varepsilon}\,e^{-\alpha R}$. 

\vskip 0.5cm

\subsection{Analyticity and boundedness properties in complex $\boldsymbol{(\lambda,k)}$--space of the functions 
$\boldsymbol{v_0(\lambda,k;\cdot)}$, the operators $\boldsymbol{L(\lambda,k)}$,
and the resolvent $\boldsymbol{R(\lambda,k;g)}$}
\label{subse:aggiunte4.2}

Our assumptions on the CAM interpolation $V(\lambda;R,R')$ of the potentials 
$V_\ell$ will allow us to introduce corresponding CAM interpolations 
$v_0(\lambda,k;\cdot)$, $L(\lambda,k)$, and $R^{(\mathrm{tr})}(\lambda,k;g)$ 
for the respective sequences $\{v_{\ell,0}(k;\cdot)\}$, $\{L_\ell(k)\}$, and 
$\{R^{(\mathrm{tr})}_\ell(k;g)\}$, defined earlier in \eqref{f:27b}, \eqref{f:27c},
\eqref{f:454}, and \eqref{f:46'}.
The derivation of the properties of all these CAM interpolations
relies not only on the assumptions on $V(\lambda;R,R')$, but also on 
the properties of the spherical functions $j_\lambda(kR)$ and of the
CAM Green function $G(\lambda,k;R,R')$ for 
$(\lambda,k)\in\C^+_{-\frac{1}{2}}\times \C^{\cut}$, where 
$\C^{\cut}=\C\setminus (-\infty,0]$. These properties have been 
established in Appendix \ref{appendix:a} (Subsections \ref{subappendix:a.complex}
and \ref{subappendix:a.complements}) and, in particular, it has been shown there
that for nonintegral values of $\ell$, these analytic functions of the two complex
variables ${\lambda,k}$ are holomorphic with respect to $k$
in a ramified domain with branch point at $k=0$, from which we only retain here
a distinguished sheet $\C^{\cut}$. For all $k$ in this domain, the functions 
$j_\lambda(kR)$ and $G(\lambda,k;R,R')$ are CAM interpolations of the corresponding 
sequences of functions $\{j_\ell(kR)\}_{\ell=0}^\infty$ and $\{G_\ell(k;R,R')\}_{\ell=0}^\infty$; 
however, it is to be noted that these interpolations
are {\bf non--Carlsonian} for general values of $k$. 
It is only when $k=\rmi\kappa$, $\kappa>0$, that 
(in view of bound \eqref{a:20}) the Green function
$G(\lambda,\rmi\kappa;R,R')$ appears to be the 
(unique) Carlsonian interpolation of the corresponding
sequence $\{G_\ell(\rmi\kappa;R,R')\}_{\ell=0}^\infty$, thus implying similar
Carlsonian properties for the functions $L(\lambda,i\kappa)$ 
and $R^{(\mathrm{tr})}(\lambda,\rmi\kappa;g)$ ($\kappa>0$). 
The (non--Carlsonian) interpolations obtained for general values of $k$ in $\C^{(cut)}$
are the analytic continuations of the latter with respect to $k$. 

\vskip 0.1cm

The previous considerations concerning the occurrence of a branch--point at $k=0$
and of non--Carlsonian bounds in the half--plane $\Real \lambda> -\frac{1}{2}$ lead us 
to introduce the following domains:

\begin{itemize}
\addtolength{\itemsep}{0.1cm}
\item[(i)] In the complex $\lambda$--plane: for each pair of positive numbers
$(\gamma,\delta)$, the (\emph{truncated}) \emph{angular sector} 
\beq
\Lambda_\gamma^{(\delta)} \doteq
\left\{\lambda\in\C\,:\,|\Imag\lambda|<\frac{\gamma}{3\pi}\left(\Real\lambda +\frac{1}{2}\right); \,
\Real\lambda >-\frac{1}{2}+\delta\right\},
\label{AG.5}
\eeq
whose closure is denoted by $\overline{\Lambda}_\gamma^{(\delta)}$.
\item[(ii)] In the complex $k$--plane:\\ \null\hspace{0.2cm} the \emph{cut--strip}
$\Omega_\alpha^{\cut} \doteq \Omega_\alpha\setminus (-\infty,0]$,
and its closure $\overline{\Omega}_\alpha^{\cut}$;\\ 
\null \hspace{0.2cm} the \emph{cut half--plane} $\Pi_\alpha^{\cut} \doteq \Pi_\alpha\setminus (-\infty,0]$, 
and its closure $\overline{\Pi}_\alpha^{\cut}$.
\item[(iii)] In the complex $(\lambda,k)$--space $\C^2$: the domain
$D^{(\delta)}_{\gamma,\alpha}\doteq\Lambda_\gamma^{(\delta)}\times\Omega_\alpha^{\cut}$,
whose closure is 
$\overline{D}^{(\delta)}_{\gamma,\alpha}\doteq\overline{\Lambda}_\gamma^{(\delta)}\times\overline{\Omega}_\alpha^{\cut}$.
\end{itemize}

\subsubsection{The vector--valued function $\boldsymbol{(\lambda,k) \mapsto v_0(\lambda,k;R)}$.}  
\label{subsubse:v-0-lambda}

The CAM interpolation of the sequence of functions $\{v_{\ell,0}(k;R)\}_{\ell=0}^\infty$
(see Eq. \eqref{f:27b}) is formally defined by the following expression,
whose analyticity and boundedness properties are stated below:
\beq
v_0(\lambda,k;R) = \int_0^{+\infty}V(\lambda;R,R')\,kR'\,j_\lambda(kR')\,\rmd R'.
\label{AG.11}
\eeq
In order to obtain a majorization of the integrand on the r.h.s. of the latter,
we shall use bound \eqref{a:41} for the spherical Bessel function $j_\lambda(kR)$,
which holds for $(\lambda,k)\in\C_{-\frac{1}{2}}^+\times \C^{\cut}$,
and yields (compare to \eqref{f:j0} and \eqref{f:j1}):
\beq
\begin{split}
&\left\|{k\cdot} \ j_\lambda(k\cdot)\right\|^*_{\wea}
\doteq
\left(\int_0^{+\infty}\frac{|kR j_\lambda(kR)|^2}{{\we}(R)\,e^{2\alpha R}}\,\rmd R\right)^{{1/2}} \\
&\quad\leqslant
\sqrt{\frac{\pi}{2}}\,|k|^\frac{1}{2}\, 
e^{\frac{3\pi}{2}|\Imag\lambda|}\left(\frac{3}{2}+\frac{1}{\pi(\Real\lambda+\frac{1}{2})}\right)
\ \left(\int_0^{+\infty}
\frac{R\, e^{-2(\alpha-|\Imag k|)R}}{{\we}(R)}\,\rmd R
\right)^{1/2}\!\!\!.
\end{split}
\label{AG.14}
\eeq
By taking Eqs. \eqref{wf:2} and \eqref{f:e0} into account for majorizing
the latter integral, we then obtain, for 
$(\lambda,k)\in \C^+_{-\frac{1}{2}}\times \overline \Omega_\alpha^{\cut}$:
\beq
\left\|{k\cdot} \ j_\lambda(k\cdot)\right\|^*_{\wea} \leqslant
\sqrt{\frac{\pi}{2}} \, A_\varepsilon\, |k|^\frac{1}{2}\, 
e^{\frac{3\pi}{2}|\Imag\lambda|}\left(\frac{3}{2}+\frac{1}{\pi(\Real\lambda+\frac{1}{2})}\right).
\label{AG.15} 
\eeq
We then have: 

\skd

\begin{theorem}
\label{the:AG.1}
For every nonlocal potential $V$ in $\cN^{\,\gamma}_{\wea}$,
the corresponding function $(\lambda,k) \mapsto v_0(\lambda,k;\cdot)$, 
is well--defined by the integral \eqref{AG.11} as a vector--valued function in  
the set $\C^+_{-\frac{1}{2}}\times {\overline \Omega}_\alpha^{\cut}$,
holomorphic in $\C^+_{-\frac{1}{2}}\times {\Omega}_\alpha^{\cut}$,
taking its values in $X_{\wea}$; 
the corresponding norm $\|v_0(\lambda,k;\cdot)\|_{\wea}$ admits the following
majorization in $\C^+_{-\frac{1}{2}}\times {\overline \Omega}_\alpha^{\cut}$:
\beq
\left\|v_0(\lambda,k;\cdot)\right\|_{\wea}
\leqslant\sqrt{\frac{\pi}{2}} \, A_\varepsilon \, \cK C(V_*)
\, |k|^\frac{1}{2}\, e^{-\gamma\Real \lambda}e^{\frac{3\pi}{2}|\Imag\lambda|} 
\left(\frac{3}{2}+\frac{1}{\pi(\Real\lambda+\frac{1}{2})}\right),
\label{AG.12}
\eeq
(the constants on the r.h.s. being defined by Eqs. \eqref{f:e0},
\eqref{g:5'}, \eqref{AG.6'}, and \eqref{AG.8}).
Moreover, the function $k^{-\frac{1}{2}}\ v_0(\lambda,k;\cdot)$ 
is defined as a continuous and uniformly bounded 
vector--valued function of $(\lambda,k)$ in any closed set  
$\overline{D}^{(\delta)}_{2\gamma,\alpha}$ (for any $\delta >0$). 
\end{theorem}

\begin{proof} 
The proof is quite similar to the one of Lemma \ref{lemma:v}.
By applying the Schwarz inequality to integral \eqref{AG.11}
and taking Eqs. \eqref{g:5'} and \eqref{AG.6'} into account, 
one obtains (for a.e. $R$): 
\beq
|v_0(\lambda,k;R)|
\leqslant \cK \, e^{-\gamma \Real\lambda} \,
V_\ell^{\left(\we\right)}(R)
\, \left\|{k\cdot} \ j_\lambda(k\cdot)\right\|^*_{\wea}.
\label{f:v2'}
\eeq
It then follows from \eqref{AG.8} and \eqref{AG.15}  
that the function $R \mapsto v_0(\lambda,k;R)$ belongs to 
$X_{\wea}$ for all 
$(\lambda,k)\in\C^+_{-\frac{1}{2}}\times\overline\Omega_\alpha^{\cut}$,
and satisfies bound \eqref{AG.12}.
Then, it follows from the latter that,
for $(\lambda,k)$ varying in any set 
$\overline{D}^{(\delta)}_{2\gamma,\alpha}$ 
(defined for any $\delta >0$ via \eqref{AG.5}), 
the function $(\lambda,k)\mapsto |k|^{-\frac{1}{2}}\,\|v_0(\lambda,k;\cdot)\|_{\wea}$ 
is uniformly bounded by the constant 
$\sqrt{\frac{\pi}{2}} A_\varepsilon \cK \cC(V_*) e^{\gamma/2}\left(\frac{3}{2}+\frac{1}{\pi\delta}\right)$.

\vskip 0.1cm

Finally, the holomorphy and continuity properties of the vector--valued
function $(\lambda,k)\mapsto v_0(\lambda,k;\cdot)$ are obtained as
in Lemma \ref{lemma:v} by a direct application of 
Lemma \ref{lemma:B9} (giving the holomorphy and continuity properties in
$(\lambda,k)$ of the integral \eqref{AG.11}, for a.e. $R$) 
and of Lemma \ref{lemma:B8}, by noting that the 
function $(\lambda,k,R)\mapsto v_0(\lambda,k;R)$ 
belongs to a relevant class ${\cC}(D, {\mu},p)$, with  
$D= \C^+_{-\frac{1}{2}}\times \overline\Omega_\alpha^{\cut}$
(resp., $\overline{D}^{(\delta)}_{2\gamma,\alpha}$ for the continuity property),  
$\mu(R)= \we(R)\, e^{2\alpha R}$ and $p=2$.
\end{proof}

\subsubsection{The operator--valued function $\boldsymbol{L(\lambda,k)}$.}  
\label{subsubse:L-lambda}

The CAM interpolation of the sequence of kernels $\{L_\ell(k;R,R')\}_{\ell=0}^\infty$ 
(see Eq. \eqref{f:27c}) is formally defined by the following expression,
whose analyticity and boundedness properties are stated below:
\beq
L(\lambda,k;R,R')=\int_0^{+\infty}V(\lambda;R,R'')G(\lambda,k;R'',R')\,\rmd R''.
\label{AG.19}
\eeq
As for the case of the kernels $L_\ell(k;R,R')$ (see Section \ref{se:rotationally}),
we are going to show that $L(\lambda,k;R,R')$ is bounded by
a kernel of rank one. To this effect, we shall use bounds
\eqref{a:30} and \eqref{a:32} for
the complex angular momentum Green function,
which hold respectively for $\Imag k \geqslant 0$ and $\Imag k <0$ in the domain
$(\lambda,k)\in\C_{-\frac{1}{2}}^+\times \C^{\cut}$.
These bounds imply the following global majorization, which holds for
$(\lambda,k)\in\C_{-\frac{1}{2}}^+\times 
\overline{\Pi}_\alpha^{\cut}$:
\beq
\left|G(\lambda,k;R,R')\right| \leqslant c\,\sqrt{R R'} \, e^{\alpha(R+R')} \,
e^{3\pi|\Imag\lambda|}\left(1+\frac{1}{2\Real\lambda+1}\right).
\label{AG.3}
\eeq
Now, if the potential $V(\lambda;R,R')$ belongs to the class
$\cN^{\,\gamma}_{\wea}$,
it follows from \eqref{g:5'} and \eqref{AG.3} that the following majorization
holds for a.e. $(R,R',R'')$ and  
$(\lambda,k)\in\C_{-\frac{1}{2}}^+\times \overline{\Pi}_\alpha^{\cut}$:
\beq
\left|V(\lambda;R,R'') G(\lambda,k;R'',R')\right| 
\leqslant M_\gamma(\lambda)\, V_*(R,R'')\,\sqrt{R'' R'} \, e^{\alpha(R''+R')},
\label{AG.4}
\eeq
where:
\beq
M_\gamma(\lambda) \doteq
c\,\cK\,e^{3\pi|\Imag\lambda|}\,e^{-\gamma\Real\lambda}
\left(1+\frac{1}{2\Real\lambda+ 1}\right).
\label{AG.4'}
\eeq
In view of the latter, we obtain the following bound for the integral \eqref{AG.19}: 
\beq
|L(\lambda,k;R,R')| \leqslant 
M_\gamma(\lambda) \, \sqrt{R'} \, e^{\alpha R'}
\int_0^{+\infty} V_*(R,R'') \, \sqrt{R''} \, e^{\alpha R''}\, \rmd R'',
\label{AG.a}
\eeq
which yields, by taking \eqref{AG.6'} into account, using the bound 
$\int_0^{+\infty} \frac{R''}{\we(R'')}\,\rmd R'' \leqslant A_\varepsilon^2$
(with $A_\varepsilon$ given by \eqref{f:e0}), and Schwarz's inequality: 
\beq
|L(\lambda,k;R,R')| \leqslant 
M_\gamma(\lambda)\, A_\varepsilon\, V_*^{\left(\we\right)}(R)\, \sqrt{R'} \, e^{\alpha R'}. 
\label{AG.b} 
\eeq
As in Section \ref{se:rotationally} (see Theorem \ref{theorem:4}), 
it is then appropriate to introduce the Hilbert space 
$\widehat{X}_{\wea}$
of Hilbert--Schmidt kernels $K(R,R')$ with respect to the
measure $\mu(R)\,\rmd R= \we(R) \, e^{2\alpha R} \,\rmd R
=R^{1-\varepsilon}(1+R)^{1+2\varepsilon}\,e^{2\alpha R}\,\rmd R$ 
(see Appendix \ref{subappendix:b.products}, formula \eqref{b:8}), whose norm is given by 
\beq
\left\|K\right\|^2_\mathrm{HS}
\doteq \left\|K\right\|^2_{\widehat{X}_{\wea}}
=\int_0^{+\infty}\!\!\! \frac{e^{-2\alpha R}}{\we(R)}\,\rmd R
\int_0^{+\infty}\!\!\! \we(R')\,e^{2\alpha R'}\left|K(R',R)\right|^2\,\rmd R'.
\label{AG.23}
\eeq
In fact, the kernel of rank one on the r.h.s. of Eq. \eqref{AG.b} belongs to this space,
its norm being expressed and majorized as follows:
\beq
\left\|V_*^{\left(\we\right)}\right\|_{\wea} 
\left\|\sqrt{(\cdot)}\, e^{(\alpha\,\cdot)}\right\|^*_{\wea} =
C(V_*)\, \left[\int_0^{+\infty}\frac{R''}{\we(R'')}\,\rmd R''\right]^{1/2} \leqslant 
C(V_*)\,A_\varepsilon.
\label{AG.23'} 
\eeq
It then follows from \eqref{AG.b} that, for each 
$(\lambda,k)\in\C_{-\frac{1}{2}}^+\times 
\overline{\Pi}_\alpha^{\cut}$, 
the kernel $L(\lambda,k;R,R')$ belongs to
$\widehat{X}_{\wea}$
and satisfies the following norm inequality:
\beq
\left\|L(\lambda,k)\right\|_\mathrm{HS} \leqslant
M_\gamma(\lambda)\ C(V_*)\ A_\varepsilon^2.
\label{AG.c}
\eeq
We can then state the following theorem.

\skd

\begin{theorem}
\label{the:AG.2}
For every nonlocal potential $V\in\cN_{\wea}^{\,\gamma}$,
the corresponding kernels $L(\lambda,k;R,R')$ (formally defined by \eqref{AG.19})
are well--defined as compact operators $L(\lambda,k)$ of Hilbert--Schmidt--type
acting in the Hilbert space $X_{\wea}$ 
for all $(\lambda,k)$ in 
$\C_{-\frac{1}{2}}^+ \times \overline{\Pi}_\alpha^{\cut}$, 
and the corresponding Hilbert--Schmidt norm 
$\left\|L(\lambda,k)\right\|_{\mathrm{HS}}$ of $L(\lambda,k)$ in 
$\widehat{X}_{\wea}$ admits the following global majorization:
\beq
\left\|L(\lambda,k)\right\|_\mathrm{HS} \leqslant
c\,\cK\,C(V_*)\,A_\varepsilon^2\,e^{3\pi|\Imag\lambda|}\,e^{-\gamma\Real\lambda}
\left(1+\frac{1}{\Real\lambda+\frac{1}{2}}\right).
\label{AG.d}
\eeq
Moreover, the $\mathrm{HS}$--operator--valued function $(\lambda,k) \mapsto L(\lambda,k)$,
taking its values in $\widehat{X}_{\wea}$,
is holomorphic in $\C_{-\frac{1}{2}}^+\times {\Pi}_\alpha^{\cut}$, 
and is continuous and uniformly bounded in any set of the form
$\overline{\Lambda}_\gamma^{(\delta)}\times \overline{\Pi}_\alpha^{\cut}$. 
\end{theorem}

\begin{proof}
The main part of it has been given in the previous argument; in particular, the 
global majorization \eqref{AG.d} (which is a rewriting of \eqref{AG.c} and \eqref{AG.4'})
is seen to give a uniform bound for
$\left\|L(\lambda,k)\right\|_\mathrm{HS}$
in any set of the form 
$\overline{\Lambda}_\gamma^{(\delta)}\times 
\overline{\Pi}_\alpha^{\cut}$ (see Eq. \eqref{AG.5}). 
As in Theorem \ref{theorem:4}, the holomorphy and continuity properties of the
HS--operator--valued function $(\lambda,k)\mapsto L(\lambda,k)$
are directly obtained by applying Lemma \ref{lemma:B10} 
(now with $\zeta=(\lambda,k)$) to integral \eqref{AG.19}.
\end{proof}

\subsubsection{Smithies' formalism for the resolvent $\boldsymbol{R(\lambda,k;g)}$.}  
\label{subsubse:smithies-lambda}

We can now formally write the following expression for the resolvent:
\beq
R(\lambda,k;g)=\left[\I-gL(\lambda,k)\right]^{-1}.
\label{AG.26}
\eeq
The fact that $L(\lambda,k)$ is a Hilbert--Schmidt operator on the Hilbert space
$X_{\wea}$ allows us to use Smithies' formulae and bounds.
Accordingly, we can write:
\beq
R(\lambda,k;g)=\I+ g \frac{N(\lambda,k;g)}{\sigma(\lambda,k;g)},
\label{AG.27}
\eeq
where the operators 
$N(\lambda,k;g)$ and the functions $\sigma(\lambda,k;g)$
are defined by extending formally 
all the formulae (3.39) through (3.45) of Smithies' formalism 
from non--negative integral values of $\ell$ to complex values of $\lambda$
in $\C^+_{-\frac{1}{2}}$. More precisely, we can now prove the following theorems.

\begin{theorem}
\label{the:AG.3}
For every nonlocal potential $V\in\cN_{\wea}^{\,\gamma}$, the function
$(\lambda,k,g)\mapsto \sigma(\lambda,k;g)$
is holomorphic in $\C^+_{-\frac{1}{2}}\times 
{\Pi}_\alpha^{(\mathrm{cut})}\times \C$
and continuous in $\C^+_{-\frac{1}{2}}\times 
\overline{\Pi}_\alpha^{(\mathrm{cut})}\times \C$.
Moreover, at fixed $g$, it is uniformly bounded in any closed set of the form 
$\overline{\Lambda}_\gamma^{(\delta)}\times 
\overline{\Pi}_\alpha^{\cut}$,
and the function $\sigma(\lambda,k;g)-1$ tends uniformly to zero 
for $|\lambda|$ tending to infinity in any subset 
$\overline{\Lambda}_{\gamma'}^{(\delta)}\times 
\overline{\Pi}_\alpha^{\cut}$ with $\gamma' <\gamma$. 
\end{theorem}

\begin{proof}
In view of the holomorphy and continuity properties of
$L(\lambda,k)$ stated in Theorem \ref{the:AG.2}, one can then
follow the argument given in the proof of Theorem \ref{theorem:5}
for justifying successively (and for every $n\geqslant 2$)
the corresponding holomorphy and continuity properties 
of the functions $\rho_n(\lambda,k)= \Tr [L^n(\lambda,k)]$,
$Q_n(\lambda,k)$ and $\sigma_n(\lambda,k)$ (defined as in Eqs.
\eqref{f:49} and \eqref{f:48}).

Now, by combining the basic inequalities of Smithies' theory 
with the bound \eqref{AG.c} on $\|L(\lambda,k)\|_\mathrm{HS}$, 
one obtains the following majorizations, similar to \eqref{f:54}:
\beq
|\sigma_n(\lambda,k)|\leqslant\left(\frac{e}{n}\right)^{n/2}
\left\|L(\lambda,k)\right\|^n_\mathrm{HS}\leqslant
\left(\frac{e}{n}\right)^{n/2}\left[
C(V_*)\,A_\varepsilon^2\,M_\gamma(\lambda)\right]^n.
\label{AG.28}
\eeq
It follows that the series $\sigma(\lambda,k;g)=\sum_{n=0}^\infty\sigma_n(\lambda,k)\,g^n$
(with $\sigma_0 =1$) is dominated, for all $(\lambda,k,g)$ in
$\C^+_{-\frac{1}{2}}\times 
\overline{\Pi}_\alpha^{(\mathrm{cut})}\times \C$
by a convergent series with positive terms.
By associating with the latter the entire function $\Phi(z)$
as in the proof of Theorem \ref{theorem:5} (see Eq. \eqref{f:Phi}),
one then concludes from inequality \eqref{AG.28}, written for all values of $n$, that 
the sum of the series $\sigma(\lambda,k;g)$ is well--defined and satisfies 
the following global majorization:
\beq
|\sigma(\lambda,k;g)-1|\leqslant\Phi\left(|g|\ \|L(\lambda,k)\|_\mathrm{HS}\right)
\leqslant\Phi\left(|g|\,C(V_*)\,A_\varepsilon^2\,M_\gamma(\lambda)\right).
\label{AG.30}
\eeq
By applying Lemma \ref{lemma:B0} to the sequence of functions
$\{(\lambda,k,g) \mapsto \sigma_n(\lambda,k)g^n;\,n\in\N\}$,
inequalities \eqref{AG.28} and \eqref{AG.30} entail 
that the function $\sigma(\lambda,k;g)$ is holomorphic in the domain
$\C^+_{-\frac{1}{2}}\times{\Pi}_\alpha^{(\mathrm{cut})}\times\C$ of $\C^3$, 
and also defined and continuous for $k\in \overline{\Pi}_\alpha^{(\mathrm{cut})}$. 
Moreover, since the function $M_\gamma(\lambda)$ (see Eq. \eqref{AG.4'}) 
is uniformly bounded in any set $\overline{\Lambda}_\gamma^{(\delta)}$  
and tends uniformly to zero for $|\lambda|$ tending to infinity
in any set $\overline{\Lambda}_{\gamma'}^{(\delta)}$, the last statement 
of the theorem directly follows from majorization \eqref{AG.30}.
\end{proof}

\skd

\begin{theorem}
\label{the:AG.4}
For every nonlocal potential $V\in\cN_{\wea}^{\,\gamma}$, the operators
$N(\lambda,k;g)$ exist as Hilbert--Schmidt operators acting on $X_{\wea}$
for all $(\lambda,k,g)$ in the subset 
$\C^+_{-\frac{1}{2}}\times 
\overline{\Pi}_\alpha^{(\mathrm{cut})}\times \C$ of $\C^3$.
The $\mathrm{HS}$--operator--valued function
$(\lambda,k,g) \mapsto N(\lambda,k;g)$, taking its values in 
$\widehat{X}_{\wea}$, is holomorphic in $\C^+_{-\frac{1}{2}}\times 
{\Pi}_\alpha^{(\mathrm{cut})}\times\C$
and continuous in $\C^+_{-\frac{1}{2}}\times 
\overline{\Pi}_\alpha^{(\mathrm{cut})}\times \C$.
Moreover, at fixed $g$,
$\|N(\lambda,k;g)\|_\mathrm{HS}$ is uniformly bounded 
in any closed set of the form $\overline{\Lambda}_\gamma^{(\delta)}\times 
\overline{\Pi}_\alpha^{\cut}$
and tends uniformly to zero for $|\lambda|$ tending to infinity in any subset 
$\overline{\Lambda}_{\gamma'}^{(\delta)}\times
\overline{\Pi}_\alpha^{\cut}$ with $\gamma' <\gamma$.
\end{theorem}

\begin{proof}
By defining successively (and for every $n\geqslant 1$)
the bounded--operator--valued functions $\Delta_n(\lambda,k)$ and 
the HS--operator--valued functions $N_n(\lambda,k)$ in terms of $L(\lambda,k)$ 
as in Eqs. \eqref{f:53} and \eqref{f:52}, one deduces from Lemma \ref{lemma:B6} 
that all these functions satisfy the same holomorphy and continuity
properties as those of $L(\lambda,k)$ specified in
Theorem \ref{the:AG.2}. Moreover, in view of Smithies' theory, 
there hold the following inequalities, similar to \eqref{f:60'}: 
\beq
\|N_n(\lambda,k)\|_\mathrm{HS} 
\leqslant \|\Delta_n(\lambda,k)\|\ \|L(\lambda,k)\|_\mathrm{HS}
\leqslant 
\frac{e^{(n+1)/2}}{n^{n/2}} \left\|L(\lambda,k)\right\|_\mathrm{HS}^{n+1}.
\label{AG.32}
\eeq
In view of the latter, the series $\sum_{n=0}^\infty N_n(\lambda,k)\,g^n$ is dominated term by term
in the $\mathrm{HS}$--norm by a convergent series; 
the sum of this operator--valued entire series is therefore well--defined as a
$\mathrm{HS}$--operator $N(\lambda,k;g)$ for all $(\lambda,k,g)$ in
$\C^+_{-\frac{1}{2}}\times\overline{\Pi}_\alpha^{(\mathrm{cut})}\times\C$.
Now, it follows from \eqref{AG.32} and from the bound \eqref{AG.c} on 
$\|L(\lambda,k)\|_\mathrm{HS}$ that the norm of 
$N(\lambda,k;g)$ in $\widehat{X}_{\wea}$ satisfies the bound:
\beq
\|N(\lambda,k;g)\|_\mathrm{HS} \leqslant 
\frac{1}{|g|} \Psi\left(|g|\,\|L(\lambda,k)\|_\mathrm{HS}\right)
\leqslant \frac{1}{|g|} \,
\Psi\left(|g|\,C(V_*)\,A_\varepsilon^2\,M_\gamma(\lambda)\right),
\label{AG.34}
\eeq
where $\Psi(z)$ is the entire function introduced in the proof of
Theorem \ref{theorem:6} (see Eq. \eqref{f:Psi}). The fact that
$\Psi(z)$ is an increasing function of $z$ for $z\geqslant 0$ has been used for
obtaining the second inequality in \eqref{AG.34}.

By applying Lemma \ref{lemma:B0} to the sequence of functions
$\{(\lambda,k,g) \to N_n(\lambda,k)\,g^n;\,n\in\N\}$,
inequalities \eqref{AG.32} and \eqref{AG.34}
entail that the HS--operator--valued function
$N(\lambda,k;g)$ is holomorphic in the domain
$\C^+_{-\frac{1}{2}}\times {\Pi}_\alpha^{(\mathrm{cut})}\times \C$ of $\C^3$, 
and also defined and continuous for  
$k\in \overline{\Pi}_\alpha^{(\mathrm{cut})}$. 
Finally, the last statement of the theorem directly follows from 
\eqref{AG.34} and from the expression \eqref{AG.4'} of
$M_\gamma(\lambda)$ (as for the last statement of Theorem \ref{the:AG.3}).
\end{proof}

\subsection
{Meromorphy properties of the resolvent and their physical interpretation} 
\label{subse:aggiunte4.3}

\subsubsection{General structure}
\label{subsubse:general}

Let us introduce, as in Section \ref{se:rotationally}, the truncated Fredholm 
resolvent $R^{(\mathrm{tr})}(\lambda,k;g)$ as follows:
\beq
R^{(\mathrm{tr})}(\lambda,k;g) = \frac{N(\lambda,k;g)}{\sigma(\lambda,k;g)}.
\label{AG.35}
\eeq
It follows from Theorems \ref{the:AG.3} and \ref{the:AG.4} that 
$R^{(\mathrm{tr})}(\lambda,k;g)$ is a  HS--operator--valued meromorphic
function of $(\lambda,k,g)$, whose ``poles" are localized on the
various possible connected components of the complex analytic set 
defined by the equation $\sigma(\lambda,k;g) =0$.
Now, in view of the last property stated in
Theorem \ref{the:AG.3}, this set cannot contain any component
of the form $g=g_0$, and therefore for each fixed $g$ (in $\C$) the
subset
\beq
D_\alpha (V;g)\doteq 
\left\{(\lambda,k)\in\C^+_{-\frac{1}{2}}\times{\Pi}_\alpha^{(\mathrm{cut})}
\,:\,\sigma(\lambda,k;g) \neq 0\right\}
\label{AG.35'}
\eeq
is a domain, which is the complement in  
$\C^+_{-\frac{1}{2}}\times {\Pi}_\alpha^{(\mathrm{cut})}$ of a one--dimensional analytic set.
We can then state:

\skd

\begin{theorem}
\label{the:AG.5}
For every nonlocal potential $V\in\cN_{\wea}^{\,\gamma}$,
the operators $R^{(\mathrm{tr})}(\lambda,k;g)$ exist  
as Hilbert--Schmidt operators acting on $X_{\wea}$
for all $(\lambda,k,g)$ in the dense subdomain of 
$\C^+_{-\frac{1}{2}}\times 
\overline{\Pi}_\alpha^{(\mathrm{cut})}\times \C$ where $\sigma(\lambda,k;g)\neq 0$.
Moreover, the $\mathrm{HS}$--operator--valued function
$(\lambda,k,g)\! \mapsto\! R^{(\mathrm{tr})}(\lambda,k;g)$, taking its values in 
$\widehat{X}_{\wea}$,
is a meromorphic function whose restriction to each fixed value of $g$ is 
holomorphic in $D_\alpha (V;g)$.
\end{theorem}

As in Section \ref{se:rotationally} (see Theorem \ref{theorem:6"}), 
one can also state the following property of the ``complete resolvent"
$R(\lambda,k;g)= [\I -g L(\lambda,k)]^{-1} =\I +gR^{(\mathrm{tr})}(\lambda,k;g)$:

\skd

\begin{theorem}
\label{the:AG.6}
For any fixed $g$, the function $R(\lambda,k;g)$ is holomorphic 
in the domain $D_\alpha (V;g)$ as an operator--valued function, 
taking its values in the space of bounded operators in $X_{\wea}$.
\end{theorem}

\subsubsection{Symmetry properties in the complex variables $\boldsymbol{(\lambda,k,g)}$.} 
\label{subsubse:symmetry}

In Section \ref{se:rotationally} we have shown that the basic symmetry properties
\eqref{f:24'}, \eqref{f:symd}, and \eqref{f:symc} imply 
the corresponding invariance  of the quantities
$\sigma_\ell (k;g)$, $N_\ell(k;g)$, and $R_\ell(k;g)$ under the 
transformation $(k\to -\overline k, g\to \overline g)$ in their respective 
analyticity domains of the complex $(k,g)$--space (see Theorems \ref{theorem:5} (b) 
and \ref{theorem:6} (b)). Here one can similarly justify the invariance 
of the corresponding interpolated quantities in the CAM--plane, 
namely $\sigma(\lambda,k;g)$, $N(\lambda,k;g)$, and $R(\lambda,k;g)$, 
under the transformation 
$(\lambda \to \overline \lambda, k \to -\overline k, g\to \overline g)$.
Note however that some additional specifications must be given
in view of the occurrence of the branch point at $k=0$ and of 
the affiliated cut on the real negative $k$--axis, which are effective 
not only for complex values of $\lambda$ in $\C^+_{-\frac{1}{2}}$, but already for 
nonintegral real values of $\lambda$ ($\lambda >-\frac{1}{2}$). 
As a matter of fact, it can be seen that for the quantities
previously mentioned, their invariance under the transformation 
$(\lambda \to \overline \lambda, k \to -\overline k, g\to \overline g)$
can be first established for $k$ varying in the upper half--plane; 
the extension of this invariance property to the cut--plane 
$\Pi_\alpha^{\cut}$ then follows by analytic continuation, provided one 
introduces a $k\to -\overline k$--invariant ramified
analyticity domain over $\Pi_\alpha\setminus \{k=0\}$ 
(thus including also a cut--domain with a cut along 
the positive real $k$--axis, which is the symmetric domain of 
$\Pi_\alpha^{\cut}$). Here again, these properties are based on:
\begin{itemize}
\item[(a)] the symmetry properties of the potential, 
i.e.,  $V(\overline{\lambda};R,R')=\overline{V(\lambda;R,R')}$,
and $V(\lambda;R,R')=V(\lambda;R',R)$ for all $\lambda\in\C^+_{-\frac{1}{2}}$,
which correspond to the Carlsonian interpolation of \eqref{f:24'}; \\[-6pt]
\item[(b)] the relation 
$G(\overline \lambda,-\overline k;R,R')=
\overline {G(\lambda,k;R,R')}$ for all 
$(\lambda,k)$ such that $\lambda\in\C^+_{-\frac{1}{2}}$ and $\Imag k >0$. 
This relation, which extrapolates \eqref{f:symd}, is easily derived from the
integral representation \eqref{a:16} of $G(\lambda,\rmi\kappa; R,R')$
and the analytic continuation of the latter to complex values of $\lambda$ and $\kappa$. 
\end{itemize}
From (a) and (b), one then derives the analog of \eqref{f:symc}, namely:

\begin{itemize}
\item[(c)] \, $L(\overline{\lambda},-\overline{k};R,R')=
\overline{L(\lambda,k;R,R')}$, and subsequently (by arguments
similar to those given in the proofs of Theorems \ref{theorem:5} (b) and 
\ref{theorem:6} (b)); \\[-6pt]
\item[(d)] \,
$\sigma(\overline\lambda,-\overline k; \overline g)
=\overline{\sigma(\lambda,k;g)}$, \,
$N(\overline\lambda,-\overline k; \overline g)
=\overline{N(\lambda,k;g)}$, \,
$R(\overline\lambda,-\overline k; \overline g)
=\overline{R(\lambda,k;g)}$.
\end{itemize}

\subsubsection{Poles of the resolvent and solutions of the Schr\"odinger--type equation} 
\label{subsubse:poles-complex}

By the same analysis as in Subsection \ref{subse:mero} for the case $\ell$ integer,
we can say that the existence of a pole $k=k(\lambda,g)$ of 
the meromorphic function $k \mapsto R(\lambda,k;g)$, namely a value of $k$  
such that $\sigma(\lambda,k(\lambda,g);g)=0$ with 
$N(\lambda,k;g) \neq 0$, is equivalent to the existence of 
at least one non--zero solution $x=x(R) $ in $X_{\wea}$
of the homogeneous Fredholm equation $gL(\lambda,k)x=x$.

Concerning the terminology, we prefer to say that such a solution $x$
is associated with a \emphsl{singular pair} $(\lambda,k)$
(i.e., a pair satisfying the equation $\sigma(\lambda,k;g) =0$ for a fixed value of $g$) 
rather than with the variable--dependent notion of ``pole".
At a singular pair $(\lambda,k)$ indeed, 
it can be advantageous as well to consider the pole 
of the meromorphic function $\lambda \mapsto R(\lambda,k;g)$,
at the value $\lambda=\lambda(k,g)$ such that $\sigma(\lambda(k,g),k;g) =0$  
(instead of the pole in the variable $k$, as always considered before).

As shown below in Lemma \ref{lemma:viivii}, one can associate with any function
$x(R) $ in $X_{\wea}$, and for every $\lambda \in \C^+_{-\frac{1}{2}}$ 
and $k\in {\overline \Pi}_\alpha^{\cut}$, the function
\beq
\psi(R)= g \int_0^{+\infty}G(\lambda,k;R,R')\,x(R')\,\rmd R',
\label{xtopsi'}
\eeq
which satisfies the equation $D_{\lambda,k}\ \psi = g \, x$,
since $G(\lambda,k;R,R')$ is (for every complex pair $(\lambda,k)$) the 
Green function of the corresponding differential operator
$D_{\lambda,k}$ (defined as $D_{\ell,k}$ by complexification of $\ell$
in Eq. \eqref{f:24}). Now, in view of Eq. \eqref{xtopsi'}, the definition
\eqref{AG.19} of $L(\lambda,k)$ implies the following equality:
\beq
g \ [L(\lambda,k) x](R) = 
\int_0^{+\infty} V(\lambda;R,R') \, \psi(R')\, \rmd R'.
\label{psitoLx'}
\eeq
So, if $x$ is a non--zero solution of the homogeneous Fredholm equation 
$gL(\lambda,k)x=x$ associated with a given singular pair $(\lambda,k)$ 
of the resolvent $ R^{(\mathrm{tr})}(\lambda,k;g)$ ($g$ being fixed), 
then one has:
\beq
D_{\lambda,k}\psi(R) =\ g\,x(R) =\ 
g\int_0^{+\infty}V(\lambda;R,R')\,\psi(R')\,\rmd R'.
\label{f:124}
\eeq
Thus $\psi$ appears as a non--zero solution of the extension to complex $\lambda$
of the Schr\"odinger--type equation \eqref{f:24}.

As in Section \ref{se:rotationally} (see Theorem \ref{theorem:7}), 
some specific properties of this type of solution $\psi(R)$ will be given below. 

\subsubsection{Some results on the locations of the poles 
of $\boldsymbol{R^{(\mathrm{tr})}(\lambda,k;g)}$ } 
\label{subsubse:some}

First of all, we shall take into account the fact  
that each class of nonlocal potentials 
$\cN_{\wea}^{\,\gamma}$ (for any $\gamma>0$) 
is contained in the corresponding class $\cN_{\wea}$
introduced in Section \ref{se:rotationally}.
It follows that all the properties of the poles 
in $k$ at fixed integer $\ell$ proved above
hold true for the potentials in $\cN_{\wea}^{\,\gamma}$.
By incorporating the results of this previous analysis
and taking $g$ real, we are led to distinguish 
between two situations,
whose specifications and interest will be justified below.

\begin{itemize}
\item[(a)] 
{\bf $\boldsymbol{\lambda}$ real and larger than $\boldsymbol{-\frac{1}{2}}$, 
$\boldsymbol{k}$ complex in $\boldsymbol{\overline{\Pi}_\alpha^{(\mathrm{cut})}\!\!.}$} \
In this case we shall see that one obtains a ``natural" extension of the
results  obtained in Section \ref{se:rotationally} for $\lambda=\ell$ integer 
(which leads us to use the same terminology): \\[-5pt]
\begin{itemize}
\addtolength{\itemsep}{0.1cm}
\item[(a.1)] \emph{bound states}: zeros of $\sigma$ sitting on the positive imaginary axis; 
no other zeros of $\sigma$ can occur in the upper half--plane $\Imag k >0$;
\item[(a.2)] \emph{spurious bound states}: zeros of $\sigma$ sitting on the positive real axis 
$\Imag k=0$, $\Real k>0$;
\item[(a.3)] \emph{anti--bound states}: zeros of $\sigma$ sitting on the negative imaginary axis
$-\alpha<\Imag k<0$, $\Real k=0$;
\item[(a.4)] \emph{resonances}: zeros of $\sigma$ in the half--strip $-\alpha<\Imag k<0$,
$\Real k >0$. \\[-5pt]
\end{itemize}
\item[(b)]
{\bf $\boldsymbol{\lambda}$ complex in $\boldsymbol{\C^+_{-\frac{1}{2}}}$ and 
$\boldsymbol{k}$ real in $\boldsymbol{\R^+}$.} \ 
In this case we may have: \\[-5pt]
\begin{itemize}
\addtolength{\itemsep}{0.1cm}
\item[(b.1)] zeros of $\sigma$ in the \emphsl{first quadrant}
of the $\lambda$--plane ($\Imag\lambda>0$, $\Real \lambda> -\frac{1}{2}$),
which correspond to an alternative description of \emph{resonances};
\item[(b.2)] zeros of $\sigma$ in the \emphsl{fourth quadrant} of the $\lambda$--plane
($\Imag\lambda<0$, $\Real\lambda> -\frac{1}{2}$), corresponding to 
\emph{antiresonances} (see the end of Subsection \ref{subsubse:some}, 
and Subsection \ref{subse:analysis}).
\end{itemize}
\end{itemize}

\vskip 0.1cm

Note that in this description a dissymmetric role is played by the 
first and fourth quadrant of the $\lambda$--plane, but we have to keep in
mind that here the interesting range of $k$ has been restricted to
the region $\Real k \geqslant 0$. As explained in the previous subsection, the
symmetry $\lambda \to \overline \lambda$ would be restored only
if accompanied by the transformation $k\to -\overline k$.

\vskip 0.3cm

We shall first prove the following variant of Lemma \ref{lemma:vii}.

\skd

\begin{lemma}
\label{lemma:viivii} 
For every function $x$ in ${X}_{\wea}$
and for all $(\lambda,k)$ such that $\Real\lambda >-\frac{1}{2}$
and $\Imag k\geqslant 0$, the corresponding function
\beq
\psi_{x;\lambda,k}(R)= g \int_0^{+\infty} \!\! G(\lambda,k;R,R')\,x(R')\,\rmd R'
\label{AG:62}
\eeq
is well--defined as a locally bounded function, contained in the space 
$X^*_{\wea}$. Moreover, it enjoys the following properties:

\vskip 0.2cm

\noindent
{\rm (i)} If $\Real\lambda>0$, there holds a global majorization of the following form 
for $R$ varying on the whole half--line $\{R\in\R^+\}$:
\begin{align}
\hspace{-0.1cm}
|\psi_{x;\lambda,k}(R)| &\leqslant |g|\ \|x\|_{\wea}
\,\widehat{c}_0(\lambda)\ [\Phi(k)]^{-(\Real\lambda+\frac{1}{2})}
\ \min\left(R,\,\beta(\lambda,\alpha)\,R^{-{\Real\lambda}}\right), \label{f:viivii} \\
\intertext{where:}
\widehat{c}_0(\lambda) &= c \, e^{\pi |\Imag\lambda|}
\left(1+\frac{1}{2\Real\lambda+1}\right)\ \left[\frac{1}{2(\Real\lambda+1)} 
+ \frac{1}{2\Real \lambda}\right]^{1/2},
\label{AG:61}
\end{align}
$\Phi(k) =1$ if $\,\,|\!\Arg(-\rmi k)|\leqslant\frac{\pi}{4}$,
while $\Phi(k)=\sin 2|\phi(k)|$ if
$\frac{\pi}{4}\leqslant |\phi(k)| =|\Arg(-\rmi k)|<\frac{\pi}{2}$,
and $\beta$ denotes a suitable positive function of $\lambda$ and $\alpha$; 

\vskip 0.2cm

\noindent
{\rm (ii)} the derivative $\psi'_{x;\lambda,k}(R)$ 
of $\psi_{x;\lambda,k}$ is well--defined on $\R^+$ and 
satisfies a global majorization of the form:
\beq
|\psi'_{x;\lambda,k}(R)| 
\leqslant |g|\ \frac{\widehat{c}_1(\lambda,k)}{\sqrt{2\alpha}} \ 
\|x\|_{\wea} \ R^{-1/2};
\label{f:viivii'}
\eeq

\vskip 0.3cm

\noindent
{\rm (iii)} for any potential $V$ in a class $\cN^{\,\gamma}_{\wea}$,
the corresponding double integral
\beq
\int_0^{+\infty} \!\! \rmd R \int_0^{+\infty} \!\! \rmd R' \
\overline{\psi}(R) \, V(\lambda;R,R') \, \psi (R')
\label{AG:66} \\
\eeq
is absolutely convergent.
\end{lemma}

\begin{proof}
By using the assumption that $x$ belongs to 
${X}_{\wea}$ and the Schwarz inequality,
we obtain the following majorization of the expression
\eqref{AG:62} of $\psi_{x;\lambda,k}(R)$:
\beq
|\psi_{x;\lambda,k}(R)|\leqslant
|g|\, \|x\|_{\wea}
\left[\int_0^{+\infty} |G(\lambda,k;R,R')|^2 \, \frac{e^{-2\alpha R'}}{\we(R')} 
\,\rmd R'\right]^{1/2},
\label{AG:63}
\eeq
with 
$w^{(\varepsilon)}$ defined by Eq. \eqref{wf:2}. 
In the latter integral, we can now plug the bound
\eqref{a:31"} on $|G(\lambda,k;R,R')|$, which is valid 
for $\Imag k \geqslant 0$ and $\lambda \in \C^+_{-\frac{1}{2}}$. 
A straightforward majorization then shows that
the integral in \eqref{AG:63} is convergent and bounded by 
$R\,[c(\lambda,k)]^2/(2\alpha)$, which implies that 
$\psi_{x;\lambda,k}(R)$ is well--defined and bounded by 
$|g|\,\|x\|_{\wea}\,\frac{c(\lambda,k)}{\sqrt {2\alpha}}\,\sqrt{R}$. 
It then also follows that
\beq
\left\|\psi_{x;\lambda,k}\right\|^*_{\wea}\leqslant
\frac{c(\lambda,k)}{\sqrt {2\alpha}} \ |g|\,\|x\|_{\wea}
\left[\int_0^{+\infty} \frac{R\,e^{-2\alpha R}}{w^{(\varepsilon)}(R)} \,\rmd R\right]^{1/2}\  < +\infty.
\label{AG:64-ex}
\eeq
\emph{Proof of {\rm (i)}.}
\, Using \eqref{a:30'} (along with inequality $R^{\varepsilon}\,(1+R)^{-1-2\varepsilon}< 1$)
allows one to majorize the integral in \eqref{AG:63} by the expression
\beq
c^2(\lambda,k) \ [\Phi(k)]^{-(2\Real \lambda +1)}
R\,\left[\int_0^R \!\! e^{-2\alpha R'}\left(\frac{R'}{R}\right)^{\!\! 2\Real\lambda+1}\!\!\rmd R'
+\int_R^{+\infty} \!\! e^{-2\alpha R'}\left(\frac{R}{R'}\right)^{\!\! 2\Real\lambda +1} \!\! \rmd R'\right],
\label{AG:64'}
\eeq
which can itself be majorized by either one of the 
following two expressions (by respectively majorizing $e^{-2\alpha R'}$ by one or not):  
\begin{align}
& \mathrm{(a)} \qquad 
c^2(\lambda,k) \, [\Phi(k)]^{-(2\Real \lambda +1)}
R^2 \left[\frac{1}{2(\Real\lambda +1)} + \frac{1}{2\Real \lambda}\right], \label{AG:64"} \\
\intertext{which holds under the additional condition $\Real \lambda>0$.}
& \mathrm{(b)} \qquad 
c^2(\lambda,k) \, [\Phi(k)]^{-(2\Real \lambda +1)}
\left[R^{-2\Real \lambda}\,\frac{\Gamma(2\Real\lambda +2)}{{(2\alpha)}^{2\Real\lambda+2}} 
+(2\alpha)^{-1} R \ e^{-2\alpha R}\right]. \label{AG:64}
\end{align}
Here, the expression inside the bracket can itself be
majorized by $\beta^2(\lambda,\alpha)\ R^{-2\Real \lambda}$
(in terms of a suitable positive constant $\beta(\lambda,\alpha)$).
Therefore the inequalities \eqref{AG:64"} and \eqref{AG:64} imply a global
majorization of the r.h.s. of \eqref{AG:63} of the form \eqref{f:viivii}, 
by noting that, in view of \eqref{a:30"}, the expression $\widehat{c}_0(\lambda)$, 
as defined by Eq. \eqref{AG:61}, is such that 
$\widehat{c}_0(\lambda) = \left[\frac{1}{2(\Real\lambda+1)} 
+ \frac{1}{2\Real \lambda}\right]^{1/2} \times \sup_{\{k\,:\,\Imag k\geqslant 0\}} c(\lambda,k)$.

\vskip 0.3cm

\noindent
\emph{Proof of {\rm (ii)}:}
In view of Eq. \eqref{AG:62}, one is led to establish the convergence 
and boundedness properties of the integral
$\cJ(R)= \int_0^\infty g \,\frac{\partial G}{\partial R}(\lambda,k;R,R') \, x(R') \,\rmd R'$,
which will then define the function $\psi'_{x;\lambda,k}(R)$. 
Similarly to formula \eqref{AG:63}, one can now write,
in view of the Schwarz inequality and of bound \eqref{a:131"},
the following successive majorizations: 
\beq
\begin{split}
& |\cJ(R)| \leqslant
|g| \, \|x\|_{\wea}
\left[\int_0^{+\infty}\left|\frac{\partial G}{\partial R}(\lambda,k;R,R')\right|^2 
\frac{e^{-2\alpha R'}}{w^{(\varepsilon)}(R')} \,\rmd R'\right]^{1/2} \\
&\quad\leqslant 
\frac{\widehat{c}_1(\lambda,k)}{\sqrt{R}} \ |g| \, \|x\|_{\wea}
\left[\int_0^{+\infty} \frac{R \, e^{-2\alpha R}}{w^{(\varepsilon)}(R)} \,\rmd R\right]^{1/2}\  < +\infty,
\end{split}
\label{AG:63'}
\eeq
which therefore proves that 
$\psi'_{x;\lambda,k}(R)$ is well--defined on $\R^+$ and satisfies the
global majorization \eqref{f:viivii'}.

\vskip 0.3cm

\noindent
\emph{Proof of {\rm (iii)}:}
In view of the conditions \eqref{AG.7} and \eqref{AG.10} 
on the potential $V$ and of the fact that $\psi$ belongs to
the space $X^*_{\wea}$, it results 
from the Schwarz inequality that the integral \eqref{AG:66} 
is absolutely convergent and bounded by the constant 
$ C(V(\lambda;..))\, \left(\left\|\psi\right\|^*_{\wea}\right)^2$. 
\end{proof}

\noindent
We can now give an appropriate variant of the Wronskian Lemma \ref{lemma:vii'}.

\skd

\begin{lemma}[Wronskian Lemma]
\label{lemma:vii'vii'} 
For given values of $\lambda$, $k$, $g$ such that $\Real\lambda>0$,
$\Imag k>0$, and $g\in\R$, let $\psi(R)$ be a solution of the following
integro--differential equation:
$D_{\lambda,k}\,\psi(R) = g \int_0^{+\infty} V(\lambda;R,R') \,\psi(R')\,\rmd R'$, 
associated with a solution $x \in X_{\wea}$ of the 
corresponding equation $x=gL(\lambda,k) \, x$ via Eq. \eqref{xtopsi'}.
Then there holds the following identity, in which all the integrals are
absolutely convergent:
\beq
\begin{split}
& 2\Imag k \, \Real k \int_0^{+\infty} {\overline \psi}(R) \, \psi(R)\,\rmd R
-\Imag \lambda \ (2\Real \lambda +1)
\int_0^{+\infty} \frac{{\overline \psi}(R)\, \psi(R)}{{R}^2}\,\rmd R \\
&\quad = g\int_0^{+\infty} \!\! \rmd R\int_0^{+\infty} \!\! \rmd R'
\,\frac{[V(\lambda;R,R')-V({\overline \lambda};R,R')]}{2\rmi} \,\overline{\psi}(R) \, \psi(R').
\end{split}
\label{f:wrwr}
\eeq
\end{lemma}

\begin{proof}
The following equation (analogous to Eq. \eqref{f:wr1})
results from the extension of Eq. \eqref{f:24} to  
complex values of $(\lambda,k)$:
\beq 
\begin{split}
& \overline{\psi}(R)\,\psi''(R)-\overline{\psi''} (R)\,\psi(R) 
+(k^2-{\overline k}^{\,2}){\overline \psi}(R) \, \psi(R) 
-[\lambda(\lambda+1)-\overline{\lambda}(\overline{\lambda}+1)]
\frac{{\overline \psi}(R) \, \psi(R)}{R^2} \\
& \quad = \overline{\psi} (R)\ [D_{\lambda,k} \, \psi](R)
-[\overline{D_{\lambda,k}\psi}](R)\,\psi(R) \\
& \quad = g \int_0^{+\infty} 
\left[V(\lambda;R,R')\overline{\psi}(R)\,\psi(R')-
V(\overline{\lambda};R,R')\,\overline{\psi}(R')\,\psi (R)\right]
\, \rmd R'
\end{split}
\label{f:wrwr1}
\eeq
Then, since $\psi(R)$ satisfies all the properties
described in Lemma \ref{lemma:viivii}, 
it is legitimate to integrate  
over $R$ from $0$ to $+\infty$ both sides of the latter equation;
more precisely, one deduces from (i) and (iii) of Lemma \ref{lemma:viivii} 
that, under the conditions $\Real\lambda >0$ and $\Imag k >0$,
the r.h.s. of \eqref{f:wrwr1}, as well as the function  
${{\overline \psi}(R) \psi(R)/R^2}$ are integrable 
over $R$ from $0$ to $+\infty$. Moreover, the majorizations \eqref{f:viivii}
and \eqref{f:viivii'} imply that the function
$[\overline{\psi} (R)\,\psi' (R)-\overline{\psi'}(R)\,\psi(R)]$
tends to zero at both ends of the half--line $[0,+\infty)$. 
It therefore follows that the remaining 
integrated term of Eq. \eqref{f:wrwr1}, namely,
$4\rmi\Imag k \Real k \int_0^{\widehat{R}} \overline{\psi}(R)\,\psi(R)\,\rmd R$,
has a finite limit when $\widehat{R}$ tends to $+\infty$.
Finally, by taking into account the symmetry property of
$V(\lambda;R,R')$, the integral 
in Eq. \eqref{f:wrwr1} can be rewritten under the form 
of Eq. \eqref{f:wrwr}, in which all terms are well--defined.
\end{proof}

\vskip 0.3cm

\noindent
{\bf Application: regions of $\boldsymbol{(\lambda,k)}$--space free of singularities}

\vskip 0.2cm

In Theorem \ref{theorem:7} (c), we had proved (as an application of
the Wronskian Lemma) that no singularity of the resolvent $R_\ell(k;g)$ 
can occur in the upper half--plane of $k$, except on the imaginary axis.
In other words, there are no singular pairs $(\lambda,k)$ with 
$\lambda=\ell$ integer $(\ell\geqslant 0)$ and $\Imag k>0$, $\Real k \neq 0$.
Now, there are some extensions of that result 
to the location of the singularity manifolds in complex 
$(\lambda,k)$--space of the ``interpolated resolvent" $R(\lambda,k;g)$, 
which similarly follow from Lemma \ref{lemma:vii'vii'}. 

\vskip 0.3cm

\noindent
(a) For $\lambda$ real, Eq. \eqref{f:wrwr} reduces to:
\beq
\Imag k \, \Real k \int_0^{+\infty} {\overline \psi}(R) \, \psi(R) \,\rmd R=0,
\label{f:wrwr2}
\eeq
which entails that there are no singular pairs $(\lambda,k)$ 
with $\lambda$ real positive and $\Imag k >0$, $\Real k \neq 0$, 
since Eq. \eqref{f:wrwr2} excludes the possibility of a non--zero 
function $\psi$ in these situations.

\vskip 0.3cm

\noindent
(b) If the potential $V(\lambda; R,R')$ is constant with respect to $\lambda$,
the r.h.s. of Eq. \eqref{f:wrwr} still vanishes. Then one sees that no singular 
pairs $(\lambda,k)$ can occur such that $\Imag\lambda<0$ with $\Imag k>0$ and 
$\Real k>0$ (and in the symmetric region obtained by 
$(\lambda,k) \to ({\overline \lambda}, -{\overline k})$).
In all such situations indeed, the l.h.s. of Eq. \eqref{f:wrwr} 
would have to be strictly positive for any non--zero function $\psi$.

Note that, via an argument of continuity, this result also implies that
there is no singular pair $(\lambda,k)$ with $\Imag\lambda<0$ and $ k>0$.

We notice that an analog of this situation is always encountered in the 
complex angular momentum formalism of the theory of local potentials, 
since the correspondence with the present case is simply given formally 
by the relation $V(\lambda;R,R')= \delta(R-R')V(R)$:
in other words, the complex angular momentum interpolation of the potential
is always constant in $\lambda$. As a matter of fact, in that framework, 
the previous result appeared as a basic theorem, which was proved by T. Regge 
\cite{DeAlfaro}: according to the latter, all the singularity manifolds of 
$R(\lambda,k;g)$, such as those which manifest themselves as resonances 
(i.e., containing pairs $(\lambda,k)$ with $\lambda=\ell$ integer and $\Imag k<0$),
can only manifest themselves at $ k>0$ in the region $\Imag\lambda>0$.

\vskip 0.3cm

\noindent
(c) An interesting simple class of nonlocal potentials 
(already mentioned in Subsection \ref{subse:interpolation}) are the 
potentials of the form $V(\lambda;R,R')= V_*(R,R') \ \widetilde{F}(\lambda)$,
where the function $\widetilde{F}$ is holomorphic, bounded and of Hermitian--type 
in the half--plane $\C^+_{-\frac{1}{2}}$.
As typical simple examples, we mentioned 
$\widetilde{F}(\lambda)= e^{-\gamma \lambda}$ or
$\widetilde{F}(\lambda)= Q_\lambda(e^\gamma)$.
For such a class, let $\{{\cL}_j;\  j\in \Z\}$ be the set of lines in
$\C^+_{-\frac{1}{2}}$ on which one has $\Imag\widetilde{F}(\lambda)= 0$,
${\cL}_j$ and ${\cL}_{-j}$ being complex conjugate of each other
with the following property: all the points of a curve 
${\cL}_j$ with $j>0$ (resp., $j<0$) belong to the region $\Imag\lambda>0$
(resp., $\Imag\lambda<0$), ${\cL}_0$ being along the real positive axis.
Then, at any point $\lambda\in {\cL}_j$ we have 
$V(\lambda;R,R')=V({\overline\lambda};R,R')$, and therefore the r.h.s. of
Eq. \eqref{f:wrwr} vanishes at such points. Therefore, on the basis of 
Eq. \eqref{f:wrwr} as in (b), no singular pairs $(\lambda,k)$ can occur 
with $\lambda$ in any line ${\cL}_j$, $j\leqslant 0$, and with $\Imag k \geqslant 0$ 
and $\Real k \geqslant 0$. This property indicates the 
\emphsl{possible occurrence of singularity manifolds containing branches in the region 
$\Imag k\geqslant 0$, $\Real k \geqslant 0$ and $\Real\lambda>0$, $\Imag\lambda<0$, 
but always located in ``strips" of this fourth quadrant of the $\lambda$--plane, 
well--separated from one another by the set of lines ${\cL}_j$, $j\leqslant 0$}.
These possible singularities, which did not exist for the case of local potentials, 
can be seen to enjoy properties which are related to the notion of \emphsl{antiresonance}
(see Section \ref{se:AG.5}).

\subsection{Analyticity and boundedness properties in complex 
$\boldsymbol{(\lambda,k)}$--space of the partial scattering 
amplitude $\boldsymbol{T(\lambda,k;g)}$}
\label{subse:aggiunte4.4}

We now extend Eq. (3.20) from non--negative integral values $\ell$ 
to complex $\lambda$ by considering the integral equation
\beq
[\I-g L(\lambda,k)]\,v(\lambda,k;g;\cdot)=v_0(\lambda,k;\cdot),
\label{AG.36}
\eeq
where, according to Theorems \ref{the:AG.1} and \ref{the:AG.2}, 
the vector--valued and operator--valued functions
$(\lambda,k) \mapsto v_0(\lambda,k;\cdot)$ and 
$(\lambda,k) \mapsto L(\lambda,k)$ are well--defined in the set 
$\C^+_{-\frac{1}{2}}\times {\overline\Omega}_\alpha^{\,\cut}$,
and holomorphic in the domain $\C^+_{-\frac{1}{2}}\times \Omega_\alpha^{\cut}$.
Then there holds the following analog of Theorem \ref{theorem:3'},
in which we have put 
$\underline{D}_{\,\alpha}(V;g)=D_\alpha(V;g)\cap\Omega_\alpha^{\cut}$ 
(see Eq. \eqref{AG.35'}).

\skd

\begin{theorem}
\label{the:AG.7}
For any nonlocal potential $V$ in a class 
$\cN_{\wea}^{\,\gamma}$, 
the inhomogeneous equation \eqref{AG.36}
admits for any $g\in\C$ and $(k,\lambda)\in{\underline D}_{\,\alpha}(V;g)$ 
a unique solution $v(\lambda,k;g;\cdot)$ in $X_{\wea}$, 
which is well--defined by the formula:
\beq
v(\lambda,k;g;\cdot)=R(\lambda,k;g) \, v_0(\lambda,k;\cdot).
\label{AG.37}
\eeq
Furthermore, for any $g$, the function $v(\lambda,k;g;\cdot)$
is holomorphic in ${\underline D}_{\,\alpha}(V;g)$ as a 
vector--valued function, taking its values in $X_{\wea}$.
\end{theorem}

\begin{proof}
The solution \eqref{AG.37} of equation \eqref{AG.36}, which follows from \eqref{AG.26},
defines $v(\lambda,k;g;\cdot)$ as an element of 
$X_{\wea}$, in view of the fact that $v_0(\lambda,k;\cdot)$ belongs to 
$X_{\wea}$ (see Theorem \ref{the:AG.1})
and that $R(\lambda,k;g)$ is a bounded operator in 
$X_{\wea}$ (see Theorem \ref{the:AG.6})
for all $g\in\C$ and $(k,\lambda)\in D_\alpha(V;g)$. 
Moreover, the holomorphy properties of 
$R(\lambda,k;g)$ and $v_0(\lambda,k;\cdot)$, established in
Theorems \ref{the:AG.6} and \ref{the:AG.1} respectively, 
imply the corresponding property for $v(\lambda,k;g;\cdot)$
in view of Lemma \ref{lemma:B1'}(ii).
\end{proof}

Now, by taking into account the expression \eqref{AG.27} of 
$R(\lambda,k;g)$ in Eq. \eqref{AG.37}, we can re--express 
$v(\lambda,k;g;\cdot)$ as follows, for all 
$(\lambda,k)\in{\underline D}_{\,\alpha}(V;g)$: 
\begin{align}
v(\lambda,k;g;\cdot) &= \frac{u(\lambda,k;g;\cdot)}{\sigma(\lambda,k;g)}, \label{AG.38} \\
\intertext{where:}
u(\lambda,k;g;\cdot) &\doteq
\left[\sigma(\lambda,k;g)+g N(\lambda,k;g)\right]v_0(\lambda,k;\cdot).\label{AG.39}
\end{align}
In fact, in view of Theorems \ref{the:AG.3} and \ref{the:AG.4},
for every $(\lambda,k)\in \C^+_{-\frac{1}{2}}\times {\Omega}_\alpha^{\cut}$,
one can act with the multiplier $\sigma(\lambda,k;g)$  
and with the Hilbert--Schmidt operator $N(\lambda,k;g)$ on the vector 
$v_0(\lambda,k;\cdot)\in X_{\wea}$. One can then state  
(also in view of Lemma \ref{lemma:B1'}(ii)) the following theorem.

\skd

\begin{theorem}
\label{the:AG.7'}
For any nonlocal potential $V$ in a class
$\cN_{\wea}^{\,\gamma}$, the quantity 
$u(\lambda,k;g;\cdot)$ is well--defined for every $g\in\C$ and
every $(\lambda,k)\in\C^+_{-\frac{1}{2}}\times {\Omega}_\alpha^{\cut}$,
as a vector in $X_{\wea}$, depending holomorphically on $(\lambda,k)$, such that 
\beq
\left\|u(\lambda,k;g;\cdot)\right\|_{\wea}\leqslant
\left\{|\sigma(\lambda,k;g)| + |g| \left\|N(\lambda,k;g)\right\|_\mathrm{HS}\right\} 
\, \|v_0(\lambda,k;g;\cdot)\|_{\wea}.
\label{AG.41-ex}
\eeq
\end{theorem}

\begin{remark}
It is clear that the various functions $(\lambda,k)\mapsto \sigma(\lambda,k;g)$,
$N(\lambda,k;g)$, $u(\lambda,k;g;\cdot)$ are respectively CAM interpolations of 
the corresponding sequences of functions $k\mapsto \sigma_\ell(k;g)$, $N_\ell(k;g)$,
$u_\ell(k;g;\cdot)$ (see Eq. \eqref{f:73}), since all the integral relations of the 
present section reduce legitimately to the corresponding equations of Section 
\ref{se:rotationally} for $\lambda =\ell \in \N$ (any potential in 
$\cN_{\wea}^{\,\gamma}$ being contained in $\cN_{\wea}$).
We also stress the fact that while the functions $\sigma$ and $N$
are holomorphic in $\C^+_{-\frac{1}{2}}\times \Pi_\alpha^{\cut}$, 
the domain of the function $u$ is only 
$\C^+_{-\frac{1}{2}}\times \Omega_\alpha^{\cut}$, which is the maximal 
domain in which the function $v_0$ can be proved to be holomorphic 
(see Theorem \ref{the:AG.1}).
\label{rem:5}
\end{remark}

Next, by substituting the complex variable $\lambda$ to the integer $\ell$ in formulae
\eqref{f:72'}, \eqref{f:73}, \eqref{f:82}, and \eqref{f:83},
we shall introduce an analytic interpolation $T(\lambda,k;g)$
of the sequence of partial scattering amplitudes 
$\{T_\ell(k;g)\}_{\ell=0}^\infty$ by the following formula:
\begin{align}
T(\lambda,k;g) &= -g\int_0^{+\infty}\hspace{-0.2cm}
R'\,j_\lambda(kR') \, v(\lambda,k;g;R')\,\rmd R' = 
\frac{\Theta(\lambda,k;g)}{\sigma(\lambda,k;g)}, \label{AG.40} \\
\intertext{where:} 
\Theta(\lambda,k;g) &= -g \int_0^{+\infty}\hspace{-0.2cm}
R'\,j_\lambda(kR') \, u(\lambda,k;g;R')\,\rmd R'. \label{AG.40'}
\end{align}
$T(\lambda,k;g)$ will be called here the CAM--partial--scattering--amplitude.
We can in fact prove the following 

\skd

\begin{theorem}
\label{the:AG.8}
For every nonlocal potential $V$ belonging to the class 
$\cN_{\wea}^{\,\gamma}$, the following properties hold:
\begin{itemize}
\item[\rm (i)] the function $(\lambda,k,g)\mapsto \Theta(\lambda,k;g)$ is defined and 
holomorphic in $\C^+_{-\frac{1}{2}}\times {\Omega}_\alpha^{\cut}\times \C$;
at fixed $g$, it is uniformly bounded in any sector
$\overline{D}^{\,(\delta)}_{\gamma,\alpha}$ (for any $\delta >0$);
\item[\rm (ii)] for every $g$ the function $T(\lambda,k;g)$
is meromorphic in $\C^+_{-\frac{1}{2}}\times {\Omega}_\alpha^{\cut}$
and holomorphic in ${\underline D}_{\,\alpha}(V;g)$.
Moreover, for every $\gamma'$, with $\gamma'<\gamma$,
there exists a number $\delta_0$ (depending on $\gamma'$ and $g$) 
such that $T(\lambda,k;g)$ is holomorphic in the corresponding 
truncated sector $\overline D^{\,(\delta_0)}_{\gamma'\!,\alpha}$
and satisfies an exponentially decreasing bound of the following form:
\beq 
|T(\lambda,k;g)|\leqslant {\bf c}_{\gamma'\!,g} \ e^{-(\gamma -\gamma')\Real\lambda}.
\label{AG.41}
\eeq
\end{itemize}
\end{theorem}

\begin{proof}
Since for every $(\lambda,k)\in \C^+_{-\frac{1}{2}}\times {\Omega}_\alpha^{\cut}$
the function $Rj_\lambda(kR)$ is a vector in the dual space 
$X^*_{\wea}$, the quantity $\Theta(\lambda,k;g)$
is well--defined by Eq. \eqref{AG.40'} and such that 
\beq
\begin{split}
&|\Theta(\lambda,k;g)| \leqslant 
|g|\, \left\|\, \cdot \ j_\lambda(k\cdot)\right\|^*_{\wea} \
\|u(\lambda,k;g;\cdot)\|_{\wea} \\
& \quad\leqslant
|g| \, \left\|\, \cdot \ j_\lambda(k\cdot)\right\|^*_{\wea} \
\|v_0(\lambda,k;g;\cdot)\|_{\wea}
\ \left\{|\sigma(\lambda,k;g)| + |g| \, \|N(\lambda,k;g)\|_\mathrm{HS}\right\}. 
\end{split}
\label{AG.42}
\eeq
The fact that the function $(\lambda,k,g)\mapsto \Theta(\lambda,k;g)$ is holomorphic
in $\C^+_{-\frac{1}{2}}\times {\Omega}_\alpha^{\cut}\times \C$
is then directly implied by Lemma \ref{lemma:B2}. By now taking the majorizations 
\eqref{AG.12}, \eqref{AG.15}, \eqref{AG.30}, and \eqref{AG.34} into account, 
we derive from \eqref{AG.42} the following global bound:
\beq
\begin{split}
&|\Theta(\lambda,k;g)| \leqslant 
\frac{\pi}{2}\,|g| \, A^2_\varepsilon \,\cK \, C(V_*)
\ e^{-\gamma\Real\lambda} \ e^{3\pi |\Imag\lambda|} 
\left(\frac{3}{2}+\frac{1}{\pi(\Real\lambda+\frac{1}{2})}\right)^2 \\
&\quad \times 
\left\{1+[\Phi+\Psi]\left(|g|\,C(V_*)\,A_\varepsilon^2\,M_\gamma(\lambda)\right)\right\},
\end{split}
\label{AG.43}
\eeq
with $M_\gamma(\lambda)$ given by Eq. \eqref{AG.4'}. In view of Eq. \eqref{AG.5},
one then easily checks that for $(\lambda,k)\in\overline{D}^{\,(\delta)}_{\gamma,\alpha}$ 
(for any $\delta >0$), the r.h.s. of \eqref{AG.43} is uniformly majorized by 
$\frac{\pi}{2}|g| A^2_\varepsilon \cK C(V_*)\, e^{\gamma/2}\left(\frac{3}{2}+\frac{1}{\pi \delta}\right)^2
\times\left\{1+[\Phi+\Psi]\left(|g|C(V_*)\,A_\varepsilon^2\,c\,\cK \, e^{\gamma/2} 
\left(1+\frac{1}{2\delta}\right)\right)\right\}$,
which is independent of $\lambda$ and $k$ (here one takes into account the 
regularity properties of the entire functions $\Phi$ and $\Psi$, defined by 
Eqs. \eqref{f:Phi} and \eqref{f:Psi}). This ends the proof of (i).

\vskip 0.1cm

The first part of (ii) is a straightforward 
consequence of Eq. \eqref{AG.40}. The holomorphy property of 
$T(\lambda,k;g)$ and its majorization of the form \eqref{AG.41}
in the truncated sectors $\overline D^{\,(\delta_0)}_{\gamma'\!,\alpha}$
is implied by the fact that $|\sigma(\lambda,k;g)-1|$ tends uniformly to zero
for $\lambda$ going to infinity (see Theorem \ref{the:AG.3}),
and that $|\Theta(\lambda,k;g)|$ satisfies the global bound
\eqref{AG.43}. When using the latter, one now takes into account 
Eq. \eqref{AG.5} with $\gamma$ replaced by $\gamma'$ ($\gamma' <\gamma$),
which finally yields the exponential factor on the r.h.s. of \eqref{AG.41}.
\end{proof}

\section{Watson resummation of the partial wave amplitudes.
Resonances and antiresonances}
\label{se:AG.5}

Since each class of nonlocal potentials $\cN_{\wea}^{\,\gamma}$
is a subclass of $\cN_{\wea}$,
one can apply all the results of scattering theory obtained in 
Subsections \ref{subse:partial-wave} and \ref{subse:aggiunte4.4}
to the case of potentials $V$ in any given class   
$\cN_{\wea}^{\,\gamma}$.

In particular, we shall rely on the expansion
\eqref{f:119} of the scattering amplitude $F(k,\cos\theta ;g)$
in terms of the partial waves
$a_\ell(k;g)=T_\ell(k;g)/k$ (see Eq. \eqref{f:118}), 
whose finiteness at $k=0$ is a consequence of the threshold behavior \eqref{f:120}
(see Proposition \ref{pro:k2}). 
A physically important related function, introduced in \eqref{f:118}, is the 
phase--shift $\delta_\ell(k;g)$ (see also Eqs. \eqref{f:84}, \eqref{f:85}, and \eqref{f:86}).

Then, in Subsection \ref{subse:aggiunte4.3} an analytic interpolation $T(\lambda,k;g)$
of the sequence $\{T_\ell(k;g)\}_{\ell=0}^\infty$ has been defined 
as a meromorphic function in 
$\C^+_{-\frac{1}{2}}\times {\Omega}_\alpha^{\cut}\times \C$, whose various properties
have been listed in Theorem \ref{the:AG.8}. 

In view of the exponential decrease properties of $a_\ell(k;g)$ (resp.,
$a(\lambda,k;g)$) for $\ell$ (resp., $\Real\lambda$) tending to infinity,
specified in formulae \eqref{f:117} and \eqref{f:118} (resp., \eqref{AG.41}), 
we can safely apply the Watson resummation method
to expansion \eqref{f:119}, written for any fixed values of 
$k$ in $\R^+$, $g$ real and $\theta$ in the interval $0<\theta\leqslant\pi$. It yields:
\beq
\sum_{\ell=0}^\infty(2\ell+1) \, a_\ell(k;g) \, P_\ell(\cos\theta)=
\frac{\rmi}{2}\int_C\frac{(2\lambda+1)\,a(\lambda,k;g)\,
P_{\lambda}(-\cos\theta)}{\sin\pi\lambda}\,\rmd\lambda,
\label{p:1}
\eeq
where the path $C$ encircles the positive real semi--axis 
in the $\lambda$--plane (see Fig. \ref{fig:AGG.1}a).
This path must be chosen with some care in order to include only 
the singularities of the integrand on the
r.h.s. of formula \eqref{p:1} which are the poles 
generated by the zeros of $\sin\pi\lambda$: other
singularities in the first and in the fourth quadrants 
of the $\lambda$--plane, but close to
the real semi--axis, must be avoided.

We now introduce for every $(\gamma',\delta)$ such that $0<\gamma' <\gamma$ and
$0<\delta < \frac{1}{2}$, the path $\Gamma = \Gamma_{\gamma'}^{(\delta)}$,
whose support is the boundary of the truncated angular sector
$\Lambda^{(\delta)}_{\gamma'}$ (see Eq. \eqref{AG.5})
and whose orientation is given by continuous
distortion from $C$ to $\Gamma$ in the $\lambda$--plane (see Fig. \ref{fig:AGG.1}b).
According to Theorem \ref{the:AG.8}, the number $N$ of
poles $\lambda=\lambda_n(k,g)$ of $T(\lambda,k;g)$ which are contained in
$\Lambda^{(\delta)}_{\gamma'}$ is \emphsl{finite},
since all these poles must be confined in the bounded
region $\Lambda^{(\delta)}_{\gamma'}\setminus\Lambda^{(\delta_0)}_{\gamma'}$. 
Then, in view of the exponentially decreasing majorization 
\eqref{AG.41} on $a(\lambda,k;g)=T(\lambda,k;g)/k$ and of 
the following bound on the Legendre function
(see \cite[p. 709, formula II.107]{Bros5}):
\beq
|P_\lambda(\cos\theta)| \leqslant 
C(\cos\theta) \ e^{\pi|\Imag\lambda|} \qquad (0 \leqslant \theta < \pi),
\label{p:3}
\eeq
which is compensated by $|\sin \pi \lambda|^{-1}$,
the integration contour $C$ in Eq. \eqref{p:1} 
can be legitimately replaced by $\Gamma$,
provided the contributions of the poles $\lambda_n(k;g)$ be taken into account via the
residue theorem. We thus obtain: 
\beq
\begin{split}
& \frac{\rmi}{2}\int_C
\frac{(2\lambda+1)\,a(\lambda,k;g)\,P_{\lambda}(-\cos\theta)}{\sin\pi\lambda}\,\rmd\lambda \\
&\quad =\frac{\rmi}{2}
\int_\Gamma\frac{(2\lambda+1)\,a(\lambda,k;g)\,P_{\lambda}(-\cos\theta)}{\sin\pi\lambda}\,\rmd\lambda
-\pi\sum_{n=1}^N\frac{\rho_n(k;g)\,P_{\lambda_n}(-\cos\theta)}{\sin\pi\lambda_n(k;g)}.
\end{split}
\label{p:2}
\eeq
According to the analysis of Subsection \ref{subsubse:some},
the poles $\lambda_n(k;g)$ may lie  
either in the upper or in the lower half--plane. 
In the generic case these poles can be considered as first order poles
(see also the considerations on the spectrum of the resolvent
and the physical arguments given by R.G. Newton in \cite{Newton3};
see Sections 9.1, p. 240 and 9.3, p. 257).
In \eqref{p:2}, the factors $\rho_n(k;g)$ are the
corresponding residues of the function $[(2\lambda+1)a(\lambda,k;g)]$.

\noindent
The previous analysis can thus be summarized  in the following

\skd

\begin{theorem}
\label{the:28}
For every nonlocal potential $V\in\cN_{\wea}^{\,\gamma}$ 
the following representation of the total scattering amplitude holds:
\beq
F(k,\cos\theta;g) = 
\frac{\rmi}{2}\int_\Gamma\frac{(2\lambda+1)\,a(\lambda,k;g)\,
P_{\lambda}(-\cos\theta)}{\sin\pi\lambda}\,\rmd\lambda
-\pi \sum_{n=1}^N\frac{\rho_n(k;g)\,P_{\lambda_n}(-\cos\theta)}
{\sin\pi\lambda_n(k;g)},
\label{p:4}
\eeq
for $0<\theta\leqslant\pi$, $k\in\R^+$, and $\Gamma$ denotes any choice 
$\Gamma = \Gamma_{\gamma'}^{(\delta)}$ such that  
$\gamma' <\gamma$ and $0<\delta < \frac{1}{2}$.
\end{theorem}

\subsection{Extension of representation $\boldsymbol{\protect\eqref{p:4}}$ 
in the two complex variables $\boldsymbol{k}$ and $\boldsymbol{\cos\theta}$.} 
\label{subse:extension}

Formula (II.107) of \cite{Bros5} provides us with the following majorization
on $P_\lambda(\cos\theta)$, which is valid for all $\cos\theta$ in the cut--plane 
$\C\setminus ]-\infty,-1]$ and $\lambda\in \C^+_{-\frac{1}{2}}$:
\beq 
|P_\lambda(\cos\theta)| \leqslant 
C(\cos\theta) \, e^{\pi|\Imag\lambda|} \, 
e^{|\Imag \theta| \Real\lambda}, 
\label{p:6}
\eeq
where $C(\cos\theta)$ denotes a suitable locally bounded function.
Then, in view of the latter and of bound \eqref{AG.41}, the 
integrability condition of the background integral of 
\eqref{p:4} on a given path $\Gamma_{\gamma'}^{(\delta)}$ is: 
\beq
|\Imag \theta| < \gamma -\gamma',
\label{p:7}
\eeq
which means that the corresponding integral representation of 
$F(k,\cos\theta; g)$ is valid 
and defines $F$ as an analytic function of the two variables 
$(k,\cos\theta)$ in the domain
$\Omega_\alpha^{\cut}\times \{\cE_{\gamma-\gamma'}\setminus [1,+\infty[\}$ 
($\cE_{\gamma-\gamma'}$ denoting the ellipse with foci $+1$ and
$-1$ and major semi--axis $\cosh (\gamma-\gamma')$).
Note that the maximal ellipse $\cE_\gamma$ 
of analyticity is obtained for a choice of
$\gamma'$ arbitrarily small, namely for the choice 
of the original path $\Gamma=C$.

Finally, by using this path $C$, it can be seen that  
$F(k,\cos\theta; g)$ \emphsl{can be analytically continued in 
the product domain} $\Omega_\alpha^{\cut}\times\cE_{\gamma}$. 
This is based on the following argument:

\vskip 0.1cm

(i) As noticed in \cite[p. 7]{Newton4}, the discontinuity of the
integral of \eqref{p:1} across the cut 
$\cos\theta\in [1,\cosh\gamma[$ is proportional to  
$\int_C (2\lambda +1) \, a(\lambda,k;g) \, P_\lambda(\cos\theta)\, \rmd\lambda$,
which vanishes in view of the Cauchy theorem. 

\vskip 0.1cm

(ii) For $\cos\theta$ tending to $1$, the limit of the integral
\eqref{p:1} is infinite since $\lim P_\lambda(-z) = \infty$
for $z$ tending to one. However, since $P_\lambda(-z)$ is bounded 
by a multiple of $\ln (z-1)$, a similar bound holds for the 
integral in the neighborhood of $z=1$; so this point cannot be an 
isolated singularity for the holomorphic function $F(k,z;g)$.

\begin{remark}
\label{rem:6}
In the applications of the CAM method to high energy physics, 
it is essential that the path $\Gamma$ of the
so--called \emphsl{background integral} of the Watson--Regge 
representation \eqref{p:4} can be taken along the imaginary axis at
$\Real\lambda=-\frac{1}{2}$. Of course, this requires that the region 
of the $\lambda$--plane in which $a(\lambda,k;g)$ is decreasing uniformly 
with respect to $\Real \lambda$ is the full half--plane 
$\C^+_{-\frac{1}{2}}$; this condition is fulfilled if 
$ F(k,\cos\theta;g)$ is analytic in a cut--plane of $\cos\theta$ for  
each $k$ fixed (i.e., in particular if the Mandelstam representation is
satisfied; for example, this is the case for the scattering 
amplitude of the theory of Yukawa--type local potentials).

Indeed, if such a property is valid, use can be made of the 
following asymptotic behaviour of the Legendre function 
$P_\lambda(z) \ (z=\cos\theta)$ as $|z|\to +\infty$, 
for $\Real\lambda \geqslant -\frac{1}{2}$ (see \cite[p. 6, formula 2.5]{Newton4}):
\beq
\label{p:5}
P_\lambda(z)\simeq\pi^{-1/2} \, 2^\lambda \,
z^\lambda \, \frac{\Gamma(\lambda+\frac{1}{2})}{\Gamma(\lambda+1)}
\qquad (z\in\C\setminus(-\infty,-1]),
\eeq
which implies that the background integral 
is of the order $z^{-1/2}$ as $|z|\to+\infty$.
Assuming that the number of poles of the scattering amplitude is finite 
(as in the case of Yukawian potentials), one obtains the leading term in the 
asymptotic behaviour of the scattering amplitude as $|z|\to +\infty$
from the pole term with the largest real part. 

However, in the present framework, we are working in the
physical region of $\cos\theta$ ($-1\leqslant\cos\theta\leqslant 1$),
and we are not interested in the asymptotic behaviour of the 
scattering amplitude for large momentum transfer; accordingly, 
we shall fully exploit representation \eqref{p:4} with its path integral 
along $\Gamma^{(\delta)}_{\gamma'} $, as shown in Fig. \ref{fig:AGG.1}b.
\end{remark}

\begin{figure}[tb]
\begin{center}
\leavevmode
\psfig{file=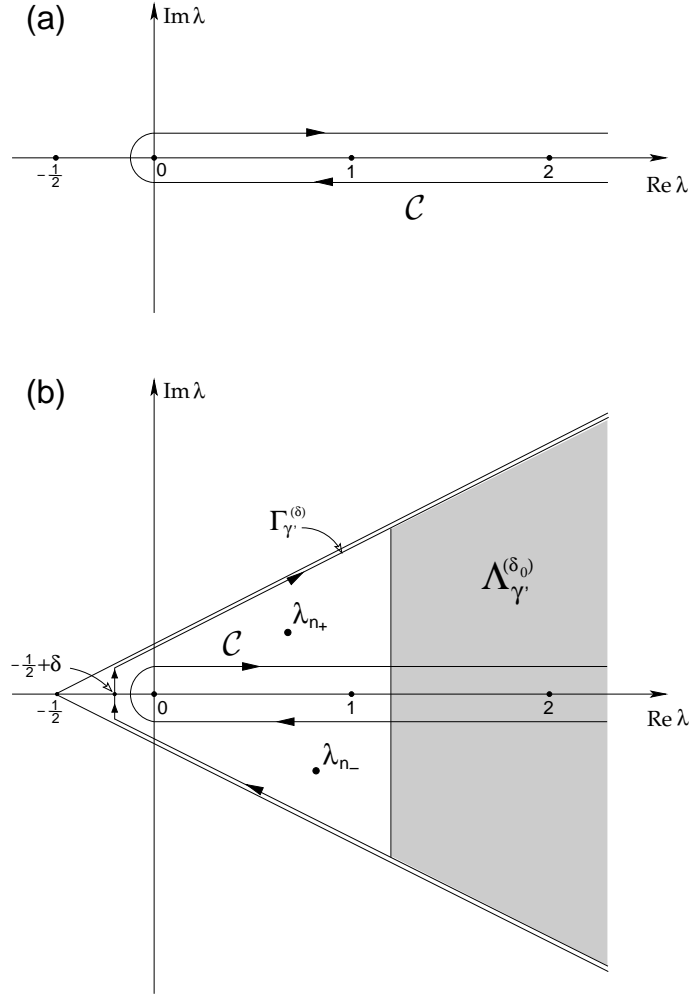,width=9cm}
\caption{\label{fig:AGG.1}
(a): the integration path $\cC$ in formula \eqref{p:1}.
(b): the integration path $\Gamma =\Gamma_{\gamma'}^{(\delta)}$
in formulae \eqref{p:2} and \eqref{p:4}. The grey sector 
$\Lambda^{(\delta_0)}_{\gamma'}$ represents a region
without poles for the function $a(\lambda,k;g)$.
Typical poles, namely $\lambda_{n_+}$ and $\lambda_{n_-}$, 
are indicated in the first and fourth quadrants, respectively.}
\end{center}
\end{figure}

\subsection{Analysis of resonances and antiresonances as contributions of poles 
in the $\boldsymbol{\lambda}$--plane.} 
\label{subse:analysis}

We shall now apply the previous Watson--type representation \eqref{p:4}
of the total scattering amplitude $F(k,\cos\theta;g)$
to the computation of the set of partial waves $a_\ell (k;g)$, $(\ell=0,1,\ldots)$,
which are defined by the standard inversion formula of expansion \eqref{f:119}: 
\beq
a_\ell(k;g) =\frac{1}{2}\int_{-1}^{+1} F(k,\cos\theta;g) \, P_\ell (\cos\theta) \,\rmd(\cos\theta).
\label{p:8'}
\eeq
For this purpose, we shall use the basic projection formula 
(see \cite[Vol. 1, p. 170, Eq. (7)]{Bateman}):
\beq
\frac{1}{2}\int_{-1}^{1}P_\ell(\cos\theta)\,P_\lambda(-\cos\theta)\,\rmd(\cos\theta)
=\frac{\sin\pi\lambda}{\pi(\lambda-\ell)(\lambda+\ell+1)}.
\label{p:7'}
\eeq
By plugging the expression \eqref{p:4} of $F$ in \eqref{p:8'} and applying \eqref{p:7'}
to the various terms, one obtains:
\beq
a_\ell(k;g) =
\frac{\rmi}{2\pi}\int_\Gamma \frac{(2\lambda + 1)\,a(\lambda,k;g)}
{(\lambda-\ell)(\lambda +\ell +1)}\, \rmd\lambda
-\sum_{n=1}^N \frac{\rho_n(k;g)}{(\lambda_n(k;g)-\ell)(\lambda_n(k;g) +\ell +1)}.
\label{p:9'}
\eeq
In the r.h.s. of the latter, we distinguish the so--called 
``background integral" over $\Gamma$ (which is always convergent in view of the
exponential decrease property \eqref{AG.41}) from the individual contributions 
of the poles $\lambda=\lambda_n(k;g)$ of $a(\lambda,k;g)$. 

We will show that under certain simple assumptions these poles 
can be seen to induce properties of the partial waves $a_\ell(k;g)$,
which are characteristic of \emph{resonances} and \emph{antiresonances};
these properties are:

\vskip 0.1cm

\noindent
(i) the rapid variation in the momentum variable $k$
(or the energy $E= k^2$), 
including the passage through $\frac{\pi}{2}$ ($\mathrm{mod.} \pi$),
of the phase--shift function $k \mapsto \delta_\ell(k;g) $, which satisfies 
(in view of \eqref{f:86}):
\beq
a_\ell(k;g) = \frac{e^{\rmi\delta_\ell(k;g)}\sin\delta_\ell(k;g)}{k}
\qquad(k\in\R^+),
\label{p:10'}
\eeq
and therefore:
\beq 
\delta_\ell(k;g) = \Arg a_\ell(k;g) \ \mathrm{mod.}\ \pi.
\label{p:11'}
\eeq
Resonances (resp., antiresonances) are characterized by
the upward (resp., downward) passage of the phase--shift through $\pm\frac{\pi}{2}$
at a certain value $ k=k_\mathrm{r}$ (resp., $k=k_\mathrm{ar}$) 
at which (in view of \eqref{p:10'}) $a_\ell(k;g) = \rmi/k$.

In the basic literature on the subject (see, e.g., \cite[Chapter 2, Subsection 2.11 (c)]{Nussenzveig}), 
the quantity
$2 \, {\partial\delta_\ell/\partial E}$ ($=\frac{1}{k}\times{\partial\delta_\ell/\partial k}$), 
whose positive or negative sign plays a role in the previous description of
resonances or antiresonances, has been interpreted in terms of the  
\emph{time--delay} or \emph{time--advance} that the incident wave--packet 
undergoes in the scattering process, in a sense which has been introduced by Eisenbud. 

\vskip 0.1cm

(ii) The production of a ``bump" around $k= k_\mathrm{r}$ or $k_\mathrm{ar}$ 
in the plot of the function $k\mapsto |a_\ell(k;g)|$
and therefore of the cross--section $\sigma_\mathrm{c}(k;g)$, since one can write
\beq
\sigma_\mathrm{c} (k;g)= 4\pi \,(2\ell +1) \,|a_\ell(k;g)|^2
+ 4\pi \sum_{\ell'\neq \ell} (2\ell' + 1) \, |a_{\ell'}(k;g)|^2,
\label{p:12'}
\eeq
the sum at the r.h.s. of the latter being subdominant near 
$k= k_\mathrm{r} $ or $k_\mathrm{ar}$.

\vskip 0.3cm

\noindent
\textbf{Assumptions.} We shall concentrate on the function $a_\ell(k;g)$ for given fixed values of $\ell$
and $g$, and assume that among the various poles $\lambda=\lambda_n(k;g)$ of  
the meromorphic function $a(\lambda,k;g)$ which contribute to the sum on the r.h.s. of 
Eq. \eqref{p:9'}, there exists a distinguished pole, denoted simply by $\lambda(k;g)$,
with the corresponding residue $\rho(k;g) =|\rho(k;g)| e^{\rmi\varphi(k;g)}$.
We assume that the corresponding analytic function $k\mapsto \lambda(k;g)$
satisfies the following properties:

\vskip 0.2cm

There exists a finite interval $I =\{k\,:\,0<k_{\min} < k < k_{\max}\}$ such that for $k\in I$:
\begin{itemize}
\item[(a)] The function $k\mapsto\alpha(k;g)=\Real\lambda(k;g)$ is an increasing function 
such that for a certain value $k=k_\ell\in I$ one has: $\alpha(k_\ell;g)=\ell$. \\[-5pt]
\item[(b)] The function  $k\mapsto\beta(k;g)=\Imag\lambda(k;g)$ is such that  
$0<|\beta(k;g)|\ll 1$ and 
$|{(\partial\beta/\partial k)}(k;g)| \ll {(\partial\alpha/\partial k)}(k;g)$. \\[-5pt]
\item[(c)] ``One--pole dominance": $|a_\ell(k;g)-\widehat{a}_\ell(k;g)|\ll 1$, where we have put: 
\beq
\widehat{a}_\ell(k;g)= 
\frac{-\rho(k;g)}{[\lambda(k;g)-\ell]\,[\lambda(k;g)+\ell+1]}.
\label {p:13"}
\eeq
\end{itemize}
We shall then consider that the unitarity relation 
\beq
k \, |a_\ell(k;g)|^2 - \Imag a_\ell(k;g) =0,
\label{p:13'}
\eeq
(implemented by the parametrization \eqref{p:10'}) is 
approximately satisfied by this one--pole dominant contribution $\widehat{a}_\ell(k;g)$ 
itself. Applying this approximation yields the following relation between 
the modulus and the argument $\varphi$ of the residue $\rho$:
\beq
|\rho| = \frac{1}{k} 
\left\{(\sin\varphi)\,\left[-(\alpha-\ell)(\alpha+\ell +1)+\beta^2\right]
+(\cos\varphi)\,\beta\,(2\alpha +1)\right\},
\label {p:14'}
\eeq
which, in particular, yields for $k=k_\ell$ (in view of (a)):
\beq
|\rho(k_\ell;g)| = 
\frac{(\cos\varphi)\,\beta(k_\ell)\,(2\ell +1)}{k_\ell} + O\left([\beta(k_\ell)]^2\right),
\label {p:15'}
\eeq
and therefore, in view of assumption (c):
\beq
|a_\ell(k_\ell;g)|\,\approx\,
|\widehat{a}_\ell(k_\ell;g)| \approx \frac{|\cos\varphi(k_\ell;g)|}{k_\ell}. 
\label {p:15"}
\eeq

\subsubsection{Variation of the phase--shift near $\boldsymbol{k=k_\ell}$.}
\label{subsubse:variation}

In view of assumption (c) and of \eqref{p:11'}, and 
by taking the arguments of both sides of Eq. \eqref{p:13"}, one then obtain:
\beq
\delta_\ell (k;g) \approx 
\varphi(k;g) \pm \pi -\Arg(\lambda(k;g)-\ell)-\Arg(\lambda(k;g)+\ell+1).
\label {p:16'}
\eeq
A simple geometrical analysis, making an essential use of assumption (a),
shows that: 

\vskip 0.1cm

(1) \ If $\beta(k;g)$ is positive, the function $k \mapsto\delta_\ell (k;g)$ 
admits an upward variation from a value of the form $(\varphi(k;g) +\varepsilon)$ to 
$(\varphi(k;g)+\pi-\varepsilon)$ for $k$ varying on a short interval
$[k_-,k_+]$ centered at $k_\ell$. (In view of the assumption 
$\beta(k;g) \ll 1$, the size of this interval  
is such that: $\ell-1 < \alpha(k_-)<\ell<\alpha(k_+)< \ell+1$,
provided $\varepsilon$ is chosen sufficiently small).
We thus conclude that in a ``generic way" (i.e., except if 
$\varphi(k;g) = \pm\frac{\pi}{2}$), the interval $[k_-,k_+]$ 
will contain a value $k=k_\mathrm{r}$ at which
$\delta_\ell (k_\mathrm{r};g)= \pm \frac{\pi}{2}$,
which therefore exhibits the typical behaviour of a resonance
at $k=k_\mathrm{r}$.

\vskip 0.1cm

(2) \ Similarly, if $\beta(k;g)$ is negative, the function  
$k \to \delta_\ell (k;g)$ admits a downward variation from 
$(\varphi(k;g)-\varepsilon)$ to $(\varphi(k;g)-\pi+\varepsilon)$ 
for $k$ varying on a short interval $[k_-,k_+]$ centered at  
$k_\ell$. The latter also contains (in a generic way) a value 
$k=k_\mathrm{ar}$ at which $\delta_\ell (k_\mathrm{ar};g)= \pm\frac{\pi}{2}$, 
which then exhibits the typical behaviour of an antiresonance
at $k=k_\mathrm{ar}$.

\vskip 0.3cm

\noindent
\textbf{Bump for the cross--section near $\boldsymbol{k=k_\ell}$.} \
In both cases (1) and (2) of the previous analysis, one has, 
in view of \eqref{p:10'}:
\beq
|a_\ell(k_{\mathrm{r},\mathrm{ar}})| = \frac{1}{\ k_{\mathrm{r},\mathrm{ar}}},
\label{p:17'}
\eeq
(where $k_{\mathrm{r},\mathrm{ar}} = k_\mathrm{r}$ or $k_\mathrm{ar}$),
and therefore the plot of the partial wave $a_\ell(k;g)$ exhibits a bump 
around $k=k_{\mathrm{r},\mathrm{ar}}$ which is tangent to the upper 
limiting curve $|a_\ell^{\max}(k;g)|= {1/k}$. One can also notice the 
difference between \eqref{p:17'} and the value \eqref{p:15"} of
$|\widehat{a}_\ell(k_\ell;g)|$, which (if $|\cos\varphi(k_\ell;g)|\neq 1$)
expresses the rapid variation of the phase--shift between $k_\ell$ and 
$k_{\mathrm{r},\mathrm{ar}}$. Finally, in view of \eqref{p:12'}, the plot of 
the cross--section $\sigma_\mathrm{c}(k;g)$ will also present a bump in an interval
of the momentum variable containing the values $k_\ell$ and 
$k_{\mathrm{r},\mathrm{ar}}$.

\subsection{Connection between descriptions of phenomena in the 
$\boldsymbol{\lambda}$--plane and in the $\boldsymbol{k}$--plane.} 
\label{subse:connection}

The description of resonances and antiresonances of a given partial wave 
$a_\ell(k;g)$, which has been given above, appeared to be completely symmetric. 
In both cases indeed, it was based on the assumption of
a dominant one--pole approximation of the partial scattering function
$a(\lambda,k;g)$ in the complex $\lambda$--plane,
such that the dominant pole $\lambda=\lambda(k;g)$ be located at a very small
distance $|\beta(k;g)|$ from the real axis; the two cases are distinguished 
from each other by the sign of the function $k \mapsto \beta(k;g)$.

\vskip 0.1cm

We are now going to show that, in spite of the previous apparent symmetry,
these two cases necessarily correspond to completely different types
of analyticity properties of the dominant function $k\mapsto\widehat{a}_\ell(k;g)$ 
in the complex $k$--plane.

\vskip 0.3cm

\noindent
(1) \emph{The case of resonances.}

In the r.h.s. of \eqref{p:13"}, it is the factor 
$\lambda(k;g)-\ell = \alpha(k;g) -\ell +{\rmi} \beta (k;g)$
at the denominator which is responsible for the rapid variation of $\delta_\ell(k;g)$  
near $k=k_\ell$, and whose vanishing in the complex $k$--plane must be analysed. 
Since the function $(k\ \mathrm{real})\mapsto \lambda(k;g)$
is holomorphic in a small complex neighbourhood $\cV$ of $k\ell$, one can postulate
the validity (in $\cV$) of the following first--order Taylor approximation
$\lambda(k;g)\approx\ell+\rmi\beta(k_\ell;g)+
\frac{\partial\lambda}{\partial k}(k_\ell;g)\ (k-k_\ell)$,
which therefore yields:
\beq
\lambda(k;g)-\ell \approx 
\frac{\partial\lambda}{\partial k}(k_\ell;g) \ [(k-k_\ell) +\rmi\gamma],
\label{p:11}
\eeq
where we have put:
\beq
\gamma = \frac{\beta(k_\ell;g)}{\frac{\partial\lambda}{\partial k}(k_\ell;g)}.
\label{p:12"}
\eeq
Since it was supposed that 
$\left|\frac{\partial\beta}{\partial k}\right| \ll \frac{\partial\alpha}{\partial k}$,
we can say that 
$\frac{\partial\lambda}{\partial k}(k_\ell;g)\approx\frac{\partial\alpha}{\partial k}$, 
which is real and positive, and (in view of \eqref{p:12"}) this positivity property is 
also true for $\gamma$. In view of \eqref{p:11}, the expression
\eqref{p:13"} of $\widehat{a}_\ell(k;g)$ therefore factorizes a pole located at
$k=k_\ell-\rmi\gamma$, and $\gamma$ can thus be related to the standard \emph{width} 
parameter of the Breit--Wigner one--pole approximation.

\vskip 0.2cm

\noindent
\emphsl{The shape of the bump for the cross--section:} 

\vskip 0.1cm

If one first considers the case when the residue $\rho(k;g)$ is real 
(for $k$ real), so that, in view of \eqref{p:14'}:
$\rho = \frac{\beta(2\alpha +1)}{k}>0$, we can see that the expression  
\eqref{p:13"} of $\widehat{a}_\ell(k;g)$ can be approximated by
\beq
\frac{-\beta}{k\,(\alpha-\ell +\rmi\beta)}
\approx 
\frac{-\gamma}{k\,(k-k_\ell +\rmi\gamma)}.
\label{p:13}
\eeq
This approximation, which is valid for $k$ varying in a suitable interval
centered at $k_\ell$, gives the standard Lorentzian contribution to the 
cross--section, namely: 
\beq
\widehat{\sigma}_\ell(k;g) \approx 
\frac{4\pi\gamma^2}{k^2[(k-k_\ell)^2 +\gamma^2]}.
\label{p:14-ex}
\eeq
In the general case, the residue is complex and satisfies Eq. \eqref{p:14'}.
Then, by taking Eqs. \eqref{p:11} and \eqref{p:12"} into account,
expression \eqref{p:13"} now yields a contribution
to the cross--section which is of the following form:
\beq
\widehat{\sigma}_\ell(k;g) \approx
\frac{4\pi [\sin\varphi(k;g) \ (k-k_\ell)-\gamma\,\cos\varphi(k;g)]^2}
{k^2\,[(k-k_\ell)^2 +\gamma^2]}.
\label{p:14}
\eeq
This asymmetric contribution corresponds to 
the generalized form of a Breit--Wigner one--pole model, 
when the unitary partial wave $S_\ell(k)$ includes an additional 
phase function $\varphi(k)$, namely:
\beq
S_\ell(k)= e^{2\rmi\varphi(k)}\ \frac{k-k_0-\rmi\gamma}{k-k_0+\rmi\gamma}.
\label{p:15}
\eeq

\noindent
(2) \emph{The case of antiresonances.}

In view of the apparent symmetric treatment of resonances and antiresonances 
that we have given above, one would be tempted to apply again the previous 
Taylor expansion argument to the analysis of the factor $[\lambda(k;g) -\ell]$ 
when $\beta(k;g)$ is negative. However, the analog of formula \eqref{p:11} 
would exhibit a pole in the upper half--plane at $k=k_\ell-\rmi\gamma$, 
with $\gamma <0$, corresponding to the real value $\lambda =\ell$.
But such a result is contradictory with the constraint imposed by the 
Wronskian Lemma (see, after Lemma \ref{lemma:vii'vii'}, the paragraph (a)
in ``Application: regions of $(\lambda,k)$--space free of singularities").
This impossibility of having singular pairs $(\lambda,k)$ at $\lambda$ real and
$\Real k>0$, $\Imag k >0$, which was known to hold in general for local potentials,
has indeed be extended here to the large classes $\cN^{\,\gamma}_{\wea}$
of nonlocal potentials. 

\vskip 0.2cm

In order to have a full account of antiresonances associated with dominant poles 
of $a(\lambda,k;g)$ that may be produced by the theory of nonlocal potentials,
one is thus faced to imagine the following type of mathematical model:
construct a holomorphic function $k\mapsto\lambda(k;g)$ 
satisfying the three previous assumptions (a), (b), and (c) with 
$\beta(k;g) = \Imag \lambda(k;g) <0$ for $k \in [k_-,k_+]$,
and such that, in addition, $\Imag \lambda(k;g)$ remains strictly negative 
when $k$ varies in the upper half--plane. 

\vskip 0.2cm

As a tutorial model, one can propose the following function:
\beq
\lambda(k) =\ell +\rmi\,\frac{\beta_0}{2}\left[1+\frac{e^{c(k-k_\ell)^2}}{1-\rmi(k-k_\ell)}\right]
\qquad (\beta_0<0;\ \  \ell = {\rm fixed}\ \ {\rm integer}).
\label{p:16}
\eeq
In fact, one can check that if $\Real k $ varies in some positive interval  
$I$ depending on $c$ ($c>0$) and centered at $k_\ell$, then:
\begin{itemize}
\item[(i)] for real $k$, ${\partial\alpha/\partial k}$ is a positive function, and $\beta(k)$ remains 
of the order of $\beta_0$; \\[-5pt]
\item[(ii)] for all $k= \Real k + \rmi\gamma$, with $\gamma>0$, one has $\Imag\lambda(k)<0$.
\end{itemize}

\vskip 0.1cm

Then, if we choose as a dominant--pole contribution the function 
$\widehat{a}_\ell(k)$ associated with $\lambda (k)$ by formula \eqref{p:13"},
it satisfies the characteristic phase--shift properties of an antiresonance 
near $k=k_\ell$; thus, according to the previous general analysis,
its corresponding contribution to the cross--section
produces a bump near $k=k_\ell$, but the mathematical 
description of this bump involves a Gaussian--type behaviour 
instead of the familiar Lorentzian behaviour,
as it is produced by the Breit--Wigner--type pole in the case of a resonance. 

\vskip 0.2cm

Of course, the value of the previous example is limited to 
the description of local phenomena related to a given partial wave amplitude
$a_\ell(k)$. As a matter of fact, one should expect that, in the same way as a 
``Regge trajectory" $\lambda=\lambda(k)$ with $\beta >0$ is able to describe 
a sequence of resonances, whose angular momentum $\ell$ increases with $k$, 
a similar ``image--trajectory" $\lambda=\lambda(k)$ with $\beta <0$ might 
correspondingly describe a sequence of antiresonances, whose angular momentum 
$\ell$ would also increase with $k$. The expectation of an alternating sequence
of resonances and antiresonances associated with increasing values of 
$\ell$ and $k$ seems indeed suggested by the phenomenological study
of various nuclear scattering processes (see \cite{DeMicheli1} and   
\cite{DeMicheli2}). A strong hope exists that the theory of nonlocal potentials
(in particular within the classes that have been studied in the present paper)
may be able to produce such type of coupled trajectories belonging respectively
to the first and to the fourth quadrant in the $\lambda$--plane,
a possibility which was forbidden by the usual theory of local potentials.

\vfill\eject

\renewcommand{\thesubsection}{A.\Roman{subsection}}

\appendix

\section{Bounds on the ``angular--momentum Green function"}
\label{appendix:a}

\subsection{Discrete angular momentum analysis} 
\label{subappendix:a.discrete}

We shall use the partial wave expansion of the (free--Hamiltonian) Green function:
\begin{align}
g(\cos \theta,k;R,R') &\doteq \frac{1}{4\pi}\,\frac{e^{\rmi k|\bR-\bR'|}}{|\bR-\bR'|}
=\frac{1}{4\pi}\,\frac{e^{\rmi k(R^2 +{R'}^2 -2RR'\cos\theta)^{1/2}}}
{(R^2 +{R'}^2 -2RR'\cos\theta)^{1/2}}, \label{a:1} \\
\intertext{namely:}
g(\cos\theta,k;R,R') &= -\sum_{\ell=0}^{\infty}\ (2\ell+1)
\ \frac{G_\ell(k;R,R')}{4\pi RR'} \ P_\ell(\cos\theta), \label{a:2}
\end{align}
in which the set of coefficients $G_{\ell}(k;R,R')$, called
``angular--momentum Green function" are given by:
\beq
G_\ell(k;R,R')=-2\pi RR' \int_{-1}^{+1}
g(\cos \theta,k;R,R')\ P_\ell(\cos\theta)\, \rmd\cos\theta.
\label{a:3}
\eeq
Relations \eqref{f:27d}: i.e.,
$G_\ell(k;R,R')=-\rmi kRR'j_\ell[k\min (R,R')] \ h_\ell^{(1)}[k\max (R,R')]$,
(implied by the fact that $G_\ell$ satisfies Bessel equations
with appropriate boundary conditions separately with respect to $R$ and $R'$)
will not be used here, since we shall obtain relevant bounds on
$G_\ell$ by direct use of formula \eqref{a:3}.

In formula \eqref{a:3}, $k$ may be real or complex, namely the functions
$G_\ell$ are defined for all $R>0$, $R'>0$ as entire functions of $k$, and one has:
\beq
4\pi\, |g(\cos \theta,k;R,R')|=
\frac{e^{-\Imag k\ [(R-R')^2 +2RR'(1- \cos\theta)]^\frac{1}{2}}}
{[(R-R')^2 +2RR'(1- \cos\theta)]^\frac{1}{2}}.
\label{a:4}
\eeq
For $\Imag k \geqslant 0$, the latter is uniformly bounded by
$[2RR'(1- \cos\theta)]^{-\frac{1}{2}}$, while for $\Imag k < 0$ 
it is uniformly bounded by
$\frac{e^{|\Imag k|\ (R+R')}}{[2RR'(1-\cos\theta)]^\frac{1}{2}}$.
From \eqref{a:3} one thus obtains the following global majorization:
\beq
\begin{split}
& |G_\ell(k;R,R')|\leqslant  \frac{1}{2}(RR')^\frac{1}{2}
\max\left(1, e^{-\Imag k (R+R')}\right)\int_{-1}^{+1} \frac{|P_\ell(t)|}{\sqrt{2(1-t)}}\,\rmd t \\
& \quad\leqslant
\frac{1}{2} \max\left(1, e^{-\Imag k (R+R')}\right)
\left(\frac{\pi RR'}{2\ell +1}\right)^\frac{1}{2}.
\end{split}
\label{a:5}
\eeq
For writing the rightmost inequality of \eqref{a:5},
we have used the Martin inequality (see \cite{Martin2})
$|P_\ell(\cos \theta)|< \min\left(1,2\,[\ell\pi\sin\theta]^{-1/2}\right)$.
(Note that for $k$ real, bound \eqref{a:5} itself lies in \cite{Martin2}).

\vskip 0.2cm

We shall now derive an alternative bound on $G_\ell$, which exhibits a decrease
property with respect to $|k|$ ($k\in \C$).
Let us rewrite Eq. \eqref{a:3} as follows, by introducing the change of integration variable
$u= [R^2 +{R'}^2 -2RR'\cos\theta]^\frac{1}{2}$, i.e.,
$\cos\theta (u)= \frac{R^2+{R'}^2 -u^2}{2RR'}$:
\beq
G_\ell(k;R,R')=\frac{1}{2} \int_{|R-R'|}^{R+R'} e^{\rmi ku}\ P_\ell(\cos\theta(u))\, \rmd u.
\label{a:7}
\eeq
We now have (by using integration by parts):
\beq
\begin{split}
& \rmi k G_\ell(k;R,R')=\frac{1}{2} \int_{|R-R'|}^{R+R'}
\frac{\rmd}{\rmd u}[e^{\rmi ku}]\ P_\ell(\cos\theta(u))\, \rmd u \\
&\quad=\frac{1}{2 RR'} \int_{|R-R'|}^{R+R'} e^{\rmi ku}\ P'_\ell(\cos\theta(u))\, u\, \rmd u \,
+\frac{1}{2}\left[e^{\rmi k(R+R')}- (-1)^{\ell}e^{\rmi k|R-R'|}\right],
\end{split}
\label{a:8}
\eeq
(where we have used the fact that $P_\ell(1) =1$, $P_\ell(-1) =(-1)^\ell$).

\vskip 0.3cm

From the integral representation
$P_\ell(\cos\theta)= \frac{1}{\pi}\int_0^{\pi}(\cos\theta +\rmi\sin \theta\cos\alpha)^\ell\,\rmd \alpha$, 
one readily obtains the following bound
\beq
|P'_\ell(\cos\theta)|\leqslant \frac{\ell}{|\sin \theta|}.
\label{a:9}
\eeq
The latter allows one to give a majorization for the r.h.s. of Eq. \eqref{a:8}, which yields:
\beq
|k| \, |G_\ell(k;R,R')|\leqslant
\, \max \left(1, e^{-\Imag k (R+R')}\right) \left[1 + \frac{\ell}{2}
\int_{|R-R'|}^{R+R'} \frac{u \,\rmd u}{RR' \, |\sin\theta(u)|}\right].
\label{a:10}
\eeq
Since
$2RR' \sin \theta(u) = [(R+R'+u)(R+R'-u)(u+R-R')(u+R'-R)]^\frac{1}{2}$,
one then gets the following majorization:
\beq
\frac{\ell}{2}\int_{|R-R'|}^{R+R'} \frac{u\,\rmd u}{RR'\, |\sin\theta|} \leqslant
\ell\int_{|R-R'|}^{R+R'} \frac{\rmd u}{(R+R'-u)^\frac{1}{2}\,(u-|R-R'|)^\frac{1}{2}} = \ell \pi,
\label{a:11}
\eeq
and therefore, from \eqref{a:10}:
\beq
|G_\ell(k;R,R')|\leqslant
\, \max \left(1, e^{-\Imag k (R+R')}\right) \left(\frac{1 + \ell\pi}{|k|}\right).
\label{a:12}
\eeq
As a result of \eqref{a:5} and \eqref{a:12}, we can thus write the following
global uniform bound, which exhibits decrease properties with respect to both
variables $\ell$ and $|k|$ when they go to infinity:
\beq
|G_\ell(k;R,R')|\leqslant\max\left(1, e^{-\Imag k (R+R')}\right)
[1+R]^\frac{1}{2}\,[1+R']^\frac{1}{2}
\min\left(\frac{\ell\pi +1}{|k|}, \frac{1}{2}\sqrt{\frac{\pi}{2\ell +1}}\,\right).
\label{a:13}
\eeq

\subsection{Complex angular momentum analysis} 
\label{subappendix:a.complex}

We shall now introduce a function $G(\lambda,k; R,R')$, called
the complex--angular--momentum Green function, defined for
all complex $\lambda$ in the half--plane
${\C^+_{-\frac{1}{2}}}= \left\{\lambda\,:\, \Real\lambda > -\frac{1}{2}\right\}$,
such that for all positive integers $\ell$, one has:
$G_\ell(k;R,R')=G(\ell,k; R,R')$. For every $R$, $R'$ ($R>0$, $R'>0$),
the function $G$ will be uniquely defined as a holomorphic function
of $(\lambda,k)$ in the product $\C^+_{-\frac{1}{2}}\times\widehat{\C}$,
where $\widehat{\C}$ denotes the universal covering of $\C\setminus\{0\}$.
The uniqueness of this interpolation of $G_\ell$ will be ensured by the fact
that for $k=\rmi\kappa$, $\kappa>0$, $G(\lambda,\rmi\kappa; R,R')$ is a 
Carlsonian interpolation, which will thus allow us to specify the ``basic" 
first--sheet ${\C^{\cut}}\doteq\C\setminus(-\infty, 0]$ of $G$.
Our purpose now is the derivation of uniform bounds for
$|G(\lambda,k; R,R')|$ in 
$\left\{(\lambda,k)\in\C^+_{-\frac{1}{2}}\times\C^{\cut}\right\}$.

\skt

\noindent
\textbf{(1) Analysis for $\boldsymbol{k=\rmi\kappa}$, $\boldsymbol{\kappa>0}$.}

\sku

Let $z_0 \equiv z_0(R,R')= \frac{1}{2}\left(\frac{R}{R'}+\frac{R'}{R}\right)$. 
For $k=\rmi\kappa$, $\kappa>0$, the Green function $g$ (see \eqref{a:1}) can be 
conveniently rewritten as follows in terms of the complex variable $z=\cos\theta$:
\beq
g(z,\rmi\kappa;R,R') =  \frac{1}{4\pi}\,
\frac{e^{-\kappa (2RR')^\frac{1}{2}(z_0-z)^\frac{1}{2}}}{(2RR')^\frac{1}{2}\,(z_0-z)^\frac{1}{2}}.
\label{a:14}
\eeq
Since the function $u(R,R';z)\doteq (2RR')^\frac{1}{2}(z_0-z)^\frac{1}{2}$ (specified 
as being positive for $z$ real, $z< z_0$) is such that $\Real  u(R,R';z) \geqslant 0$ for z
varying in the (closed) cut--plane ${\C^{\cut}_{z_0}} \doteq\C\setminus [z_0,+\infty[$, 
the following uniform bound holds:
\beq
\mathrm{For} \  z\in\C^{\cut}_{z_0}, \qquad
|g(z,\rmi\kappa;R,R')|\leqslant \frac{1}{4\pi}
\frac{1}{(2RR')^\frac{1}{2}\,|z_0-z|^\frac{1}{2}}.
\label{a:15}
\eeq
It follows that, as a holomorphic function of $z$,
$g(z,\rmi\kappa;R,R')$ satisfies the conditions
of the Froissart--Gribov theorem (or Laplace--transformation on the
one--sheeted hyperboloid in the sense of \cite{Bros4}). Therefore, there
exists a function $G(\lambda,\rmi\kappa; R,R')$, holomorphic in the half--plane
$\left\{\lambda \in {\C^+_{-\frac{1}{2}}}\right\}$ such that for all integers $\ell$ 
($\ell\geqslant 0$), one has $G(\ell,\rmi\kappa;R,R')=G_\ell(\rmi\kappa;R,R')$.
Moreover, this function $G$ is given by the following integral
in terms of the discontinuity $\Delta g$ of $g$ across the cut
$z\in [z_0,+\infty[$, namely 
$\Delta g (z,\rmi\kappa;R,R')=\frac{1}{4\pi}\frac{\cos\left[\kappa(2RR')^\frac{1}{2}(z-z_0)^\frac{1}{2}\right]}
{(2RR')^\frac{1}{2}(z-z_0)^\frac{1}{2}}$:
\beq
G(\lambda,\rmi\kappa;R,R') = -\frac{RR'}{2}\int_{z_0(R,R')}^{+\infty}
\frac{\cos\left(\kappa \, (2RR')^\frac{1}{2}(z-z_0)^\frac{1}{2}\right)}{(2RR')^\frac{1}{2}(z-z_0)^\frac{1}{2}}
\,Q_{\lambda}(z)\,\rmd z.
\label{a:16}
\eeq
In this equation, $Q_{\lambda}$ denotes the second--kind Legendre function,
and we note that the integral is convergent for all $\lambda$ in ${\C^+_{-\frac{1}{2}}}$.

\skt

\noindent
\emph{Bounds on $G(\lambda,i\kappa;R,R')$:}

\sku

{(a)} \ In view of \eqref{a:16}, we have:
\beq
|G(\lambda,\rmi\kappa;R,R')|\leqslant \frac{(RR')^\frac{1}{2}}{2\sqrt 2}
\int_{z_0(R,R')}^{+\infty}
\frac{|Q_{\lambda}(z)|}{[z-z_0(R,R')]^\frac{1}{2}}\,\rmd z.
\label{a:17}
\eeq
Then, by using the following integral representation of $Q_{\lambda}$
(see, e.g., \cite[formula (III--11)]{Bros4}):
\beq
Q_{\lambda}(z) =\frac{1}{\pi} 
\int_{v\doteq \cosh^{-1}\!z}^{+\infty}
e^{-(\lambda + \frac{1}{2})w} \ [2(\cosh w- \cosh v)]^{-\frac{1}{2}} \,\rmd w,
\label{a:18}
\eeq
and the relation 
$\cosh^{-1}\! z_0(R,R')= \left|\ln \frac{R}{R'}\right|$, we obtain:
\beq
\int_{z_0(R,R')}^{+\infty}\frac{|Q_{\lambda}(z)|}{[z-z_0(R,R')]^\frac{1}{2}}\,\rmd z
\leqslant \frac{1}{\pi}\int_{\left|\ln \frac{R}{R'}\right|}^{+\infty}
e^{-(\Real \lambda + \frac{1}{2})w}\,\rmd w
\int_{z_0(R,R')}^{\zeta\doteq \cosh w}
\frac{\rmd z}{\{2(\zeta-z)\,[z-z_0(R,R')]\}^\frac{1}{2}}.
\label{a:19}
\eeq
But, since the subintegral in the r.h.s of \eqref{a:19} is
equal to the constant $\frac{\pi}{\sqrt 2}$, we obtain, in view of \eqref{a:17}
and \eqref{a:18}:
\beq
|G(\lambda,\rmi\kappa;R,R')|\leqslant
\frac{(RR')^\frac{1}{2}}{2(2\Real \lambda +1)} \
\left[\min\left(\frac{R}{R'},\frac{R'}{R}\right)\right]^{(\Real \lambda + \frac{1}{2})}.
\label{a:20}
\eeq

\skd

{(b)} \, We shall now derive an alternative bound on $G$, which exhibits a decrease
property with respect to $\kappa$.
By making use of the integration variable 
$\widehat u =\widehat u(z) = (2RR')^\frac{1}{2}(z-z_0)^\frac{1}{2}$ and of the inverse mapping
$z=z(\widehat u) = z_0 + \frac{\widehat{u}^2}{2RR'}$, we rewrite Eq. \eqref{a:16} as follows:
\beq
G(\lambda,\rmi\kappa;R,R') = -\frac{1}{2}
\int_{0}^{+\infty}\cos\kappa\widehat u \ Q_{\lambda}(z(\widehat u)) \,\rmd\widehat{u}.
\label{a:21}
\eeq
We then have:
\beq
\kappa \, G(\lambda,\rmi\kappa;R,R')= -\frac{1}{2}\int_{0}^{+\infty}
\frac{\rmd}{\rmd\widehat{u}}[\sin\kappa\widehat u] \ Q_{\lambda}(z(\widehat u)) \,\rmd\widehat{u} 
= \frac{1}{2}\int_{0}^{+\infty}
\sin\kappa\widehat{u} \ \, Q'_{\lambda}(z(\widehat u))\ \frac{\widehat{u}}{RR'}\,\rmd\widehat{u}.
\label{a:22}
\eeq
From \eqref{a:18} (and making use of a partial integration procedure)
we can deduce the following integral representation for $Q'_{\lambda}$:
\beq
Q'_{\lambda}(z) =
-\frac{\lambda + \frac{1}{2}}{\pi} \int_{z}^{+\infty}
\frac{e^{-(\lambda + \frac{1}{2})w} \ \rmd\zeta}{[2(\zeta -z)]^\frac{1}{2}\,(\zeta^2 -1)}
-\frac{1}{\pi} \int_{z}^{+\infty}
\frac{e^{-(\lambda + \frac{1}{2})w}\ \zeta\,\rmd\zeta}
{[2(\zeta -z)]^\frac{1}{2}\ (\zeta^2 -1)^\frac{3}{2}},
\label{a:23}
\eeq
in which $w= \cosh^{-1}\!\zeta$. By taking the latter into account in Eq. \eqref{a:22}
and inverting the integrations over $z$ and $\zeta$, one obtains:
\beq
\begin{split}
& \kappa \, G(\lambda,\rmi\kappa;R,R')
= -\frac{\lambda + \frac{1}{2}}{2\pi} \int_{z_0}^{+\infty}
\frac{e^{-(\lambda + \frac{1}{2})w}}{\zeta^2 -1}\,\rmd\zeta
\int_{z_0}^{\zeta}\frac{\sin\left(\kappa\,(2RR')^\frac{1}{2}(z-z_0)^\frac{1}{2}\right)}{[2(\zeta -z)]^\frac{1}{2}}
\,\rmd z \\
&\quad -\frac{1}{2\pi} \int_{z_0}^{+\infty}
\frac{e^{-(\lambda + \frac{1}{2})w}}{(\zeta^2 -1)^\frac{3}{2}}\ \zeta\,\rmd\zeta
\int_{z_0}^{\zeta}\frac{\sin\left(\kappa\,(2RR')^\frac{1}{2}(z-z_0)^\frac{1}{2}\right)}{[2(\zeta -z)]^\frac{1}{2}}
\,\rmd z.
\end{split}
\label{a:24}
\eeq
A uniform bound for the first term on the r.h.s. of Eq. \eqref{a:24}
is obtained by simply majorizing the sine--function by one. In fact,
this term is majorized in the whole half--plane $\{\lambda \in {\C^+_{-\frac{1}{2}}}\}$ by:
\beq
\frac{\Real \lambda + \frac{1}{2}}{2\pi} \int_{z_0(R,R')}^{+\infty}
\frac{\{2[\zeta -z_0(R,R')]\}^\frac{1}{2}}{(\zeta^2 -1)}\,\rmd\zeta \
\leqslant\ A_1 \times  (2\Real \lambda +1),
\label{a:25}
\eeq
where $A_1$ is a numerical constant (independent of $k$, $R$ and $R'$).
For obtaining a uniform bound for the second term on the r.h.s. of Eq. \eqref{a:24}
it is necessary to majorize the sine--function:  (a) by one in the range
$z\geqslant z_1\doteq z_0 + (2RR'\kappa^2)^{-1}$, and (b) by
$\kappa (2RR')^\frac{1}{2} (z-z_0)^\frac{1}{2}$ in the range
$z_0\leqslant z \leqslant z_1$.
One is then led to introduce a partition of the integration
region into three subregions of the $(z,\zeta)$--plane, namely
$R_1 = \{\zeta\geqslant z_1,\, z_0\leqslant z\leqslant z_1\}$,
$R_2 = \{z_0\leqslant\zeta\leqslant z_1,\, z_0\leqslant z\leqslant\zeta\}$,
and $R_3 = \{\zeta\geqslant z_1,\, z_1\leqslant z\leqslant\zeta\}$.
While the integrals in $R_1$ and $R_2$ yield uniform majorizations
by numerical constants, the integral in $R_3$ is majorized by an expression of
the form  $[a_1+a_2 \max\left(\ln\,(RR'\kappa^2)^{-1},0\right)]$
($a_1$ and $a_2$ being numerical constants).
By taking all these estimates into account for the r.h.s. of \eqref{a:24},
one obtains a majorization of the following form in the half--plane 
$\left\{\lambda\in\C^+_{-\frac{1}{2}}\right\}$:
\beq
\begin{split}
& |G(\lambda,\rmi\kappa;R,R')| \leqslant
\frac{A_1 (2\Real\lambda+1)+A_2+A_3\ln\left(1+(RR'\kappa^2)^{-1}\right)}{\kappa} \\
& \quad\leqslant C \,\frac{2(\Real\lambda+1)+\ln\left(1+\kappa^{-2}\right)}{\kappa}
\left[1 + \ln\left(1+\frac{1}{R}\right)\right]
\left[1 + \ln \left(1+\frac{1}{R'}\right)\right],
\end{split}
\label{a:26}
\eeq
where $A_1$, $A_2$, $A_3$, and $C$ are suitable numerical constants.

As a result of \eqref{a:20} and \eqref{a:26}, we can thus write the following
global uniform bound, which exhibits decrease properties with respect to both
variables $\lambda$ and $\kappa$ when they go to infinity:
\beq
\begin{split}
& |G(\lambda,\rmi\kappa;R,R')| \leqslant
\left[1 + \ln \left(1+\frac{1}{R}\right) +\sqrt{R}\right]
\ \left[1 + \ln \left(1+\frac{1}{R'}\right) +\sqrt{R'}\right] \\
& \quad\times
\min\left[\frac{1}{2(2\Real\lambda +1)}, \
C\,\frac{2(\Real \lambda +1) +\ln \left(1+\kappa^{-2}\right)}{\kappa}\right].
\end{split}
\label{a:27}
\eeq

\skt

\noindent
\emph{Bounds on $\frac{\partial}{\partial R}\,G(\lambda,\rmi\kappa;R,R')$:}

\sku

We shall obtain a relevant expression for this derivative of $G$  
by computing the derivative of the double integral on the r.h.s. 
of Eq. \eqref{a:24} with respect to $R$. Using the fact that
the successive integrands of the latter in the variables $\zeta$ and $z$
vanish at their common threshold $z_0(R,R')$, we obtain the following 
integral representation (in which $w= \cosh^{-1}\!\zeta$):
\beq
\begin{split}
& 2\pi \, \frac{\partial}{\partial R}\, G(\lambda,\rmi\kappa;R,R') =
\left(\frac{R'}{2R}\right)^\frac{1}{2}
\int_{z_0(R,R')}^{+\infty}\rmd\zeta \
e^{-(\lambda + \frac{1}{2})w}\,
\left[\frac{\lambda + \frac{1}{2}}{(\zeta^2 -1)} +
\frac{\zeta}{(\zeta^2 -1)^\frac{3}{2}} \right] \\
& \quad\times\!\!\int_{z_0(R,R')}^{\zeta} \!\! \rmd z \
\frac{\cos\left(\kappa\,(2RR')^\frac{1}{2}(z-z_0)^\frac{1}{2}\right)}{[2(\zeta -z)]^\frac{1}{2}}
\left\{\!\frac{\left[\frac{R}{R'}-\frac{R'}{R}\right]}{2[z-z_0(R,R')]^\frac{1}{2}} 
-[z-z_0(R,R')]^\frac{1}{2} \!\right\}\!.
\end{split}
\label{a:128}
\eeq
By proceeding as for the bound \eqref{a:20} on $G$, we now deduce the
following bound from \eqref{a:128}:
\beq
\begin{split}
& \left|\frac{\partial}{\partial R}\ G(\lambda,\rmi\kappa;R,R')\right| \\
& \quad\leqslant\frac{1}{4\sqrt{2}}\left(\frac{R'}{2R}\right)^\frac{1}{2}
\left|\frac{R}{R'}-\frac{R'}{R}\right|\int_{z_0(R,R')}^{+\infty}
e^{-(\Real\lambda + \frac{1}{2})w} \,
\left[\frac{|\lambda + \frac{1}{2}|}{(\zeta^2 -1)} + \frac{\zeta}{(\zeta^2-1)^\frac{3}{2}}\right]\,\rmd\zeta \\ 
& \qquad+\frac{1}{4\sqrt 2}\left(\frac{R'}{2R}\right)^\frac{1}{2}
\int_{z_0(R,R')}^{+\infty}
\! e^{-(\Real\lambda + \frac{1}{2})w}
\left[\frac{|\lambda + \frac{1}{2}|}{(\zeta^2 -1)}+\frac{\zeta}{(\zeta^2 -1)^\frac{3}{2}}\right] 
[\zeta-z_0(R,R')]\,\rmd\zeta.
\end{split}
\label{a:129}
\eeq
By using the inequality $\zeta-z_0 \leqslant\zeta -1$ one readily
obtains that the latter integral in \eqref{a:129} is convergent
and bounded by a ($\lambda$--dependent) constant in the whole half--plane 
$\C^+_{-\frac{1}{2}}$. By now making the change of variable 
$u=\sinh w=\sqrt{\zeta^2-1}$ in the former integral of \eqref{a:129},
one also sees that this integral can be majorized  
(for $\lambda\in\C^+_{-\frac{1}{2}}$) by
$\left(\left|\lambda + \frac{1}{2}\right| +1\right)
\times \int_{\frac{1}{2}\left|\frac{R}{R'}-\frac{R'}{R}\right|}^\infty u^{-2} \,\rmd u$. 
As a result, one can replace the inequality \eqref{a:129} by a simple majorization of 
the following form:
\beq
\left|\frac{\partial}{\partial R}\,G(\lambda,i\kappa;R,R')\right|\leqslant
c_1(\lambda)\times\left(\frac{R'}{R}\right)^\frac{1}{2},
\label{a:130}
\eeq
which is valid for all $k=\rmi\kappa$ ($\kappa>0$) and $\lambda \in \C^+_{-\frac{1}{2}}$.

\skt

\noindent
\textbf{(2) Analytic continuation in $\boldsymbol{k}$.}

\vskip 0.1cm

We start from the definition \eqref{a:16} of
$G$, in which we insert the integral representation
\eqref{a:18} of $Q_{\lambda}$ and then invert the integrations over
$z$ and $\zeta = \cosh w$. We thus obtain:
\beq
G(\lambda,\rmi\kappa;R,R') = -\frac{(RR')^\frac{1}{2}}{2\sqrt{2}\pi} \!\!
\int_{z_0(R,R')}^{+\infty} \!
\frac{e^{-(\lambda + \frac{1}{2})w}}{(\zeta^2 -1)^\frac{1}{2}}\,\rmd\zeta
\!\int_{z_0}^{\zeta}\!
\frac{\cos\left(\!\kappa\,(2RR')^\frac{1}{2}(z-z_0)^\frac{1}{2}\!\right)}
{\{2(\zeta-z)\,[z-z_0(R,R')]\}^\frac{1}{2}}\,\rmd z.
\label{a:28}
\eeq
We now wish to define the analytic continuation of this double integral
with respect to the complex variable $k$ by putting $k=\rmi\kappa e^{-\rmi\phi}$;
$\phi$ will be taken in the interval $-\frac{\pi}{2}\leqslant\phi\leqslant\frac{3\pi}{2}$ 
so that $G$ be defined in the ``basic first sheet" ${\C^{\cut}}$ of the $k$--plane.

\vskip 0.3cm

\noindent
(a) For $|\phi|\leqslant\frac{\pi}{2}$, this analytic continuation of
$G(\lambda,k;R,R')= G(\lambda, \rmi\kappa e^{-\rmi\phi};R,R')$
is well--defined by shifting in $\C^2$ the integration
region from its initial situation at $k=\rmi\kappa$, namely
$\Gamma_0\doteq \{(\zeta,z)\,:\, z_0 = z_0(R,R')\leqslant z\leqslant \zeta <+\infty\}$
to the set
$\Gamma_{\phi}\doteq \{(\zeta,z)\,:\, \zeta- z_0= |\zeta-z_0| e^{2\rmi\phi},
z- z_0= |z-z_0| e^{2\rmi\phi};\, z_0\leqslant |z|\leqslant |\zeta| <+\infty\}$.
The corresponding rotation of angle $\phi$ of $(z-z_0)^\frac{1}{2}$ will then cancel
the rotation of angle $-\phi$ of $\kappa$ in the cosine--factor under the
integral on the r.h.s. of \eqref{a:28}, so that this factor can always be
bounded by one. It follows that one obtains a majorization for the analytic continuation
at $k= \rmi\kappa e^{-\rmi\phi}$ of the r.h.s. of
\eqref{a:28} which involves the same subintegral over $z$ as in
\eqref{a:19} (equal to the constant $\frac{\pi}{\sqrt 2}$), namely:
\beq
|G(\lambda,k; R,R')|\leqslant
\frac{(RR')^\frac{1}{2}}{4}
\int_{\gamma_{\phi}(z_0(R,R'))}
\left|e^{-(\lambda +\frac{1}{2})(w+\rmi\varphi)}\right|\, |\rmd (w+i\varphi)|;
\label{a:29}
\eeq
in \eqref{a:29}, $\gamma_{\phi}(z_0(R,R'))$ is the image of
$\widehat{\gamma}_{\phi} \doteq\{\zeta= z_0(R,R') + \rho e^{2\rmi\phi},\, \rho\in [0,+\infty)\}$
by the mapping $\zeta \mapsto \widehat{w} \doteq w +\rmi\varphi =\cosh^{-1} \zeta$.
One can check that in the path $\gamma_{\phi}(z_0(R,R'))$ the variables $|\varphi|$ and 
$w$ vary respectively in the intervals $|\varphi|\in [0,2\phi]$ and $w\in [w_0(R,R'), +\infty)$, 
where $w_0(R,R')$ is positive and such that: 
\begin{itemize}
\item[(i)] if $|\phi|\leqslant\frac{\pi}{4}$,\
$\cosh w_0= z_0(R,R') = \frac{1}{2}\left(\frac{R}{R'} + \frac{R'}{R}\right)$,\ i.e.:
\beq
e^{-w_0(R,R')}= \min \left(\frac{R}{R'},\frac{R'}{R}\right); 
\label{a:29'}  
\eeq
\item[(ii)] if $\frac{\pi}{4} \leqslant |\phi|\leqslant \frac{\pi}{2}$, \
$\cosh^2 w_0(R,R')= z_0^2(R,R') \sin^2 2\phi + \cos^2 2\phi$, which yields: 
$\sinh w_0(R,R')=\frac{1}{2}\left| \frac{R}{R'} - \frac{R'}{R}\right|\,\sin 2|\phi|$,
and thereby:
\beq
e^{-w_0(R,R')}\leqslant (\sin 2|\phi|)^{-1}\,\min\left(\frac{R}{R'},\frac{R'}{R}\right). 
\label{a:29"}  
\eeq
\end{itemize}
Note that in case (i), $\gamma_{\phi}(z_0(R,R'))$ 
defines $|\varphi|$ as an increasing function of $w$ in 
the interval $[w_0(R,R'), +\infty)$, while in case (ii),
$\gamma_{\phi}(z_0(R,R'))$ is tangent at the line $w= w_0(R,R')$ 
at some point $\varphi_0$ (with $0\leqslant\varphi_0\leqslant\frac{\pi}{2}$). \\
By taking these geometrical facts into account, one then deduces from \eqref{a:29}
a majorization of the following form, which is valid for all
$k$ in the closed upper half--plane: 
\beq
|G(\lambda,k;R,R')|\leqslant
\frac{\sqrt{RR'}}{4}
\int_{\gamma_{\phi}(z_0)}\!\!\!
e^{-(\Real\lambda + \frac{1}{2})w + (\Imag \lambda) \varphi}\, |\rmd (w+\rmi\varphi)|
\leqslant  c(\lambda,k) \, \sqrt{RR'} \,
e^{-(\Real\lambda + \frac{1}{2})w_0(R,R')},
\label{a:30}
\eeq
where:
\beq 
c(\lambda,k)= c \,\max\left(e^{2\Imag\lambda\,\phi(k)},1\right) \
\left(1+\frac{1}{2\Real\lambda+1}\right).
\label{a:30"}
\eeq
In the latter, $c$ is a numerical constant and $\phi(k)=-\Arg(-\rmi k)$; 
more precisely, $k= \rmi\kappa e^{-\rmi\phi}$, with $\phi=\phi(k)$ such that 
$|\phi|\leqslant\frac{\pi}{2}$.

Moreover, by taking Eqs. \eqref{a:29'} and \eqref{a:29"} into account, we see that 
\eqref{a:30} implies the following majorization, which is valid globally in the set 
$\{(\lambda,k)\,:\,\lambda\in\C^+_{-\frac{1}{2}};\Imag k>0\}$:
\beq
|G(\lambda,k; R,R')|\leqslant
c(\lambda,k) \, [\Phi(k)]^{-(\Real\lambda+\frac{1}{2})} \ 
(RR')^\frac{1}{2}\left[\min \left(\frac{R}{R'},\frac{R'}{R}\right)\right]^{\Real\lambda+\frac{1}{2}},
\label{a:30'}
\eeq
in which we have put:
\begin{align}
\Phi(k) &= 1 \quad \mathrm{if} \quad |\phi(k)|\leqslant \frac{\pi}{4}, \label{a:31} \\
\intertext{and}
\Phi(k) &= (\sin 2|\phi(k)|) \quad \mathrm{if} \quad \frac{\pi}{4}\leqslant |\phi(k)| < \frac{\pi}{2}.
\label{a:31'}
\end{align}
We also notice that, since $w_0(R,R') \geqslant 0$, majorization \eqref{a:30} also yields
\emph{for all $k$ in the closed upper half--plane} $\Imag k \geqslant 0$:
\beq
|G(\lambda,k; R,R')|\leqslant
c(\lambda,k)\ (RR')^\frac{1}{2}
\leqslant c \, e^{\pi|\Imag \lambda|}
\left(1+\frac{1}{2\Real\lambda+1}\right)\, (RR')^\frac{1}{2}. 
\ \ \ \
\label{a:31"}
\eeq
By performing the same contour distortion argument on the integral 
\eqref{a:128} for defining the analytic continuation in $k$  of the function 
$\frac{\partial G}{\partial R}(\lambda,\rmi\kappa; R,R')$, and by proceeding 
as for the derivation of bound \eqref{a:130}, we obtain an extension of the 
latter to the full half--plane $\Imag k>0$, which is of the following form:
\beq
\left|\frac{\partial G}{\partial R}(\lambda,k; R,R')\right|\leqslant
\widehat{c}_1(\lambda,k)\ \left(\frac{R'}{R}\right)^\frac{1}{2}.
\label{a:131"}
\eeq

\vskip 0.3cm

\noindent
(b) For $\frac{\pi}{2}< \phi\leqslant \frac{3\pi}{2}$, the analytic continuation
of $G(\lambda, \rmi\kappa e^{-\rmi\phi};R,R')$ may be pursued, but the
``rotated cycle" $\Gamma_{\phi}$ now acquires an additional part
whose support is the real set
$\Gamma_\mathrm{r}\doteq \{(\zeta,z)\,:\, -1\leqslant \zeta \leqslant z\leqslant z_0(R,R')\}$;
in fact, the inclusion of $\Gamma_\mathrm{r}$ is more easily seen in the
$\widehat{w}$--plane, since for $\phi\geqslant \frac{\pi}{2}$, the contour
$\gamma_{\phi}(z_0(R,R'))$ may be distorted so as to contain
the broken line $[v_0,0]\cup [0,-2\rmi\pi]\cup [-2\rmi\pi,-2\rmi\pi +v_0]$,
completed by an infinite branch whose asymptot is the line
$\varphi = 2\phi$ (i.e., the image of $\widehat{\gamma}_{\phi}$ from a second sheet).
In this new situation, the bound that one obtains for the analytic continuation
at $k= \rmi\kappa e^{-\rmi\phi}$ of the r.h.s. of \eqref{a:28} still contains
the constant subintegral over $z$ (equal to $\frac{\pi}{\sqrt{2}}$), but the latter 
is now obtained after a majorization of the cosine--factor by
$\cosh [\kappa (R+R') \cos\phi] = \cosh [\Imag k (R+R')]$. 
Moreover, since the range of values of $\varphi$ in this analytic continuation of
\eqref{a:28} admits $3\pi$ as its maximal value, the latter majorization \eqref{a:31} 
must now be replaced by
\beq
|G(\lambda,k; R,R')|\leqslant c \, (RR')^\frac{1}{2}
\cosh [\Imag k (R+R')] \max\left(e^{3\pi\Imag\lambda}, 1\right) \,
\left(1+\frac{1}{2\Real\lambda+1}\right),
\label{a:32}
\eeq
which is valid for all $k$ in the lower half--plane (of the
basic first sheet) and $\lambda \in \C^+_{-\frac{1}{2}}$.
(For simplicity, we have used the same constant $c$ in
\eqref{a:32} as in \eqref{a:30} and \eqref{a:31}, being 
not concerned with the best values of these constants).

\subsection{Complements on Bessel and Hankel functions} 
\label{subappendix:a.complements}

\noindent
(a) \emph{Bounds on the spherical Bessel and Hankel functions 
$j_\ell $ and $h_\ell^{(1)}$ for $\ell$ integer $(\ell\geqslant 0)$:}

\vskip 0.2cm

The following inequalities have been established in \cite{Newton2,Abramowitz}; 
for all $k\in \C$ and $R\geqslant 0$, there hold: 
\begin{align}
\left|kRj_\ell(kR)\right| &\leqslant 
c_\ell \left(\frac{|k|R}{1+|k|R}\right)^{(\ell+1)}
e^{R\,|\Imag k|}, \label{a:0} \\
\left|kRh_\ell^{(1)}(kR)\right| &\leqslant 
c_\ell^{(1)} \left(\frac{1+|k|R}{|k|R}\right)^\ell
e^{-R\,\Imag k}, \label{a:0'}
\end{align}
where $c_\ell$ and $c_\ell^{(1)}$ are constants whose dependence on 
$\ell$ is not exploited here.

\skt

\noindent
(b) \emph{The derivatives of the spherical Bessel and Hankel functions:} 

\sku

Starting from the following relation (see \cite[Vol. 2, p. 11, formula (50)]{Bateman}):
\beq
\frac{\rmd}{\rmd z}\left[z^\nu J_\nu(z)\right]=z^\nu J_{\nu-1}(z) \qquad (\nu\in\C),
\label{a:1b}
\eeq
and recalling that
\beq
j_\ell(z)=\sqrt{\frac{\pi}{2z}}\, J_{\ell+1/2}(z), 
\label{a:1b'}
\eeq
one obtains:
\beq
\frac{\rmd}{\rmd z}\left[z j_\ell(z)\right]=-\ell \, j_\ell(z)+ z \, j_{\ell-1}(z),
\label{a:2b}
\eeq
which yields, in view of \eqref{a:0}:
\beq
\left|\frac{\rmd}{\rmd R}\left[Rj_\ell(kR)\right]\right| \, \leqslant \, 
c'_\ell \,\left(\frac{|k|R}{1+|k|R}\right)^\ell \,e^{R\,|\Imag k|}, \label{a:3b}
\eeq
where $c'_\ell = \ell \, c_\ell+c_{\ell-1}$. This bound is valid 
for all $k\in \C$ and $R\geqslant 0$.

By using a similar formula for the Hankel functions $H_\nu^{(1)}(z)$,
namely (see, e.g., \cite[p. 361, Eq. (9.1.30)]{Abramowitz}):
\beq
\frac{\rmd}{\rmd z}\left[z^\nu H_\nu^{(1)}(z)\right]=z^\nu H_{\nu-1}^{(1)}(z) \qquad (\nu\in\C),
\label{a:4b}
\eeq
together with the relation
$h_\ell^{(1)}(z)=\sqrt{\frac{\pi}{2z}}\,H_{\ell+1/2}^{(1)}(z)$, one also obtains:
\beq
\frac{\rmd}{\rmd z}\left[z h_\ell^{(1)}(z)\right]=-\ell \, h_\ell^{(1)}(z)+ z \, h_{\ell-1}^{(1)}(z).
\label{a:5b}
\eeq
In view of \eqref{a:0'}, the latter equality yields the following bound, which is valid 
for all $k\in \C$ and $R\geqslant 0$:
\beq
\left|\frac{\rmd}{\rmd R}\left[Rh_\ell^{(1)}(kR)\right]\right| \, \leqslant \,  
{c'_\ell}^{(1)}\left(\frac{1+|k|R}{|k|R}\right)^{\ell+1}
e^{-R\Imag k}, 
\label{a:6b}
\eeq
where ${c'_\ell}^{(1)} = \ell \, c_\ell^{(1)} +c_{\ell-1}^{(1)}$. 

\skt

\noindent 
(c) \emph{The Sommerfeld condition for the spherical Hankel functions:} 

\sku

We now want to prove the following property:

\vskip 0.1cm

\noindent
\emphsl{For all $k$ such that $k\neq 0$, there holds the following behaviour in the limit $R\to +\infty$}:
\beq
e^{R\Imag k}
\left|\frac{\rmd}{\rmd R}\left[R h_\ell^{(1)}(kR)\right]- \rmi kR h_{\ell}^{(1)}(kR)\right|
=\frac{1}{|k|^3} \ O\!\left(\frac{1}{R^2}\right).
\label{a:8'}
\eeq
Formula \eqref{a:8'} derives from the following representation of the
Hankel functions \cite[p. 117]{Sommerfeld} by an asymptotic series (in the
sense of Poincar\'e):
\beq
H_\nu^{(1)}(\rho)=\sqrt{\frac{2}{\pi\rho}}\,e^{\rmi[\rho-(\nu+\frac{1}{2})\frac{\pi}{2}]}
\sum_{m=0,1,2,\ldots}\frac{(\nu,m)}{(-2\rmi\rho)^m},
\label{a:4'}
\eeq
where:
\beq
(\nu,m)=\frac{(4\nu^2-1)(4\nu^2-9)\cdots
(4\nu^2-\{2m-1\}^2)}{2^{2m} \, m!}, \qquad (\nu,0)=1.
\label{a:5'}
\eeq
This series reduces exceptionally to a finite sum whenever the subscript
$\nu$ takes a half--integral value. In fact, one can easily verify that if 
$\nu = \ell +\frac{1}{2}$ the symbols $(\nu,m)$ are zero 
for all integers $m$ such that $m>\ell$. In these cases,
the series \eqref{a:4'} represents exactly the Hankel functions.
In view of the relation
$R h_\ell^{(1)}(kR)=\sqrt{\frac{\pi}{2}}\sqrt{\frac{R}{k}} H_{\ell+1/2}^{(1)}(kR)$,
we deduce from \eqref{a:4'} the following representation:
\beq
R h_\ell^{(1)}(kR)=\sqrt{\frac{\pi}{2}}\sqrt{\frac{R}{k}} \, H_{\ell+1/2}^{(1)}(kR)=
\frac{1}{k}e^{\rmi[kR-(\ell+1)\frac{\pi}{2}]}
\sum_{0\leqslant m\leqslant \ell}\frac{\left(\ell+\frac{1}{2},m\right)}{(-2\rmi kR)^m}.
\label{a:6'}
\eeq
Then we get by a direct computation
\beq
\begin{split}
& \frac{\rmd}{\rmd R}\left[R h_\ell^{(1)}(kR)\right]
= \rmi k\left\{\frac{1}{k}\,e^{\rmi[kR-(\ell+1)\frac{\pi}{2}]}
\!\!\sum_{0\leqslant m\leqslant \ell}\frac{\left(\ell+\frac{1}{2},m\right)}{(-2\rmi kR)^m}\right\}
-\frac{1}{k}\,e^{\rmi[kR-(\ell+1)\frac{\pi}{2}]}
\!\!\sum_{1\leqslant m\leqslant \ell}\frac{m\left(\ell+\frac{1}{2},m\right)}
{(-2\rmi kR)^{m+1}} \\
& \quad= \rmi kR h_\ell^{(1)}(kR) 
+ \frac{1}{k}\,e^{\rmi[kR-(\ell+1)\frac{\pi}{2}]} \
\frac{P^{(\ell-1)}(kR)}{(kR)^{\ell +1}},
\end{split}
\label{a:7'}
\eeq
in which $P^{(\ell-1)}$ denotes a polynomial of degree $\ell-1$
whose all coefficients are different from zero. One readily checks that the
latter yields limit \eqref{a:8'} (for all $k$ such that $k\neq 0$).

\newpage

\noindent
(d) \emph{Bounds on the spherical Bessel functions $j_{\lambda}(z)$ for
$\lambda\in \C^+_{-\frac{1}{2}}$ and $z\in \C^{\cut}$}:

\sku

We recall that
$j_\lambda(z)=\sqrt{\frac{\pi}{2z}}\,J_{\lambda+1/2}(z)$,
and start from the following integral representation of the Bessel function $J_\lambda(z)$ \cite{Tichonov},
which is valid for $z \in \R^+$ and $\Real \lambda >0$:
\beq
J_\lambda(z)
=\frac{1}{2\pi}\int_{-\pi}^\pi e^{-\rmi z\sin\varphi+\rmi\lambda\varphi}\,\rmd\varphi
-\frac{\sin\pi\lambda}{\pi}\int_0^{+\infty} e^{-z\sinh\xi-\lambda\xi}\,\rmd\xi.
\label{a:33}
\eeq
This representation defines $J_\lambda(z)$ as the sum of two holomorphic functions
$J^{(1)}$ and $J^{(2)}$ of $\lambda$ and $z$.
As an integral on the interval $[-\pi,\pi]$, the first term on the r.h.s. of
\eqref{a:33} defines $J^{(1)}$ as an entire function of $(\lambda,z)$ satisfying 
the following global bound in $\C^2$:
\beq
\left|J^{(1)}(\lambda, z)\right|\,\leqslant\, e^{|\Imag z|}\ e^{\pi|\Imag \lambda|}.
\label{a:34}
\eeq
Consider now the function $J^{(2)}(\lambda,z)$ to be defined 
by the second term on the r.h.s. of \eqref{a:33}. For
$(\lambda,z)\in \C^+\times \C^+$, the integral in that term 
is easily majorized by a convergent integral 
(thanks to the minoration $\sinh \xi > \xi$ in the exponential under the integral),
which shows that $J^{(2)}$ is well--defined and analytic in this set and such that: 
\beq
\left|J^{(2)}(\lambda, z)\right|
\leqslant\frac{e^{\pi|\Imag \lambda|}}{\pi (\Real \lambda +\Real z)}
\leqslant\frac{e^{\pi|\Imag \lambda|}}{\pi \, \Real \lambda}.
\label{a:35}
\eeq
We now obtain an analytic continuation of $J^{(2)}$ 
and an extension of the previous bound for  
$(\lambda,z)\in \C^+\times \C^{\cut}$
by distorting the integration path $\gamma_0 = [0,+\infty[$ of
the second integral of \eqref{a:33} into any path
$\gamma_\phi$ whose support is the following set:
$\{\xi\in [0,-\rmi\phi]\}\cup \{\xi= -\rmi\phi+\beta\,:\, \beta\geqslant 0\}$, with 
$|\phi|\leqslant \frac{\pi}{2}$. The corresponding integral can then be replaced by
\beq
\int_0^{\phi} e^{\rmi z \sin u +\rmi\lambda u}\,\rmd u  + e^{\rmi\phi}
\int_0^{+\infty} e^{-z(\cos\phi \sinh\beta -\rmi\sin\phi \cosh\beta)}
\, e^{-\lambda \beta}\,\rmd\beta, 
\label{a:35'}
\eeq
which can be bounded in modulus by
\beq
e^{|\Imag \lambda|\,|\phi|}\left[|\phi|\,e^{|\Imag z|\,|\sin \phi|} +
\int_0^\infty e^{-(\Real z \cos\phi +\Imag z\sin\phi)\sinh \beta} 
\, e^{-(\Real \lambda)\beta}\,\rmd\beta \right],
\label{a:36}
\eeq
in the half--plane  $\varepsilon(\phi) \Imag z >0$ ($\varepsilon (\phi)$ 
denoting the sign of $\phi$). In the sector with equation 
$\Real z +\Imag z\tan\phi >0$ of this half--plane, 
the previous integral can then be majorized by 
$(\Real\lambda)^{-1}$ (for any value of $\phi$).
For $z$ varying in $\C^{\cut}$, one then obtains a global 
bound for the expression \eqref{a:36} with the choice $\phi =\pm\frac{\pi}{2}$,
which is equal to 
$e^{|\Imag \lambda|\frac{\pi}{2}}\left[\frac{\pi}{2}e^{|\Imag z|}+\frac{1}{\Real\lambda}\right]$. 
It then yields:
\beq
\left|J^{(2)}(\lambda, z)\right|
\leqslant \, e^{|\Imag \lambda|\frac{3\pi}{2}}\left[\frac{1}{2}\,e^{|\Imag z|} +
\frac{1}{\pi\Real\lambda}\right].
\label {a:37}
\eeq
By now putting together the bounds \eqref{a:34}, \eqref{a:35}, and \eqref{a:37},
we obtain the following majorizations for the function
$\sqrt{kR} \, j_\lambda(kR)=\sqrt{\frac{\pi}{2}}\,J_{\lambda+1/2}(kR)$: \\[+4pt]
\null\hspace{0.3cm} for $(\lambda,k)\in{\C^+_{-\frac{1}{2}}}\times {\C^+}$,
\beq
\left|\sqrt{kR} \, j_\lambda(kR)\right| \leqslant \sqrt{\frac{\pi}{2}}
\ e^{|\Imag k|R}\ e^{\pi|\Imag \lambda|}\left[1 + \frac{1}{\pi(\Real \lambda +\frac{1}{2})}\right];
\label{a:40}
\eeq
\null\hspace{0.3cm} for $(\lambda,k)\in \C^+_{-\frac{1}{2}}\times \C^{\cut}$,
\beq
\left|\sqrt{kR} \, j_\lambda(kR)\right| \leqslant \sqrt{\frac{\pi}{2}}
\ e^{|\Imag k|R}\ e^{\frac{3\pi}{2}|\Imag \lambda|}
\left[\frac{3}{2}+ \frac{1}{\pi(\Real\lambda +\frac{1}{2})}\right].
\label{a:41}
\eeq

\renewcommand{\thesubsection}{B.\Roman{subsection}}

\section{Continuity and holomorphy properties of
vector--valued and operator--valued functions}
\label{appendix:b}

We recall some general facts about continuous and
holomorphic functions taking their values in complete normed spaces
on the field of complex numbers,
denoted by $A$ (resp., $B$ or $C$) in the following. The norm of an element $a$
of $A$ is denoted by $\|a\|_A $, or simply $\|a\|$ if there is no ambiguity.
$D$ will denote a given domain either in ${\R}^m $ or in 
${\C}^m$. The real or complex variables whose range is $D$ are
called $\z=(\z_1,\ldots,\z_m)$, and we shall consider vector--valued functions
$\z \mapsto a[\z]$ such that for all $\z \in D$, $a[\z]$
belongs to $A$ and $\|a[\z]\|$ is uniformly bounded on
every compact subset of $D$.  By definition,
the function $\z \mapsto a[\z]$ is continuous in $D$ if
$\lim_{h \to 0} \left\|a[\z+h]-a[\z]\right\| =0$
for every $\z \in D$ and $h$
varying in a neighborhood $V_{\z}$ of zero
such that $\z+V_{\z}\subset D$.
When $\z$ is complex, the function $\z \mapsto a[\z]$ is holomorphic in $D$ if
there exist $m$ functions $\z \mapsto \dot{a}_j[\z]$ defined and
continuous in $D$ with values in $A$, called the partial derivatives of $a[\z]$,
such that:
$\lim_{{\rm h} \to 0} \left\|\frac{a[\z+h_j]-a[\z]}{{\rm h}}
- \dot{a}_j[\z]\right\| =0$
for every $\z \in D$ and $h_j=(0,\ldots,0,{\rm h},0,\ldots,0)$, 
${\rm h}$ being the $j^{\rm th}$--component of $h_j$; 
the increments $h_j$, $(1\leqslant j\leqslant m)$, are supposed to
vary in a neighborhood $V_{\z}$ of zero, which is chosen 
such that $\z+V_{\z}\subset D$.
These definitions of ``continuity" and ``holomorphy" (or ``analyticity")
for vector--valued functions, (also called more precisely
``strong continuity" and ``strong holomorphy or analyticity") 
will be used directly in the following survey of
various product operations given in Subsection \ref{subappendix:b.products} and of the vector--valued
functional interpretation of functions depending on real or complex 
parameters that we give in Subsection \ref{subappendix:b.passage}. For completeness, we shall also
summarize in Subsection \ref{subappendix:b.complement} the various (equivalent) 
``weak" and ``strong" characterizations of ``vector--valued holomorphy", 
by presenting them directly in the several--variable case. 

\vskip 0.2cm

\noindent
From the previous definitions and
by using the norm inequality in $A$, one readily checks that
any sum of continuous (resp., holomorphic) vector--valued functions
is a continuous (resp., holomorphic) vector--valued function. This
property extends to uniformly convergent series of vector--valued functions:

\skd

\begin{lemma}
\label{lemma:B0}
Let $\{\z \mapsto a_n[\z];\, n\in \N \}$ be a sequence of vector--valued functions
in $D$ taking their values in the complete normed space $A$ and such that
for all $\z$ in $D$ there holds a majorization of the form:
$\|a_n[\z]\|\leqslant u_n$, where the sequence $\{u_n\}_{n=0}^\infty$
is such that $M = \sum_{n=0}^{\infty} u_n < \infty$.
Then there exists a vector--valued function $\z \mapsto s[\z]$
taking its values in $A$, such that for all $\z$ in $D$ one has
$s[\z] =\sum_{n=0}^{\infty} a_n[\z]$, as the sum of a convergent series in $A$,
with $\|s[\z]\| \leqslant M$. Moreover,
\begin{itemize}
\item[\rm (i)] if the functions $a_n[\z]$ are continuous in $D$, then $s[\z]$ is
continuous in $D$; \\[-4pt]
\item[\rm (ii)] if $\z$ is complex and if the functions $a_n[\z]$ are holomorphic 
in $D$, then $s[\z]$ is holomorphic in $D$.
\end{itemize}
\end{lemma}

\begin{proof}
For all $\z$, the series with general term
$\|a_n[\z]\|$ is dominated by the series with general term $u_n$
and therefore convergent. Then the norm inequality
$\|\sum_{n=N_1}^{N_2} a_n[\z]\| \leqslant \sum_{n=N_1}^{N_2} \|a_n[\z]\|$ 
and the completeness property of $A$ imply the convergence
in $A$ of the sequence
$\{s_N[\z]\doteq \sum_{n=0}^{N} a_n[\z]; \ N\in \N\} $
to a vector $s[\z]$ such that $\|s[\z]\| \leqslant M$.
Moreover,

\vskip 0.1cm

(i) For any given $\varepsilon$, let $N_{\varepsilon}$ be such that
$\sum_{p=N_{\varepsilon}}^{\infty} u_p \leqslant \frac{\varepsilon}{3}$.
If the functions $a_n[\z]$ are all continuous, for any given
$\z$ in $D$ there exists a neighborhood ${\cV}_{\z,\varepsilon}$ of
zero such that $\|s_{N_{\varepsilon}}[\z+h]-s_{N_{\varepsilon}}[\z]\|
\leqslant\frac{\varepsilon}{3}$ for all $h\in {\cV}_{\z,\varepsilon}$.
Then by writing the norm inequality
$\|s[\z+h]-s[\z]\| \leqslant \|s[\z+h]-s_{N_{\varepsilon}}[\z+h]\|+
\|s_{N_{\varepsilon}}[\z+h]-s_{N_{\varepsilon}}[\z]\|+
\|s_{N_{\varepsilon}}[\z]-s[\z]\|$, one sees that each term on the
r.h.s. of this inequality is bounded by $\frac{\varepsilon}{3}$,
and therefore the l.h.s. is majorized by $\varepsilon$, which proves the
continuity of the function $\z \mapsto s[\z]$.

\vskip 0.1cm

(ii) If $\z$ is complex and if the functions
$a_n[\z]$ are all holomorphic, one uses a similar
$N_{\varepsilon}$--argument with the holomorphic functions
$\frac{s_N[\z+h_j]-s_N[\z]}{{\rm h}} -{(\dot s_N)}_j[\z]$ (instead of
$s_N[\z+h]-s_N[\z]$), after having proved that
the series of vector--valued functions with general terms
$(a_n)_{h_j}[\z]\doteq \frac{a_n[\z+h_j]-a_n[\z]}{{\rm h}}$ 
and ${(\dot a_n)}_j[\z]$are uniformly majorized (term by term) 
by $c \, u_n$, where $c$ is some constant independent of $h$. The proof of
the latter relies on Cauchy--type inequalities for 
vector--valued holomorphic functions, which are given in 
Subsection \ref{subappendix:b.complement}
(see our argument after formula \eqref{b:rem}).
\end{proof}

\subsection{Products which preserve continuity and holomorphy} 
\label{subappendix:b.products}

Being given three complete normed spaces $A,B,C$, we shall denote
by $\Pi$ any mapping $(a,b) \mapsto c=\Pi (a,b)$ from the direct product
$A \times B$ into $C$, which is bilinear
with respect to the two variables $a$ and $b$ and bicontinuous in the
following sense; for all pairs $(a,b)$, there holds:
\beq
\|\Pi(a,b)\|_C \leqslant \|a\|_A \times \|b\|_B.
\label{b:1}
\eeq
(A general constant factor, different from one, could be inserted
on the r.h.s. of the latter, but it would be of no use in the applications
and can always be avoided by a suitable rescaling of the norms).
Then we have:

\skd

\begin{lemma}
\label{lemma:B1}
Being given any bilinear and bicontinuous mapping $\Pi$
from $A\times B$ into $C$:
\begin{itemize}
\item[\rm (i)] if $a[\z]$ and $b[\z]$ (are continuous functions
in $D$, respectively vector--valued in $A$ and $B$, then
the function $\z \mapsto \Pi (a[\z],b[\z])$
is continuous in $D$, as a vector--valued function with values in $C$. \\[-5pt]
\item[\rm (ii)] If $\z$ is complex and if $a[\z]$ and $b[\z]$ are
holomorphic functions in $D$, respectively vector--valued in
$A$ and $B$,  then the function $\z \mapsto \Pi (a[\z],b[\z])$
is holomorphic in $D$, as a vector--valued function with values in $C$.
\end{itemize}
\end{lemma}

\begin{proof}
(i) In view of the bilinearity of $\Pi$, of the norm inequality
in $C$, and of \eqref{b:1}, we have:
\beq
\begin{split}
& \|\Pi( a[\z +h],b[\z+h])-\Pi (a[\z],b[\z]) \|_{C} \\
& \quad\leqslant \|\Pi(a[\z +h],(b[\z+h]-b[\z]))\|_{C}
+\|\Pi( (a[\z +h]-a[\z]),b[\z]) \|_{C} \\
& \quad\leqslant \| a[\z +h]\|_{A}\times \|b[\z+h]-b[\z]\|_{B} +
\|a[\z +h]-a[\z]\|_{A}\times \|b[\z]\|_{B}.
\end{split}
\label{b:2}
\eeq
For every $\z\in D$, $h$ is submitted to vary in such a sufficiently
small neighborhood of zero that $\z+h$ remains in a
compact subset of $D$, so that
$\|a[\z+h]\|_{A}$ remains uniformly bounded. Then in view of the
continuity of $a[\z]$ and $b[\z]$, the last member of \eqref{b:2}
tends to zero with $h$, which implies the continuity of
$\Pi (a[\z],b[\z])$ in $D$.

\vskip 0.2cm

(ii) Let $\dot{a}_j[\z]\in A$
and $\dot{b}_j[\z]\in B$
denote respectively the partial derivatives of $a[\z]$
and $b[\z]$ with respect to $\z_j$ in the complex domain $D$.
We shall then show that the function
$\Pi (a[\z],b[\z])$
admits a partial derivative with respect to $\z_j$ in $D$, 
which is equal to
$ \Pi(a[\z],\dot{b}_j[\z]) +
\Pi(\dot{a}_j[\z],b[\z])$. In fact, 
in view of the bilinearity
of $\Pi$ and of the norm inequality in $C$, we can write:
\beq
\begin{split} 
& \left\|\frac{\Pi(a[\z +h_j],b[\z+h_j]) -\Pi(a[\z],b[\z])}{{\rm h}}
-(\Pi(a[\z],\dot{b}_j[\z]) + \Pi(\dot{a}_j[\z],b[\z]))\right\|_{C} \\
& \quad\leqslant \left\|\frac{\Pi(a[\z +h_j],b[\z]) - \Pi(a[\z],b[\z])}{{\rm h}}
-\Pi(\dot{a}_j[\z],b[\z])\right\|_{C} \\
& \qquad +\left\|\frac{\Pi(a[\z],b[\z+h_j]) - \Pi(a[\z],b[\z])}{{\rm h}}
-\Pi(a[\z],\dot{b}_j[\z])\right\|_{C} \\
& \qquad +\left\|\frac{\Pi((a[\z +h_j]-a[\z]),(b[\z+h_j]-b[\z]))}{{\rm h}}\right\|_{C}.
\end{split}
\label{b:3}
\eeq
By using again the bilinearity of $\Pi$ and applying
the bicontinuity inequality \eqref{b:1}
to each term of the r.h.s. of
\eqref{b:3}, we can majorize the latter by
\beq
\begin{split}
& \left\|\frac{a[\z +h_j] - a[\z]}{{\rm h}} - \dot{a}_j[\z]\right\|_{A}
\times \|b[\z]\|_{B} +\|a[\z]\|_{A} \times
\left\|\frac{b[\z +h_j] - b[\z]}{{\rm h}} - \dot{b}_j[\z]\right\|_{B} \\
& \quad +|{\rm h}| \times \left\|\frac{a[\z +h_j] - a[\z]}{{\rm h}}\right\|_{A} \times
\left\|\frac{b[\z +h_j] - b[\z]}{{\rm h}} \right\|_{B}.
\end{split}
\label{b:4}
\eeq
Then, for any given $\z$ in $D$, each of the three terms of
\eqref{b:4} tends to zero with ${\rm h}$, in view of the hypothesis that
$a[\z]$ and $b[\z]$ are holomorphic vector--valued functions in $D$
whose values are bounded in the norm in every compact subset of $D$;
this implies that for each $j$, ($1\leqslant j\leqslant m$), the l.h.s.
of \eqref{b:3} tends to zero with ${\rm h}$, and therefore that
the function $\Pi(a[\z],b[\z])$ is holomorphic in $D$.
\end{proof}

\noindent
An immediate corollary of the previous lemma is obtained
by taking $C=A$, and defining $B$ as the space $\cL(A)$
of bounded linear operators $\{L: a \mapsto L(a);\, a\in A\}$ on $A$
equipped with the usual norm $\|L\| = \sup_{a\in A} \frac{\|L(a)\|_A}{\|a\|_A}$.
The inequality
\beq
\|L(a)\|_{A} \leqslant \|L\|\times \|a\|
\label{b:4'}
\eeq
plays the role of \eqref{b:1} and there holds

\skd

\begin{lemma}
\label{lemma:B1'}
Let $a[\z]$ and $L[\z]$ denote functions in $D$ which are
respectively vector--valued in $A$ and
$\cL(A)$, and let $\z \mapsto L(a)[\z]=L[\z](a[\z])$
denote the image function which is vector--valued in $A$. Then:
\begin{itemize}
\item[(i)] if $a[\z]$ and $L[\z]$ are continuous in $D$, $L(a)[\z]$ is
continuous in $D$; \\[-5pt]
\item[(ii)] if $\z$ is complex and if $a[\z]$ and $L[\z]$ are
holomorphic in $D$, $L(a)[\z]$ is holomorphic in $D$.
\end{itemize}
\end{lemma}

We shall now give applications of Lemma \ref{lemma:B1} to
particular structures which are relevant at several places 
of this paper.

\vskip 0.2cm

\noindent
\textbf{(1)} Let $A$ be the Hilbert space
$X_{\mu}\doteq L_2([0,\infty), \mu(R)\,\rmd R)$,
where $\mu$ denotes a given continuous and strictly positive
function on the interval $[0,\infty)$.
For any element $x =x(R)$ in $X_{\mu}$, we put:
\beq
\|x\|_{\mu}\doteq
\left[\int_0^\infty |x(R)|^2 \,\mu(R) \,\rmd R\right]^\frac{1}{2}.
\label{b:5}
\eeq
Let $B$ be the dual space $X_\frac{1}{\mu}$ of $X_{\mu}$,
$C = \C$, and $\Pi$ the bilinear form which associates with each pair
$(x,y)\in X_{\mu}\times X_\frac{1}{\mu}$ the ``quasi--scalar product"
\beq
\langle y,x\rangle \doteq
\int_0^\infty y(R) \, x(R)\,\rmd R.
\label{b:6}
\eeq
Note that it differs from the usual scalar
product which is sesquilinear,
but since $X_{\mu}$
is stable under the operation  $x\mapsto \overline{x}$,
where $\overline{x}$ denotes the complex conjugate function 
$R\mapsto \overline{x(R)}$ of $R \mapsto x(R)$, there still 
holds the Schwarz inequality:
\beq
 |\langle y,x \rangle| \leqslant \|x\|_{\mu} \times \|y\|_\frac{1}{\mu},
\label{b:7}
\eeq
which appears as a
bicontinuity inequality of the type \eqref{b:1}.
We can then state as a special case of Lemma \ref{lemma:B1}:

\skd

\begin{lemma}
\label{lemma:B2}
{\rm (i)} If $x[\z]$ and $y[\z]$ are continuous functions
in $D$, respectively vector--valued in
$X_{\mu}$ and $X_\frac{1}{\mu}$, the quasi--scalar--product--function
$\z \mapsto \langle y[\z],x[\z]\rangle$ is continuous in $D$.

{\rm (ii)} If $\z$ is complex and if $x[\z]$ and $y[\z]$ are
holomorphic functions in $D$, respectively vector--valued in
$X_{\mu}$ and $X_\frac{1}{\mu}$, the function
$\z \mapsto \langle y[\z],x[\z]\rangle$ is holomorphic in $D$.
\end{lemma}

\vskip 0.2cm

\noindent
\textbf{(2)} We take for $A$ the Hilbert space $\widehat{X}_{\mu}$
(called ``HS--kernel space") of kernels
$K(R,R')$ on $X_{\mu}$ equipped with the Hilbert--Schmidt--type norm
\beq
\|K\|_{(\mu)}  \doteq
\left[\int_0^{+\infty} \rmd R
\int_0^{+\infty} \rmd R' \,\frac{\mu(R)}{\mu(R')}\,
|K(R,R')|^2\right]^\frac{1}{2}.
\label{b:8}
\eeq
As it can be seen by applying Schwarz's inequality, this definition
of the HS--norm of $K$ ensures that
the linear operator defined by the formula
\beq
(K x)(R) = \int_0^{+\infty} K(R,R') \ x(R') \,\rmd R', \nonumber
\eeq
associates with every element
$x$ of $X_{\mu}$ an element $Kx$ of $X_{\mu}$.
As in \textbf{(1)}, we take for $B$ the dual space
$\widehat{X}_\frac{1}{\mu} = L_2([0,\infty)\times [0,\infty),\,\mu^{-1}(R)\,\mu(R')\,\rmd R\, \rmd R')$
of $\widehat{X}_{\mu}$, $C=\C$, and we choose for $\Pi$
the corresponding quasi--scalar--product of pairs
$(K,K')\in\widehat{X}_{\mu}\times\widehat{X}_\frac{1}{\mu}$:
\beq
\prec\! K',K \!\succ \ \doteq  
\int_0^\infty \!\rmd R\int_0^{+\infty}\!\rmd R'\ K'(R,R') \, K(R,R'),
\label{b:9}
\eeq
which also satisfies Schwarz's inequality:
\beq
|\prec\! K',K\!\succ| \leqslant
\|K\|_{(\mu)} \times
\|K'\|_{\left(\frac{1}{\mu}\right)}.
\label{b:11}
\eeq
Note that by introducing the transposed kernel $K'_\mathrm{t}$ of $K'$,
which is such that $K'_\mathrm{t} \in \widehat X_{\mu}$,
one can rewrite the previous formulae \eqref{b:9} and \eqref{b:11}
in terms of the \emph{trace} formalism, namely:
\beq
\prec\! K',K\!\succ \ = \Tr\left[K K'_\mathrm{t}\right], \qquad
|\Tr\left[KK'_\mathrm{t}\right]| \leqslant \|K\|_{(\mu)} \times
\|K'_\mathrm{t}\|_{(\mu)}.
\label{b:11'}
\eeq
Then, by specializing Lemma \ref{lemma:B1} to the present case, we obtain

\skd

\begin{lemma}
\label{lemma:B3}
{\rm (i)} If $K[\z]$ and $K'[\z]$ are continuous {\rm HS}--operator--valued functions
in $D$, taking their values respectively in
$\widehat{X}_{\mu}$ and $\widehat{X}_\frac{1}{\mu}$,
the quasi--scalar--product--function
$\z \mapsto \ \prec\! K'[\z],K[\z]\!\succ \ = \Tr\left[K[\z] K'_\mathrm{t}[\z]\right]$
is continuous in $D$.

{\rm (ii)} If $\z$ is complex and
if $K[\z]$ and $K'[\z]$ are holomorphic {\rm HS}--operator--valued functions
in $D$, taking their values respectively in
$\widehat{X}_{\mu}$ and $\widehat{X}_\frac{1}{\mu}$,
then the function $\z \mapsto \ \prec\! K'[\z],K[\z] \!\succ \ = \Tr\left[K[\z] K'_\mathrm{t}[\z]\right]$
is holomorphic in $D$.
\end{lemma}

\vskip 0.2cm

\noindent
\textbf{(3)} Let A=B=C denote the Hilbert space of HS--kernels
$\widehat{X}_{\mu}$, and $\Pi$ denote the composition of kernels:
for any pair $(K_1,K_2)$ in $\widehat{X}_{\mu}\times\widehat{X}_{\mu}$,
the kernel $K= K_1K_2\doteq K_1 \circ K_2$, defined by
\beq
K(R,R')\doteq \int_0^{+\infty} K_1(R,R'') \ K_2(R'',R')\,\rmd R'',
\label{b:13}
\eeq
belongs to $\widehat{X}_{\mu}$. In fact, the proof of the standard
HS--norm inequality, which corresponds to the choice $\mu=1$ (see, e.g.,
\cite{Smithies} and references therein) can be directly reproduced 
for the case of $\widehat{X}_{\mu}$, with an arbitrary function 
$\mu$ ($\mu>0$), namely;
\beq
\|K\|_{(\mu)} \leqslant
\|K_1\|_{(\mu)}\times \|K_2\|_{(\mu)}.
\label{b:130}
\eeq
(To check it, just introduce the ``renormalized" kernels
$(K_j)_{\rm ren}(R,R')= \sqrt{\frac{\mu(R)}{\mu(R')}} \, K_j(R,R')$, $j=1,2$,
which are such $\|(K_j)_{\rm ren}\|_{(1)} = \|K_j\|_{(\mu)}$).
Since \eqref{b:130} is a bicontinuity inequality of the type
\eqref{b:1}, Lemma \ref{lemma:B1} applies and yields

\begin{lemma}
\label{lemma:B6}
{\rm (i)} If $K_1[\z]$ and $K_2[\z]$
are continuous HS--operator--valued functions
in $D$, with values in $\widehat{X}_{\mu}$, then
the composition--product--function
$K[\z]=K_1[\z] \circ K_2[\z]$
is continuous in $D$ as an operator--valued function
with values in $\widehat{X}_{\mu}$.

{\rm (ii)} If $\z$ is complex, and if $K_1[\z]$ and $K_2[\z]$
are holomorphic HS--operator--valued functions
in $D$, with values in $\widehat{X}_{\mu}$, then
the composition--product--function
$K[\z]= K_1[\z] \circ K_2[\z]$
is holomorphic in $D$ as an operator--valued function
with values in $\widehat{X}_{\mu}$.
\end{lemma}

\vskip 0.2cm

\noindent
\textbf{(4)} Taking $A=\widehat{X}_{\mu}$, $B=\C$, $C=\widehat{X}_{\mu}$, and
the product $\Pi (K, \lambda) = \lambda  K$,
($K\in \widehat{X}_{\mu},\,\lambda\in\C$), which is such that
$\|\lambda K\| = |\lambda| \, \|K\|$, we immediately obtain from
Lemma \ref{lemma:B1} the continuity (resp., holomorphy) property of
any product $\lambda[\z] K[\z]$ of continuous (resp., holomorphic) 
functions $\z\mapsto K[\z]\in \widehat{X}_{\mu}$, $\z \mapsto \lambda[\z]\in \C$.
By combining this property with the result of Lemma \ref{lemma:B6}
applied iteratively to any power
$K^n[\z]\doteq K[\z]\circ\cdots\circ  K[\z]$ ($n$ factors), we obtain

\skd

\begin{lemma}
\label{lemma:B7}
{\rm (i)} If $\z\to K[\z]$ is a continuous {\rm HS}--operator--valued function
in $D$, taking its values in $\widehat{X}_{\mu}$,
then any polynomial function of the form
$\z\mapsto P_n(K)[\z] = \sum_{j=1}^{n} a_j[\z]
\, K^j[\z]$, where the $a_j\!$'s are complex--valued continuous functions in $D$,
is a continuous {\rm HS}--operator--valued function in $D$, with values in
$\widehat{X}_{\mu}$.

{\rm (ii)} If $\z$ is complex, and if $\z \mapsto K[\z]$ is a holomorphic {\rm HS}--operator--valued function
in $D$, taking its values in $\widehat{X}_{\mu}$,
then any polynomial function of the form
$\z\mapsto P_n(K)[\z] = \sum_{j=1}^{n} a_j[\z]\, K^j[\z]$,
where the $a_j\!$'s are holomorphic functions in $D$,
is a holomorphic {\rm HS}--operator--valued function
in $D$, with values in $\widehat{X}_{\mu}$.
\end{lemma}

\subsection{Passage from functions depending
continuously or holomorphically of parameters $\boldsymbol{\z}$
to continuous or holomorphic vector--valued functions of $\boldsymbol{\z}$} 
\label{subappendix:b.passage}

We now introduce for each strictly positive function $\mu$
and each positive number $p$, the functional space
$X_{\mu}^{(p)} \doteq L^p([0,\infty), \mu(R)\,\rmd R)$, 
of all functions $f(R)$ (defined almost everywhere on 
$[0,+\infty)$) with norm
\beq
\|f\|_{\mu,p} \doteq
\left[\int_0^{+\infty} |f(R)|^p \mu (R)\,\rmd R\right]^\frac{1}{p},
\label{b:14}
\eeq
{\bf Spaces $\boldsymbol{\cC(D,\mu,p)}$:}
Keeping the same notations as in \ref{subappendix:b.products},
we introduce ${\cC}(D, {\mu},p)$ as the space of all functions
$(\z,R) \mapsto f(\z; R)$ which are defined on 
$D\times [0,+\infty)$ \emph{for almost every (a.e.) $R$, namely
up to a subset of measure zero in $\{R\in [0,+\infty)\}$}, and 
which enjoy the following property.
For each function $f$, there exists a positive function $M(R)$ in $X_{\mu}^{(p)}$
such that the following uniform majorization holds, for all 
$\z\in D$ and a.e. $R\in [0,+\infty)$:
\beq
|f(\z;R)|\leqslant  M(R).
\label{b:15}
\eeq
It follows from this definition that every function $f$ in
${\cC}(D, {\mu},p)$ defines a vector--valued function
$\z \mapsto f[\z](\cdot)= f(\z;\cdot)$ in $D$, which takes its values in
$X_{\mu}^{(p)}$, since (in view of \eqref{b:15}), one has
for all $\z\in D$: $\|f[\z]\|_{\mu,p} \leqslant \|M\|_{\mu,p}$.
We shall now prove:

\skd

\begin{lemma}
\label{lemma:B8}
{\rm (i)} Let $(\z,R) \mapsto f(\z; R)$ be a function in
${\cC}(D, {\mu},p)$ such that for a.e. $R$, $f(\cdot;R)$ is a continuous
function of $\z$ in $D$. Then there exists a continuous vector--valued
function $\z \mapsto f[\z]$ in $D$ which takes its values in
$X_{\mu}^{(p)}$ and such that $f[\z](R) = f(\z;R)$.

{\rm (ii)} Let $D$ be a domain of the space $\C^m$ 
of the complex variables $\z=(\z_1,\ldots,\z_m)$, and let
$(\z,R) \mapsto f(\z; R)$ be a function in
${\cC}(D, {\mu},p)$ such that for a.e. $R$, $f(\cdot;R)$ is a holomorphic
function of $\z$ in $D$. Then there exists a holomorphic vector--valued
function $\z \mapsto f[\z]$ in $D$ which takes its values in
$X_{\mu}^{(p)}$ and such that $f[\z](R) = f(\z;R)$.
\end{lemma}

\begin{proof}
(i) $\z$ being fixed in $D$, one considers the family of functions
$R\mapsto f_{\z,h}(R) \doteq [f(\z+h;R)-f(\z;R)]$ which, in view of the
uniform majorization \eqref{b:15} (true for a.e. $R$), satisfy
for all $h$ such that $\z+h \in D$  the following uniform bound:
\beq 
\|f_{\z,h}\|_{\mu,p}^p = \int_0^{+\infty} |f_{\z,h}(R)|^p \, \mu(R)\,\rmd R
\leqslant 2^p \int_0^{+\infty} M(R)^p \, \mu(R) \,\rmd R = 2^p \|M\|_{\mu,p}^p.
\label{b:16}
\eeq
Since, by the continuity assumption, one has $\lim_{h\to 0} f_{\z,h}(R) =0$
for a.e. $R$, it then follows from Lebesgue--Fatou's theorem
that the integral on the l.h.s. of \eqref{b:16}, and therefore
$\|f_{\z,h}\|_{\mu,p}= \|f[\z+h]-f[\z]\|_{\mu,p}$, tends to zero
with $h$. Since this holds for all $\z \in D$,
this proves the continuity in $D$ of the function
$ f[\z](R) = f(\z;R)$ as a vector--valued function 
with values in $X_{\mu}^{(p)}$.

(ii) Let $\z$ be fixed in the complex domain $D$ 
at a ``$j$\emph{--distance}" $r_j(\z)$
from the boundary of $D$ (by $j$\emph{--distance}, we mean the distance 
of $\z$ from the boundary of the section of $D$ by
the complex one--dimensional submanifold passing at $\z$ and parallel
to the $\z_j$--plane). We then introduce
the following family of functions
(labeled by $h_j=(0,\ldots,0,{\rm h},0,\dots,0)$; $1\leqslant j\leqslant m$,
with the condition $\z+ h_j \in D$):
\beq
R\longmapsto g_{h_j}[\z](R) \doteq \left(\frac{f(\z+h_j;R)-f(\z;R)}{{\rm h}}-
\frac{\partial f}{\partial \z_j}(\z;R)\right). \nonumber
\label{derivative}
\eeq
In the latter, the derivative
$\frac{\partial f}{\partial \z_j}$ of the holomorphic function $f$ satisfies
for a.e. $R$ a Cauchy integral representation of the form:
\beq
\frac{\partial f}{\partial \z_j}(\z;R) =
\frac{1}{2\pi\rmi}
\int_{\gamma_r} \frac{f(\z';R)}{(\z'_j-\z_j)^2} \,\rmd\z'_j,
\label{b:17}
\eeq
where
$\gamma_{r}$ denotes the circle centered at $\z_j$ with radius $r$,
and where $\z'$ has all its components $\z'_k$, $k\neq j$, 
respectively equal to $\z_k$.
In view of the uniform upper bound \eqref{b:15}, and since $r$ can be chosen arbitrarily
close to $r_j(\z)$, there holds the following bound (for a.e. $R$):
\beq
\left|\frac{\partial f}{\partial \z_j}(\z;R)\right| \leqslant
\frac{M(R)}{r_j(\z)}.
\label{b:18}
\eeq
Now, by taking $h_j$ such that $|{\rm h}|< r$, we can write a
Cauchy integral representation on $\gamma_r$ for the holomorphic
function $\z \mapsto g_{h_j}[\z](R)$ (for a.e. value of $R$); by combining
Eq. \eqref{b:17} with the usual Cauchy representation for
$ f(\z)$ and $f(\z+h_j)$, one obtains:
\beq
g_{h_j}[\z](R)=
\frac{{\rm h}}{2\pi\rmi}
\int_{\gamma_{r}} \frac{f(\z';R)}{(\z'_j-\z_j)^2(\z'_j-\z_j-{\rm h})} \,\rmd\z'_j.
\label{b:19}
\eeq
By restricting ${\rm h}$ to vary in a neighborhood of zero such as, e.g.,
$\{{\rm h}\in \C\,:\, |{\rm h}| \leqslant \frac{r}{2}\}$, one obtains a uniform
majorization for the r.h.s. of Eq. \eqref{b:19} which yields for a.e. $R$
(since $r$ can be chosen arbitrarily close to $r_j(\z)$): 
\beq
|g_{h_j}[\z](R)|\leqslant \frac{2\,|{\rm h}|\, M(R)}{[r_j(\z)]^2}.
\label{b:20}
\eeq
From \eqref{b:18} and \eqref{b:20} one deduces that:

\vskip 0.1cm

(a) there exists (for each $j$) a vector--valued function
$\z \mapsto \dot{f}_j[\z](R) \doteq \frac{\partial f}{\partial \z_j}(\z;R)$
taking its values in $X_{\mu}^{(p)}$ and such that (for all $\z\in D$):
$\left\|\dot{f}_j[\z]\right\|_{\mu,p} \leqslant \frac{\|M\|_{\mu,p}}{r_j(\z)}$.

\vskip 0.1cm

(b) There holds:
$\left\|g_{h_j}[\z]\right\|_{\mu,p}\leqslant \frac{2|{\rm h}|}{[r_j(\z)]^2}
\|M\|_{\mu,p}$, which proves that the vector--valued function
$\z \mapsto f[\z](R) \doteq f(\z;R) $ is such that
$\left\|\frac{f[\z+h_j]-f[\z]}{{\rm h}}-\dot{f}_j[\z]\right\|_{\mu,p}$ 
tends to zero with ${\rm h}$ for all $\z$ in $D$.
We have thus proved that the function
$\z \mapsto f[\z](R) $ is holomorphic in $D$
as a vector--valued function taking its values in $X_{\mu}^{(p)}$.
\end{proof}

In the text, we shall have to apply
directly Lemma \ref{lemma:B8} for the case $p=2$,
and with a weight--function of the form $\mu(R) =w(R)\,e^{2\alpha R}$,
$w$ being specified in Subsection \ref{subse:properties-L}.

We also need to apply the previous result
to a more involved situation, which is described below in
Lemma \ref{lemma:B10}. For this purpose, we shall first state 
the following property, which appears as a variant of Lemma \ref{lemma:B8}
for the case $p=1$ (we also need this result only for $\mu(R) =1$).

\skd

\begin{lemma}
\label{lemma:B9}
{\rm (i)} Let $(\z,R) \mapsto f(\z;R)$ be a function in
${\cC}(D,1 ,1)$ such that for a.e. $R$, $f(\cdot;R)$ 
is a continuous function of $\z$ in $D$. Then
the integral $I(\z)\doteq \int_0^{+\infty} f(\z;R)\,\rmd R$
is continuous in $D$.

{\rm (ii)} Let $\z$ be complex, $D$ a domain of $\C^m$ and let
$(\z,R) \mapsto f(\z; R)$ be a function in
${\cC}(D,1,1 )$ such that for a.e. $R$, 
$f(\cdot;R)$ is a holomorphic function of $\z$ in $D$. Then
the integral $I(\z)\doteq \int_0^{+\infty} f(\z;R)\,\rmd R$
is holomorphic in $D$.
\end{lemma}

\begin{proof}
(i) One just has to check that
$|I(\z)|\leqslant \|f(\z;\cdot)\|_{1,1}\leqslant \|M\|_{1,1}$ and that
$|I(\z+h)-I(\z)|\leqslant \|f_{\z,h}\|_{1,1}$, 
which tends to zero with $h$ as in Lemma \ref{lemma:B8} (i).

(ii) As in the proof of Lemma \ref{lemma:B8} (ii), one considers the function
$g_{h_j}[\z](R)$ and its majorization \eqref{b:20}, which allows one to check that:
\beq
\left|\frac{I(\z+h_j)-I(\z)}{{\rm h}} - 
\frac{\partial I}{\partial \z_j} (\z)\right| =
\left|\int_0^{+\infty} g_{h_j}[\z](R)\,\rmd R\right| \leqslant
\left\|g_{h_j}[\z]\right\|_{1,1}\leqslant \frac{2\,|{\rm h}|\, \|M\|_{1,1}}{[r_j(\z)]^2}.
\label{b:21}
\eeq
It follows that the l.h.s. of \eqref{b:21} tends to zero with ${\rm h}$,
which proves the holomorphy property of $I(z)$ in $D$.
\end{proof}

\begin{lemma}
\label{lemma:B10}
{\rm (i)} Let $(\z,R,R') \mapsto F(\z; R,R')$ be a function defined
for a.e. $(R,R')$ on $D\times [0,+\infty)\times [0,+\infty)$
by a convergent integral of the following form:
\beq
F(\z;R,R')= \int_0^{+\infty} F_1(\z;R,R'') \ F_2(\z;R'',R') \,\rmd R'',
\label{b:23}
\eeq
under the following assumptions :

\vskip 0.1cm

{\rm (a)} the functions $F_j(\z;R,R')$ $(j=1,2)$, are defined
for a.e. $(R,R')$ on $D\times [0,+\infty)\times [0,+\infty)$
and satisfy uniform bounds $|F_j(\z;R,R')|\leqslant G_j(R,R')$
on this set, such that the integral
$G(R,R')= \int_0^{+\infty} G_1(R,R'') G_2(R'',R')\,\rmd R''$ is convergent
for almost every value of $(R,R')$ and the function $(R,R')\mapsto G(R,R')$
belongs to $\widehat{X}_{\mu}$, (see paragraph \ref{subappendix:b.products}-(2) of this Appendix).

\vskip 0.1cm

{\rm (b)} For a.e. $(R,R')$, the functions $F_1(\cdot;R,R')$ and $F_2(\cdot,R,R')$ are continuous
functions of $\z$ in $D$.

\vskip 0.1cm

\noindent
Then there exists a continuous {\rm HS}--operator--valued
function $\z \mapsto K[\z]$ in $D$ which takes its values in
$\widehat{X}_{\mu}$ and such that $K[\z](R,R') = F(\z;R,R')$.

\vskip 0.1cm

\noindent
{\rm (ii)} Let $\z$ be complex, $D$ a domain of $\C^m$, and
$(\z,R,R') \mapsto F(\z; R,R')$ a function of the form \eqref{b:23}
satisfying the previous conditions (a) together with the following additional 
condition:

\vskip 0.1cm

{\rm (b')} for a.e. $(R,R')$, the functions $F_1(\cdot;R,R')$ and 
$F_2(\cdot;R,R')$ are holomorphic functions of $\z$ in $D$.

\vskip 0.1cm

Then there exists a holomorphic {\rm HS}--operator--valued
function $\z \mapsto K[\z]$ in $D$ which takes its values in
$\widehat{X}_{\mu}$ and such that $K[\z](R,R') = F(\z;R,R')$.
\end{lemma}

\begin{proof}
The function $R''\mapsto F_{R,R'}(\z;R'')\doteq F_1(\z;R,R'') \, F_2(\z;R'',R')$, 
which is defined for a.e. $(R,R')$, continuous in $\z$ in case (i), 
holomorphic in $\z$ in case (ii), is uniformly bounded by
the function in $L^1$: $R''\mapsto G_{R,R'}(R'')\doteq G_1(R,R'') G_2(R'',R')$.
Then Lemma \ref{lemma:B9} entails that, for a.e. $(R,R')$, 
$F(\z;R,R')\doteq \int_0^{+\infty} F_{R,R'}(\z;R'')\,\rmd R''$
is continuous (resp., holomorphic) with respect to $\z$
and uniformly bounded by $G(R,R')= \int_0^{+\infty} G_{R,R'}(R'')\,\rmd R''$.
Since the function $(R,R')\mapsto G(R,R')$ belongs to $\widehat{X}_{\mu}$,
it follows that $F(\z;R,R')$ satisfies conditions which are similar to those of
Lemma \ref{lemma:B8} for $p=2$ (with either continuity or
holomorphy properties in $\z$ according to the respective cases (i) or (ii)),
\emph{up to the replacement} of the integration space
$\{R\in \R^+\}$ by $\{(R,R')\in \R^+\times \R^+\}$
and of the Hilbert space $X_{\mu}$ by $\widehat{X}_{\mu}$.
The results (i) and (ii) then correspond directly to
those of Lemma \ref{lemma:B8}.
\end{proof}

\subsection{Complement on the various criteria of 
vector--valued analyticity in several complex variables} 
\label{subappendix:b.complement}

We shall first recall the equivalence between
several criteria of analyticity for the case of 
numerical functions of several complex variables
defined in a domain of $\C^m$.
These criteria are:
\begin{itemize}
\item[(1)] \, ``differentiability criterium":
existence of partial derivatives with respect to the
complex variables $\z_j$ ($1\leqslant j \leqslant m$) 
at all points of $D$;
\item[(2)] \, solution of the system of Cauchy--Riemann equations in $D$; 
\item[(3)] \, ``Cauchy integral criterium" in $D$ and the associated
Cauchy integral representations for the function and
all its successive (partial) derivatives, implying
corresponding Cauchy inequalities;
\item[(4)] \, convergence of the Taylor series in an
appropriate complex neighborhood of each point of $D$.
\end{itemize}

\noindent
For each of these characteristic properties, there is a corresponding 
``weak criterium of analyticity" for the vector--valued functions
$\z\mapsto a[\z]$ of several complex variables taking their values in the 
complete normed space $A$; it consists in stating that
for every element $\varphi$ of the dual space $A'$ of $A$, 
the numerical ``scalar--product" function
$\z\mapsto \langle\varphi,a[\z]\rangle$ satisfies the corresponding 
analyticity criterium.

\vskip 0.2cm

\noindent
Then, it turns out that each weak criterium is equivalent to
a ``strong criterium of analyticity", which involves either the notion
of limit or that of integral in the sense of the norm in $A$.
 
In particular, all the results that have been derived in this
Appendix have made use of the differentiability criterium, 
which postulates the existence of partial derivatives 
$\z\mapsto \dot a_j[\z]$ of $\z \mapsto a[\z]$ as vector--valued 
functions obtained in the sense of strong limits in $A$.
The fact that it is implied by the corresponding
weak criterium (1) for the numerical functions
$\z\mapsto \langle\varphi,a[\z]\rangle$ (for all $\varphi \in A'$)
is obtained by a direct adaptation of the Dunford theorem 
(see, e.g., \cite[p. 128]{Yosida}) to the several variable
case. The main ingredient of this ``weak to strong passage" consists 
in the use of the ``maximum boundedness theorem" (through its corollary
called ``Resonance Theorem" in \cite{Yosida}) and of the completeness 
property of $A$.

\noindent
The importance of another ``weak to strong passage"
concerns the Cauchy integral criterium (3), since in particular 
the latter allows one to give a direct proof of Cauchy--type inequalities, 
which majorize the norms $\|\dot a_j[\z]\|$ of the partial derivatives at $\z$
in terms of the maximum of $\|a[\z']\|$ in a neighborhood of $\z$,
and which also imply the strong convergence of the Taylor series
in a complex neighborhood of $\z$.

\vskip 0.2cm

\noindent
Before giving a further description of the ``strong Cauchy
integral criterium", and in order to make clear how 
the latter results from the ``strong differentiability 
criterion" through the corresponding implications
for the weak criteria, let us recall how the 
implication $(1) \Longrightarrow (3)$
is obtained for numerical functions of
several complex variables.

\vskip 0.2cm

\noindent
$(1) \Longrightarrow (2)$: choosing the increment
${\rm h}$ either real or purely imaginary in the following
definition, one checks that the existence of
continuous  partial derivatives 
$\frac{\partial f}{\partial \z_j}(\z)$ for $f(\z)$,  
namely
$\lim_{\{{\rm h} \to 0\,{\rm in}\, \C\}} 
\left\|\frac{f(\z+h_j)-f(\z)]}{{\rm h}}- \dot{f}_j(\z)\right\| =0$
for every $\z \in D$ and $h_j=(0,\ldots,0,{\rm h},0,\ldots,0)$, 
implies that $\widehat{f}(x,y) \doteq f(\z)$ (with $\z_j= x_j+\rmi y_j$;
$1\leqslant j\leqslant m$) satisfies the system of Cauchy--Riemann equations in $D$: 
$\frac{\partial\widehat{f}}{\partial x_j}(x,y)=-\rmi \frac{\partial\widehat{f}}{\partial y_j}(x,y)$
or, by passing formally to the variables $(\z_j,\,\overline{\z}_j=x_j-\rmi y_j)$: 
$\partial\widehat{f}/\partial\overline{\z}_j(x,y)=0$.

\vskip 0.2cm

\noindent
$(2) \Longrightarrow (3)$:
considering $\widehat{f}(x,y)$ as a $0$--form in a domain of
$\R^{2m} \equiv \C^m$, and introducing the differential $1$--form
$\rmd\widehat{f}(x,y) = \sum_{1\leqslant j\leqslant m} 
\frac{\partial \widehat f}{\partial x_j}(x,y)\,\rmd x_j
+ \frac{\partial \widehat f}{\partial y_j}(x,y)\,\rmd y_j$
or, equivalently,
$\rmd\widehat f = \sum_{1\leqslant j\leqslant m} 
\frac{\partial \widehat f}{\partial \z_j}\,\rmd\z_j
+ \frac{\partial \widehat f}{\partial \overline \z_j}(x,y)\,\rmd\overline \z_j$, 
the Cauchy--Riemann system can be equivalently written as:
$\rmd\left(\widehat{f}(x,y)\,\rmd\z_1\wedge\cdots\wedge\rmd\z_m\right)=0$. 
Then, in view of Stokes' theorem, the latter is equivalent to
the fact that the following Cauchy--type integral formula:
\beq
\int_\Gamma \widehat{f}(x,y)\, \rmd\z_1\wedge\cdots\wedge\rmd\z_m=0,
\label{b:24}
\eeq
holds for every $m$--real--dimensional integration cycle
$\Gamma$ of the form $\Gamma = \partial \Delta$, where
$\Delta$ can be any $(m+1)$--cycle whose support is contained in
$D$. An usual and convenient choice for $\Gamma$ is the
``distinguished boundary" of a polydisk--type domain,
namely $\Gamma= \gamma_1\times \cdots \times \gamma_m$,
where each $\gamma_j$ ($1\leqslant j\leqslant m$) is the boundary of
a domain $\delta_j$ homeomorphic to a disk and such that
the polydisk--type domain $\widehat\Delta\doteq 
\delta_1\times\cdots\times\delta_m$ be contained in $D$.
As in the case of one complex variable, 
Eq. \eqref{b:24} implies the corresponding integral
representation  
\beq
f(\z) = \left(\frac{1}{2\pi\rmi}\right)^m
\int_{\gamma_1\times \cdots\times \gamma_m}
\frac{f(\z')\ \rmd\z'_1\wedge \cdots \wedge \rmd\z'_m}{(\z'_1-\z_1)\cdots (\z'_m-\z_m)},
\label{b:25}
\eeq
which is valid for every point $\z=(\z_1,\ldots,\z_m)$ in $\widehat\Delta$.
Moreover, there also holds integral representations ``of partial type", namely
with respect to any subset $J$ ($J\subset \{1,2,\ldots,m\}$) 
of variables $\z'_j$ integrated on the $1$--cycle $\gamma_j$
enclosing $\z_j$, as those used in Eqs. \eqref{b:17} and \eqref{b:19}.
(This corresponds to exploiting the 
Cauchy--Riemann system under the form: 
$\left.\rmd\left(\widehat{f}(x,y)\, \wedge_{\{j\in J\}} 
\rmd\z_j\right)\right|_{\{\z_k=0, \forall k\notin J\}}=0$, 
i.e., by Stokes' theorem: $\int_{\{\prod_{j\in J}\gamma_j\}}
\widehat{f}(x,y)\,\wedge_{\{j\in J\}}\rmd\z_j=0$.)

\vskip 0.2cm

\noindent
Coming back to the case of vector--valued functions $\z\mapsto a[\z]$,
the ``weak to strong passage" for integral relations such as \eqref{b:24} 
and \eqref{b:25} can be presented as follows.
Considering, e.g., the case of \eqref{b:24}, the fact that
for all $\varphi \in A'$, the numerical functions
$\langle\varphi,a[\z]\rangle$ satisfy the equation
$\int_\Gamma\langle\varphi,a[\z]\rangle\,\rmd\z=0$
(with $\rmd\z\doteq\rmd\z_1\wedge\cdots\wedge\rmd\z_m$) 
implies the corresponding vector--valued equation in $A$,
namely $\int_\Gamma a[\z]\,\rmd\z=0$. This implication is based 
on the following argument:

\vskip 0.1cm

(a) The strong differentiability of $a[\z]$ implies its strong
continuity in $\z$, which allows one to define (simple or
multiple) integrals of the form $\int_\Gamma a[\z] \,\rmd\z$ with values
in $A$ (namely, as strong limits of Riemann sums in $A$).

\vskip 0.1cm

(b) In view of the continuity and the linearity of
each $\varphi \in A'$, and by applying the weak criterium, one has:
$\langle\varphi, \int_\Gamma a[\z]\,\rmd\z\rangle
=\int_\Gamma \langle\varphi,a[\z]\rangle\,\rmd\z=0$.

\vskip 0.1cm

(c) If a vector $I\in A$ (such as $I=\int_\Gamma a[\z]\,\rmd z$) satisfies 
$\langle\varphi,I\rangle=0$ for all $\varphi \in A'$, 
$I$ is necessarily the zero--vector in $A$. This results from
the following

\skd

\begin{lemma}[see p. 108 of \textbf{\cite{Yosida}}] 
For any given $I$ in $A$ such as $I\neq 0$, there exists a continuous 
linear form $\varphi_0$ such that $\langle\varphi_0,I\rangle = \|I\|_A$.
\end{lemma}

\noindent
Now, being given a function $a[\z]$ holomorphic in $D$
with values in $A$, and satisfying a
uniform bound $\|a[\z]\| \leqslant M$, one can write 
Cauchy--type vector--equations similar to
Eqs. \eqref{b:17} and \eqref{b:19}, namely (by putting
$a_{h_j}[\z] = \frac{a[\z+h_j]-a[\z]}{{\rm h}}$):
\begin{align}
\dot{a}_j[\z] &= \frac{1}{2\pi\rmi}
\int_{\gamma_{r}} \frac{a[\z']}{(\z_j'-\z_j)^2} \,\rmd\z_j', \label{b:26} \\
g_{h_j}[\z] &\doteq a_{h_j}[\z]-\dot{a}_j[\z] =
\frac{{\rm h}}{2\pi\rmi}\int_{\gamma_{r}} 
\frac{a[\z']}{(\z_j'-\z_j)^2(\z_j'-\z_j-{\rm h})} \,\rmd\z_j'. \label{b:27}
\end{align}
By using norm inequalities under the integration signs in the
latter, one then obtains majorizations for these integrals 
which are similar to \eqref{b:18} and \eqref{b:20} and yield
the Cauchy--type inequalities:
\beq
\left\|\dot{a}_j[\z]\right\| \leqslant
\frac{M}{r_j(\z)},
\qquad
\left\|g_{h_j}[\z]\right\|\leqslant \frac{2\,|{\rm h}|\,M}{[r_j(\z)]^2}.
\label{b:rem}
\eeq
As an application, we notice that the end of the proof of Lemma \ref{lemma:B0}
can be obtained by applying \eqref{b:rem}
to the functions $\left(\dot{a}_n\right)_j[\z]$ and
$(g_n)_{h_j}[\z]\doteq (a_n)_{h_j}[\z]-\left(\dot{a}_n\right)_j[\z]$, 
the positive constant $M$ being then replaced by
$u_n$. In view of these inequalities, the uniform
majorization (by $\frac{\varepsilon}{3}$) 
of the remainders of the series with general terms
$\left(\dot{a}_n\right)_j[\z]$ and $(a_n)_{h_j}[\z]$
is then ensured (for all $h_j\in\cV_{\z,\varepsilon}$) 
by the convergence of the majorizing series with general term $u_n$.

\end{document}